\documentclass[11pt]{article}
\usepackage{times}

\usepackage{thm-restate}
\usepackage{subcaption}
\usepackage{caption}
\usepackage[paperwidth=199.8mm,paperheight=297mm,centering,hmargin=20mm,vmargin=2cm]{geometry}

\usepackage[dvipsnames]{xcolor}
\usepackage{framed}
\definecolor{shadecolor}{rgb}{0.9,0.9,0.9}
\usepackage{mathtools}
\usepackage{amsmath}
\usepackage[shortlabels]{enumitem}
\usepackage[titletoc,title]{appendix}
\usepackage{graphicx,epic,eepic,epsfig,amsmath,latexsym,amssymb,verbatim,color}
\usepackage{dsfont}

\usepackage{float}
\usepackage{tikz}
\usepackage[strict]{changepage}
\usepackage{hyperref}
\hypersetup{colorlinks=true,citecolor=blue,linkcolor=blue,filecolor=blue,urlcolor=blue,breaklinks=true}

\usepackage[marginal]{footmisc}
\usepackage{url}
\usepackage{theorem}
\newtheorem{definition}{Definition}
\newtheorem{proposition}[definition]{Proposition}
\newtheorem{lemma}[definition]{Lemma}

\newtheorem{theorem}[definition]{Theorem}
\newtheorem{corollary}[definition]{Corollary}

\usepackage{colortbl}
\usepackage{pifont}
\newcommand{\cmark}{\ding{51}}%
\definecolor{Gray}{gray}{0.92}
\definecolor{Gray2}{gray}{0.75}
\definecolor{maroon}{cmyk}{0,0.87,0.68,0.32}
\usepackage{booktabs} 

\def\squareforqed{\hbox{\rlap{$\sqcap$}$\sqcup$}}
\def\qed{\ifmmode\squareforqed\else{\unskip\nobreak\hfil
\penalty50\hskip1em\null\nobreak\hfil\squareforqed
\parfillskip=0pt\finalhyphendemerits=0\endgraf}\fi}
\def\endenv{\ifmmode\;\else{\unskip\nobreak\hfil
\penalty50\hskip1em\null\nobreak\hfil\;
\parfillskip=0pt\finalhyphendemerits=0\endgraf}\fi}
\newenvironment{proof}{\noindent \textbf{{Proof~} }}{\hfill $\blacksquare$}

\newcounter{remark}
\newenvironment{remark}[1][]{\refstepcounter{remark}\par\medskip\noindent%
\textbf{Remark~\theremark #1 } }{\medskip}

\newcounter{example}

\mathchardef\ordinarycolon\mathcode`\:
\mathcode`\:=\string"8000
\def\vcentcolon{\mathrel{\mathop\ordinarycolon}}
\begingroup \catcode`\:=\active
  \lowercase{\endgroup
  \let :\vcentcolon
  }

\usepackage{cleveref}
\usepackage{graphicx}
\usepackage{xcolor}

\RequirePackage[framemethod=default]{mdframed}
\newmdenv[skipabove=7pt,
skipbelow=7pt,
backgroundcolor=darkblue!15,
innerleftmargin=5pt,
innerrightmargin=5pt,
innertopmargin=5pt,
leftmargin=0cm,
rightmargin=0cm,
innerbottommargin=5pt,
linewidth=1pt]{tBox}

\newmdenv[skipabove=7pt,
skipbelow=7pt,
backgroundcolor=blue2!25,
innerleftmargin=5pt,
innerrightmargin=5pt,
innertopmargin=5pt,
leftmargin=0cm,
rightmargin=0cm,
innerbottommargin=5pt,
linewidth=1pt]{dBox}
\newmdenv[skipabove=7pt,
skipbelow=7pt,
backgroundcolor=darkkblue!15,
innerleftmargin=5pt,
innerrightmargin=5pt,
innertopmargin=5pt,
leftmargin=0cm,
rightmargin=0cm,
innerbottommargin=5pt,
linewidth=1pt]{sBox}
\definecolor{darkblue}{RGB}{0,76,156}
\definecolor{darkkblue}{RGB}{0,0,153}
\definecolor{blue2}{RGB}{102,178,255}
\definecolor{darkred}{RGB}{195,0,0}

\newcommand{\nc}{\newcommand}
\nc{\rnc}{\renewcommand}
\nc{\beg}{\begin{equation}}
\nc{\eeq}{{\end{equation}}}
\nc{\beqa}{\begin{eqnarray}}
\nc{\eeqa}{\end{eqnarray}}
\nc{\lbar}[1]{\overline{#1}}
\nc{\bra}[1]{\langle#1|}
\nc{\ket}[1]{|#1\rangle}
\nc{\ketbra}[2]{|#1\rangle\!\langle#2|}
\nc{\braket}[2]{\langle#1|#2\rangle}

\nc{\proj}[1]{| #1\rangle\!\langle #1 |}
\nc{\avg}[1]{\langle#1\rangle}
\nc{\Rank}{\operatorname{Rank}}
\nc{\smfrac}[2]{\mbox{$\frac{#1}{#2}$}}
\nc{\tr}{\operatorname{Tr}}
\nc{\ox}{\otimes}
\nc{\dg}{\dagger}
\nc{\dn}{\downarrow}
\nc{\cA}{{\cal A}}
\nc{\cB}{{\cal B}}
\nc{\cC}{{\cal C}}
\nc{\cD}{{\cal D}}
\nc{\cE}{{\cal E}}
\nc{\cF}{{\cal F}}
\nc{\cG}{{\cal G}}
\nc{\cH}{{\cal H}}
\nc{\cI}{{\cal I}}
\nc{\cJ}{{\cal J}}
\nc{\cK}{{\cal K}}
\nc{\cL}{{\cal L}}
\nc{\cM}{{\cal M}}
\nc{\cN}{{\cal N}}
\nc{\cO}{{\cal O}}
\nc{\cP}{{\cal P}}
\nc{\cQ}{{\cal Q}}
\nc{\cR}{{\cal R}}
\nc{\cS}{{\cal S}}
\nc{\cT}{{\cal T}}
\nc{\cV}{{\cal V}}
\nc{\cX}{{\cal X}}
\nc{\cY}{{\cal Y}}
\nc{\cZ}{{\cal Z}}
\nc{\cW}{{\cal W}}
\nc{\csupp}{{\operatorname{csupp}}}
\nc{\qsupp}{{\operatorname{qsupp}}}
\nc{\var}{{\operatorname{var}}}
\nc{\rar}{\rightarrow}
\nc{\lrar}{\longrightarrow}
\nc{\polylog}{{\operatorname{polylog}}}
\nc{\1}{{\mathds{1}}}
\nc{\wt}{{\operatorname{wt}}}
\nc{\av}[1]{{\left\langle {#1} \right\rangle}}
\nc{\supp}{{\operatorname{supp}}}

\def\a{\alpha}
\def\b{\beta}
\def\di{\diamondsuit}

\def\ve{\varepsilon}

\def\x{\xi}

\def\o{\omega}

\def\U{\Upsilon}

\nc{\RR}{{{\mathbb R}}}
\nc{\CC}{{{\mathbb C}}}
\nc{\FF}{{{\mathbb F}}}
\nc{\NN}{{{\mathbb N}}}
\nc{\ZZ}{{{\mathbb Z}}}
\nc{\PP}{{{\mathbb P}}}
\nc{\QQ}{{{\mathbb Q}}}
\nc{\UU}{{{\mathbb U}}}
\nc{\EE}{{{\mathbb E}}}
\nc{\DD}{{{\mathbb D}}}
\nc{\id}{{\operatorname{id}}}

\nc{\CHSH}{{\operatorname{CHSH}}}

\nc{\be}{\begin{equation}}
\nc{\ee}{{\end{equation}}}
\nc{\bea}{\begin{eqnarray}}
\nc{\eea}{\end{eqnarray}}
\nc{\<}{\langle}
\rnc{\>}{\rangle}
\nc{\rU}{\mbox{U}}

\nc{\ob}[1]{#1}

\nc{\SEP}{{\text{\rm SEP}}}
\nc{\NS}{{\text{\rm NS}}}
\nc{\LOCC}{{\text{\rm LOCC}}}
\nc{\PPT}{{\textbf{\rm PPT}}}
\nc{\EXT}{{\text{\rm EXT}}}
\nc{\Sym}{{\operatorname{Sym}}}

\nc{\ERLO}{{E_{\text{r,LO}}}}
\nc{\ERLOCC}{{E_{\text{r,LOCC}}}}
\nc{\ERPPT}{{E_{\text{r,PPT}}}}
\nc{\ERLOCCinfty}{{E^{\infty}_{\text{r,LOCC}}}}
\nc{\Aram}{{\operatorname{\sf A}}}

\usepackage{tikz}
\RequirePackage{doi}
\usepackage{hyperref}
\hypersetup{colorlinks=true,citecolor=blue,linkcolor=blue,filecolor=blue,urlcolor=blue,breaklinks=true}

\makeatletter
\def\grd@save@target#1{%
  \def\grd@target{#1}}
\def\grd@save@start#1{%
  \def\grd@start{#1}}
\tikzset{
  grid with coordinates/.style={
    to path={%
      \pgfextra{%
        \edef\grd@@target{(\tikztotarget)}%
        \tikz@scan@one@point\grd@save@target\grd@@target\relax
        \edef\grd@@start{(\tikztostart)}%
        \tikz@scan@one@point\grd@save@start\grd@@start\relax
        \draw[minor help lines,magenta] (\tikztostart) grid (\tikztotarget);
        \draw[major help lines] (\tikztostart) grid (\tikztotarget);
        \grd@start
        \pgfmathsetmacro{\grd@xa}{\the\pgf@x/1cm}
        \pgfmathsetmacro{\grd@ya}{\the\pgf@y/1cm}
        \grd@target
        \pgfmathsetmacro{\grd@xb}{\the\pgf@x/1cm}
        \pgfmathsetmacro{\grd@yb}{\the\pgf@y/1cm}
        \pgfmathsetmacro{\grd@xc}{\grd@xa + \pgfkeysvalueof{/tikz/grid with coordinates/major step}}
        \pgfmathsetmacro{\grd@yc}{\grd@ya + \pgfkeysvalueof{/tikz/grid with coordinates/major step}}
        \foreach \x in {\grd@xa,\grd@xc,...,\grd@xb}
        \node[anchor=north] at (\x,\grd@ya) {\pgfmathprintnumber{\x}};
        \foreach \y in {\grd@ya,\grd@yc,...,\grd@yb}
        \node[anchor=east] at (\grd@xa,\y) {\pgfmathprintnumber{\y}};
      }
    }
  },
  minor help lines/.style={
    help lines,
    step=\pgfkeysvalueof{/tikz/grid with coordinates/minor step}
  },
  major help lines/.style={
    help lines,
    line width=\pgfkeysvalueof{/tikz/grid with coordinates/major line width},
    step=\pgfkeysvalueof{/tikz/grid with coordinates/major step}
  },
  grid with coordinates/.cd,
  minor step/.initial=.2,
  major step/.initial=1,
  major line width/.initial=2pt,
}
\makeatother

\usepackage{thmtools}
\usepackage{thm-restate}
\usepackage{etoolbox}
\makeatletter
\def\problem@s{}
\newcounter{problems@cnt}

\newcommand{\allproblems}{\problem@s}
\makeatother

\usepackage{stmaryrd}
\usepackage{soul}
\usepackage{tcolorbox}
\usepackage{mathrsfs}
\usepackage{relsize}
\usetikzlibrary{shapes,arrows,spy,positioning,snakes}
\usepackage{fancybox, graphicx}
\usepackage{etoolbox}
\appto\appendix{\addtocontents{toc}{\protect\setcounter{tocdepth}{1}}}
\appto\listoffigures{\addtocontents{lof}{\protect\setcounter{tocdepth}{1}}}
\appto\listoftables{\addtocontents{lot}{\protect\setcounter{tocdepth}{1}}}

\usepackage{multirow}
\usepackage{makecell}
\usepackage{diagbox}

\definecolor{colorone}{rgb}{1,0.36,0.03}
\definecolor{colortwo}{rgb}{0.54,0.71,0.03}
\definecolor{colorthree}{rgb}{0.01,0.51,0.93}
\definecolor{colorfour}{rgb}{0.47,0.26,0.58}

\makeatletter
\newcommand*\rel@kern[1]{\kern#1\dimexpr\macc@kerna}
\newcommand*\widebar[1]{%
  \begingroup
  \def\mathaccent##1##2{%
    \rel@kern{0.8}%
    \overline{\rel@kern{-0.8}\macc@nucleus\rel@kern{0.2}}%
    \rel@kern{-0.2}%
  }%
  \macc@depth\@ne
  \let\math@bgroup\@empty \let\math@egroup\macc@set@skewchar
  \mathsurround\z@ \frozen@everymath{\mathgroup\macc@group\relax}%
  \macc@set@skewchar\relax
  \let\mathaccentV\macc@nested@a
  \macc@nested@a\relax111{#1}%
  \endgroup
}
\makeatother

\nc{\supre}{\text{supremum} \ }
\nc{\sdp}{\text{sdp}}

\nc{\Renyi}{\text{R\'{e}nyi} }
\nc{\bcV}{\boldsymbol \cV}
\nc{\bD}{\boldsymbol D}
\nc{\bR}{\boldsymbol R}
\nc{\sfT}{\mathsf T}
\nc{\EBset}{\boldsymbol{\cV}_{\Sigma}}

\nc{\CP}{\text{\rm CP}}
\nc{\herm}{\text{\rm Herm}}
\nc{\OPT}{\text{\rm OPT}}
\nc{\plsdagger}[1]{#1^{\mathsf H}}
\nc{\DPScone}{\mathcal{DPS}}
\nc{\SEPcone}{\mathcal{SEP}}
\nc{\BP}{\mathcal{BP}}
\nc{\cppp}{\text{cppp}}
\nc{\Pos}{\text{Pos}}
\nc{\bi}{\text{\rm bi}}

\nc{\mana}{\mathbb{M}}
\nc{\mA}{\mathbb{A}}

\usepackage{adjustbox}
\newcommand{\ssum}{\mathord{\adjustbox{valign=C,totalheight=.7\baselineskip}{$\sum$}}}

\nc{\dbp}[1]{\left\llbracket#1 \right\rrbracket_{\mathsf P}}
\nc{\dbe}[1]{\left\llbracket#1 \right\rrbracket_{\mathsf E}}
\nc{\dbh}[1]{\left\llbracket#1 \right\rrbracket_{\mathsf H}}
\nc{\dbl}[1]{\left\llbracket#1 \right\rrbracket_{\mathsf L}}

\nc{\dbpbig}[1]{\big\llbracket#1 \big\rrbracket_{\mathsf P}}
\nc{\dbebig}[1]{\big\llbracket#1 \big\rrbracket_{\mathsf E}}
\nc{\dbhbig}[1]{\big\llbracket#1 \big\rrbracket_{\mathsf H}}
\nc{\dblbig}[1]{\big\llbracket#1 \big\rrbracket_{\mathsf L}}

\nc{\dbpBig}[1]{\Big\llbracket#1 \Big\rrbracket_{\mathsf P}}
\nc{\dbeBig}[1]{\Big\llbracket#1 \Big\rrbracket_{\mathsf E}}
\nc{\dbhBig}[1]{\Big\llbracket#1 \Big\rrbracket_{\mathsf H}}
\nc{\dblBig}[1]{\Big\llbracket#1 \Big\rrbracket_{\mathsf L}}

\nc{\dbpbigg}[1]{\bigg\llbracket#1 \bigg\rrbracket_{\mathsf P}}
\nc{\dbebigg}[1]{\bigg\llbracket#1 \bigg\rrbracket_{\mathsf E}}
\nc{\dbhbigg}[1]{\bigg\llbracket#1 \bigg\rrbracket_{\mathsf H}}
\nc{\dblbigg}[1]{\bigg\llbracket#1 \bigg\rrbracket_{\mathsf L}}

\nc{\dbpBigg}[1]{\Bigg\llbracket#1 \Bigg\rrbracket_{\mathsf P}}
\nc{\dbeBigg}[1]{\Bigg\llbracket#1 \Bigg\rrbracket_{\mathsf E}}
\nc{\dbhBigg}[1]{\Bigg\llbracket#1 \Bigg\rrbracket_{\mathsf H}}
\nc{\dblBigg}[1]{\Bigg\llbracket#1 \Bigg\rrbracket_{\mathsf L}}

\nc{\bu}{\boldsymbol u}
\nc{\by}{\boldsymbol y}
\nc{\bv}{\boldsymbol v}
\nc{\btheta}{\boldsymbol \theta}
\nc{\SO}{\text{\rm SO}}
\nc{\STAB}{\text{\rm STAB}}
\nc{\CSPO}{\text{\rm CSPO}}
\nc{\CPWP}{\text{\rm CPWP}}

\usepackage{accents}

\begin{document}

\title{\textbf{Geometric \Renyi Divergence and its Applications\\ in Quantum Channel Capacities}}

\author{\normalsize Kun Fang~\thanks{Department of Applied Mathematics and Theoretical Physics, University of Cambridge, UK. \ \emph{kf383@cam.ac.uk}} \and \normalsize Hamza Fawzi~\thanks{Department of Applied Mathematics and Theoretical Physics, University of Cambridge, UK. \ \emph{h.fawzi@damtp.cam.ac.uk}}}
     
\date{\today}
\maketitle

\begin{abstract}

We present a systematic study of the \emph{geometric \Renyi divergence} (GRD), also known as the maximal \Renyi divergence, from the point of view of quantum information theory. We show that this divergence, together with its extension to channels, has many appealing structural properties. For example we prove a chain rule inequality that immediately implies the ``amortization collapse'' for the geometric \Renyi divergence, addressing an open question by Berta et al. [arXiv:1808.01498, Equation~(55)] in the area of quantum channel discrimination.
As applications, we explore various channel capacity problems and construct new channel information measures based on the geometric \Renyi divergence, sharpening the previously best-known bounds based on the max-relative entropy 
while still keeping the new bounds single-letter efficiently computable. A plethora of examples are investigated and the improvements are evident for almost all cases.

\end{abstract}

\newpage
{
  \hypersetup{linkcolor=black}
  \tableofcontents
}

\newpage

\section{Introduction}

In information theory, an imperfect communication link between a sender and a receiver is modeled as a noisy channel. The \emph{capacity} of such a channel is defined as the maximum rate at which information can be transmitted through the channel reliably. This quantity establishes the ultimate boundary between communication rates that are achievable in principle by a channel coding scheme and those that are not. A remarkable result by Shannon~\cite{Shannon1948} states that the capacity of a classical channel is equal to the mutual information of this channel, thus completely settling this capacity problem by a single-letter formula. 
Quantum information theory generalizes the classical theory, incorporating quantum phenomena like entanglement that have the potential to enhance communication capabilities. Notably, the theory of quantum channels is much richer but less well-understood than that of its classical counterpart. For example, quantum channels have several distinct capacities, depending on what one is trying to use them for, and what additional resources are brought into play. These mainly include the classical capacity, private capacity and quantum capacity, with or without the resource assistance such as classical communication and prior shared entanglement. The only solved case for general quantum channels is the entanglement-assisted classical capacity, which is given by the quantum mutual information of the channel~\cite{Bennett2002} and is believed as the most natural analog to Shannon's formula. The capacities in other communication scenarios are still under investigation. Some recent works (e.g~\cite{Wang2019channelmagic,seddon2019quantifying}) also extend the use of quantum channels to generate quantum resources such as magic state, a key ingredient for fault-tolerant quantum computation. The capability of a channel to generate such resource is thus characterized by its corresponding generation capacity.

In general, the difficulty in finding exact expressions for the channel capacities has led to a wide body of works to construct achievable (lower) and converse (upper) bounds. We will defer the detailed discussion of these bounds to the following individual sections. There are several important and highly desirable criteria that one would like from any bound on channel capacities. Specifically, one is generally interested in bounds that are:
\begin{itemize}
    \item \textbf{single-letter;} i.e., the bound depends only on a single use of the channel. Several well-established channel coding theorems state that the quantum channel capacity is equal to its corresponding regularized information measure (e.g. the quantum capacity of a channel is equal to its regularized coherent information~\cite{Lloyd1997,Shor2002a,Devetak2005a}). However, these regularized formulas are simply impossible to evaluate in general using finite computational resources, thus not informative enough in spite of being able to write down as formal mathematical expressions. A single-letter formula could be more mathematically tractable and provides a possibility of its evaluation in practice.
    \item \textbf{computable;} i.e., the formula can be explicitly computed for a given quantum channel. This is essentially required by the nature of capacity that quantifies the ``capability'' of a channel to transmit information or generate resource.  An (efficiently) computable converse bound can help to assess the performance of a channel coding scheme in practice and can also be used as a benchmark for the succeeding research. Note that a single-letter formula is not sufficient to guarantee its computability.  An example can be given by the quantum squashed entanglement, which admits a single-letter formula but whose computational complexity is proved to be NP-hard~\cite{Huang2014}.  
    \item \textbf{general;} i.e., the bounds holds for arbitrary quantum channels without requiring any additional assumption on their structure, such as degradability or covariance. There are bounds working well for specific quantum channels with a certain structure or sufficient symmetry. However, the noise in practice can be much more versatile than expected and more importantly does not necessarily possess the symmetry we need. A general bound is definitely preferable for the sake of practical interest.
    \item \textbf{strong converse;} i.e., if the communication rate exceeds this bound, then the success probability or the fidelity of transmission of any channel coding scheme converges to zero as the number of channel uses increases. In contrast, the (weak) converse bound only requires the convergence to a scalar not equal to one. Thus a strong converse bound is conceptually more informative than a weak converse bound, leaving no room for the tradeoff between the communication rate and its success probability or fidelity. If a strong converse bound is tight for a channel, then we call this channel admits the strong converse property. This property is known to hold for all memoryless channels in the classical information theory~\cite{Wolfowitz1978} while it remains open in the quantum regime in general (except for the entanglement-assisted classical capacity~\cite{Bennett2014}). A strong converse bound may witness the strong converse property of certain quantum channels (e.g.~\cite{tomamichel2017strong,Wilde2016c}), further sharpening our understanding of the quantum theory.
\end{itemize}


\subsection{Main contributions}

In this paper we propose new bounds on quantum channel capacities that satisfy all criteria mentioned above and that improve on previously known bounds. The main novelty of this work is that our bounds all rely on the so-called \emph{geometric \Renyi divergence}.
We establish several remarkable properties for this \Renyi divergence that are particularly useful in quantum information theory and show how they can be used to provide bounds on quantum channel capacities.

\paragraph{Geometric \Renyi divergence}
The \emph{geometric \Renyi divergence} (GRD), is defined as~\cite{matsumoto2015new}
\begin{align*}
  \widehat D_\a(\rho\|\sigma)\equiv\frac{1}{\a-1} \log \tr \left[\sigma^{\frac12} \left(\sigma^{-\frac12}\rho \sigma^{-\frac12}\right)^{\a}\sigma^{\frac12}\right], \quad \alpha \in (1,2].
\end{align*}
The quantity $\widehat D_\a$ is also known as the maximal-\Renyi divergence \cite{matsumoto2015new} as it can be shown to be the maximal divergence among all quantum \Renyi divergences satisfying the data-processing inequality.
Different from the widely studied Petz \Renyi divergence~\cite{petz1986quasi}  or sandwiched \Renyi divergence \cite{muller2013quantum,Wilde2014a}, the GRD converges to the Belavkin-Staszewski relative entropy~\cite{belavkin1982c} when $\a \rightarrow 1$. 
The geometric \Renyi divergence of two channels $\cN$ and $\cM$ is defined in the usual way as:
\begin{align*}
\widehat D_{\a}(\cN\|\cM) \equiv \max_{\rho_A\in \cS(A)} \widehat D_{\a}(\cN_{A'\to B} (\phi_{AA'})\|\cM_{A'\to B}(\phi_{AA'})),
\end{align*}
where $\cS(A)$ is the set of quantum states and $\phi_{AA'}$ is a purification of $\rho_A$.
We establish the following key properties of GRD which hold for any $\alpha \in (1,2]$:
\begin{enumerate}
  \item It lies between the Umegaki relative entropy and the max-relative entropy for all $\alpha \in (1,2]$,
  \begin{align*}
    D(\rho\|\sigma) \leq \widehat D_{\a}(\rho\|\sigma) \leq D_{\max}(\rho\|\sigma).
  \end{align*}
  \item Its channel divergence admits a closed-form expression,
  \begin{align*}
  \widehat D_\a(\cN_{A\to B}\|\cM_{A\to B}) = \frac{1}{\a-1} \log \left\|\tr_B \left[J_{\cM}^{\frac12} \left(J_{\cM}^{-\frac12}J_{\cN} J_{\cM}^{-\frac12}\right)^{\a}J_{\cM}^{\frac12}\right]\right\|_{\infty},
\end{align*}
where $J_{\cN}$ and $J_{\cM}$ are the corresponding Choi matrices of $\cN$ and $\cM$ respectively.
  \item Its channel divergence is additive under tensor product of channels,
  \begin{align*}
    \widehat D_{\a}(\cN_1\ox \cN_2\|\cM_1\ox \cM_2) = \widehat D_{\a}(\cN_1\|\cM_1) + \widehat D_{\a}(\cN_2\|\cM_2).
  \end{align*}
  \item Its channel divergence is sub-additive under channel composition,
  \begin{align*}
    \widehat D_{\a}(\cN_2\circ \cN_1\|\cM_2\circ \cM_1) \leq \widehat D_{\a}(\cN_1\|\cM_1) + \widehat D_{\a}(\cN_2\|\cM_2).
  \end{align*}  
  \item It satisfies the chain rule for any quantum states $\rho_{RA}$, $\sigma_{RA}$ and quantum channels $\cN$ and $\cM$,
  \begin{align*}
    \widehat{D}_{\alpha}(\cN_{A\to B}(\rho_{RA}) \| \cM_{A \to B}(\sigma_{RA})) &\leq \widehat{D}_{\alpha}( \rho_{RA} \| \sigma_{RA}) + \widehat{D}_{\alpha}(\cN \| \cM).
  \end{align*}
\end{enumerate}
These properties set a clear difference of GRD with other \Renyi divergences. Of particular importance is the chain rule property, which immediately implies that the ``amortization collapse'' for the geometric \Renyi divergence, addressing an open question from~\cite[Eq.~(55)]{Berta2018} in the area of quantum channel discrimination. Moreover, due to the closed-form expression of the channel divergence and the semidefinite representation of the matrix geometric means~\cite{fawzi2017lieb}, any optimization $\min_{\cM \in \bcV} \widehat D_{\a}(\cN\|\cM)$ can be computed as a semidefinite program if $\bcV$ is a set of channels characterized by semidefinite conditions.

\paragraph{Applications in quantum channel capacities}
We utilize the geometric \Renyi divergence to study several different channel capacity problems, including (1) unassisted quantum capacity, (2) two-way assisted quantum capacity, (3) two-way assisted quantum capacity of bidirectional quantum channels, (4) unassisted private capacity, (5) two-way assisted private capacity, (6) unassisted classical capacity, (7) magic state generation capacity, as listed in Table~\ref{tab: capacity tasks}. Most existing capacity bounds are based on the max-relative entropy due to its nice properties, such as triangle inequality or semidefinite representations. However, these bounds are expected to be loose as the max-relative entropy stands at the top among the family of quantum divergences. For the bounds based on the Umegaki's relative entropy, they are unavoidably difficult to compute in general due to their minimax optimization formula. In this work, we construct new channel information measures based on the geometric \Renyi divergence, sharpening the previous bounds based on the max-relative entropy in general while still keeping the new bounds single-letter efficiently computable.~\footnote{For unassisted and two-way assisted private capacities, the new bounds are efficiently computable for general qubit channels.}  
A plethora of examples are analyzed in each individual sections and the improvements are evident for almost all cases. 

\setlength\extrarowheight{2pt}
\begin{table}[H]
\centering
\begin{tabular}{c|cl|c|r}
\toprule[2pt]
Tasks & & Capacities & Previous bounds ($D$ or $D_{\max}$) & New bounds ($\widehat D_{\alpha}$)\\
\hline
\multirow{3}{*}{Quantum} & (1) & unassisted $Q $ & $R$~\cite{tomamichel2017strong}, $R_{\max}$~\cite{Wang2017d} & $\widehat R_{\a}$ [Thm.~\ref{thm: main result quantum unassisted}]\\
& (2) & two-way $Q^{\leftrightarrow}$ & $R_{\max}$~\cite{Berta2017a} & $\widehat R_{\a,\Theta}$ [Thm.~\ref{thm: main result quantum assisted}]\\
& (3) & two-way $Q^{\bi,\leftrightarrow}$ & $R^\bi_{\max}$~\cite{Bauml2018} & $\widehat R^\bi_{\a,\Theta}$ [Thm.~\ref{thm: main result quantum assisted bidirectional}]\\
\midrule[1pt]
\multirow{2}{*}{Private} & (4) & unassisted $P$ & $E_R$~\cite{Pirandola2015b,Wilde2016c}, $E_{\max}$~\cite{Christandl2016} & $\widehat E_{\a}$ [Thm.~\ref{thm: main result private unassisted}]\\
 & (5) & two-way $P^{\leftrightarrow}$ & $E_{\max}$~\cite{Christandl2016} & $\widehat E_{\a,\Sigma}$ [Thm.~\ref{thm: main result private assisted}]\\
\midrule[1pt]
\multirow{1}{*}{Classical} & (6) & unassisted $C$ & $C_{\beta}$, $C_{\zeta}$~\cite{Wang2016g} & $\widehat \Upsilon_{\a}$ [Thm.~\ref{thm: main result classical}]\\
\midrule[1pt]
\multirow{1}{*}{Magic} & (7) & adaptive $C_{\psi}$ & $\theta_{\max}$~\cite{Wang2019channelmagic} & $\widehat \theta_{\a}$ [Thm.~\ref{thm: main result magic}]\\
\bottomrule[2pt]  
\end{tabular}
\caption{\small Quantum information tasks studied in this paper, and new bounds on capacities obtained using the geometric \Renyi divergence $\widehat D_{\a}$.
}
\label{tab: capacity tasks}
\end{table}

The significance of this work is at least two-fold. First, from the technical side, we showcase that the geometric \Renyi divergence, which has not been exploited so far in the quantum information literature, is actually quite useful for channel capacity problems. We regard our work as an initial step towards other interesting applications and expect that the technical tools established in this work can also be used in, for example, quantum network theory, quantum cryptography, as the max-relative entropy also appears as the key entropy in these topics. We include another explict application in quantum channel discrimination task in Appendix~\ref{app: Quantum channel discrimination}. Second, our new capacity bounds meet all the aforementioned desirable criteria and improve the previously best-known results in general, making them suitable as new benchmarks for computing the capacities of quantum channels.

\section{Preliminaries}
\label{sec: preliminaries}

A quantum system, denoted by capital letters (e.g., $A$, $B$), is usually modeled by finite-dimensional Hilbert spaces (e.g., $\cH_A$, $\cH_B$). The set of linear operators and the set of positive semidefinite operators on system $A$ are denoted as $\cL(A)$ and $\cP(A)$, respectively. The identity operator on system $A$ is denoted by $\1_A$.
The set of quantum state on system $A$ is denoted as $\cS(A)\equiv \{\rho_A\,|\, \rho_A \geq 0,\,\tr\rho_A = 1\}$. A sub-normalized state is a positive semidefinite operator with trace no greater than one.
For any two Hermitian operators $X$, $Y$, we denote $X \ll Y$ if their supports has the inclusion $\supp(X) \subseteq \supp(Y)$. The trace norm of $X$ is given by $\|X\|_1 \equiv \tr \sqrt{X^\dagger X}$. The operator norm $\|X\|_\infty$ is defined
as the maximum eigenvalue of $\sqrt{X^\dagger X}$.
The set of completely positive (CP) maps from $A$ to $B$ is denoted as $\CP(A:B)$. 
A \emph{quantum channel or quantum operation} $\cN_{A\to B}$ is a completely positive and trace-preserving linear map from $\cL(A)$ to $\cL(B)$. A \emph{subchannel or suboperation} $\cM_{A\to B}$ is a completely positive and trace non-increasing linear map from $\cL(A)$ to $\cL(B)$. Let $\ket{\Phi}_{A'A} = \sum_{i} \ket{i}_{A'}\ket{i}_A$ be the unnormalized maximally entangled state. Then the \emph{Choi matrix} of a linear map $\cE_{A'\to B}$ is defined as $J_{AB}^{\cE} = (\cI_{A}\ox\cE_{A'\to B})(\ket{\Phi}\bra{\Phi}_{A'A})$. We will drop the identity map $\cI$ and identity operator $\1$ if they are clear from the context. The logarithms in this work are taken in the base two.

\subsubsection*{Notation for semidefinite representation}

For the simplicity of presenting a semidefinite program, we will introduce some new notations to denote semidefinite conditions. Denote the positive semidefinite condition $X \geq 0$ as $\dbp{X}$, the equality condition $X = 0$ as $\dbe{X}$, the Hermitian condition $X = X^\dagger$ as $\dbh{X}$ and the linear condition $\dbl{X}$ if $X$ is certain linear operator. We also denote the Hermitian part of $X$ as $\plsdagger{X} \equiv X + X^\dagger$.

\subsubsection*{Quantum divergences}

A functional $\bD: \cS \times \cP \to \mathbb R$ is a \emph{generalized divergence} if it satisfies the data-processing inequality 
\begin{align}
\bD(\cN(\rho)\|\cN(\sigma)) \leq \bD(\rho\|\sigma).
\end{align} 
The \emph{sandwiched \Renyi divergence} is defined as~\cite{muller2013quantum,Wilde2014a}
\begin{align}
     \widetilde D_\a(\rho\|\sigma) \equiv \frac{1}{\a-1} \log \tr \left[\sigma^{\frac{1-\a}{2\a}} \rho \sigma^{\frac{1-\a}{2\a}}\right]^\a,
\end{align} 
which is the smallest quantum \Renyi divergence that satisfies a data-processing inequality, and has been widely used to prove the strong converse property (e.g.~\cite{Wilde2014a,tomamichel2017strong}).
In particular, the sandwiched \Renyi divergence is non-decreasing in terms of $\a$, interpolating the \emph{Umegaki relative entropy} $D(\rho\|\sigma) \equiv \tr [\rho\, (\log \rho - \log \sigma)]$~\cite{Umegaki1962} and the \emph{max-relative entropy} $D_{\max}(\rho\|\sigma) \equiv \min\{\log t\,|\,\rho \leq t \sigma\}$~\cite{Renner2005,Datta2009} as its two extreme cases,
    \begin{align}\label{eq: sand renyi and relative entropy}
        D(\rho\|\sigma) = \lim_{\a \to 1} \widetilde D_\a(\rho\|\sigma) \leq \widetilde D_\a(\rho\|\sigma) \leq \lim_{\a \to \infty} \widetilde D_{\a}(\rho\|\sigma) = D_{\max}(\rho\|\sigma).
    \end{align}
Another commonly used quantum variant is the \emph{Petz \Renyi divergence}~\cite{petz1986quasi} defined as
\begin{align}
    \widebar D_\a(\rho\|\sigma) \equiv \frac{1}{\a - 1} \log \tr \left[\rho^\a \sigma^{1-\a}\right],
\end{align}
which attains operational significance in the quantum generalization of Hoeffding’s and Chernoff’s bound on the success probability in binary hypothesis testing~\cite{Nussbaum2009,Audenaert2007}. At the limit of $\a\to 0$, the Petz \Renyi divergence converges to the min-relative entropy~\cite{Datta2009},
\begin{align}
  \lim_{\a \to 0} \widebar D_{\a}(\rho\|\sigma) = -\log \tr \Pi_\rho \sigma \equiv D_{\min}(\rho\|\sigma), 
\end{align}
where $\Pi_\rho$ is the projector on the support of $\rho$.
Due to the the Lieb-Thirring trace inequality~\cite{Lieb1991}, it holds for all $\a \in (1,\infty)$ that
\begin{align}\label{eq: sandwiched and petz relation}
    \widetilde D_\a(\rho\|\sigma) \leq \widebar D_\a(\rho\|\sigma).
\end{align}
Both $\widetilde D_\a$ and $\widebar D_\a$ recover the Umegaki relative entropy $D$ at the limit of $\a \to 1$. But they are not easy to compute or to optimize over in general. 

For any generalized divergence $\bD$, the generalized channel divergence between quantum channel $\cN_{A'\to B}$ and subchannel $\cM_{A'\to B}$ is defined as~\cite{Leditzky2018}
\begin{align}
\bD(\cN\|\cM) \equiv \max_{\rho_A\in \cS(A)} \bD(\cN_{A'\to B} (\phi_{AA'})\|\cM_{A'\to B}(\phi_{AA'})),
\end{align}
where $\phi_{AA'}$ is a purification of $\rho_A$.
In particular, the max-relative channel divergence is independent of the input state~\cite[Lemma 12]{Berta2018}, 
\begin{align}\label{eq: Dmax channel divergence}
  D_{\max}(\cN\|\cM) = D_{\max}(J_{\cN}\|J_{\cM}),
\end{align}
where $J_{\cN}$ and $J_{\cM}$ are the corresponding Choi matrices of $\cN$ and $\cM$ respectively.

\section{Geometric \Renyi divergence}
\label{sec: Technical tools}

In this section, we investigate the geometric \Renyi divergence and its corresponding channel divergence. Our main contribution in this section is to prove several crucial properties of these divergences which are summarized in Theorem~\ref{thm: summary of properties}.  These properties will be extensively used in the following sections.

\subsection{Definitions and key properties}

\begin{definition}[\cite{matsumoto2015new}]
Let $\rho$ be a quantum state and $\sigma$ be a sub-normalized state with $\rho \ll \sigma$ and $\a \in (1,2]$, their geometric R\'{e}nyi divergence~\footnote{It is also called the maximal \Renyi divergence (see e.g.~\cite[Section 4.2.3]{Tomamichel2015b}) as it is the largest possible quantum \Renyi divergence satisfying the data-processing inequality. We here use the term ``geometric''  as its closed-form expression is given by the matrix geometric means, depicting the nature of this quantity.} is defined as
\begin{align}
    \widehat D_\a(\rho\|\sigma)\equiv\frac{1}{\a-1} \log \tr G_{1-\a}(\rho,\sigma),
\end{align}
where $G_\a(X,Y)$ is the weighted matrix geometric mean defined as
\begin{align}
  G_\a(X,Y)  \equiv X^{\frac12}\left(X^{-\frac12}Y X^{-\frac12}\right)^\a X^{\frac12}.
\end{align}
\end{definition}

\begin{remark}
  The geometric \Renyi divergence converges to the \emph{Belavkin-Staszewski relative entropy}~\cite{belavkin1982c},
\begin{align}\label{eq: geo and BS}
  \lim_{\a \to 1} \widehat D_\a(\rho\|\sigma) = \widehat D(\rho\|\sigma) \equiv \tr \rho \log \big[\rho^{1/2}\sigma^{-1}\rho^{1/2} \big].
\end{align}
Note that $D(\rho\|\sigma) \leq \widehat D(\rho\|\sigma)$ in general and they coincide for commuting $\rho$ and $\sigma$~\cite{Hiai1991}. Some basic properties such as joint-convexity, data-processing inequality and the continuity of the geometric \Renyi divergence (or more generally, maximal $f$-divergence) of states can be found in~\cite{matsumoto2015new}. Further studies on its reversibility under quantum operations are given in~\cite{hiai2017different,bluhm2019strengthened}.  Moreover, the weighted matrix geometric mean admits a semidefinite representation~\cite{fawzi2017lieb} (see also Lemma~\ref{geometric SDP general lemma} in Appendix~\ref{app: technical lemmas}).
\end{remark}

\begin{definition}
  For any quantum channel $\cN_{A'\to B}$, subchannel $\cM_{A'\to B}$, and $\a \in (1,2]$, their geometric \Renyi channel divergence is defined as
  \begin{align}
\widehat D_{\a}(\cN\|\cM) \equiv \max_{\rho_A\in \cS(A)} \widehat D_{\a}(\cN_{A'\to B} (\phi_{AA'})\|\cM_{A'\to B}(\phi_{AA'})),
\end{align}
where $\phi_{AA'}$ is a purification of $\rho_A$.
\end{definition}

The following Theorem summarizes several crucial properties of the geometric \Renyi divergence and its channel divergence. We present their detailed proofs in the next section.
\begin{theorem}[Main technical results]\label{thm: summary of properties}
The following properties of the geometric \Renyi divergence and its channel divergence hold.
\begin{enumerate}
  \item ({Comparison with $D$ and $D_{\max}$}): For any quantum state $\rho$, sub-normalized quantum state $\sigma$ with $\rho \ll \sigma$ and $\a \in (1,2]$, it holds 
  \begin{equation}
  \label{eq:comparisonDDaDmax}
    D(\rho\|\sigma) \leq \widehat D_{\a}(\rho\|\sigma) \leq D_{\max}(\rho\|\sigma).
  \end{equation}
  \item ({Closed-form expression of the channel divergence}): For any quantum channel $\cN_{A'\to B}$, subchannel $\cM_{A'\to B}$ and $\a \in (1,2]$, the geometric \Renyi channel divergence is given by
  \begin{align}
    \widehat D_\a(\cN\|\cM) = \frac{1}{\a-1}\log \Big\|\tr_B G_{1-\a}(J_{AB}^{\cN},J_{AB}^{\cM})\Big\|_\infty,
  \end{align}
  where $J_{AB}^{\cN}$ and $J_{AB}^{\cM}$ are the corresponding Choi matrices of $\cN$ and $\cM$ respectively.
    Moreover, for the Belavkin-Staszewski channel divergence, its has the closed-form expression:
  \begin{align}
    \widehat D(\cN\|\cM) = \left\|\tr_B\, \left\{ (J^{\cN}_{AB})^{\frac{1}{2}} \log\left[ (J^{\cN}_{AB})^{\frac{1}{2}} (J^{\cM}_{AB})^{-1} (J^{\cN}_{AB})^{\frac{1}{2}} \right]   (J^{\cN}_{AB})^{\frac{1}{2}} \right\}  \right\|_\infty \, .
  \end{align}
  \item ({Additivity under tensor product}):  Let $\cN_1$ and $\cN_2$ be two quantum channels and let $\cM_1$ and $\cM_2$ be two subchannels. Then for any $\a \in (1,2]$, it holds that
  \begin{align}
    \widehat D_\a(\cN_1\ox \cN_2\|\cM_1\ox \cM_2) = \widehat D_\a(\cN_1\|\cM_1) + \widehat D_\a(\cN_2\|\cM_2).
  \end{align}
  \item ({Chain rule}): Let $\rho$ be a quantum state on $\cH_{RA}$, $\sigma$ be a subnormalized state on $\cH_{RA}$ and $\cN_{A\to B}$ be a quantum channel, $\cM_{A\to B}$ be a subchannel and $\alpha \in (1,2]$. Then
\begin{align}
\widehat{D}_{\alpha}(\cN_{A\to B}(\rho_{RA}) \| \cM_{A\to B}(\sigma_{RA})) &\leq \widehat{D}_{\alpha}( \rho_{RA} \| \sigma_{RA}) + \widehat{D}_{\alpha}(\cN \| \cM) \, .
\end{align}
\item ({Sub-additivity under channel composition}): For any quantum channels $\cN^1_{A\to B}$, $\cN^2_{B\to C}$, any subchannels $\cM^1_{A\to B}$, $\cM^2_{B\to C}$ and $\a \in (1,2]$, it holds
  \begin{align}
    \widehat D_{\a}(\cN_2\circ \cN_1\|\cM_2\circ \cM_1) \leq \widehat D_{\a}(\cN_1\|\cM_1) + \widehat D_{\a}(\cN_2\|\cM_2).
  \end{align}
\item (Semidefinite representation): Let $\bcV$ be a set of subchannels from $A$ to $B$ characterized by certain semidefinite conditions. For any quantum channel $\cN_{A\to B}$ and $\alpha(\ell) = 1+2^{-\ell}$ with $\ell \in \NN$, the optimization $\min_{\cM \in \bcV} \widehat D_{\a}(\cN\|\cM)$ can be computed by a semidefinite program: 
\begin{gather}
  2^\ell\cdot \log \min\ y \quad \text{\rm s.t.}\quad \dbhbig{M,\{N_i\}_{i=0}^\ell, J_{\cM}, y},\notag\\[2pt]
   \dbp{\begin{matrix}
    M & J_{\cN}\\
    J_{\cN} & N_{\ell}
  \end{matrix}},
  \left\{\dbp{\begin{matrix}
    J_{\cN} & N_{i} \\
    N_{i} & N_{i-1}
  \end{matrix}}\right\}_{i=1}^\ell, 
  \dbebigg{N_0 - J_{\cM}},
  \dbpbigg{y\1_A - \tr_B M}, \cM \in \bcV,
  \label{eq:  SDP formula for maximal Renyi channel divergence}
\end{gather}
where $J_{\cN}$ and $J_{\cM}$ are the corresponding Choi matrices of $\cN$ and $\cM$ respectively.
Here the short notation that $\dbp{X}$, $\dbe{X}$ and $\dbh{X}$ represent the positive semidefinite condition $X \geq 0$, the equality condition $X = 0$ and the Hermitian condition $X = X^\dagger$, respectively.
\end{enumerate}
\end{theorem}

\begin{remark}
Inequality \eqref{eq:comparisonDDaDmax} acts as a starting point of our improvement on the previous capacity bounds built on the max-relative entropy. The closed-form expression of the channel divergence directly leads to the additivity property in Item 3 and the semidefinite representation in Item 6. These properties should be contrasted with the situation for the Petz or sandwiched \Renyi divergence for channels which is unclear how to be calculated efficiently. 
The chain rule is another fundamental property that sets a difference of the geometric \Renyi divergence with other variants.  Using the notion of \emph{amortized channel divergence}~\cite{Berta2018}
\begin{align}
  \widehat{D}_{\alpha}^{A}(\cN \| \cM) \equiv \max_{\rho_{RA},\sigma_{RA} \in \cS(RA)} \Big[\widehat{D}_{\alpha}(\cN_{A \to B}(\rho_{RA}) \| \cM_{A\to B}(\sigma_{RA})) -  \widehat{D}_{\alpha}( \rho_{RA} \| \sigma_{RA})\Big],
\end{align} 
the chain rule is equivalent to
\begin{align} \label{eq_amortized}
\widehat{D}_{\alpha}^{A}(\cN \| \cM) = \widehat{D}_{\alpha}(\cN \| \cM)  \quad \text{for} \quad \alpha \in (1,2].
\end{align}
That is, the ``amortization collapse'' for the geometric \Renyi divergence. This solves an open question from~\cite[Eq.~(55)]{Berta2018} in the area of quantum channel discrimination. 
\end{remark}

\begin{remark}
Note that the chain rule property in Item 4 does not hold for the Umegaki relative entropy in general unless we consider the regularized channel divergence~\cite{Fang2019chainrule}. That is, the amortization does not collapse for the Umegaki relative entropy. This indicates that the results we obtained in this work based on the geometric \Renyi divergence cannot be easily extended to the Umegaki relative entropy.
\end{remark}

\begin{remark}
Except for the condition $\cM \in \bcV$, the semidefinite representation in the above Item 6 with $\a(\ell) = 1 + 2^{-\ell}$ is described by $\ell + 3$ linear matrix inequalities, each of size no larger than $2d \times 2d$ with $d = |A||B|$. Thus the computational complexity (time-usage) for computing $\min_{\cM \in \bcV} \widehat D_{\a}(\cN\|\cM)$ is the same as computing $\min_{\cM \in \bcV}  D_{\max}(\cN\|\cM)$. In practice, taking $\ell = 0$ ($\alpha = 2$) already gives an improved result and choosing $\ell$ around $8 \sim 10$ will make the separation more significant. Moreover, a sight modification can be done~\cite{fawzi2017lieb} to compute the optimization for any $\a \in (1,2]$. But we will restrict our attention, without loss of generality, to the discrete values $\a(\ell) = 1 + 2^{-\ell}$ with $\ell \in \NN$.
\end{remark}



\subsection{Detailed proofs}

In the following, we give a detailed proof of each property listed in Theorem~\ref{thm: summary of properties}.

\begin{lemma}[Comparison with $D$ and $D_{\max}$]\label{thm: divergence chain inequality}
    For any quantum state $\rho$, sub-normalized quantum state $\sigma$ with $\rho \ll \sigma$ and $\a \in (1,2]$, the following relation holds
    \begin{align}\label{divergence chain inequality}
        D(\rho\|\sigma)\leq \widetilde D_\a(\rho\|\sigma) \leq \widebar D_\a(\rho\|\sigma)  \leq \widehat D_\a(\rho\|\sigma) \leq D_{\max}(\rho\|\sigma).
    \end{align}
\end{lemma}

\begin{proof}
    The first two inequalities follow from Eqs.~\eqref{eq: sand renyi and relative entropy} and~\eqref{eq: sandwiched and petz relation}. The third inequality follows since the geometric \Renyi divergence is the largest \Renyi divergence satisfying the data-processing inequality (see~\cite{matsumoto2015new} or~\cite[Eq.~(4.34)]{Tomamichel2015b}). It remains to prove the last one.
    Since the geometric R\'{e}nyi divergence is monotonically non-decreasing with respect to $\a$~\footnote{This is clear from the minimization formula of $\widehat D_\a$ in~\cite[Eq. (11)]{matsumoto2015new} and the monotonicity of classical \Renyi divergence.}, it suffices to show that $\widehat D_2(\rho\|\sigma) \leq D_{\max}(\rho\|\sigma)$. Recall that $D_{\max}(\rho\|\sigma) = \min\{\log t\,|\, \rho \leq t \sigma\}$. Denote the optimal solution as $t$, and we have $D_{\max}(\rho\|\sigma) = \log t$ with $0 \leq \rho \leq t \sigma$. 
    Note that 
    \begin{align}\label{hat D 2 SDP}
        \widehat D_2(\rho\|\sigma) & = \log \tr \left[\rho^{} \sigma^{-1} \rho\right]\notag\\
    & = \log \min \big\{\tr M \,\big|\, \rho \sigma^{-1} \rho \leq M \big\}
     = \log \min \bigg\{\tr M \,\bigg| \begin{bmatrix}
            M & \rho\\ \rho & \sigma
        \end{bmatrix} \geq 0\bigg\},
    \end{align}
    where the last equality follows from the Schur complement characterization of the block positive semidefinite matrix.
    Take $M = t \rho$, and we have
    \begin{align}
        \begin{bmatrix}
            M & \rho\\ \rho & \sigma
        \end{bmatrix} = \begin{bmatrix}
            t\rho & \rho\\ \rho & \sigma
        \end{bmatrix} \geq \begin{bmatrix}
            t\rho & \rho\\ \rho & t^{-1} \rho
        \end{bmatrix} = \begin{bmatrix}
            t & 1\\ 1 & t^{-1}
        \end{bmatrix}\ox \rho \geq 0.
    \end{align}
    Thus $M = t \rho$ is a feasible solution of optimization~\eqref{hat D 2 SDP} which implies $\widehat D_2(\rho\|\sigma) \leq \log \tr [t \rho] = \log t = D_{\max}(\rho\|\sigma)$. This completes the proof.
\end{proof}

\vspace{0.2cm}
Compared with $D_{\max}$, it is clear that $\widehat D_\a$ gives a tighter approximation of the Umegaki relative entropy $D$ from above. We provide a concrete example in Figure~\ref{fig: renyi compare} to give an intuitive understanding of the relations between different divergences. 

\begin{figure}[H]
	\centering
\begin{tikzpicture}
  \begin{scope}[scale=0.8]
  \draw[thick,->] (1.2,0) -- (10.7,0);
  \foreach \x in {2,4,6,8,10}
  \draw (\x cm,1.5pt) -- (\x cm,-1.5pt);
  \node at (2,-0.3) {\footnotesize $1.0$};
  \node at (4,-0.3) {\footnotesize $1.25$};
  \node at (6,-0.3) {\footnotesize $1.5$};
  \node at (8,-0.3) {\footnotesize $1.75$};
  \node at (10,-0.3) {\footnotesize $2.0$};
  \node at (11,0) {$\a$};
  \node[inner sep=0pt] at (6,3.1) {\includegraphics[width=6.48cm]{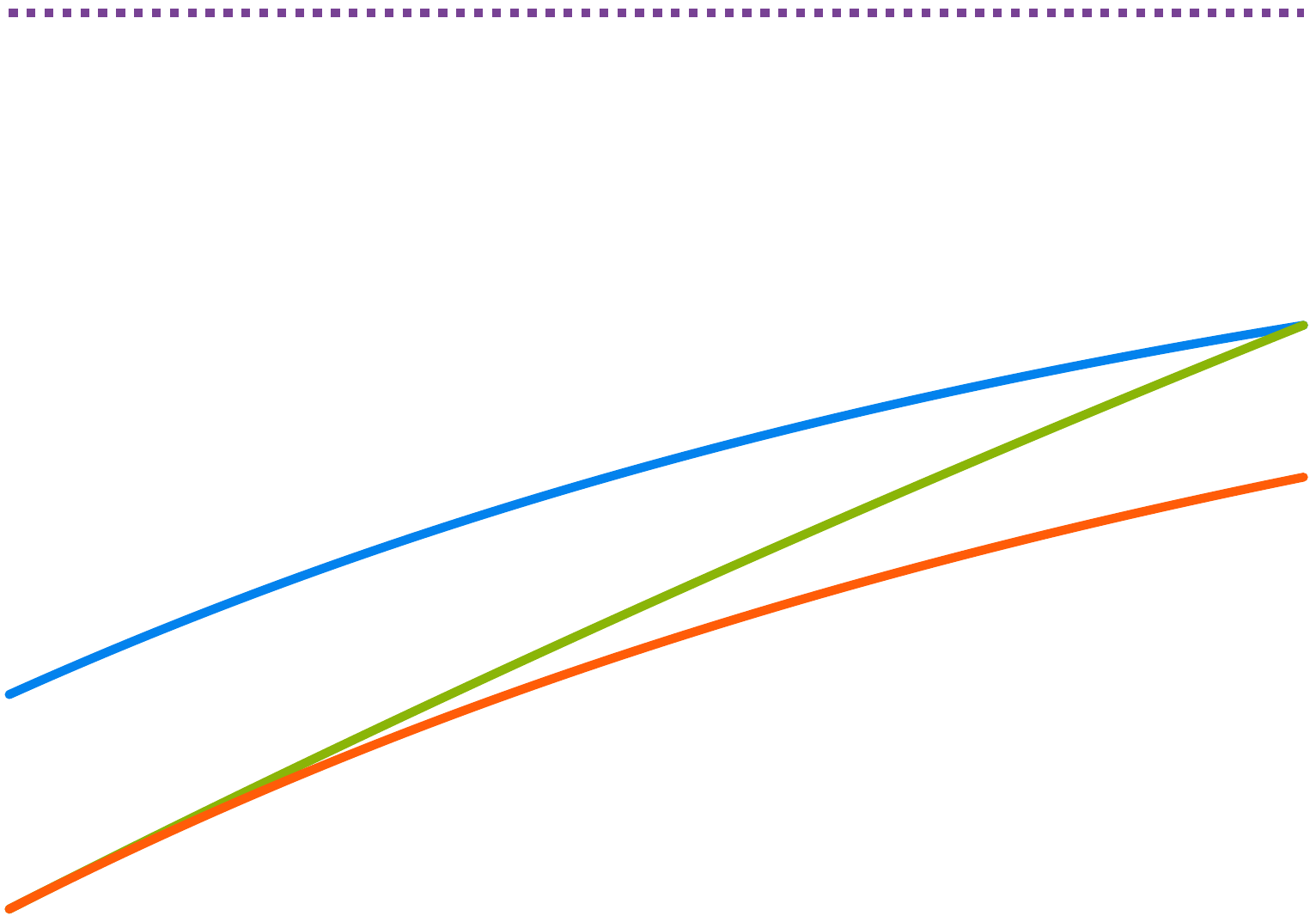}};
  \draw[very thick,gray,dashed] (2,0) -- (2,6);
  \draw[very thick,gray,dashed] (10,0) -- (10,6);  
  \node at (1.45,5.85) {\footnotesize $D_{\max}$};
   \node[draw, star, star points=5, star point ratio=.5,fill=black,black,scale=0.8] at (2,0.35) {};
  \node at (1.55,0.4) {\footnotesize $D$};
   \node at (2,1.66) [circle,fill,inner sep=2pt]{};
   \node at (1.55,1.7) {\footnotesize $\widehat D$};
   \node at (10,3.95) [circle,fill,inner sep=2pt]{};
   \node at (10.4,4) {\footnotesize $\widehat D_2$};
   \node at (6,3.4) {\footnotesize $\widehat D_\a$};
   \node at (5.5,2.4) {\footnotesize $\widebar D_\a$};
   \node at (8,2.2) {\footnotesize $\widetilde D_\a$};
   \node at (4.4,4.8) { $\rho = \frac{1}{4}\left[\begin{smallmatrix}
   	2 & 1 & 1\\[2pt] 1 & 1 & 1\\[2pt] 1 & 1 & 1
   \end{smallmatrix}\right]$};
   \node at (7.5,4.8) { $\sigma = \frac{1}{8}\left[\begin{smallmatrix}
   	4 & 0 & 0\\[2pt] 0 & 3 & 0\\[2pt] 0 & 0 & 1
   \end{smallmatrix}\right]$};       
  \end{scope}
\end{tikzpicture}
\caption{\small Relations between Umegaki relative entropy $D$, Belavkin-Staszewski relative entropy $\widehat D$, max-relative entropy $D_{\max}$, sandwiched \Renyi divergence $\widetilde D_{\a}$, Petz \Renyi divergence $\overline{D}_{\a}$ and geometric \Renyi divergence $\widehat D_{\a}$.}
\label{fig: renyi compare}
\end{figure}

\begin{lemma}[Closed-form expression]\label{lem: maximal Renyi channel divergence SDP}
  For any quantum channel $\cN_{A'\to B}$, subchannel $\cM_{A'\to B}$ and $\a \in (1,2]$, the geometric \Renyi channel divergence is given by
  \begin{align}\label{eq: maximal renyi channel divergence 1}
    \widehat D_\a(\cN\|\cM) = \frac{1}{\a-1}\log \Big\|\tr_B G_{1-\a}(J_{AB}^{\cN},J_{AB}^{\cM})\Big\|_\infty,
  \end{align}
  where $J_{AB}^{\cN}$ and $J_{AB}^{\cM}$ are the corresponding Choi matrices of $\cN$ and $\cM$ respectively.
  Moreover, for the Belavkin-Staszewski channel divergence, its has the closed-form expression:
  \begin{align}
    \widehat D(\cN\|\cM) = \left\|\tr_B\, \left\{ (J^{\cN}_{AB})^{\frac{1}{2}} \log\left[ (J^{\cN}_{AB})^{\frac{1}{2}} (J^{\cM}_{AB})^{-1} (J^{\cN}_{AB})^{\frac{1}{2}} \right]   (J^{\cN}_{AB})^{\frac{1}{2}} \right\} \right\|_\infty \, .
  \end{align}
\end{lemma}
\begin{proof}
Note that for any quantum state $\rho_{A}$ and its purification $\phi_{AA'}$, we have the relation
\begin{align} \label{eq_kun1}
\cN_{A'\to B}(\phi_{AA'}) 
=\cN_{A'\to B}\Big(\sqrt{\rho_A} \proj{\Phi}_{AA'} \sqrt{\rho_A} \Big)
= \sqrt{\rho_A} J_{AB}^{\cN}\sqrt{\rho_A} \, .
\end{align}
By definition of the geometric R\'enyi divergence we have
\begin{align}
\widehat D_{\alpha}(\cN\|\cM) & = \max_{\rho_A \in \cS(A)} \widehat D_{\alpha}(\sqrt{\rho_A} J_{AB}^{\cN}\sqrt{\rho_A}\,\|\,\sqrt{\rho_A} J_{AB}^{\cM}\sqrt{\rho_A}\,)\\
& = \frac{1}{\alpha -1} \log \max_{\rho_A \in \cS(A)} \tr\, G_{1-\alpha}(\sqrt{\rho_A} J_{AB}^{\cN}\sqrt{\rho_A},\sqrt{\rho_A} J_{AB}^{\cM}\sqrt{\rho_A}\,)\\
& = \frac{1}{\alpha -1} \log \max_{\rho_A \in \cS(A)} \tr\, \sqrt{\rho_A}\, G_{1-\alpha}(J_{AB}^{\cN},J_{AB}^{\cM})\sqrt{\rho_A} \\
& = \frac{1}{\alpha -1} \log \max_{\rho_A \in \cS(A)} \tr\, G_{1-\alpha}(J_{AB}^{\cN},J_{AB}^{\cM}) (\rho_A \ox \1_B)\\
& = \frac{1}{\alpha -1} \log \max_{\rho_A \in \cS(A)} \tr \big( \left[\tr_B \, G_{1-\alpha}(J_{AB}^{\cN},J_{EB}^{\cM})\right] \rho_A\big)\\
& = \frac{1}{\alpha -1} \log \left\|\tr_B \, G_{1-\alpha}(J_{AB}^{\cN},J_{AB}^{\cM})\right\|_\infty \, ,
\end{align}
where the third step follows from the transformer inequality given in Lemma~\ref{lem_transformer} in Appendix~\ref{app: technical lemmas} and the fact that we can assume by continuity that $\rho_A$ has full rank.\footnote{We recall that this follows from a standard continuity argument which works because the geometric R\'enyi divergence has nice continuity properties, see~\cite[Page 2]{Kubo1980} or \cite[Eq. (72)]{matsumoto2015new}.}
The last line follows from the semidefinite representation of the infinity norm $\|X\|_{\infty} = \max_{\rho \in \cS} \tr X\rho$. 

The expression for $\widehat D$ follows  exactly the same steps by using Corollary~\ref{coro: transformer inequality log} in Appendix~\ref{app: technical lemmas} and replacing the matrix geometric mean with the operator relative entropy.
\end{proof}

\begin{lemma}[Additivity]\label{lem: maximal Renyi channel divergence additivity}
  Let $\cN_1$ and $\cN_2$ be two quantum channels and let $\cM_1$ and $\cM_2$ be two subchannels. Then for any $\a \in (1,2]$ it holds that
  \begin{align}
    \widehat D_\a(\cN_1\ox \cN_2\|\cM_1\ox \cM_2) = \widehat D_\a(\cN_1\|\cM_1) + \widehat D_\a(\cN_2\|\cM_2).
  \end{align}
\end{lemma}
\begin{proof}
Due to the closed-form expression in Lemma~\ref{lem: maximal Renyi channel divergence SDP}, we have
  \begin{align}
    \hspace{-0.2cm}\widehat D_\a(\cN_1\ox \cN_2\|\cM_1\ox \cM_2) & = \frac{1}{\a-1}\log \big\|\tr_{B_1B_2} G_{1-\a}(J_{\cN_1}\ox J_{\cN_2},J_{\cM_1}\ox J_{\cM_2})\big\|_\infty\\
    & =  \frac{1}{\a-1}\log \big\|\tr_{B_1B_2} G_{1-\a}(J_{\cN_1},J_{\cM_1})\ox G_{1-\a}(J_{\cN_2},J_{\cM_2})\big\|_\infty\\
    & =  \frac{1}{\a-1}\log \big\|\tr_{B_1} G_{1-\a}(J_{\cN_1},J_{\cM_1})\big\|_{\infty} \big\|\tr_{B_2} G_{1-\a}(J_{\cN_2},J_{\cM_2})\big\|_\infty\\
    & = \widehat D_\a(\cN_1\|\cM_1) + \widehat D_\a(\cN_2\|\cM_2).
  \end{align}
  The first and last lines follow from Lemma~\ref{lem: maximal Renyi channel divergence SDP}. The second and third lines follow since the weighted matrix geometric mean and the infinity norm are multiplicative under tensor product. 
\end{proof}

\begin{lemma}[Chain rule]  \label{lem_chainRule}
Let $\rho$ be a quantum state on $\cH_{RA}$, $\sigma$ be a subnormalized state on $\cH_{RA}$ and $\cN_{A\to B}$ be a quantum channel, $\cM_{A\to B}$ be a subchannel and $\alpha \in (1,2]$. Then
\begin{align}\label{eq_chainrule_geometric}
\widehat{D}_{\alpha}(\cN_{A\to B}(\rho_{RA}) \| \cM_{A\to B}(\sigma_{RA})) &\leq \widehat{D}_{\alpha}( \rho_{RA} \| \sigma_{RA}) + \widehat{D}_{\alpha}(\cN \| \cM) \, .
\end{align}
\end{lemma}
\begin{proof}
Let $\ket{\Phi}_{SA} = \sum_i \ket{i}_S\ket{i}_A$ be the unnormalized maximally entangled state. Denote $J^{\cN}_{SB}$ and $J^{\cM}_{SB}$ as the Choi matrices corresponding to $\cN$ and $\cM$, respectively. Then we have the identities
\begin{align} \label{eq_choiOutput}
\cN_{A\to B}(\rho_{RA}) 
 = \big\<\Phi_{SA}\big|\rho_{RA} \ox J^{\cN}_{SB}\big|\Phi_{SA}\big\> \quad \text{and} \quad
\cM_{A\to B}(\sigma_{RA})  = \big\<\Phi_{SA}\big|\sigma_{RA} \otimes J^{\cM}_{SB}\big|\Phi_{SA}\big\>\, .
\end{align}
For $y = \|\tr_B\, G_{1-\alpha}(J_{SB}^{\cN},J_{SB}^{\cM})\|_\infty$, Lemma~\ref{lem: maximal Renyi channel divergence SDP} ensures that 
\begin{align}\label{eq: chain relation tmp1}
  \widehat D_{\alpha}(\cN\|\cM) = \frac{1}{\alpha-1} \log y 
\end{align}
and by definition of the infinity norm we find
\begin{align}\label{eq: chain relation tmp2}
 \tr_B\, G_{1-\alpha}(J_{SB}^{\cN},J_{SB}^{\cM}) \leq y \,\1_S \, .
\end{align}
By definition of the geometric \Renyi divergence and by using~\eqref{eq_choiOutput} we can write 
\begin{align}
  \widehat D_{\alpha}&\big(\cN_{A\to B}(\rho_{RA})\|\cM_{A\to B}(\sigma_{RA})\big) \notag \\
  & = \frac{1}{\alpha-1} \log \tr\, G_{1-\alpha} (\big\<\Phi_{SA}\big|\rho_{RA} \ox J^{\cN}_{SB}\big|\Phi_{SA}\big\>,\big\<\Phi_{SA}\big|\sigma_{RA} \otimes J^{\cM}_{SB}\big|\Phi_{SA}\big\>)\\
  & \leq \frac{1}{\alpha-1} \log \tr \big\<\Phi_{SA}\big| G_{1-\alpha} (\rho_{RA} \ox J^{\cN}_{SB}, \sigma_{RA} \otimes J^{\cM}_{SB} )\big|\Phi_{SA}\big\>\\
  & = \frac{1}{\alpha-1} \log \tr \big\<\Phi_{SA}\big| G_{1-\alpha} (\rho_{RA}, \sigma_{RA}) \ox G_{1-\alpha} ( J^{\cN}_{SB}, J^{\cM}_{SB}) \big|\Phi_{SA}\big\>\\
  & = \frac{1}{\alpha-1} \log \tr \big\<\Phi_{SA}\big| G_{1-\alpha} (\rho_{RA}, \sigma_{RA}) \ox \tr_B\, G_{1-\alpha} ( J^{\cN}_{SB}, J^{\cM}_{SB}) \big|\Phi_{SA}\big\>\\
  & \leq \frac{1}{\alpha-1} \log \tr \big\<\Phi_{SA}\big| G_{1-\alpha} (\rho_{RA}, \sigma_{RA}) \ox y\,\1_S \big|\Phi_{SA}\big\>\\
  & = \frac{1}{\alpha-1} \log \big(y\, \tr \,G_{1-\alpha} (\rho_{RA},\sigma_{RA}) \big) \label{eq_notTriv}\\
  & = \frac{1}{\alpha-1} \log y + \frac{1}{\alpha - 1} \log \tr\, G_{1-\alpha} (\rho_{RA},\sigma_{RA})\\
  & = \widehat D_{\alpha}(\cN\|\cM) + \widehat D_{\alpha}(\rho_{RA}\|\sigma_{RA}) \, ,
\end{align}
where the first inequality follows from the transformer inequality given in Lemma~\ref{lem_transformer} in Appendix~\ref{app: technical lemmas}.
The third line follows from the multiplicativity of matrix geometric mean under tensor product. 
The second inequality uses from \eqref{eq: chain relation tmp2} and the fact that $X \mapsto \tr\, K X$ is monotone for positive operator $K$. Equation~\eqref{eq_notTriv} follows from the identity $\<\Phi_{SA}| Y_{RA}\ox \1_S|\Phi_{SA}\> = \tr_A\, Y_{RA}$. 
\end{proof}

\begin{lemma}[Sub-additivity]\label{lem: subadditivity of D alpha}
  For any quantum channels $\cN^1_{A\to B}$, $\cN^2_{B\to C}$, any subchannels $\cM^1_{A\to B}$, $\cM^2_{B\to C}$ and $\a \in (1,2]$, it holds
  \begin{align}
    \widehat D_{\a}(\cN_2\circ \cN_1\|\cM_2\circ \cM_1) \leq \widehat D_{\a}(\cN_1\|\cM_1) + \widehat D_{\a}(\cN_2\|\cM_2).
  \end{align}
\end{lemma}
\begin{proof}
This is a direct consequence of the chain rule in Lemma~\ref{lem_chainRule}. For any pure state $\phi_{AR}$, we have
\begin{align}
  \widehat D_{\a}(\cN_2\circ \cN_1 (\phi_{AR})\|\cM_2\circ \cM_1(\phi_{AR})) & \leq \widehat D_{\alpha}(\cN_2\|\cM_2) +  \widehat D_{\a}(\cN_1 (\phi_{AR})\|\cM_1(\phi_{AR}))\\
  & \leq \widehat D_{\alpha}(\cN_2\|\cM_2) +  \widehat D_{\a}(\cN_1\|\cM_1) + \widehat D_{\alpha}(\phi_{AR}\|\phi_{AR})\\
  & = \widehat D_{\alpha}(\cN_2\|\cM_2) +  \widehat D_{\a}(\cN_1\|\cM_1).
\end{align}
Taking a maximization of $\phi_{AR}$ on the left hand side, we will have the desired result.
\end{proof}

\begin{lemma}[Semidefinite representation]
\label{lem: SDP representation of the channel information measure}
Let $\bcV$ be a set of subchannels from $A$ to $B$ characterized by certain semidefinite conditions. For any quantum channel $\cN_{A\to B}$ and $\alpha(\ell) = 1+2^{-\ell}$ with $\ell \in \NN$, the optimization $\min_{\cM \in \bcV} \widehat D_{\a}(\cN\|\cM)$ can be computed by a semidefinite program: 
\begin{gather}
  2^\ell\cdot \log \min\ y \quad \text{\rm s.t.}\quad \dbhbig{M,\{N_i\}_{i=0}^\ell, J_{\cM}, y},\notag\\[2pt]
   \dbp{\begin{matrix}
    M & J_{\cN}\\
    J_{\cN} & N_{\ell}
  \end{matrix}},
  \left\{\dbp{\begin{matrix}
    J_{\cN} & N_{i} \\
    N_{i} & N_{i-1}
  \end{matrix}}\right\}_{i=1}^\ell, 
  \dbebigg{N_0 - J_{\cM}},
  \dbpbigg{y\1_A - \tr_B M}, \cM \in \bcV,
  \label{eq:  SDP formula for maximal Renyi channel divergence}
\end{gather}
where $J_{\cN}$ and $J_{\cM}$ are the corresponding Choi matrices of $\cN$ and $\cM$ respectively.
\end{lemma}
\begin{proof}
This is a direct consequence of the closed-form expression in Lemma~\ref{lem: maximal Renyi channel divergence SDP} and the semidefinite representation of the weighted matrix geometric means in~\cite{fawzi2017lieb} (see also Lemma~\ref{geometric SDP general lemma} in Appendix~\ref{app: technical lemmas}), as well as the semidefinite representation of the infinity norm of an Hermitian operator $\|X\|_\infty = \min \{y\,|\, X \leq y\1\}$.
\end{proof}

\newpage
\section{Quantum communication}
\label{sec: Quantum communication}

\subsection{Backgrounds}

The \emph{quantum capacity} of a noisy quantum channel is the maximum rate at which it can reliably transmit quantum information over asymptotically many uses of the channel. There are two different quantum capacities of major concern, the (unassisted) quantum capacity $Q$ and the two-way assisted quantum capacity $Q^{\leftrightarrow}$ , depending on whether classical communication is allowed between each channel uses.
 
The well-established quantum capacity theorem shows
that the quantum capacity is equal to the regularized channel coherent information~\cite{Lloyd1997,Shor2002a,Devetak2005a,Schumacher1996a,Barnum2000,Barnum1998},
\begin{align}\label{eq: quantum channel coding theorem}
    Q(\cN) = \lim_{n\to \infty} \frac{1}{n} I_{c}(\cN^{\ox n}) = \sup_{n\in \NN} \frac{1}{n} I_{c}(\cN^{\ox n}),
\end{align}
where $I_c(\cN)\equiv \max_{\rho \in \cS} \left[H(\cN(\rho)) - H(\cN^c(\rho))\right]$ is the channel coherent information, $H$ is the von Neumann entropy and $\cN^c$ is the complementary channel of $\cN$. 
The regularization in~\eqref{eq: quantum channel coding theorem} is necessary in general since the channel coherent information is non-additive~\cite{DiVincenzo1998a,leditzky2018dephrasure} and an unbounded number of channel uses may be required to detect a channel's capacity~\cite{Cubitt2015}. For this reason, the quantum capacity is notoriously difficult to evaluate, not to mention the scenario with two-way classical communication assistance.

Substantial efforts have been made in providing single-letter lower and upper bounds (e.g. \cite{Holevo2001,Muller-Hermes2016,Smith2008a,Sutter2014,Gao2015a,Smith2008b}). Most of them require certain symmetries of the channel to be computable or relatively tight. Of particular interest is a strong converse bound given by Tomamichel et al.~\cite{tomamichel2017strong}.
Inspired by the Rains bound from entanglement theory~\cite{Rains2001}, they introduced the Rains information ($R$) of a quantum channel and further proved that it is a strong converse on quantum capacity. However, $R$ is not known to be computable in general due to its minimax optimization of the Umegaki relative entropy. For the ease of computability, Wang et al.~\cite{Wang2017d} relaxed the Umegaki relative entropy to the max-relative entropy, obtaining a variant known as the max-Rains information ($R_{\max}$). Leveraging the semidefinite representation of the max-relative entropy, they showed that $R_{\max}$ is efficiently computable via a semidefinite program.
It was later strengthened by Berta \& Wilde~\cite{Berta2017a} that $R_{\max}$ is also a strong converse on quantum capacity under two-way classical communication assistance. Since then, the max-Rains information $R_{\max}$  is arguably~\footnote{Another known strong converse bound is the entanglement-assisted quantum capacity~\cite{Bennett2014} which can be estimated by a algorithm in~\cite{fawzi2018efficient,Fawzi2017}. But this bound is usually larger than the max-Rains information since the entanglement assistance is too strong.} the best-known computable strong converse bound on both assisted and unassisted quantum capacities in general.

\subsection{Summary of results}

In this part, we aim to improve the max-Rains information in both assisted and unassisted scenarios. The structure of this part is organized as follows (see also a schematic diagram in Figure~\ref{fig: quantum communiation summary}).

In Section~\ref{sec: Maximal Renyi Rains information} we discuss the unassisted quantum communication. Based on the notion of the generalized Rains information in~\cite{tomamichel2017strong}, we exhibit that the generalized Rains information induced by the geometric \Renyi divergence ($\widehat R_{\a}$) can be computed as a semidefinite program (SDP), improving the previously known result of the max-Rains information~\cite{Wang2017d} in general. That is, we show that
\begin{align*}
Q(\cN) \leq Q^\dagger(\cN)\leq R(\cN) \leq \widehat R_{\a}(\cN) \leq R_{\max}(\cN),\quad \text{with}\ \widehat R_{\a}(\cN)\ \text{SDP computable},
\end{align*}
where $Q(\cN)$ and $Q^{\dagger}(\cN)$ denote the unassisted quantum capacity of channel $\cN$ and its corresponding strong converse capacity, respectively.

In Section~\ref{sec: Maximal Renyi Theta-information}, we study the quantum communication with PPT assistance, an assistance stronger than the two-way classical communication. We introduce the \emph{generalized Theta-information} which is a new variant of channel information inspired by the channel resource theory. More precisely, we define the generalized Theta-information as a ``channel distance'' to the class of subchannels given by the zero set of Holevo-Werner bound ($Q_{\Theta}$)~\cite{Holevo2001}. Interestingly, we show that the max-Rains information $R_{\max}$ coincides with the generalized Theta-information induced by the max-relative entropy $R_{\max,\Theta}$, i.e., $R_{\max} = R_{\max,\Theta}$, thus providing a completely new perspective of understanding the former quantity. Moreover, we prove that the generalized Theta-information induced by the geometric \Renyi divergence ($\widehat R_{\a,\Theta}$) is a strong converse on the PPT-assisted quantum capacity by utilizing an ``amortization argument''. Together with its SDP formula, we conclude that $\widehat R_{\a,\Theta}$ improves the previous result of the max-Rains information~\cite{Berta2017a} in general. That is, we show that
\begin{align*}
Q^{\PPT,\leftrightarrow}(\cN) \leq  Q^{\PPT,\leftrightarrow,\dagger}(\cN) \leq \widehat R_{\a,\Theta}(\cN) \leq R_{\max}(\cN),\quad \text{with}\ \widehat R_{\a,\Theta}(\cN)\ \text{SDP computable},
\end{align*}
where $Q^{\PPT,\leftrightarrow}(\cN)$ and $Q^{\PPT,\leftrightarrow,\dagger}(\cN)$ denote the PPT-assisted quantum capacity of channel $\cN$ and its corresponding strong converse capacity, respectively.

In Section~\ref{sec: Extension to bidirectional channels}, we consider the PPT-assisted quantum communication via bidirectional quantum channels, a more general model than the usual point-to-point channels. We extend the results in Section~\ref{sec: Maximal Renyi Theta-information} to this general model and demonstrate an improvement to the previous result of the bidirectional max-Rains information ($R_{\max}^{\bi}$)~\cite{Bauml2018}. That is, we show that
  \begin{align*}
  Q^{\bi,\PPT,\leftrightarrow}(\cN) \leq Q^{\bi,\PPT,\leftrightarrow,\dagger}(\cN) \leq \widehat R_{\a,\Theta}^{\bi}(\cN) \leq  R_{\max}^{\bi}(\cN), \quad \text{with}\ \widehat R_{\a,\Theta}^{\bi}(\cN)\ \text{SDP computable},
\end{align*}
where $Q^{\bi,\PPT,\leftrightarrow}(\cN)$ and $Q^{\bi,\PPT,\leftrightarrow,\dagger}(\cN)$ denote the PPT-assisted quantum capacity of a bidirectional channel $\cN$ and its corresponding strong converse capacity, respectively.

Finally in Section~\ref{sec: quantum capacity Examples} we investigate several fundamental quantum channels, demonstrating the efficiency of our new strong converse bounds. It turns out that our new bounds work exceptionally well and exhibit a significant improvement on the max-Rains information for almost all cases.

\begin{figure}[H]
\centering
\begin{tikzpicture}
\draw[fill=gray!10,opacity=0.3] (7,-11.5) rectangle node[gray,opacity=1,midway,shift={(0.2,2.7)}]{{\small SDP computable}} (13,-16.5);
\node (QPPT) at (0,-12) {$Q^{\text{PPT},\leftrightarrow}$};
\node (QPPTdagger) at (2.5,-12) {$Q^{\text{PPT},\leftrightarrow,\dagger}$};
\node (Qtwoway) at (0,-14) {$Q^{\leftrightarrow}$};
\node[circle,fill=magenta!10,inner sep=1pt,minimum size=2pt] (Qtwowaydagger) at (2.5,-14) {$Q^{\leftrightarrow,\dagger}$};
\node (Q) at (0,-16) {$Q$};
\node[circle,fill=magenta!10,inner sep=1pt,minimum size=2pt] (Qdagger) at (2.5,-16) {$Q^{\dagger}$};
\node (R) at (5,-16) {$R$};
\node[red] (Ra) at (7.5,-16) {$\widehat R_\a$};
\node[red] (Rasdp) at (7.5,-12) {$\widehat R_{\a,\Theta}$};
\node[red] (RmaxTheta) at (10,-12) {$R_{\max,\Theta}$};
\node (Rmax) at (10,-16) {$R_{\max}$};
\node (Qtheta) at (12.5,-16) {$Q_{\Theta}$};

\draw[very thick,->,dotted,gray] (R) -- node[black,rotate = -40,midway,shift={(-0.2,0.15)}] {\footnotesize cov.} node[black,rotate=-39,midway,shift={(-0.2,-0.2)}] {\scriptsize \cite{tomamichel2017strong}} (Qtwowaydagger);
\draw[thick,<->] (RmaxTheta) -- node[rotate = 90, midway,shift={(0.1,0.2)}] {\footnotesize Prop.~\ref{prop: Rains and Theta information}} (Rmax);
\draw[thick,->] (Qtheta) -- node[midway,shift={(0.1,0.2)}] {\scriptsize \cite{Wang2017d}} node[midway,shift={(0.1,-0.2)}] {\footnotesize $\boldsymbol\neq$}(Rmax);
\draw[thick,->] (RmaxTheta) -- node[midway,shift={(0.1,0.2)}] {\footnotesize Lem.~\ref{thm: divergence chain inequality}} node[midway,shift={(0.1,-0.2)}] {\footnotesize $\boldsymbol\neq$} (Rasdp);
\draw[thick,->] (Rmax) -- node[midway,shift={(0.1,0.2)}] {\footnotesize Lem.~\ref{thm: divergence chain inequality}} node[midway,shift={(0.1,-0.2)}] {\footnotesize $\boldsymbol\neq$} (Ra);
\draw[red,very thick,->] (Rasdp) -- node[black,midway,shift={(-0.1,0.2)}] {\footnotesize Thm.~\ref{thm: main result quantum assisted}} (QPPTdagger);
\draw[thick,->] (QPPTdagger) -- (QPPT);
\draw[thick,->] (QPPTdagger) -- (Qtwowaydagger);
\draw[thick,->] (QPPT) -- (Qtwoway);
\draw[thick,->] (Qtwowaydagger) -- (Qtwoway);
\draw[thick,->] (Qtwoway) -- (Q);
\draw[thick,->] (Qtwowaydagger) -- (Qdagger);
\draw[thick,->] (R) -- node[midway,shift={(0,0.2)}] {\scriptsize \cite{tomamichel2017strong}} node[midway,shift={(0.1,-0.2)}] {\footnotesize $\boldsymbol\neq$} (Qdagger);
\draw[thick,->] (Qdagger) -- (Q);
\draw[thick,->] (Ra) -- node[midway,shift={(0,0.2)}] {\footnotesize Lem.~\ref{thm: divergence chain inequality}} node[midway,shift={(0,-0.2)}] {\footnotesize $\boldsymbol\neq$} (R);
\draw[thick,->] (Rasdp) -- node[rotate = 90, midway,shift={(0.1,0.2)}] {\footnotesize Prop.~\ref{prop: Rains and Theta information}} node[rotate=90,midway,shift={(0,-0.2)}] {\footnotesize $\boldsymbol\neq$} (Ra);
\draw[thick,->] (Rmax) -- node[rotate=-30,midway,shift={(-0.2,0.2)}] {\scriptsize \cite{Berta2017a}} (QPPTdagger);
\draw[thick,->] (Rmax) -- (10,-17) --  node[midway,shift={(-0.2,0.2)}]{\scriptsize \cite{Wang2016a,Wang2017d}} (2.5,-17) -- (Qdagger);
\end{tikzpicture}
\caption{\small Relations between different converse bounds for quantum communication. $Q^*$ and $Q^{*,\dagger}$ are the quantum capacity with assistance $*$ and its corresponding strong converse capacity, respectively. $R$, $\widehat R_\a$ and $R_{\max}$ are the generalized Rains information induced by different quantum divergences. $\widehat R_{\a,\Theta}$ and $R_{\max,\Theta}$ are the generalized Theta-information induced by different quantum divergences. $Q_\Theta$ is the Holevo-Werner bound. The circled quantities are those of particular interest in quantum information theory. The key quantities and the main contributions in this part are marked in red. The quantity at the start point of an arrow is no smaller than the one at the endpoint. The double arrow represents that two quantities coincide. The inequality sign represents that two quantities are not the same in general. The dotted arrow represents that the relation holds under certain restrictions, where ``cov.'' stands for ``covariant''. The parameter $\a$ is taken in the interval $(1,2]$. The quantities in the shaded area are SDP computable in general.}
\label{fig: quantum communiation summary}
\end{figure}

\subsection{Unassisted quantum capacity}
\label{sec: Maximal Renyi Rains information}

In this section, we discuss converse bounds on the unassisted quantum capacity~\footnote{We refer to the work~\cite[Section II]{tomamichel2017strong} for rigorous definitions of the unassisted quantum capacity and its strong converse.}.

\begin{definition}[\cite{tomamichel2017strong}]
For any generalized divergence $\bD$, the generalized Rains bound of a quantum state $\rho_{AB}$ is defined as 
\begin{align}\label{eq: generalized Rains bound}
    \bR(\rho_{AB})\equiv \min_{\sigma_{AB} \in \PPT'(A:B)} \bD(\rho_{AB}\|\sigma_{AB}),
\end{align}
where the minimization is taken over the Rains set $\PPT'(A:B) \equiv \big\{\sigma_{AB}\,\big|\, \sigma_{AB} \geq 0,\,\big\|\sigma_{AB}^{\sfT_B}\big\|_1 \leq 1\big\}$.
\end{definition}

\begin{definition}[\cite{tomamichel2017strong}]
For any generalized divergence $\bD$, the generalized Rains information of a quantum channel $\cN_{A'\to B}$ is defined as
\begin{align}\label{generalized Rains information definition eq}
    \boldsymbol R(\cN) \equiv \max_{\rho_A \in \cS(A)} \boldsymbol R(\cN_{A'\to B}(\phi_{AA'})) =  \max_{\rho_A \in \cS(A)} \min_{\sigma_{AB} \in \PPT'(A:B)} \bD(\cN_{A'\to B}(\phi_{AA'})\|\sigma_{AB})
\end{align}
where $\phi_{AA'}$ is a purification of quantum state $\rho_A$. 
\end{definition}
In particular, the Rains information is induced by the Umegaki relative entropy~\cite{tomamichel2017strong},
\begin{align}
    R(\cN) = \max_{\rho_A \in \cS(A)}\min_{\sigma_{AB} \in \PPT'(A:B)}  D(\cN_{A'\to B}(\phi_{AA'})\|\sigma_{AB}).
\end{align}
The max-Rains information is induced by the max-relative entropy~\cite{Wang2017d},
\begin{align}
    R_{\max}(\cN) = \max_{\rho_A \in \cS(A)}\min_{\sigma_{AB} \in \PPT'(A:B)}  D_{\max}(\cN_{A'\to B}(\phi_{AA'})\|\sigma_{AB}).
\end{align}

Denote $\widehat R_{\a}$ as the generalized Rains information induced by the geometric \Renyi divergence. We have the following result.

\begin{theorem}[Main result 1]\label{thm: main result quantum unassisted}
    For any quantum channel $\cN$ and $\a \in (1,2]$, it holds
        \begin{align}
            Q(\cN) \leq Q^\dagger(\cN)\leq R(\cN) \leq \widehat R_{\a}(\cN) \leq R_{\max}(\cN),
        \end{align}
    where $Q(\cN)$ and $Q^{\dagger}(\cN)$ denote the unassisted quantum capacity of channel $\cN$ and its corresponding strong converse capacity, respectively.
\end{theorem}
\begin{proof}
The first two inequalities follow since the Rains information $R(\cN)$ has been proved to be a strong converse bound on the unassisted quantum capacity~\cite{tomamichel2017strong}. The last two inequalities are direct consequences of the inequalities in Lemma~\ref{thm: divergence chain inequality}. 
\end{proof}

\begin{remark}
  Note that in the limit of $\a \to 1$, the bound $\widehat R_{\a}$ will converge to the Rains information induced by the Belavkin-Staszewski relative entropy as mentioned in Eq.~\eqref{eq: geo and BS}.
\end{remark}

\vspace{0.2cm}
The following result shows how to compute the newly introduced bound $\widehat R_{\a}(\cN)$ as an SDP. 

\begin{proposition}[SDP formula]\label{prop: maximal renyi rains SDP formula}
For any quantum channel $\cN$ and $\a(\ell) = 1+2^{-\ell}$ with $\ell\in \mathbb N$, it holds
\begin{align}
    \widehat R_{\a}(\cN) = \ell \cdot 2^\ell - (2^\ell+1)\log(2^\ell+1) + (2^\ell +1) \log S_\a(\cN),
\end{align}
with $S_\a(\cN)$ given by the following SDP
\begin{gather}
    S_\a(\cN) = \max \ \tr\left[\left(\plsdagger{K} - \ssum_{i=1}^\ell  W_i \right)\boldsymbol\cdot J_{\cN}\right] \quad \text{\rm s.t.}\quad \dblbig{ K,\{ Z_i\}_{i=0}^\ell},\dbhbig{\{ W_i\}_{i=1}^\ell,\rho}, \notag\\
    \dbp{\begin{matrix}
        \rho\ox \1 &  K\\ {K}^\dagger & \plsdagger{ Z}_{\ell}
    \end{matrix}},
    \left\{\dbp{\begin{matrix}
         W_{i} &  Z_{i}\\
         Z_{i}^\dagger & \plsdagger{ Z}_{i-1}
    \end{matrix}}\right\}_{i=1}^\ell, 
    \dbpbigg{\rho\ox \1 \pm \left[\plsdagger{ Z}_0\right]^{\sfT_B}},\dbebigg{\tr \rho - 1}\label{Renyi Rains information max},
\end{gather}
where $J_{\cN}$ is the Choi matrix of $\cN$ and $\plsdagger{X} \equiv X + X^\dagger$ denotes the Hermitian part of $X$.
\end{proposition}
\begin{proof}
   The proof involves a non-trivial scaling technique for variables replacement, which is important for simplifying the minimax optimization of $\widehat R_{\alpha}$ to a single SDP. A detailed proof is given in Section~\ref{sec: quantum capacity detailed proofs}.

\end{proof}

\subsection{Two-way assisted quantum capacity}
\label{sec: Maximal Renyi Theta-information}

In this section, we discuss converse bounds on two-way assisted quantum capacity~\footnote{We refer to~\cite[Section 4]{Berta2017a} for rigorous definitions of the PPT/two-way assisted quantum capacity and its strong converse.}.
Recall that the Rains bound in~\eqref{eq: generalized Rains bound} is essentially established as the divergence between the given state and the Rains set --- a set of sub-normalized  states given by the zero set~\footnote{It makes no difference by considering $\|\rho_{AB}^{\sfT_B}\|_1 = 1$ or $\|\rho_{AB}^{\sfT_B}\|_1 \leq 1$.} of the log-negativity $E_N(\rho_{AB})\equiv \log \|\rho_{AB}^{\sfT_B}\|_1$ \cite{Plenio2005b}. With this in mind, we introduce a new variant of the channel's analog of Rains bound, compatible with the notion of channel resource theory. Specifically, consider the Holevo-Werner bound~\cite{Holevo2001} --- a channel's analog of the log-negativity,
\begin{align}
    Q_\Theta(\cN) \equiv \log \|\Theta_B\circ\cN_{A\to B}\|_\di,
\end{align}
where $\Theta$ is the transpose map and 
$\|\cF_{A'\to B}\|_\di\equiv \sup_{X_{AA'}\in \cL(AA')} \|\cF_{A'\to B}(X_{AA'})\|_1/ \|X_{AA'}\|_1$
 is the diamond norm~\cite{kitaev1997quantum}. In particular, this bound can be represented as the following SDP,
\begin{align}
    Q_\Theta(\cN) = \log \min \left\{y\,\Big|\, Y_{AB} \pm J_{\cN}^{\sfT_B} \geq 0,\, \tr _B Y_{AB} \leq y \1_A\right\}.
\end{align}
Inspired by the formulation of the Rains set, we define the set of subchannels given by the zero set of the Holevo-Werner bound $Q_{\Theta}$ as
\begin{align}\label{channel's rains set}
    \bcV_\Theta  \equiv \left\{\cM \in \text{CP}(A:B)\,\Big|\, \exists\, Y_{AB}, \ \text{s.t.}\ Y_{AB} \pm J_{\cM}^{\sfT_B} \geq 0,\, \tr _B Y_{AB} \leq \1_A\right\}.
\end{align}

\begin{definition}[Theta-info.]\label{def: generalized Rains theta infor}
    For any generalized divergence $\bD$, the generalized Theta-information~\footnote{The name follows from the Theta set $\bcV_\Theta$ where $\Theta$ was originally used as the transpose map in the Holevo-Werner bound.} of a quantum channel $\cN_{A'\to B}$ is defined as
    \begin{align}\label{generalized Rains theta information definition eq}
    \bR_{\Theta}(\cN)\equiv \min_{\cM \in \bcV_\Theta}\bD(\cN\|\cM) = \min_{\cM \in \bcV_\Theta} \max_{\rho_A\in \cS(A)} \bD(\cN_{A'\to B} (\phi_{AA'})\|\cM_{A'\to B}(\phi_{AA'})),
\end{align}
where ${\bcV_\Theta}$ is the Theta set in~\eqref{channel's rains set} and $\phi_{AA'}$ is a purification of quantum state $\rho_A$.
\end{definition}

\begin{remark}\label{Rains theta swap}
    On the r.h.s. of Eq.~\eqref{generalized Rains theta information definition eq}, the objective function is concave in $\rho_A$ and convex in $\cM$~\cite[Proposition 8]{wang2019converse}. Thus we can swap the min and max by using Sion's minimax theorem~\cite{Sion1958}.
\end{remark}

The following result compares the generalized Theta-information in~\eqref{generalized Rains theta information definition eq} and the generalized Rains information in~\eqref{generalized Rains information definition eq} presented in the previous section. Interestingly, these two quantities coincide for the max-relative entropy in general.

\begin{proposition}\label{prop: Rains and Theta information}
    For any generalized divergence $\bD$ and any quantum channel $\cN$, it holds    
    \begin{align}\label{eq: rains and theta info}
        \bR(\cN) \leq \bR_\Theta(\cN).
    \end{align}
  Moreover, for the max-relative entropy the equality always holds, i.e,
  \begin{align}\label{eq: max rains and theta info}
    R_{\max}(\cN) = R_{\max,\Theta}(\cN).
  \end{align}
\end{proposition}
\begin{proof}
   A detailed proof is given in Section~\ref{sec: quantum capacity detailed proofs}.
\end{proof}

\vspace{0.2cm}
We proceed to consider the geometric \Renyi divergence and show its amortization property, a key ingredient to proving the strong converse bound on the assisted quantum capacity in Theorem~\ref{thm: main result quantum assisted}.  

Suppose Alice and Bob share a quantum state $\rho_{A'AB'}$ with the system cut $A'A:B'$. Their shared entanglement with respect to the measure $\widehat R_\a$ is given by $\widehat R_\a(\rho_{A'A:B'})$. If Alice redistributes part of her system $A$ through the channel $\cN_{A\to B}$ and Bob receives the output system $B$, then their shared state becomes to $\o_{A':BB'} = \cN_{A\to B}(\rho_{A'A:B'})$ with the shared entanglement evaluated as $\widehat R_\a(\o_{A':BB'})$. The amortization inequality shows that the amount of entanglement change after the state redistribution is upper bounded by the channel's information measure $\widehat R_{\a,\Theta}(\cN)$. 

\begin{proposition}[Amortization]\label{amortization proposition}
    For any quantum state $\rho_{A'AB'}$, any quantum channel $\cN_{A\to B}$ and the parameter $\a \in (1,2]$, it holds
    \begin{align}
    \widehat R_\a(\omega_{A':BB'}) \leq \widehat R_\a(\rho_{A'A:B'}) + \widehat R_{\a,\Theta}(\cN_{A\to B})\quad \text{with} \quad \omega_{A':BB'} = \cN_{A\to B}(\rho_{A'A:B'}).
\end{align}
\end{proposition}
\begin{proof}
  This is a direct consequence of the chain rule property of the geometric \Renyi divergence in Lemma~\ref{lem_chainRule}.  A detailed proof is given in Section~\ref{sec: quantum capacity detailed proofs}. 
\end{proof}

\begin{figure}[H]
\centering
\begin{tikzpicture}
\draw[very thick, colorone] (-5.6,0.7) -- (-3.9,0.7) node[black,midway,shift={(-0.5,0.2)}] {\scriptsize $A_1'$};
\draw[very thick, colorone] (-3,0.7) -- (-1.3,0.7)node[black,midway,shift={(-0.5,0.2)}] {\scriptsize $A_2'$};
\draw[very thick, colorone] (1.6,0.7) -- (3.3,0.7)node[black,midway,shift={(-0.5,0.2)}] {\scriptsize $A_n'$};

\draw[very thick, colorthree] (-5.6,-0.7) -- (-3.9,-0.7) node[black,midway,shift={(-0.5,0.2)}] {\scriptsize $B_1'$};
\draw[very thick, colorthree] (-3,-0.7) -- (-1.3,-0.7) node[black,midway,shift={(-0.5,0.2)}] {\scriptsize $B_2'$};
\draw[very thick, colorthree] (1.6,-0.7) -- (3.3,-0.7) node[black,midway,shift={(-0.5,0.2)}] {\scriptsize $B_n'$};

\draw[very thick, colorone] (-5.6,0) -- (-5,0)  node[black,midway,shift={(0.05,0.2)}] {\scriptsize $A_1$};
\draw[very thick, colorone] (-3,0) -- (-2.4,0) node[black,midway,shift={(0.05,0.2)}] {\scriptsize $A_2$};
\draw[very thick, colorone] (1.6,0) -- (2.2,0) node[black,midway,shift={(0.05,0.2)}] {\scriptsize $A_n$};

\draw[very thick, colorthree] (-4.5,0) -- (-3.9,0) node[black,midway,shift={(0.05,0.2)}] {\scriptsize $B_1$};
\draw[very thick, colorthree] (-1.9,0) -- (-1.3,0) node[black,midway,shift={(0.05,0.2)}] {\scriptsize $B_2$};
\draw[very thick, colorthree] (2.7,0) -- (3.3,0) node[black,midway,shift={(0.05,0.2)}] {\scriptsize $B_n$};

\draw[very thick, colorone] (4.2,0.7)  -- node[black,midway,shift={(0,0.2)}] {\scriptsize $M_A$} (4.9,0.7) -- (5.5,0);
\draw[very thick, colorthree] (4.2,-0.7) -- node[black,midway,shift={(0,0.2)}] {\scriptsize $M_B$} (4.9,-0.7) -- (5.5,0);

\node at (6.2,0) {\scriptsize $\omega_{M_AM_B}$};

\draw[thick,gray,dashed] (-5.48,1.2) -- (-5.48,-1.3) node[black,below] {\scriptsize $\rho^{(1)}_{A_1'A_1B_1'}$};
\draw[thick,gray,dashed] (-2.9,1.2) -- (-2.9,-1.3) node[black,below] {\scriptsize $\rho^{(2)}_{A_2'A_2B_2'}$};
\draw[thick,gray,dashed] (1.7,1.2) -- (1.7,-1.3) node[black,below] {\scriptsize $\rho^{(n)}_{A_n'A_nB_n'}$};

\draw[thick,gray,dashed] (-4.4,1.3) node[black,above] {\scriptsize $\sigma^{(1)}_{A_1'B_1B_1'}$} -- (-4.4,-1.2);
\draw[thick,gray,dashed] (-1.8,1.3) node[black,above] {\scriptsize $\sigma^{(2)}_{A_2'B_2B_2'}$} -- (-1.8,-1.2);
\draw[thick,gray,dashed] (2.8,1.3) node[black,above] {\scriptsize $\sigma^{(n)}_{A_n'B_nB_n'}$} -- (2.8,-1.2);

\draw[very thick, rounded corners = 1] (-6.5,1.1) rectangle (-5.6,-1.1) node[midway] {\large $\cO$};
\draw[very thick, rounded corners = 1] (-3.9,1.1) rectangle (-3,-1.1) node[midway] {\large $\cO$};
\draw[very thick, rounded corners = 1] (-1.3,1.1) rectangle (-0.4,-1.1) node[midway] {\large $\cO$};
\draw[very thick, rounded corners = 1] (0.7,1.1) rectangle (1.6,-1.1) node[midway] {\large $\cO$};
\draw[very thick, rounded corners = 1] (3.3,1.1) rectangle (4.2,-1.1) node[midway] {\large $\cO$};

\draw[very thick, rounded corners = 1] (-5,0.3) rectangle (-4.5,-0.4) node[midway] {$\cN$};
\draw[very thick, rounded corners = 1] (-2.4,0.3) rectangle (-1.9,-0.4) node[midway] {$\cN$};
\draw[very thick, rounded corners = 1] (2.2,0.3) rectangle (2.7,-0.4) node[midway] {$\cN$};

\node at (-0.1,0) [circle,fill,inner sep=1pt]{};
\node at (0.1,0) [circle,fill,inner sep=1pt]{};
\node at (0.3,0) [circle,fill,inner sep=1pt]{};

\node at (-0.1,0.5) [circle,fill,inner sep=1pt]{};
\node at (0.1,0.5) [circle,fill,inner sep=1pt]{};
\node at (0.3,0.5) [circle,fill,inner sep=1pt]{};

\node at (-0.1,-0.5) [circle,fill,inner sep=1pt]{};
\node at (0.1,-0.5) [circle,fill,inner sep=1pt]{};
\node at (0.3,-0.5) [circle,fill,inner sep=1pt]{};

\end{tikzpicture}
\caption{\small A schematic diagram for the protocol of $\cO$-assisted quantum communication that uses a quantum channel $n$ times, where $\cO$ is usually chosen as $\LOCC$ or $\PPT$. Every channel use is interleaved by an operation in the class $\cO$. The goal of such a protocol is to produce an approximate maximally entangled state $\o_{M_A M_B}$ between Alice and Bob. The systems in red are held by Alice while the systems in blue are held by Bob.}
\label{fig: twoway assisted quantum capacity}
\end{figure}
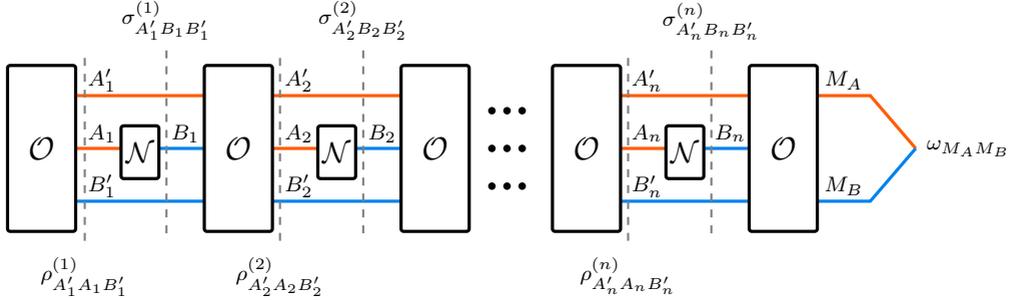

\begin{theorem}[Main result 2]\label{thm: main result quantum assisted}
    For any quantum channel $\cN$ and $\a \in (1,2]$, it holds
        \begin{align}
            Q^{\PPT,\leftrightarrow}(\cN) \leq  Q^{\PPT,\leftrightarrow,\dagger}(\cN) \leq \widehat R_{\a,\Theta}(\cN) \leq R_{\max}(\cN),
        \end{align}
    where $Q^{\PPT,\leftrightarrow}(\cN)$ and $Q^{\PPT,\leftrightarrow,\dagger}(\cN)$ denote the PPT-assisted quantum capacity of channel $\cN$ and its corresponding strong converse capacity, respectively.
\end{theorem}
\begin{proof}
    The first inequality holds by definition. The last inequality holds since we have $ \widehat R_{\a,\Theta}(\cN) \leq R_{\max,\Theta}(\cN) = R_{\max}(\cN)$ by Lemma~\ref{thm: divergence chain inequality} and Proposition~\ref{prop: Rains and Theta information}, respectively. It remains to prove the second one. Once we have the amortization inequality in Proposition~\ref{amortization proposition}, the proof of the second inequality will closely follow the one in~\cite[Theorem 3]{Berta2017a}. Consider $n$ round PPT-assisted quantum communication protocol illustrated in Figure~\ref{fig: twoway assisted quantum capacity}. For each round, denote the input state of $\cN$ as $\rho_{A'AB'}^{\scriptscriptstyle (i)}$ and the output state as $\sigma_{A'BB'}^{\scriptscriptstyle (i)}$. The final state after $n$ rounds communication is denoted as $\o_{M_AM_B}$. Then we have
    \begin{align}
         \widehat R_\a(\omega_{M_AM_B}) & \leq \widehat R_\a\big(\sigma_{A':BB'}^{\scriptscriptstyle (n)}\big)\\
         & = \widehat R_\a\big(\sigma_{A':BB'}^{\scriptscriptstyle (n)}\big) - \widehat R_{\a}\big(\rho^{\scriptscriptstyle (1)}_{A'A:B'}\big)\\
         & \leq \widehat R_\a\big(\sigma_{A':BB'}^{\scriptscriptstyle (n)}\big) + \sum\nolimits_{i=1}^{n-1} \left[\widehat R_\a\big(\sigma_{A':BB'}^{\scriptscriptstyle (i)}\big) - \widehat R_{\a}\big(\rho^{\scriptscriptstyle (i+1)}_{A'A:B'}\big)\right]- \widehat R_{\a}\big(\rho^{\scriptscriptstyle (1)}_{A'A:B'}\big)\\
         & = \sum\nolimits_{i=1}^n \Big[\widehat R_\a\big(\sigma_{A':BB'}^{\scriptscriptstyle (i)}\big) - \widehat R_{\a}\big(\rho^{\scriptscriptstyle (i)}_{A'A:B'}\big)\Big]\\
         & \leq n \widehat R_{\a,\Theta}(\cN).\label{renyi strong converse tmp}
    \end{align} 
    The first and third lines follow from the monotonicity of the geometric \Renyi Rains bound $\widehat R_\a$ with respect to the PPT operations~\cite[Eq.~(22)]{tomamichel2017strong}. The second line follows since $\rho_{A'A:B'}^{\scriptscriptstyle (1)}$ is a PPT state and thus $\widehat R_{\a}(\rho^{\scriptscriptstyle (1)}_{A'A:B'}) = 0$. The last line follows from Proposition~\ref{amortization proposition}. 

    For any communication protocol with triplet $(n,r,\ve)$, denote $k \equiv 2^{nr}$. This implies
    $\tr \Phi_k \,\omega \geq 1-\ve$ with $\Phi_k$ being the $k$-dimensional maximally entangled state. Moreover, for any $\sigma \in \PPT'$, it holds
    $\tr \Phi_k \sigma \leq 1/k$~\cite{Rains2001}. Without loss of generality, we can assume that $\ve \leq 1-2^{-nr}$. Otherwise, the strong converse would already hold for any rates above the capacity since $1-\ve < 2^{-nr}$. Thus for any $\sigma \in \PPT'$ we have
    \begin{align}
        1- \tr \Phi_k \omega \leq \ve \leq 1- 2^{-nr} \leq 1 - \tr \Phi_k \sigma.
    \end{align}
  Let $\cN(\gamma) = (\tr \Phi_k \gamma) \ket{0}\bra{0} + (\tr (\1-\Phi_k) \gamma) \ket{1}\bra{1}$. Due to the data-processing inequality, we have
    \begin{align}\label{eq: quantum twoway main theorem tmp1}
        \widehat D_\a(\omega\|\sigma) \geq \widehat D_\a(\cN(\omega)\|\cN(\sigma)) = \delta_\a( 1-\tr \Phi_k \omega\|1-\tr \Phi_k \sigma) \geq \delta_\a(\ve\|1- 2^{-nr}),
    \end{align}
    where $\delta_\a(p\|q)\equiv \frac{1}{\a-1} \log \big[p^\a q^{1-\a} + (1-p)^\a (1-q)^{1-\a}\big]$ is the binary classical \Renyi divergence. The last inequality in~\eqref{eq: quantum twoway main theorem tmp1} follows from the monotonicity property that $\delta_\a(p'\|q) \leq \delta_\a(p\|q)$ if $p \leq p' \leq q$ and $\delta_\a(p\|q') \leq \delta_\a(p\|q)$ if $p \leq q' \leq q$~\cite{Polyanskiy2010b}.
    Then we have
    \begin{align}
        \widehat R_\a(\omega) & = \min_{\sigma \in \PPT'} \widehat D_{\a}(\omega\|\sigma)\\
        & \geq \delta_\a(\ve\|1- 2^{-nr})\\
        & \geq \frac{1}{\a-1} \log (1-\ve)^\a(2^{-nr})^{1-\a} \\
        & = \frac{\a}{\a-1} \log (1-\ve) + nr.
        \label{renyi strong converse tmp1}
    \end{align}
    Combining Eqs.~\eqref{renyi strong converse tmp} and~\eqref{renyi strong converse tmp1}, we have
    \begin{align}
        \frac{\a}{\a-1} \log (1-\ve) + nr \leq n \widehat R_{\a,\Theta}(\cN),
    \end{align}
    which is equivalent to 
    \begin{align}
        1-\ve \leq 2^{-n \left(\frac{\a-1}{\a}\right)\left[r - \widehat R_{\a,\Theta}(\cN)\right]}.
    \end{align}
    This implies that if the communication rate $r$ is strictly larger than $\widehat R_{\a,\Theta}(\cN)$, the fidelity of transmission $1-\ve$ decays exponentially fast to zero as the number of channel use $n$ increases. Or equivalently, we have the strong converse inequality $Q^{\PPT,\leftrightarrow,\dagger}(\cN) \leq \widehat R_{\a,\Theta}(\cN) $ and completes the proof.
\end{proof}

\vspace{0.2cm}
Let $Q^{\leftrightarrow}$ and $Q^{\leftrightarrow,\dagger}$ be the two-way assisted quantum capacity and its strong converse capacity respectively. We have the following as a direct consequence of Theorem~\ref{thm: main result quantum assisted}, since PPT assistance is stronger.

\begin{corollary}\label{cor: two-way quantum capacity}
For any quantum channel $\cN$ and $\a\in (1,2]$, it holds
\begin{align}
    Q^{\leftrightarrow}(\cN) \leq  Q^{\leftrightarrow,\dagger}(\cN) \leq \widehat R_{\a,\Theta}(\cN) \leq R_{\max}(\cN).
\end{align}
\end{corollary}

Finally, we present how to compute $\widehat R_{\a,\Theta}(\cN)$ as an SDP.

\begin{proposition}[SDP formula]\label{prop: SDP formula for maximal Rains theta info}
    For any quantum channel $\cN_{A'\to B}$ and $\a(\ell) = 1+2^{-\ell}$ with $\ell \in \mathbb N$, the geometric \Renyi Theta-information can be computed as an SDP:
\begin{gather}
    \widehat R_{\a,\Theta}(\cN)= 2^\ell\cdot \log \min \ y \quad \text{\rm s.t.}\quad \dbhbig{M,\{N_i\}_{i=0}^\ell,R,y},\notag\\[2pt]
     \dbp{\begin{matrix}
        M & J_{\cN}\\
        J_{\cN} & N_{\ell}
    \end{matrix}},
    \left\{\dbp{\begin{matrix}
        J_{\cN} & N_{i} \\
        N_{i} & N_{i-1}
    \end{matrix}}\right\}_{i=1}^\ell, 
    \dbpbigg{R \pm N_0^{\sfT_B}},
    \dbpbigg{\1 - \tr_B R}, \dbpbigg{y \1_A - \tr_B M },\label{eq:  SDP formula for maximal Rains theta info}
\end{gather}
where $J_{\cN}$ is the Choi matrix of $\cN$.
\end{proposition}
\begin{proof}
This directly follows from Lemma~\ref{lem: SDP representation of the channel information measure} and the definition of the Theta set $\bcV_\Theta$ in~\eqref{channel's rains set}.
\end{proof}

\subsection{Extension to bidirectional channels}
\label{sec: Extension to bidirectional channels}

In this section we showcase that the above results for the PPT/two-way assisted quantum capacity can be extended to a more general scenario where Alice and Bob share a bidirectional quantum channel.

A bipartite quantum channel $\cN_{A_1B_1\to A_2B_2}$ is a completely positive trace-preserving map that sending composite system $A_1B_1$ to $A_2B_2$. This channel is called \emph{bidirectional channel} if $A_1A_2$ are held by Alice and $B_1B_2$ are held by Bob. That is, Alice and Bob each input a state to this channel and receive an output~\cite{Bennett2003bidirectional}, as depicited in Figure~\ref{fig: bidirectional channel model}. This is the most general setting for two-party communications and will reduce to the usual point-to-point channel when the dimensions of Bob's input and Alice's output are trivial, i.e., $\dim(\cH_{B_1}) = \dim(\cH_{A_2}) = 1$.

\begin{figure}[H]
\centering
\begin{tikzpicture}
\begin{scope}[shift={(3,-3)}]
\draw[very thick] (-1,0.5) rectangle node[midway] {$\cN$}(1,-0.5);
\draw[very thick,->,colorone] (-2.2,0.2)  -- node[midway,shift={(0,0.25)},black] {$A_1$}(-1,0.2);
\draw[very thick,<-,colorone] (-2.2,-0.2) -- node[midway,shift={(0,-0.25)},black] {$A_2$}(-1,-0.2);
\draw[very thick,<-,colorthree] (1,0.2)  -- node[midway,shift={(0,0.25)},black] {$B_1$}(2.2,0.2);
\draw[very thick,->,colorthree] (1,-0.2) -- node[midway,shift={(0,-0.25)},black] {$B_2$}(2.2,-0.2);
\end{scope}

\end{tikzpicture}

\caption{\small A model of bidirectional quantum channel where $A_1$, $A_2$ are held by Alice and $B_1$, $B_2$ by Bob.}
\label{fig: bidirectional channel model}
\end{figure}
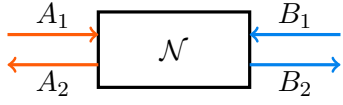

In~\cite{Bauml2018}, the authors introduced the bidirectional version of the max-Rains information as
\begin{gather}\label{eq: bi max Rains info}
  R^{\bi}_{\max}(\cN_{A_1B_1\to A_2 B_2}) \equiv \log \min \|\tr_{A_2B_2} (V_{A_1B_1A_2B_2} + Y_{A_1B_1A_2B_2})\|_{\infty}\quad \text{s.t.}\\[2pt]
  V_{A_1B_1A_2B_2} \geq 0, Y_{A_1B_1A_2B_2} \geq 0, (V_{A_1B_1A_2B_2}-Y_{A_1B_1A_2B_2})^{\sfT_{B_1B_2}} \geq J_{A_1B_1A_2B_2}^{\cN}.\notag
\end{gather}

Let $Q^{\bi,\PPT,\leftrightarrow}$ and $Q^{\bi,\PPT,\leftrightarrow,\dagger}$ be the PPT-assisted quantum capacity of a bidirectional channel and its strong converse capacity  respectively~\footnote{We refer to the work~\cite[Page 2-3]{Bauml2018} for rigorous definitions of the PPT/two-way assisted quantum capacity of a bidirectional channel and its strong converse.}. It was proved in~\cite{Bauml2018} that 
\begin{align}\label{eq: bi max Rains strong converse}
  Q^{\bi,\PPT,\leftrightarrow}(\cN_{A_1B_1\to A_2B_2})\leq Q^{\bi,\PPT,\leftrightarrow,\dagger}(\cN_{A_1B_1\to A_2B_2}) \leq R^\bi_{\max}(\cN_{A_1B_1\to A_2B_2}).
\end{align}
Following a similar approach in Section~\ref{sec: Maximal Renyi Theta-information}, we can further strengthen this bound by exploiting the geometric \Renyi divergence.

We start with a bidirectional version of the Werner-Holevo bound~\footnote{Note that this quantity was also independently introduced in~\cite{Gour2019} as well as in~\cite{Bauml2019} when studying the resource theory of bidirectional quantum channels.}
\begin{align}
  Q^\bi_{\Theta}(\cN_{A_1B_1\to A_2B_2}) \equiv \log \|\Theta_{B_2}\circ \cN_{A_1B_1\to A_2B_2}\circ \Theta_{B_1}\|_\di,
\end{align}
and define its zero set $\bcV^\bi_{\Theta}$ which admits a semidefinite representation as
\begin{align}\label{eq: bi theta set}
  \bcV^\bi_{\Theta} = \Big\{\cM\in \CP(A_1B_1:A_2B_2)\,\Big|\, \exists R_{A_1B_1A_2B_2}, \ \text{s.t.}\ R \pm J_{\cM}^{\sfT_{B_1B_2}} \geq 0,\, \tr_{A_2B_2} R \leq \1_{A_1B_1}\Big\}.
\end{align}
Using the same idea as the point-to-point scenario, we defined the generalized Theta-information of a bidirectional channel $\cN_{A_1B_1\to A_2 B_2}$ as the ``channel distance''
  \begin{align}
  \bR^\bi_\Theta(\cN_{A_1B_1\to A_2B_2})\equiv \min_{\cM \in \bcV^\bi_{\Theta}} \bD(\cN\|\cM),
\end{align}
where $\bD$ is a generalized divergence and the channel divergence follows from the usual definition
\begin{align}
  \bD(\cN\|\cM) \equiv \max_{\phi_{A_1B_1A_3B_3}}\bD(\cN_{A_1B_1\to A_2B_2}(\phi_{A_1B_1A_3B_3})\|\cM_{A_1B_1\to A_2B_2}(\phi_{A_1B_1A_3B_3}))
\end{align}
by maximizing over all the pure states $\phi_{A_1B_1A_3B_3}$.

Following a similar proof of Proposition~\ref{prop: Rains and Theta information}, we can show that the bidirectional max-Rains information defined in~\eqref{eq: bi max Rains info} coincides with the bidirectional Theta-information induced by the max-relative entropy~\footnote{Note that this relation was independently found in the recent work~\cite{Bauml2019} where the authors used this result to simplify a proof in~\cite{Bauml2018} as stated in Eq.~\eqref{eq: bi max Rains strong converse}}. That is, 
\begin{align}
  R^\bi_{\max}(\cN_{A_1B_1\to A_2B_2}) = R^\bi_{\max,\Theta}(\cN_{A_1B_1\to A_2B_2}).
\end{align}

Denote the bidirectional Rains bound as 
$\widehat R_{\a}^{\bi}(\rho) \equiv \min_{\sigma \geq 0, \|\sigma^{\sfT_{B_1B_2}}\|_1 \leq 1} \widehat D_{\alpha}(\rho\|\sigma)$. A similar proof as Proposition~\ref{amortization proposition} will give us the following amortization inequality.

\begin{proposition}[Amortization]\label{prop: amortization bidirectional}
For any quantum state $\rho_{A_1A_3:B_1B_3}$, any bidirectional quantum channel $\cN_{A_1B_1\to A_2B_2}$ and $\a \in (1,2]$, it holds
\begin{gather}
  \widehat R_{\a}^{\bi}(\omega_{A_2A_3:B_2B_3}) \leq \widehat R_{\a}^{\bi}(\rho_{A_1A_3:B_1B_3}) + \widehat R_{\a,\Theta}^{\bi}(\cN_{A_1B_1\to A_2B_2}),  
\end{gather}  
with the output state $\omega_{A_2A_3:B_2B_3} = \cN_{A_1B_1\to A_2B_2}(\rho_{A_1A_3:B_1B_3})$.
\end{proposition}

\vspace{0.2cm}
Using the amortization inequality in Proposition~\ref{prop: amortization bidirectional} and a standard argument as Theorem~\ref{thm: main result quantum assisted}, we have the analog results of Theorem~\ref{thm: main result quantum assisted} and Corollary~\ref{cor: two-way quantum capacity} for bidirectional channels as follows:
\begin{theorem}[Main result 3]\label{thm: main result quantum assisted bidirectional}
  For any bidirectional channel $\cN_{A_1B_1\to A_2B_2}$ and $\a \in (1,2]$, it holds
  \begin{align}
  Q^{\bi,\PPT,\leftrightarrow}(\cN) \leq Q^{\bi,\PPT,\leftrightarrow,\dagger}(\cN) \leq \widehat R_{\a,\Theta}^{\bi}(\cN) \leq  R_{\max}^{\bi}(\cN),
\end{align}
where $Q^{\bi,\PPT,\leftrightarrow}(\cN)$ and $Q^{\bi,\PPT,\leftrightarrow,\dagger}(\cN)$ denote the PPT-assisted quantum capacity of a bidirectional channel $\cN$ and its corresponding strong converse capacity, respectively.
As a consequence, it holds
  \begin{align}
  Q^{\bi,\leftrightarrow}(\cN) \leq Q^{\bi,\leftrightarrow,\dagger}(\cN) \leq \widehat R_{\a,\Theta}^{\bi}(\cN) \leq  R_{\max}^{\bi}(\cN),
\end{align}
where $Q^{\bi,\leftrightarrow}(\cN)$ and $Q^{\bi,\leftrightarrow,\dagger}(\cN)$ denote the two-way assisted quantum capacity of a bidirectional channel $\cN$ and its corresponding strong converse capacity, respectively.
\end{theorem}

\begin{proposition}[SDP formula]
  For any bidirectional channel $\cN_{A_1B_1\to A_2B_2}$ and $\a(\ell) = 1+2^{-\ell}$ with $\ell \in \NN$, the bidirectional geometric \Renyi Theta-information can be computed as an SDP:
\begin{gather}
  \widehat R_{\a,\Theta}^{\bi}(\cN)= 2^\ell\cdot \log \min \ y \quad \text{\rm s.t.}\quad \dbhbig{M,\{N_i\}_{i=0}^\ell,R,y}, \label{eq:  SDP formula for bidirectional maximal Rains theta info}\\[2pt]
   \dbp{\begin{matrix}
    M & J_{\cN}\\
    J_{\cN} & N_{\ell}
  \end{matrix}},
  \left\{\dbp{\begin{matrix}
    J_{\cN} & N_{i} \\
    N_{i} & N_{i-1}
  \end{matrix}}\right\}_{i=1}^\ell, 
  \dbpbigg{R \pm N_0^{\sfT_{B_1B_2}}},
  \dbpbigg{\1 - \tr_{A_2B_2} R}, \dbpbigg{y\1 - \tr_{A_2B_2} M },\notag
\end{gather}
where $J_{\cN}$ is the Choi matrix of $\cN$.
\end{proposition}
\begin{proof}
  This directly follows from Lemma~\ref{lem: SDP representation of the channel information measure} and the definition of $\bcV^\bi_{\Theta}$ in~\eqref{eq: bi theta set}.
\end{proof}

\subsection{Examples}
\label{sec: quantum capacity Examples}

In this section, we investigate several fundamental quantum channels as well as their compositions. We use these toy models to test the performance of our new strong converse bounds, demonstrating the improvement on the previous results. 
The semidefinite programs are implemented in MATLAB via the CVX package, by the solver ``Mosek'' with the best precision.~\footnote{All the data and codes can be found on the GitHub page https://github.com/fangkunfred.}

\subsubsection*{Fundamental quantum channels}

The \emph{quantum depolarizing channel} with dimension $d$ is defined as
\begin{align}\label{DP channel definition}
    \cD_p(\rho) = (1-p) \rho + p \1/d, \quad p \in [0,1].
\end{align} 
The \emph{quantum erasure channel} is defined as
\begin{align}\label{ER channel definition}
    \cE_p(\rho) = (1-p) \rho + p \ket{e}\bra{e}, \quad p \in [0,1],
\end{align}
where $\ket{e}$ is an erasure state orthogonal to the input Hilbert space.
The \emph{quantum dephasing channel} is defined as 
\begin{align}
    \cZ_p(\rho) = \left(1-p\right) \rho + p Z \rho Z, \quad p\in[0,1],
\end{align}
where $Z = \ket{0}\bra{0} - \ket{1}\bra{1}$ is the Pauli-$z$ operator. These three classes of channels are covariant with respect to the whole unitary group.
The \emph{generalized amplitude damping} (GAD) channel is defined as
\begin{align}\label{GAD definition}
    \cA_{\gamma,N} (\rho) = \sum_{i=1}^4 A_i \rho A_i^\dagger, \quad \gamma, N\in [0,1]
\end{align}
with the Kraus operators
\begin{alignat}{2}
    & A_1  = \sqrt{1-N} (\ket{0}\bra{0} + \sqrt{1-\gamma}\ket{1}\bra{1}), \quad && A_2 = \sqrt{\gamma(1-N)}\ket{0}\bra{1},\\
& A_3 = \sqrt{N} (\sqrt{1-\gamma}\ket{0}\bra{0}+\ket{1}\bra{1}), && A_4 = \sqrt{\gamma N}\ket{1}\bra{0}.
\end{alignat}
The GAD channel is one of the realistic sources of noise in superconducting-circuit-based quantum computing~\cite{Chirolli2008}, which can viewed as the qubit analogue of the bosonic thermal channel.
When $N = 0$, it reduces to the conventional \emph{amplitude damping channel} with two Kraus operators $A_1$, $A_2$.

\subsubsection*{Comparison for the unassisted quantum capacity}

For the unassisted quantum capacity, we compare the qubit depolarizing channel $\cD_p$, the qubit erasure channel $\cE_p$, the qubit dephasing channel $\cZ_p$ and the generalized amplitude damping channels $\cA_{p,N}$ with different choices of parameter $N$.

Since $\cD_p$, $\cE_p$ and $\cZ_p$ are covariant with respect to the unitary group,  the optimal input state $\rho_A$ of their Rains information is taken at the maximally mixed state~\cite[Proposition 2]{tomamichel2017strong}. Therefore, their Rains information can be computed via the algorithm in~\cite{fawzi2018efficient,Fawzi2017}. Moreover, for any parameters $\gamma,N \in [0,1]$, the GAD channel $\cA_{\gamma,N}$ is covariant with respect to the Pauli-$z$ operator $Z$. That is, $\cA_{\gamma,N}(Z\rho Z) = Z \cA_{\gamma,N}(\rho )Z$ for all quantum state $\rho$. To compute its Rains information, it suffices to perform the maximization over input states with respect to the one-parameter family of states $\rho_A = (1-p)\ket{0}\bra{0} + p \ket{1}\bra{1}$~\cite{Khatri2019}. This can be handled, for example, by MATLAB function ``fminbnd''.

The comparison results are shown in Figure~\ref{uassisted quantum capacity compare 1}. It is clear that the geometric \Renyi Rains information $\widehat R_{\a {\scriptscriptstyle (10)}}$ coincide with the Rains information $R$ for all these channels except for the particular case $\cA_{p,0}$ in subfigure (d). For all cases, $\widehat R_{\a {\scriptscriptstyle (10)}}$ sets a big difference from the max-Rains information $R_{\max}$. 

\begin{figure}[H]
\centering
\begin{adjustwidth}{-0.5cm}{0cm}
\begin{tikzpicture}
\node at (-5.7,0) {\includegraphics[width = 5.3cm]{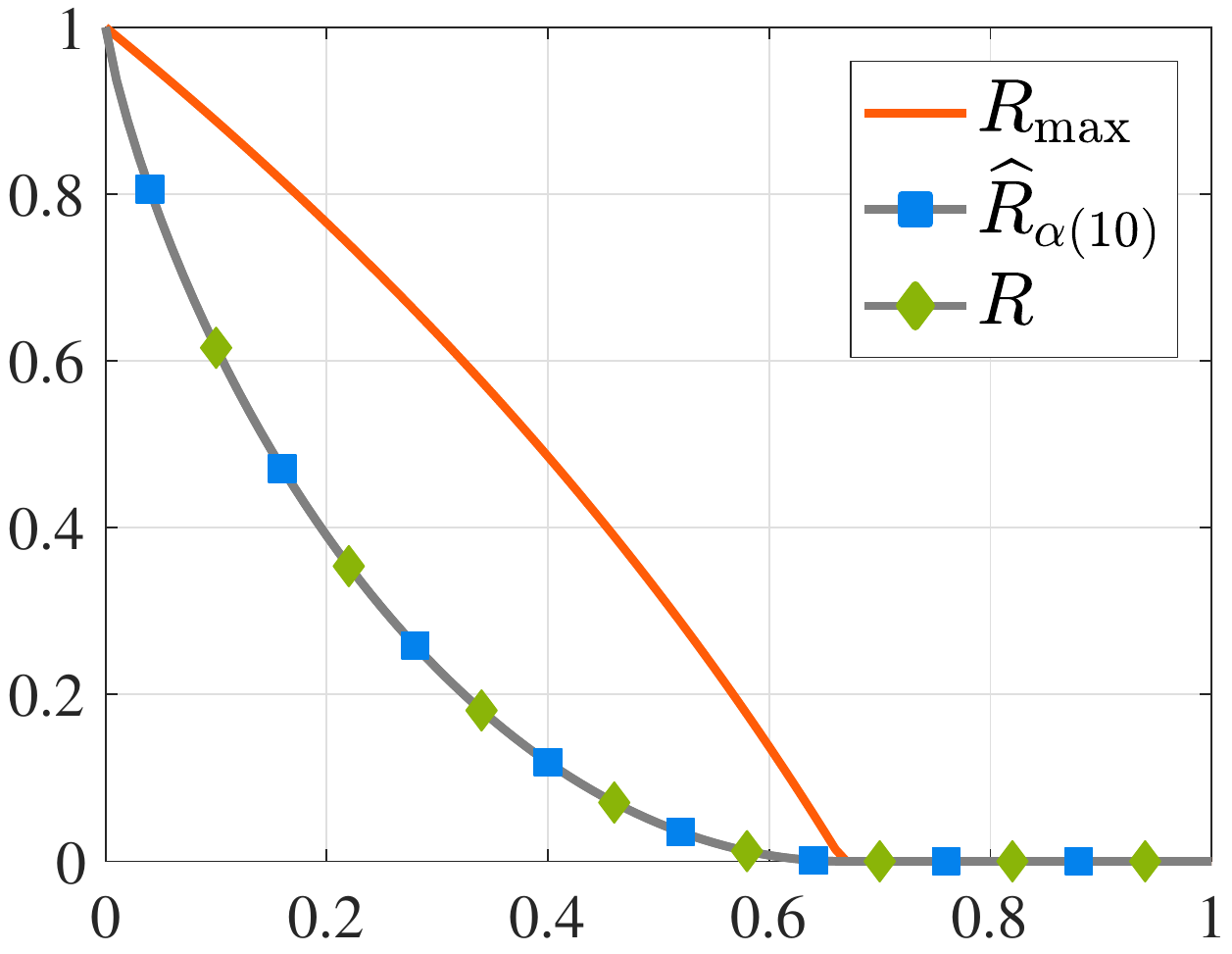}};
\node at (-5.5,-2.4) {\small (a) Qubit depolarizing channel $\cD_p.$};

\node at (0,0) {\includegraphics[width = 5.3cm]{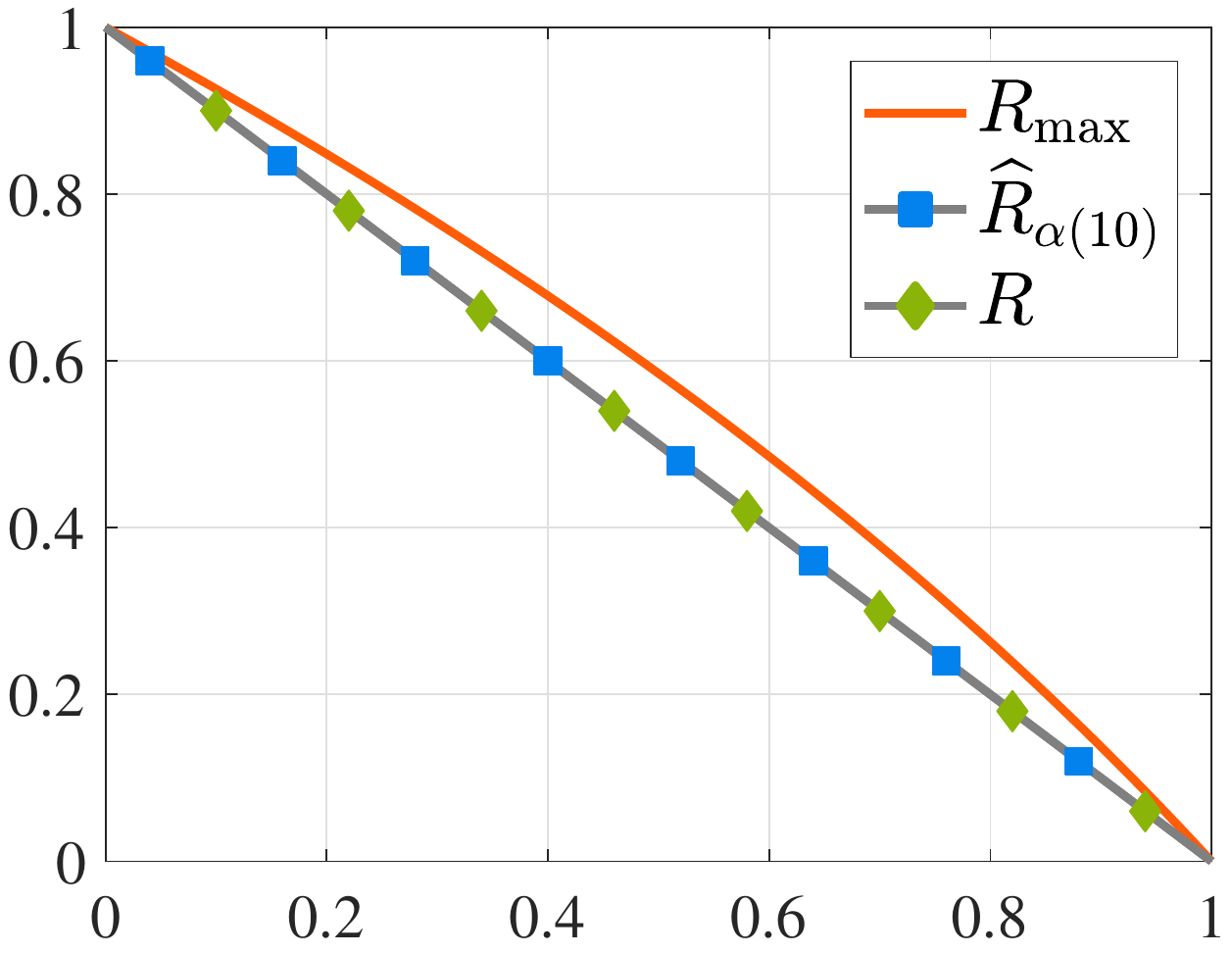}};
\node at (0.2,-2.4) {\small (b) Qubit erasure channel $\cE_p.$};

\node at (5.7,0) {\includegraphics[width = 5.3cm]{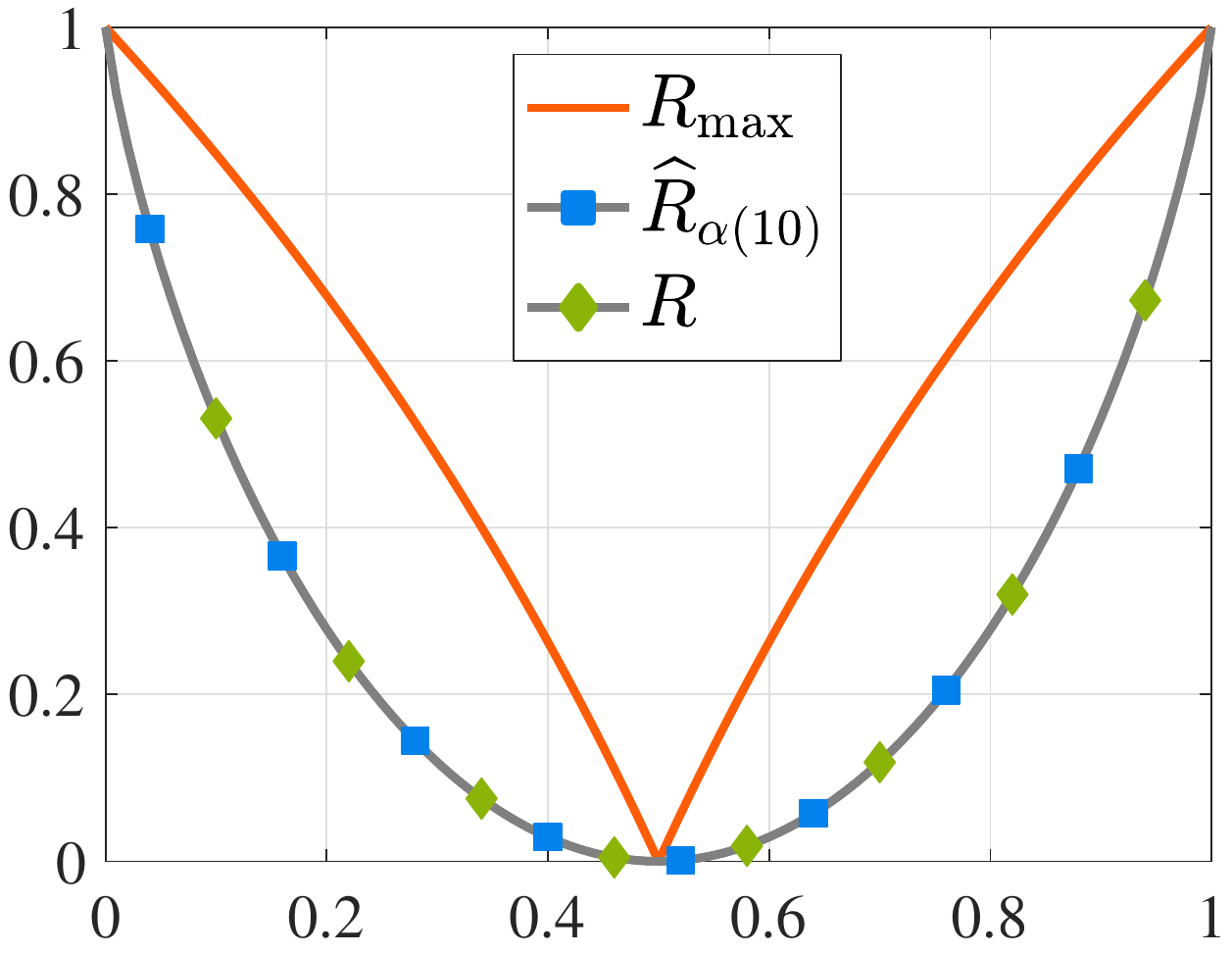}};
\node at (5.9,-2.4) {\small (c) Qubit dephasing channel $\cZ_p.$};

\begin{scope}[shift={(0,-5)}]
\node at (-5.7,0) {\includegraphics[width = 5.3cm]{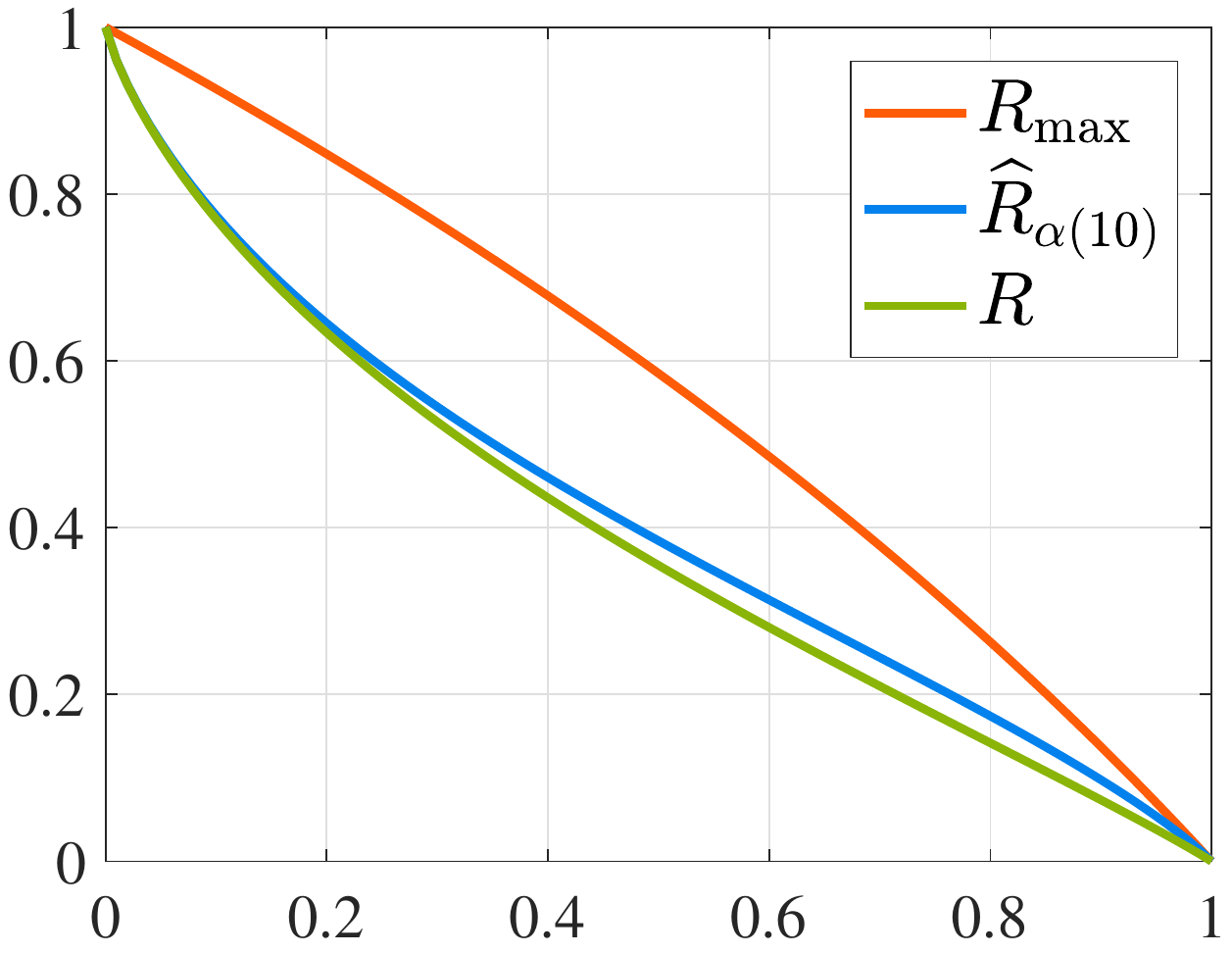}};
\node at (-5.5,-2.4) {\small (d) GAD channel with $N = 0$.};

\node at (0,0) {\includegraphics[width = 5.3cm]{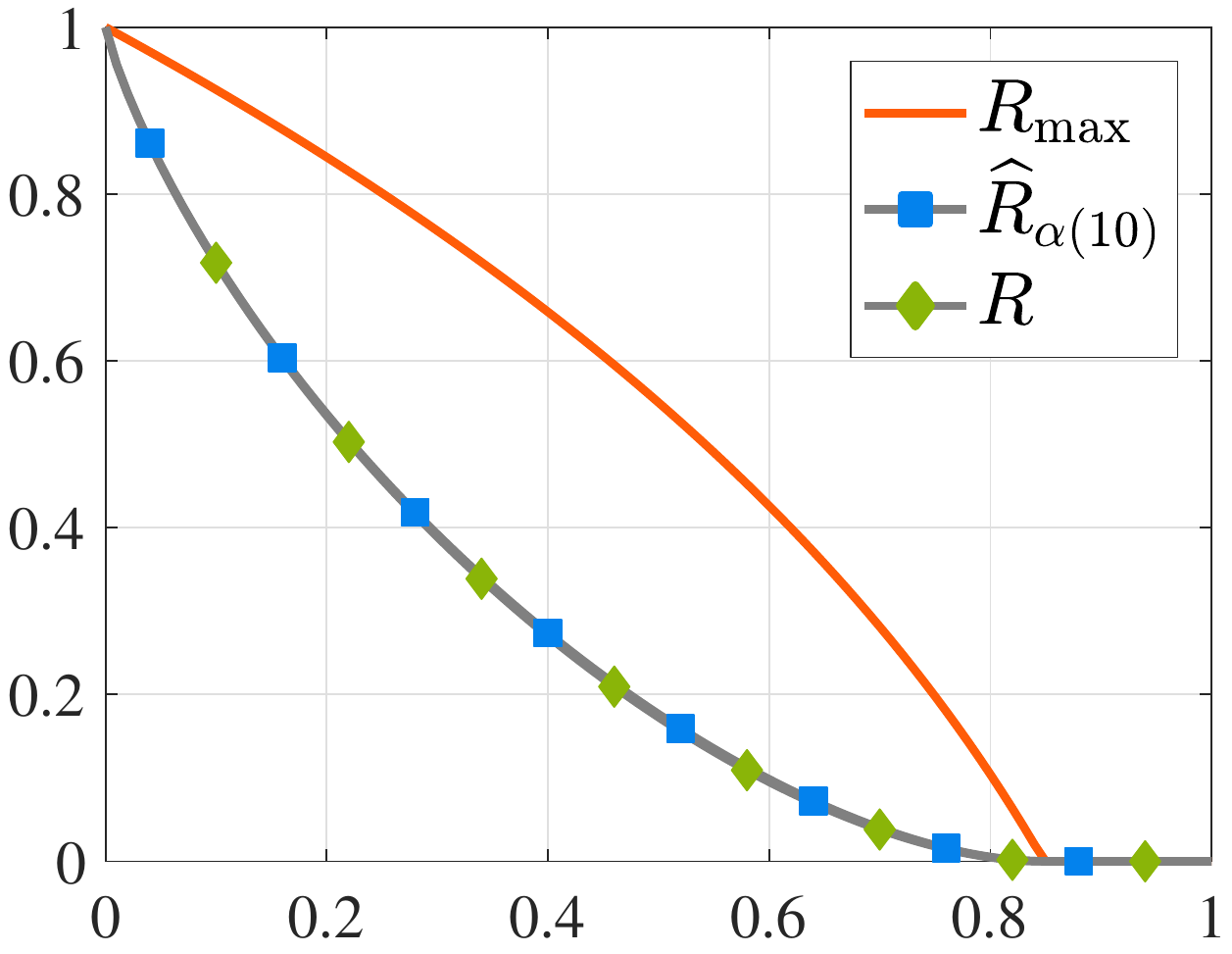}};
\node at (0.2,-2.4) {\small (e) GAD channel with $N = 0.3$.};

\node at (5.7,0) {\includegraphics[width = 5.3cm]{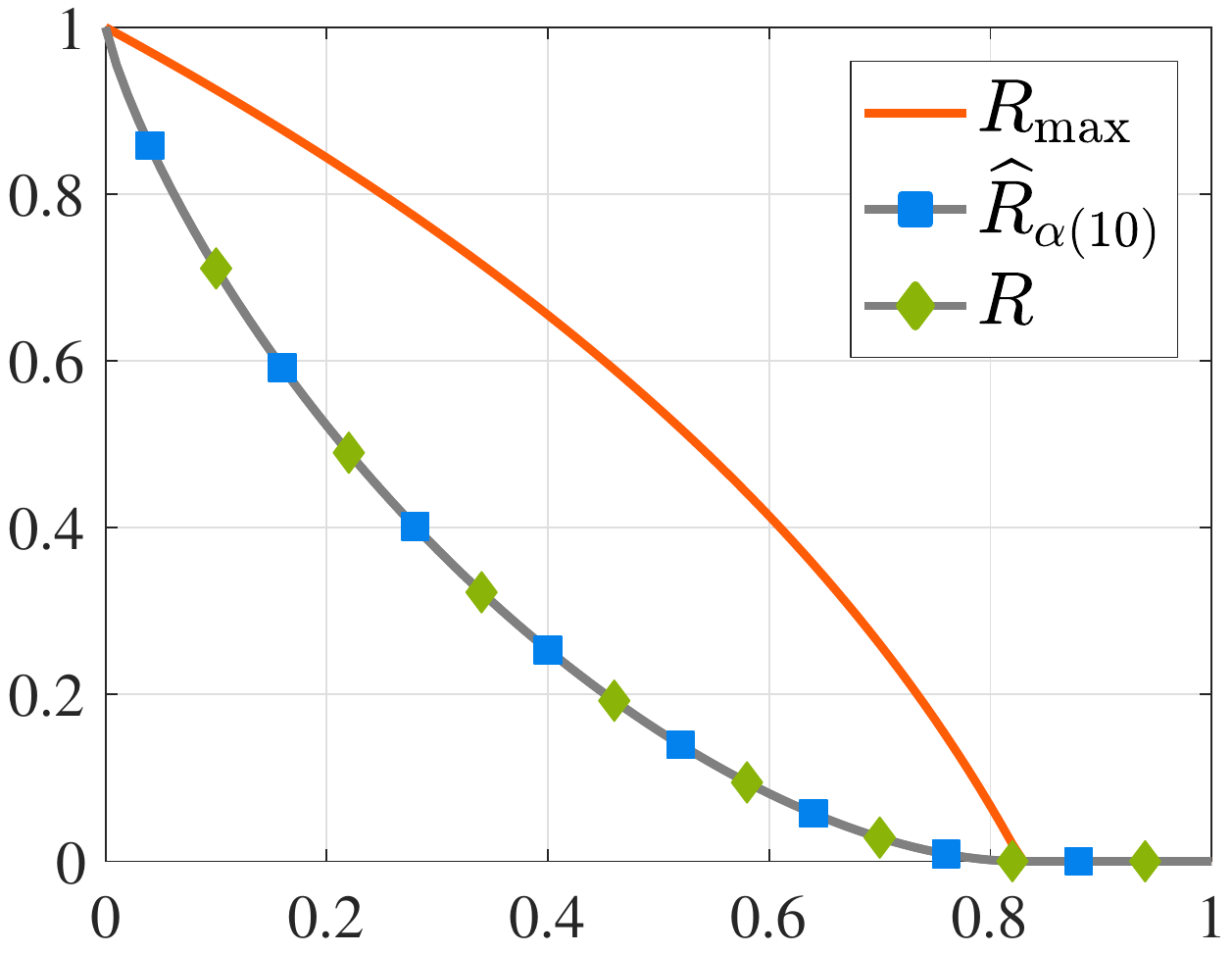}};
\node at (5.9,-2.4) {\small (f) GAD channel with $N = 0.5$.};
\end{scope}

\end{tikzpicture}

\end{adjustwidth}
\caption{\small Comparison of the strong converse bounds on the unassisted quantum capacity of the qubit depolaring channel $\cD_p$, the qubit erasure channel $\cE_p$, the qubit dephasing channel $\cZ_p$ and the generalized amplitude damping channels $\cA_{p,N}$. The horizontal axis takes value of $p \in [0,1]$. }
\label{uassisted quantum capacity compare 1}
\end{figure}
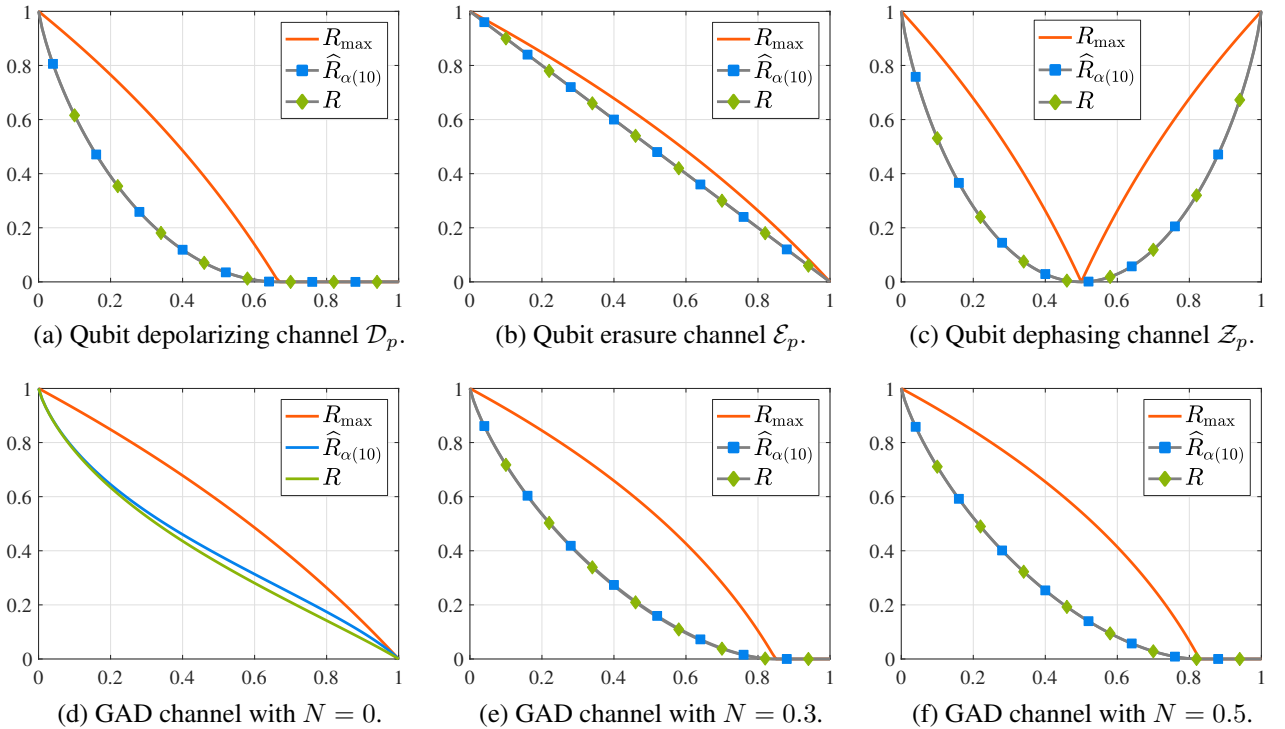

\subsubsection*{Comparison for the two-way assisted quantum capacity}

For the two-way assisted quantum capacity, we consider the channels $\cD_p$, $\cE_p$ and $\cZ_p$ composed with the amplitude damping channel $\cA_{p,0}$, and the generalized amplitude damping channel $\cA_{p,N}$ with different choices of parameter $N$. Note that because these channels are not sufficiently covariant, their Rains information are not known as valid converse bounds on the two-way assisted quantum capacity. 

The comparison result~\footnote{A detailed comparison of the GAD channels with other weak converse bounds in~\cite{Khatri2019} is given in Appendix~\ref{app: Detailed comparison for generalized amplitude damping channel}.} for the two-way assisted quantum capacity is given in Figure~\ref{two-way assisted quantum capacity compare 1}. The geometric \Renyi Theta-information $\widehat R_{\a {\scriptscriptstyle(10)},\Theta}$ demonstrates a significant improvement over the max-Rains information $R_{\max}$ for all these channels except for one particular case $\cA_{p,0}$ in subfigure (d).

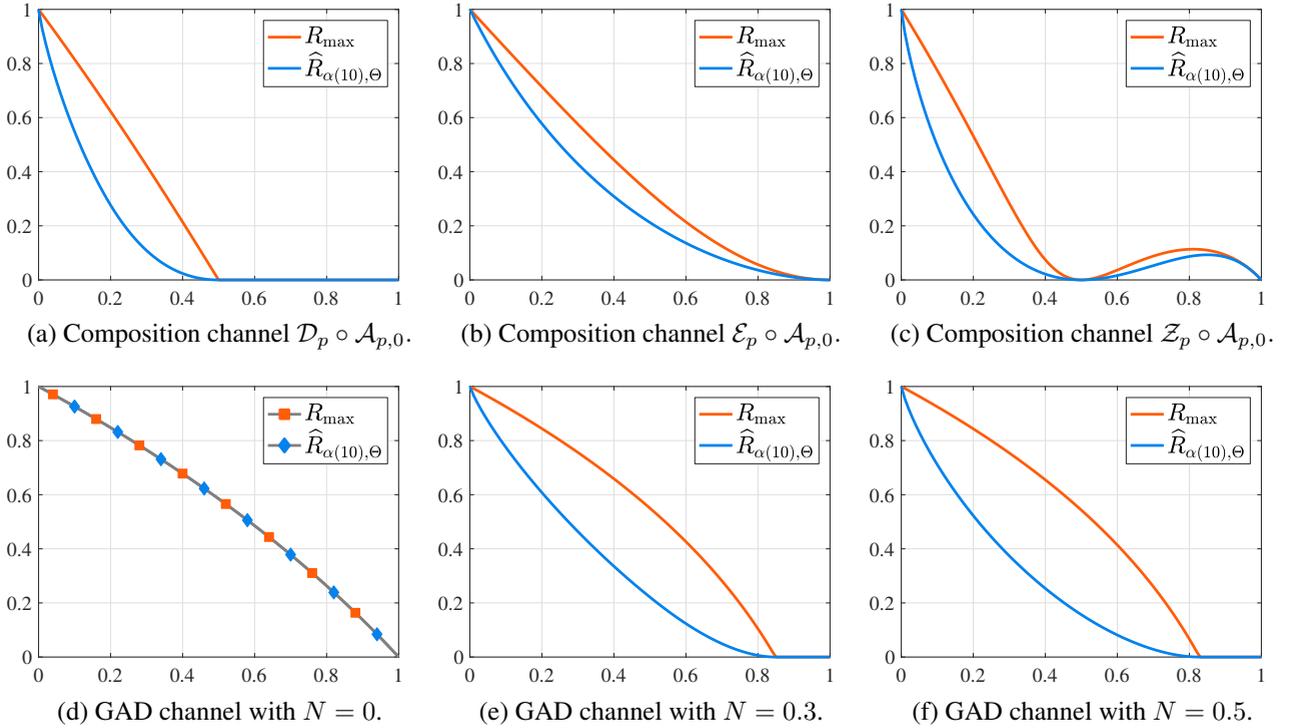
\begin{figure}[H]
\centering
\begin{adjustwidth}{-0.5cm}{0cm}
\begin{tikzpicture}
\node at (-5.7,0) {\includegraphics[width = 5.3cm]{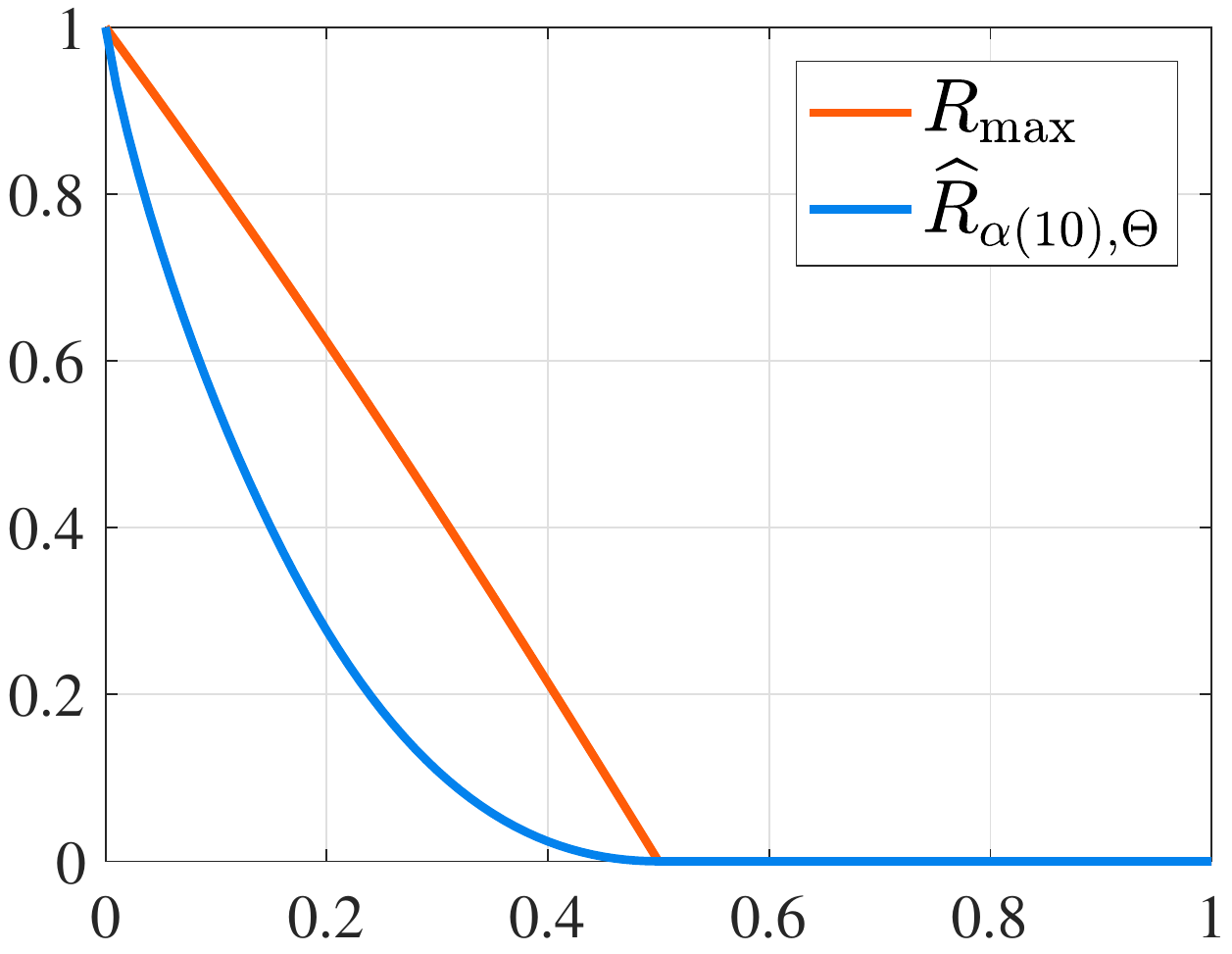}};
\node at (-5.5,-2.4) {\small (a) Composition channel $\cD_p\circ \cA_{p,0}.$};

\node at (0,0) {\includegraphics[width = 5.3cm]{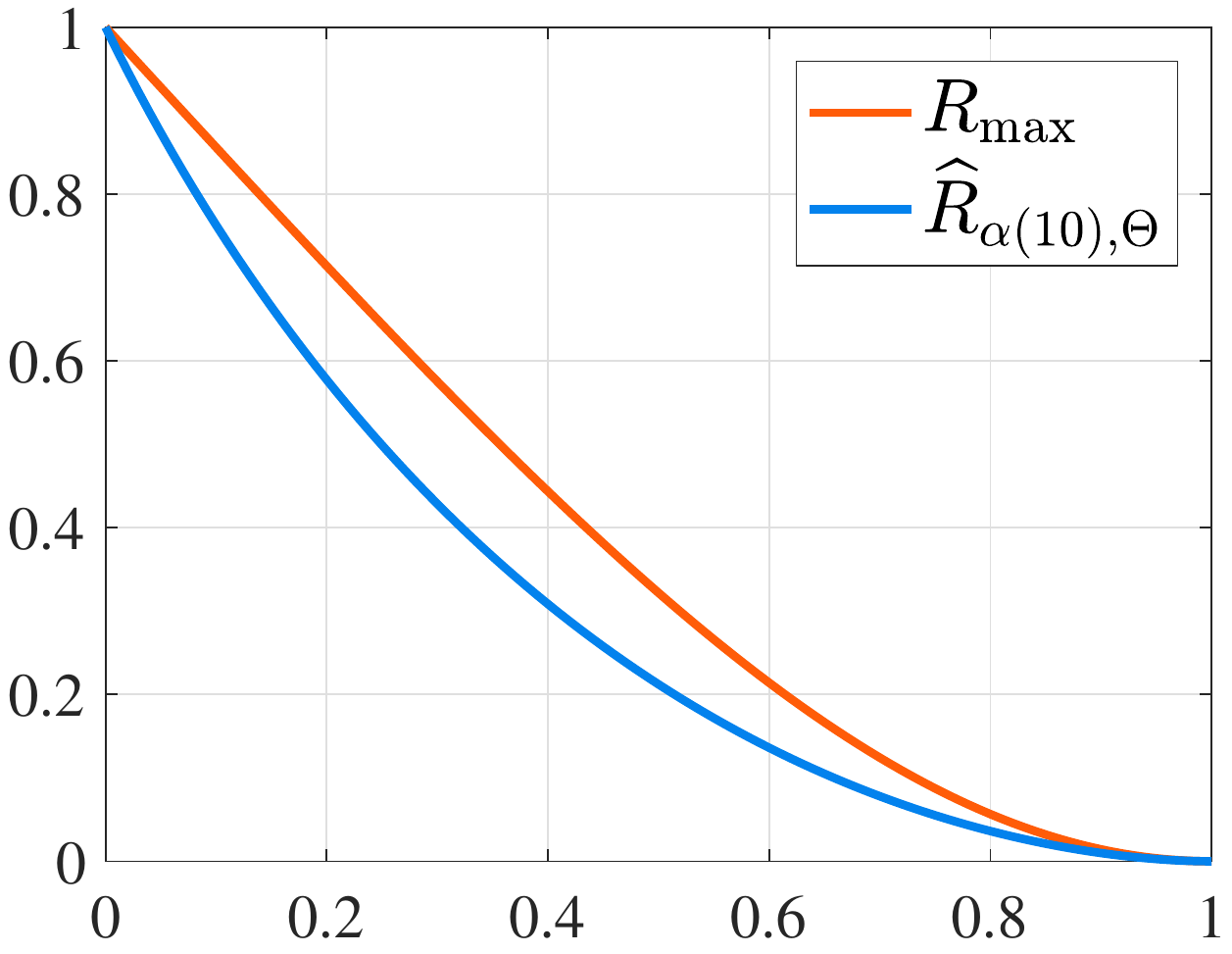}};
\node at (0.2,-2.4) {\small (b) Composition channel $\cE_p\circ \cA_{p,0}.$};

\node at (5.7,0) {\includegraphics[width = 5.3cm]{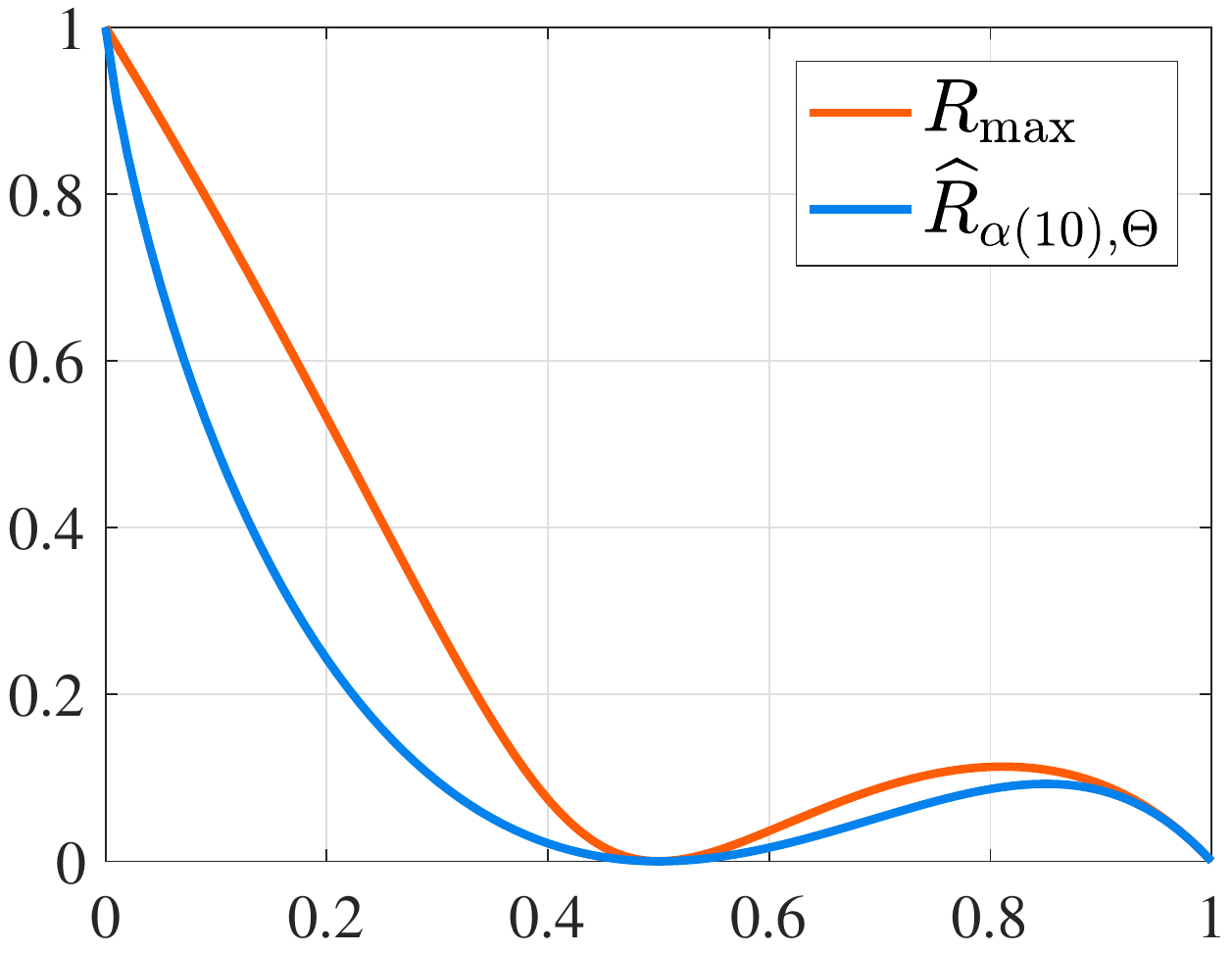}};
\node at (5.9,-2.4) {\small (c) Composition channel $\cZ_p\circ \cA_{p,0}.$};

\begin{scope}[shift={(0,-5)}]
\node at (-5.7,0) {\includegraphics[width = 5.3cm]{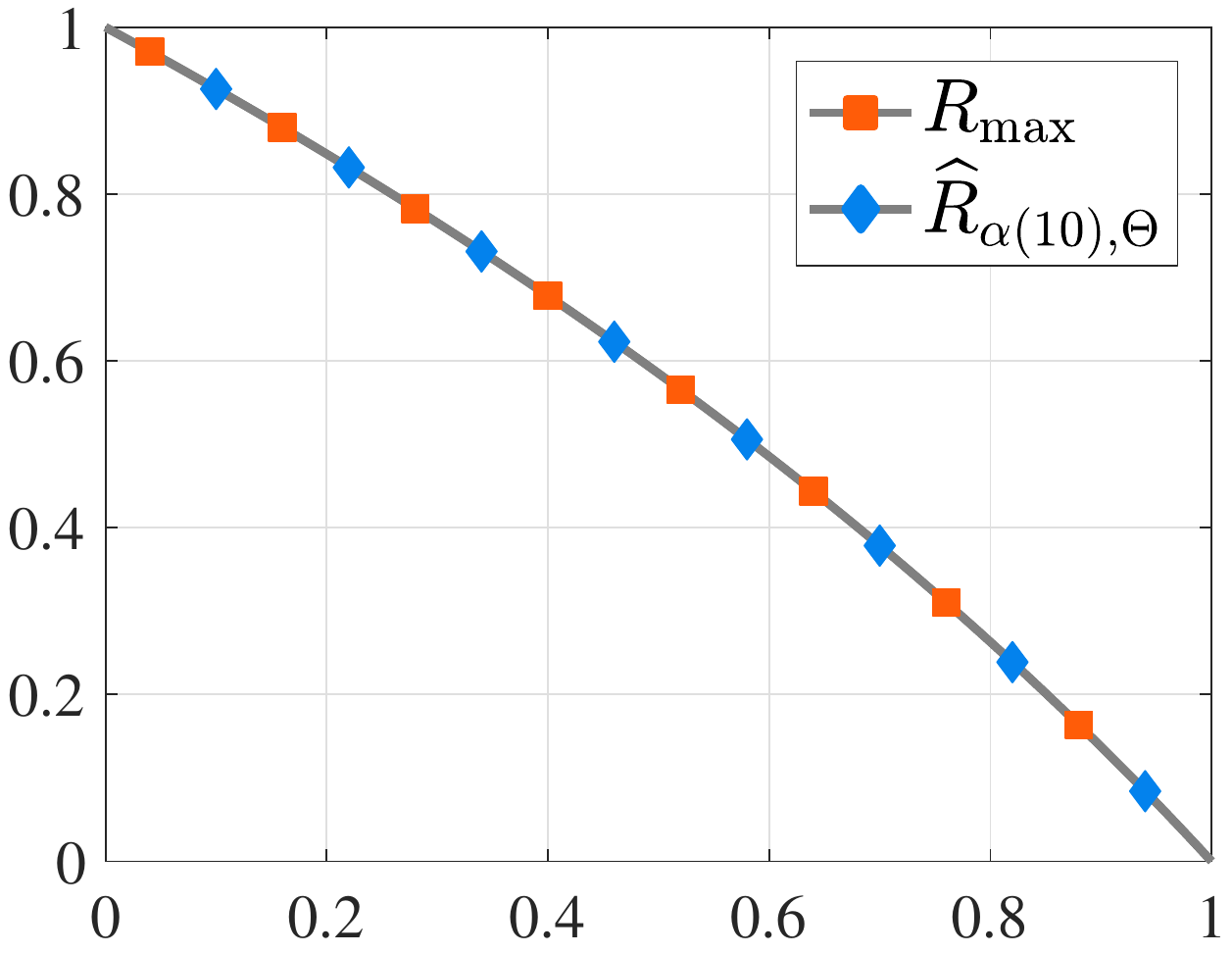}};
\node at (-5.5,-2.4) {\small (d) GAD channel with $N = 0$.};

\node at (0,0) {\includegraphics[width = 5.3cm]{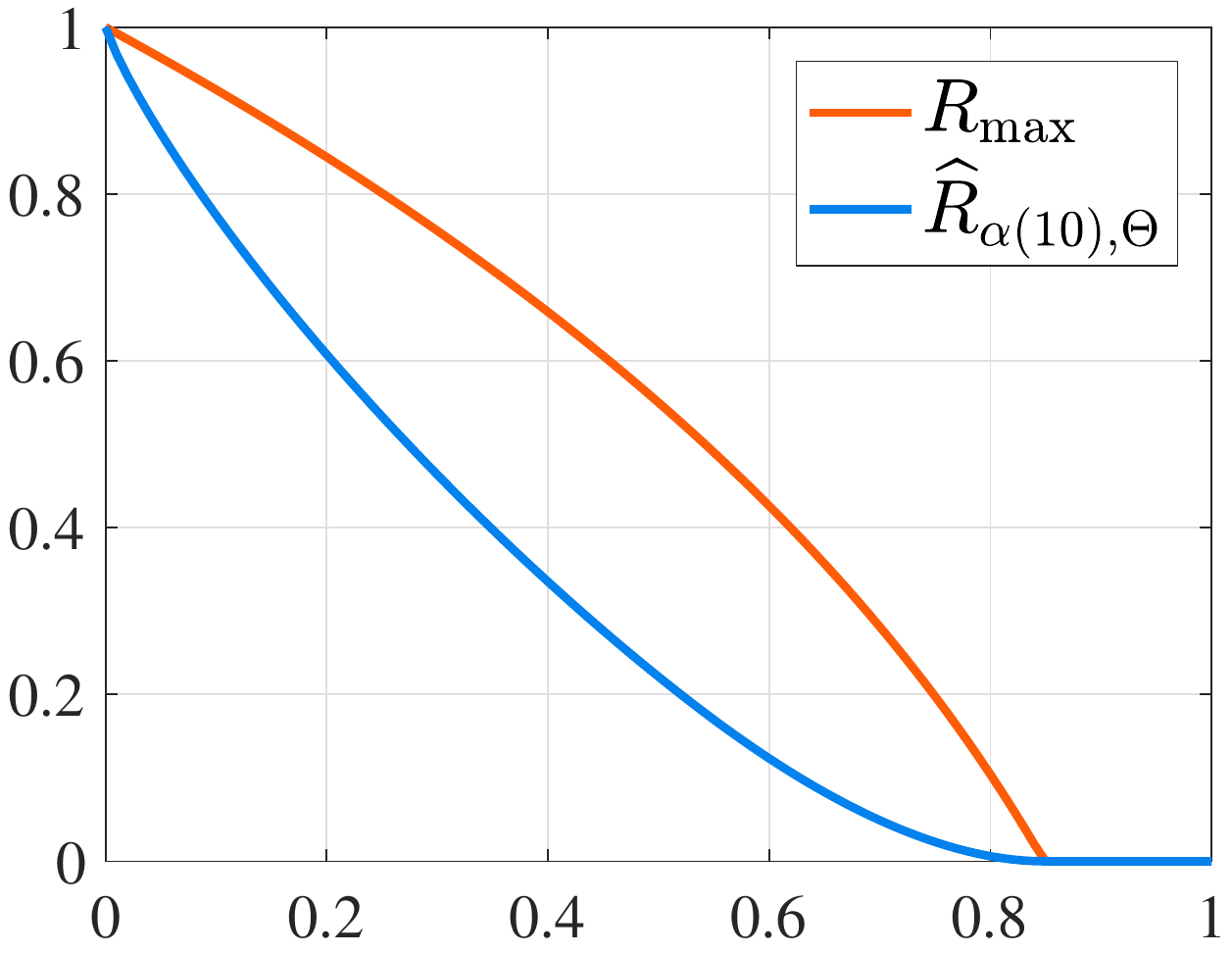}};
\node at (0.2,-2.4) {\small (e) GAD channel with $N = 0.3$.};

\node at (5.7,0) {\includegraphics[width = 5.3cm]{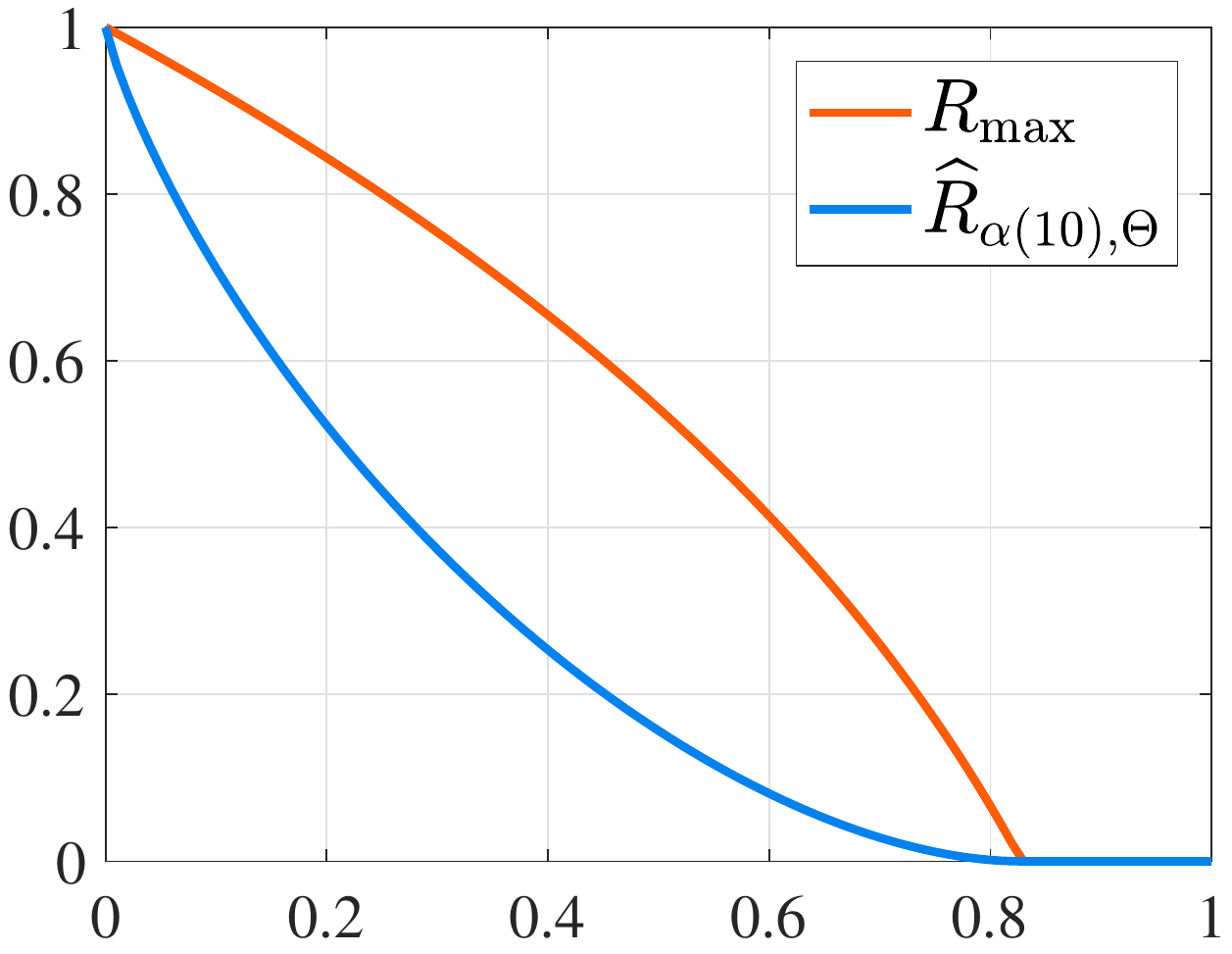}};
\node at (5.9,-2.4) {\small (f) GAD channel with $N = 0.5$.};
\end{scope}
\end{tikzpicture}
\end{adjustwidth}
\caption{\small Comparison of the strong converse bounds on the two-way assisted quantum capacity of the channels $\cD_p$, $\cE_p$ and $\cZ_p$ composed with the amplitude damping channel $\cA_{p,0}$, and the generalized amplitude damping channels $\cA_{p,N}$ with different parameters. The horizontal axis takes value of $p \in [0,1]$.}
\label{two-way assisted quantum capacity compare 1}
\end{figure}

\subsubsection*{Comparison for the two-way assisted quantum capacity of bidirectional channels}

Consider a typical noise in a quantum computer which is modeled as~\cite{Bauml2018}
\begin{align}\label{eq: swap with dephasing channel}
  \cN_{A_1B_1\to A_2B_2}(\rho) = p S\rho S^\dagger + (1-p) U_{\phi} S\rho S U_\phi^\dagger, \quad p \in [0,1]
\end{align}
where $S$ is the swap operator and $U_\phi = \ket{00}\bra{00} + e^{i\phi} \ket{01}\bra{01} + e^{i\phi} \ket{10}\bra{10} + e^{2i\phi} \ket{11}\bra{11}$ is the collective dephasing noise. The comparison result of our new bound $\widehat R^\bi_{\alpha {\scriptscriptstyle (10)},\Theta}$ with the previous bound $R^\bi_{\max}$ is given in Figure~\ref{two-way assisted quantum capacity compare bidirectional}.

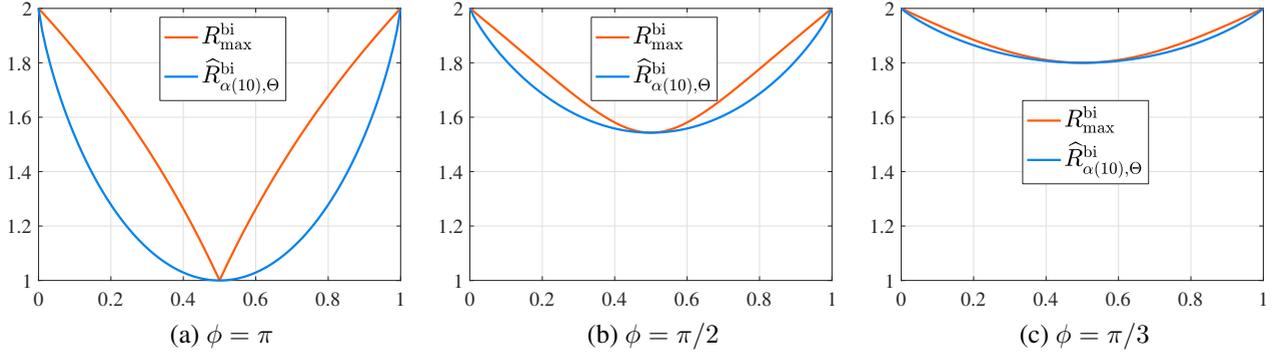
\begin{figure}[H]
\centering
\begin{adjustwidth}{-0.5cm}{0cm}
\begin{tikzpicture}
\node at (-5.7,0) {\includegraphics[width = 5.3cm]{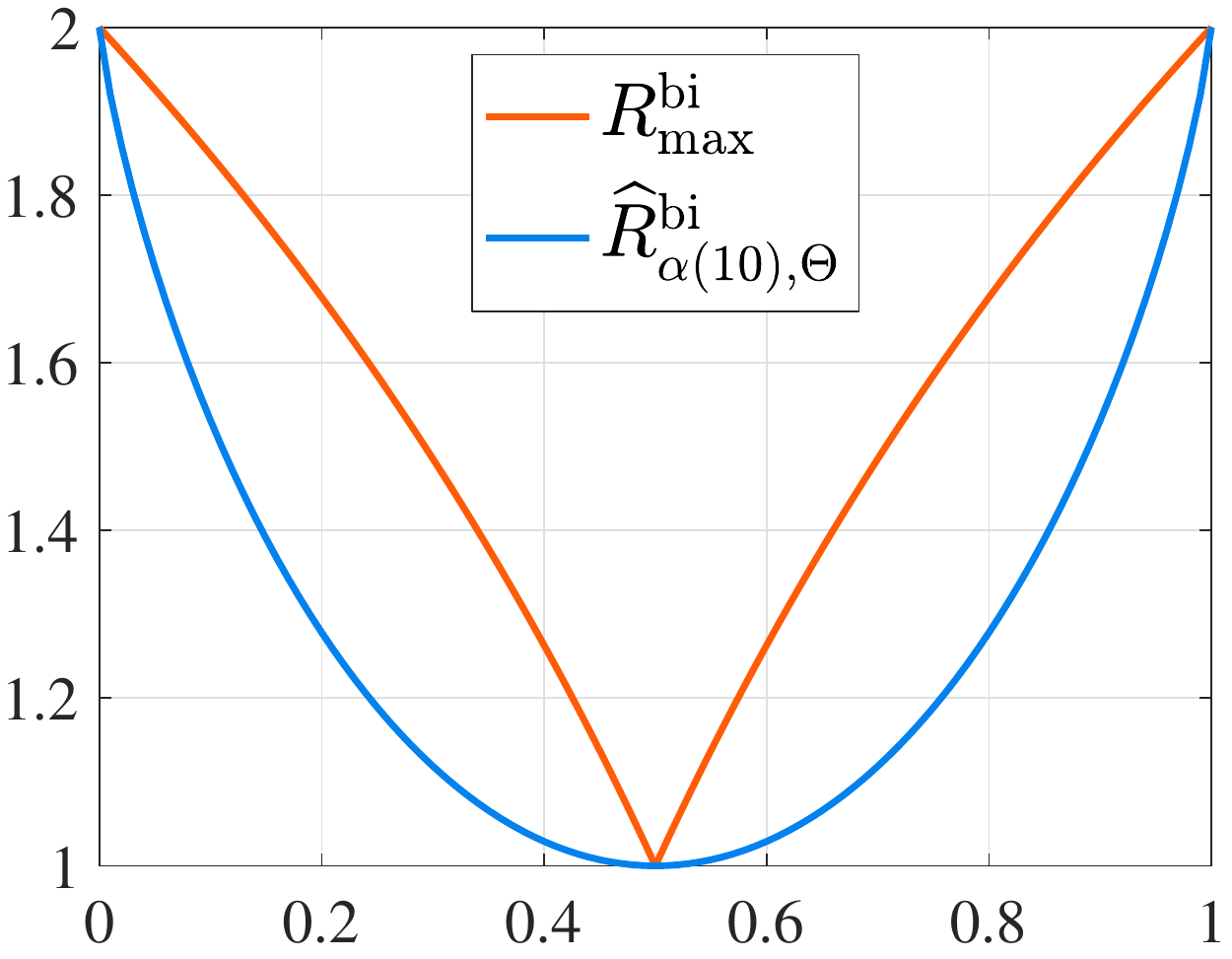}};
\node at (-5.5,-2.4) {\small (a) $\phi = \pi$};

\node at (0,0) {\includegraphics[width = 5.3cm]{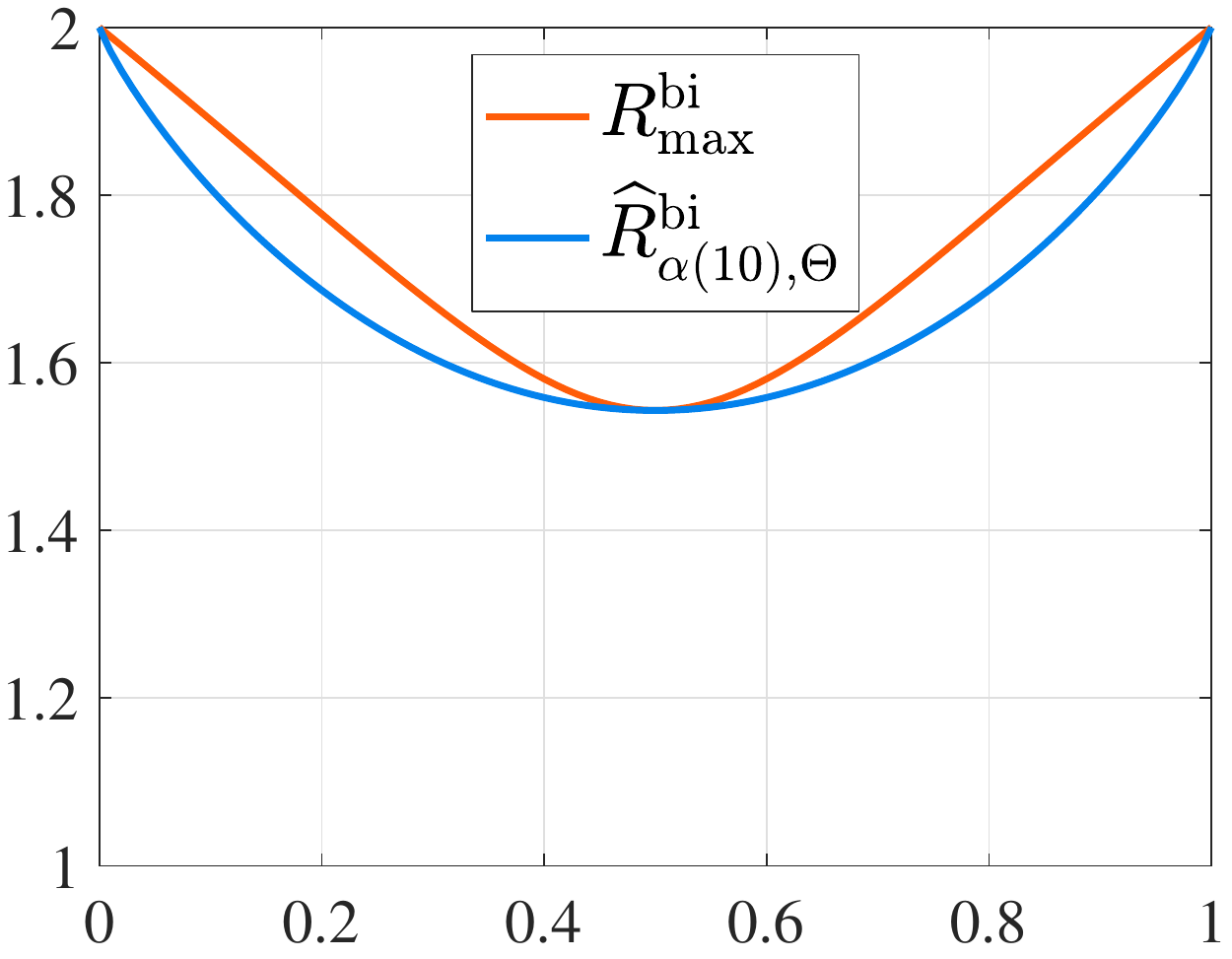}};
\node at (0.2,-2.4) {\small (b) $\phi = \pi/2$};

\node at (5.7,0) {\includegraphics[width = 5.3cm]{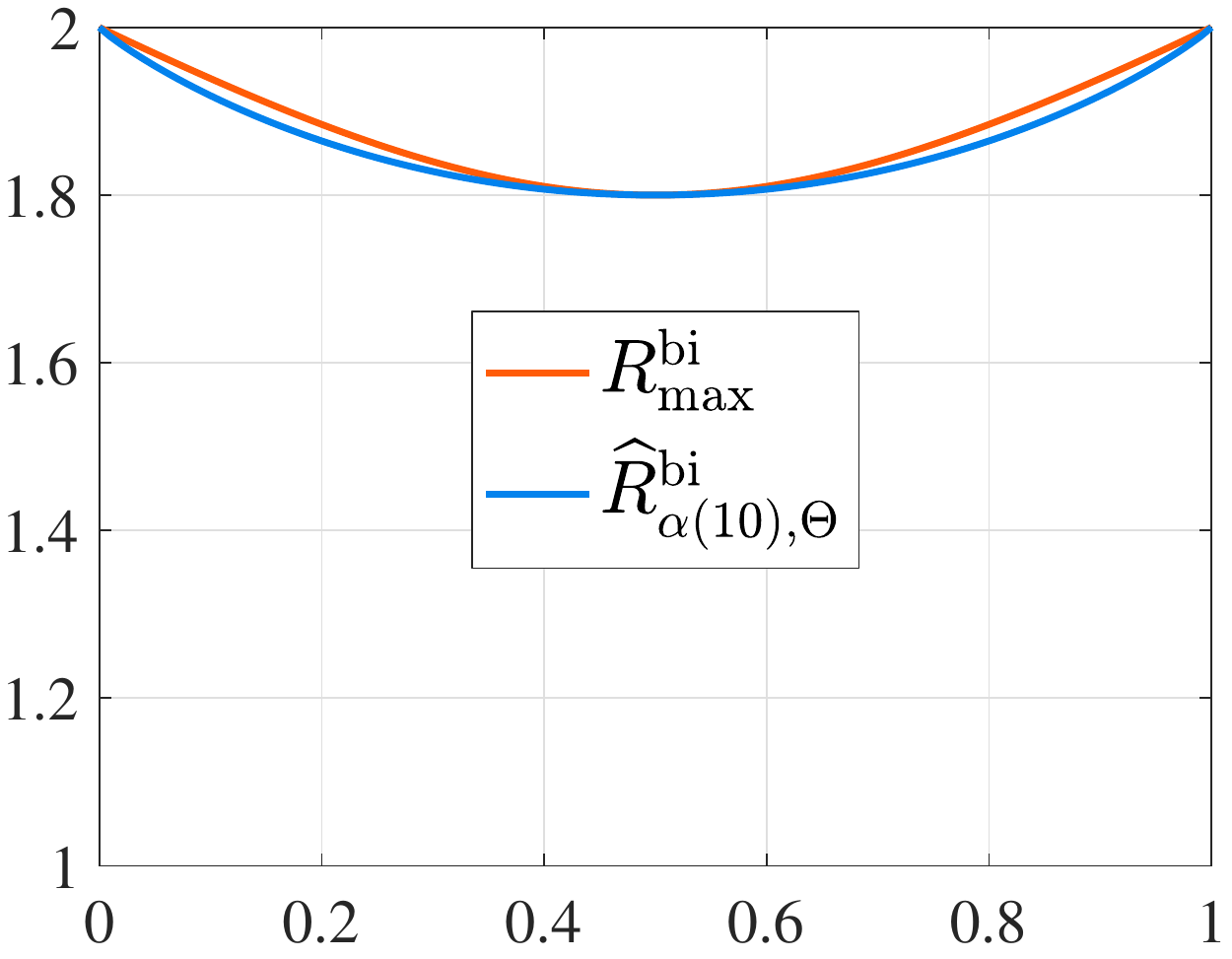}};
\node at (5.9,-2.4) {\small (c) $\phi = \pi/3$};

\end{tikzpicture}
\end{adjustwidth}
\caption{\small Comparison of the strong converse bounds on the two-way assisted quantum capacity of the bidirectional channels in~\eqref{eq: swap with dephasing channel} with the dephasing parameter choosing from $\phi \in \{\pi, \pi/2,\pi/3\}$. The horizontal axis takes value of $p \in [0,1]$.}
\label{two-way assisted quantum capacity compare bidirectional}
\end{figure}

\subsection{Some detailed proofs}
\label{sec: quantum capacity detailed proofs}

In this part, we give the detailed proofs of some aforementioned results.

\vspace{0.2cm}
\noindent \textbf{[Restatement of Proposition~\ref{prop: maximal renyi rains SDP formula}]} 
For any quantum channel $\cN$ and $\a(\ell) = 1+2^{-\ell}$ with $\ell\in \mathbb N$, it holds
\begin{align}
  \widehat R_{\a}(\cN) = \ell \cdot 2^\ell - (2^\ell+1)\log(2^\ell+1) + (2^\ell +1) \log S_\a(\cN),
\end{align}
with $S_\a(\cN)$ given by the following SDP
\begin{gather}
  S_\a(\cN) = \max \ \tr\left[\left(\plsdagger{K} - \ssum_{i=1}^\ell  W_i \right)\boldsymbol\cdot J_{\cN}\right] \quad \text{\rm s.t.}\quad \dblbig{ K,\{ Z_i\}_{i=0}^\ell},\dbhbig{\{ W_i\}_{i=1}^\ell,\rho}, \notag\\
  \dbp{\begin{matrix}
    \rho\ox \1 &  K\\ {K}^\dagger & \plsdagger{ Z}_{\ell}
  \end{matrix}},
  \left\{\dbp{\begin{matrix}
     W_{i} &  Z_{i}\\
     Z_{i}^\dagger & \plsdagger{ Z}_{i-1}
  \end{matrix}}\right\}_{i=1}^\ell, 
  \dbpbigg{\rho\ox \1 \pm \left[\plsdagger{ Z}_0\right]^{\sfT_B}},\dbebigg{\tr \rho - 1}\label{Renyi Rains information max},
\end{gather}
where $J_{\cN}$ is the Choi matrix of $\cN$ and $\plsdagger{X} \equiv X + X^\dagger$ denotes the Hermitian part of $X$.

\vspace{0.2cm}
\begin{proof}
This proof contains two steps. First we derive a \emph{suitable} SDP formula for $\widehat R_\a(\rho_{AB})$ in terms of a maximization problem. Second, we replace $\rho_{AB}$ as the channel's output state $\cN_{A'\to B}(\phi_{AA'})$ and maximize over all the input state $\rho_A$. Since the SDP maximization formula for $\widehat R_\a(\rho_{AB})$ is not necessarily unique, we need to find a suitable one which is able to give us an overall semidefinite optimization in the second step.

\noindent \textit{Step One:} Combining the semidefinite representation of the geometric \Renyi divergence in Lemma~\ref{geometric SDP general lemma} and the semidefinite representation of the Rains set
$\PPT'(A:B)=\big\{\sigma_{AB}\geq 0 \,|\, \sigma_{AB}^{\mathsf{T}_{ B}} = X_{AB} - Y_{AB},\, \tr (X_{AB} + Y_{AB}) \leq 1,\, X_{AB} \geq 0,\, Y_{AB} \geq 0\big\}$,
we have the SDP formula for the geometric \Renyi Rains bound as,
\begin{gather}
  \widehat R_{\a}(\rho_{AB})  = 2^\ell \boldsymbol\cdot \log\min\  \big[\tr M\big] \quad \text{\rm s.t.} \quad \dbhbig{M, \{N_i\}_{i=0}^\ell}, \dbp{X,Y},\notag\\[3pt]
   \dbp{\begin{matrix}
      M & \rho \\ \rho & N_\ell
    \end{matrix}},\, 
    \left\{\dbp{\begin{matrix}
      \rho & N_{i} \\ N_i & N_{i-1}
    \end{matrix}}\right\}_{i=1}^\ell, 
    \dbebigg{N_{0}^{\mathsf{T}_{ B}} - X + Y}, \dbpbigg{1-\tr (X + Y)}.\label{Renyi Rains bound SDP min}
\end{gather}
By the Lagrange multiplier method, the dual SDP is given by
\begin{gather}
  \widehat R_{\a}(\rho_{AB}) = 2^\ell \boldsymbol\cdot \log\max \ \left[\tr\Big[\Big(\plsdagger{K} - \ssum_{i=1}^\ell W_i \Big)\boldsymbol\cdot\rho\Big] - f\right] \quad \text{\rm s.t.}\quad \dblbig{K,\{Z_i\}_{i=0}^\ell},\dbhbig{\{W_i\}_{i=1}^\ell,f}, \notag\\[2pt]
  \dbp{\begin{matrix}
    \1 & K\\ K^\dagger & \plsdagger{Z}_{\ell}
  \end{matrix}},
  \left\{\dbp{\begin{matrix}
    W_{i} & Z_{i}\\
    Z_{i}^\dagger & \plsdagger{Z}_{i-1}
  \end{matrix}}\right\}_{i=1}^\ell, 
  \dbpbigg{f\1 \pm \left[\plsdagger{Z}_0\right]^{\sfT_B}}.\label{Renyi Rains bound SDP max}
\end{gather}
Due to the Slater's condition, we can easily check that the strong duality holds.
Note that both~\eqref{Renyi Rains bound SDP min} and~\eqref{Renyi Rains bound SDP max} are already SDPs for any quantum state $\rho_{AB}$. However, the last condition in~\eqref{Renyi Rains bound SDP max} will introduce an non-linear term if we perform the second step of proof at this stage. The following trick will help us get rid of the variable $f$ which is essential to obtain the final result. 
Note that the last condition above implies $f \geq 0$ and together with its precedent conditions we will necessarily have $f > 0$. Replacing the variables as
\begin{align}
  \widetilde K = f^{-{1}/({2^\ell + 1})}K,\quad \widetilde W_i = f^{-{1}/({2^\ell + 1})}W_i,\quad \widetilde Z_i = f^{-({2^{\ell-i}+1})/({2^\ell + 1})} Z_i,
\end{align} 
we obtain an equivalent SDP of $\widehat R_{\a}(\rho_{AB})$ as
\begin{gather}
   2^\ell \boldsymbol\cdot \log\max \ \left[f^{{1}/({2^\ell + 1})}\tr\Big[\Big(\plsdagger{\widetilde K} - \ssum_{i=1}^\ell \widetilde W_i \Big)\boldsymbol\cdot\rho\Big] - f\right] \quad
  \text{\rm s.t.}\quad \dblbig{\widetilde K,\{\widetilde Z_i\}_{i=0}^\ell},\dbhbig{\{\widetilde W_i\}_{i=1}^\ell,f},\notag\\[2pt]
  \dbp{\begin{matrix}
    \1 & \widetilde K\\ \widetilde{K}^\dagger & \plsdagger{\widetilde Z}_{\ell}
  \end{matrix}},
  \left\{\dbp{\begin{matrix}
    \widetilde W_{i} & \widetilde Z_{i}\\
    \widetilde Z_{i}^\dagger & \plsdagger{\widetilde Z}_{i-1}
  \end{matrix}}\right\}_{i=1}^\ell, 
  \dbpbigg{\1 \pm \left[\plsdagger{\widetilde Z}_0\right]^{\sfT_B}}.
\end{gather}
Denote the objective function $f^{{1}/({2^\ell + 1})} \cdot a - f$ with $a = \tr\big[\big(\plsdagger{\widetilde K} - \sum_{i=1}^\ell \widetilde W_i \big)\boldsymbol\cdot\rho\big] \geq 0$. For any fixed value $a$, the optimal solution is taken at $f = [{a}/({2^\ell+1})]^{1+1/2^\ell}$ with the maximal value $2^\ell [{a}/({2^\ell + 1})]^{1+1/2^\ell}$. Without loss of generality, we can replace the objective function with $2^\ell [{a}/({2^\ell + 1})]^{1+1/2^\ell}$ and get rid of the variable $f$. Direct calculation gives us 
\begin{gather}
  \widehat R_{\a}(\rho_{AB}) = \ell \cdot 2^\ell - (2^\ell+1)\log(2^\ell+1) + (2^\ell +1) \log S_\a(\rho_{AB}) \quad \text{with} \notag\\[2pt]
  S_\a(\rho_{AB}) = \max \ \tr\Big[\Big(\plsdagger{K} - \ssum_{i=1}^\ell  W_i \Big)\boldsymbol\cdot\rho\Big] \quad \text{\rm s.t.}\quad \dblbig{ K,\{ Z_i\}_{i=0}^\ell},\dbhbig{\{ W_i\}_{i=1}^\ell}, \notag\\
  \dbp{\begin{matrix}
    \1 &  K\\ {K}^\dagger & \plsdagger{ Z}_{\ell}
  \end{matrix}},
  \left\{\dbp{\begin{matrix}
     W_{i} &  Z_{i}\\
     Z_{i}^\dagger & \plsdagger{ Z}_{i-1}
  \end{matrix}}\right\}_{i=1}^\ell, 
  \dbpbigg{\1 \pm \left[\plsdagger{ Z}_0\right]^{\sfT_B}}.\label{eq: Renyi Rains bound SDP simplified}
\end{gather}

\noindent \textit{Step Two:} Note that $\cN_{A'\to B}(\phi_{AA'}) = \sqrt{\rho_A} J_{\cN} \sqrt{\rho_A}$ holds for any quantum state $\rho_A$ with purification $\phi_{AA'}$. Thus the final result is straightforward from~\eqref{eq: Renyi Rains bound SDP simplified} by replacing the input state $\rho_{AB}$ as $\sqrt{\rho_A} J_{\cN} \sqrt{\rho_A}$, replacing $K,Z_i,W_i$ as $\rho_A^{-1/2} K \rho_A^{-1/2},\rho_A^{-1/2} Z_i \rho_A^{-1/2}, \rho_A^{-1/2} W_i\rho_A^{-1/2}$ respectively and maximizing over all input state $\rho_A$.
\end{proof}

\vspace{0.2cm}
\noindent \textbf{[Restatement of Proposition~\ref{prop: Rains and Theta information}]} 
  For any generalized divergence $\bD$ and any quantum channel $\cN$, it holds  
  \begin{align}\label{eq: rains and theta info}
    \bR(\cN) \leq \bR_\Theta(\cN).
  \end{align}
  Moreover, for the max-relative entropy the equality always holds, i.e,
  \begin{align}\label{eq: max rains and theta info}
    R_{\max}(\cN) = R_{\max,\Theta}(\cN).
  \end{align}
\begin{proof}
We prove the relation~\eqref{eq: rains and theta info} first. Note that for any pure state $\phi_{AA'}$ and $\cM_{A'\to B} \in \bcV_\Theta$, we have 
\begin{align}
  \|(\cM_{A'\to B}(\phi_{AA'}))^{\sfT_B}\|_1 = \|\Theta_{B}\circ\cM_{A'\to B}(\phi_{AA'})\|_1 \leq \|\Theta_{B}\circ\cM_{A'\to B}\|_\di \leq 1.
\end{align} 
This implies $\cM_{A'\to B}(\phi_{AA'}) \in \PPT'(A:B)$.
Then it holds
\begin{align}
  \bR(\cN) & = \max_{\rho_A \in \cS(A)} \min_{\sigma_{AB} \in \PPT'(A:B)} \bD(\cN_{A'\to B}(\phi_{AA'})\|\sigma_{AB})\\
  & \leq \max_{\rho_A \in \cS(A)} \min_{\cM \in {\bcV_{\Theta}}} \bD(\cN_{A'\to B}(\phi_{AA'})\|\cM_{A'\to B}(\phi_{AA'}))\\
  & =  \min_{\cM \in {\bcV_{\Theta}}} \max_{\rho_A \in \cS(A)} \bD(\cN_{A'\to B}(\phi_{AA'})\|\cM_{A'\to B}(\phi_{AA'}))\\
  & = \bR_\Theta(\cN).  
\end{align}
The first and last line follow by definition. The inequality holds since $\cM_{A'\to B}(\phi_{AA'}) \in \PPT'(A:B)$ and thus the first line is minimizing over a larger set. In the third line, we swap the min and max by the argument in Remark~\ref{Rains theta swap}.

We next prove the equation~\eqref{eq: max rains and theta info}. Recall that the SDP formula of the max-Rains information is given by (\cite[Proposition 5]{Wang2017d} or~\cite[Eq.~(11)]{Wang2016a})
  \begin{align}
    R_{\max}(\cN) = \log \min \big\{\mu\,\big|\,(V-Y)^{\sfT_{B}} \ge J_{\cN},\tr_B(V+Y)\le \mu \1_A, Y,V\geq 0 \big\}.
  \end{align}
Replace $V$ and $Y$ with $\mu V$ and $\mu Y$ respectively, and then denote $N = (V- Y)^{\sfT_B}$, we have
  \begin{align}\label{eq: max Rains tmp1}
  R_{\max}(\cN) = \log \min \Big\{\mu\,\Big|\ J_{\cN} \leq \mu N,\, N = (V-Y)^{\sfT_{B}},\,\tr_B(V+Y)\le \1_A, Y,V\ge0 \Big\}.
\end{align}
Notice that the second to the last conditions define a set of CP maps
\begin{align}\label{theta set alternative representation}
  \bcV  \equiv \left\{\cM \in \CP(A:B)\,\big|\, \exists\, V_{AB}, Y_{AB},\ \text{s.t.}\ J_{\cM}^{\sfT_B} = V - Y, \tr_B (V+Y) \leq \1_A, Y,V \geq 0\right\}.
\end{align}
Combining \eqref{eq: max Rains tmp1} and \eqref{theta set alternative representation}, we obtain $R_{\max}(\cN) = \min_{\cM\in\bcV} D_{\max}(\cN\|\cM)$. Thus it suffices for us to show the equivalence $\bcV= \bcV_\Theta$.
For any $\cM \in \bcV_\Theta$, take $V = (R + J_{\cM}^{\sfT_B})/2$ and $Y = (R - J_{\cM}^{\sfT_B})/2$. Then $V \geq 0$, $Y \geq 0$, $J_{\cM}^{\sfT_B} = V-Y$ and $\tr_B (V+Y) = \tr_B R \leq \1_A$, which implies $\cM \in \bcV$. On the other hand, for any $\cM \in \bcV$, take $R = V + Y$. We can check that $\tr_B R = \tr_B(V+Y) \leq \1_A$, $R + J_{\cM}^{\sfT_B} = (V+Y) + (V-Y) = 2V \geq 0$ and $R - J_{\cM}^{\sfT_B} = (V+Y) - (V-Y) = 2Y \geq 0$, which implies $\cM \in \bcV_\Theta$. Finally we have
\begin{align}
R_{\max}(\cN) = \min_{\cM\in\bcV} D_{\max}(\cN\|\cM) = \min_{\cM\in\bcV_\Theta} D_{\max}(\cN\|\cM) = R_{\max,\Theta}(\cN),
\end{align}
which completes the proof.
\end{proof}

\vspace{0.2cm}
\noindent \textbf{[Restatement of Proposition~\ref{amortization proposition}]} 
  For any quantum state $\rho_{A'AB'}$, any quantum channel $\cN_{A\to B}$ and the parameter $\a \in (1,2]$, it holds
  \begin{align}
  \widehat R_\a(\omega_{A':BB'}) \leq \widehat R_\a(\rho_{A'A:B'}) + \widehat R_{\a,\Theta}(\cN_{A\to B})\quad \text{with} \quad \omega_{A':BB'} = \cN_{A\to B}(\rho_{A'A:B'}).
\end{align}
\begin{proof}
Suppose the optimal solution of $\widehat R_\a(\rho_{A'A:B'})$ and $\widehat R_{\a,\Theta}(\cN_{A\to B})$ are taken at $\sigma_{A'AB'}\in \PPT'(A'A:B')$ and $\cM \in \bcV_{\Theta}$, respectively. Let $ \gamma_{A'BB'} = \cM_{A\to B}(\sigma_{A'AB'})$. We have
  \begin{align}
    \left\|\gamma_{A'BB'}^{\sfT_{BB'}}\right\|_1 = \left\|\Theta_{B}\circ \cM_{A\to B}(\sigma_{A'AB'}^{\sfT_{B'}})\right\|_1 \leq \|\Theta_{B}\circ \cM_{A\to B}\|_{\di} \left\|\sigma_{A'AB'}^{\sfT_{B'}}\right\|_1 \leq 1,
  \end{align}
  where the first inequality follows from the definition of diamond norm and the second inequality follows from the choice of $\sigma_{A'AB'}$ and $\cM_{A\to B}$.
  Thus $\gamma_{A'BB'} \in \PPT'(A':BB')$ and forms a feasible solution for $\widehat R_\a(\omega_{A':BB'})$. Then we have
  \begin{align}
    \widehat R_\a(\omega_{A':BB'}) & \leq \widehat D_{\alpha}(\omega_{A':BB'}\|\gamma_{A'BB'})\\
    & = \widehat D_{\alpha}(\cN_{A\to B}(\rho_{A'A:B'})\|\cM_{A\to B}(\sigma_{A'AB'}))\\
    & \leq \widehat D_{\alpha}(\cN\|\cM) + \widehat D_{\alpha}(\rho_{A'A:B'}\|\sigma_{A'AB'})\\
    & = \widehat R_{\a,\Theta}(\cN_{A\to B}) + \widehat R_\a(\rho_{A'A:B'}).
  \end{align}
  The second inequality follows from the chain rule of the geometric \Renyi divergence in Lemma~\ref{lem_chainRule}, and the last line follows by the optimality assumption of $\cM$ and $\sigma$.
\end{proof}

\newpage
\section{Private communication}
\label{sec: Private communication}

\subsection{Backgrounds}

The \emph{private capacity} of a quantum channel is defined as the maximum rate at which classical information can be transmitted privately from the sender (Alice) to the receiver (Bob). By ``private'', it means a third party (Eve) who has access to the channel environment cannot learn anything about the information that Alice sends to Bob. There are also two different private capacities of major concern, the (unassisted) private capacity $P$ and the two-way assisted private capacity $P^{\leftrightarrow}$, depending on whether classical communication is allowed between each channel uses. 

In the same spirit of the quantum capacity theorem, the private capacity theorem states that the private capacity of a quantum channel is given by its regularized private information~\cite{Cai2004,Devetak2005a},
\begin{align}\label{eq: private channel coding theorem}
    P(\cN) = \lim_{n\to \infty} \frac{1}{n} I_p(\cN^{\ox n}) = \sup_{n\in \NN} \frac{1}{n} I_p(\cN^{\ox n}),
\end{align}
where $I_p(\cN) \equiv \max_{\mathscr E} \left[\chi(\mathscr E, \cN) - \chi(\mathscr E, \cN^c)\right]$
is the private information with the maximization taken over all quantum state ensembles $\mathscr E = \{p_i, \rho_i\}$, $\chi(\mathscr E, \cN) \equiv H(\sum_i p_i \cN(\rho_i)) - \sum_i p_i H(\cN(\rho_i))$ is the Holevo information of the ensemble $\mathscr E$, $H$ is the von Neumann entropy and $\cN^c$ is the complementary channel of $\cN$. The regularization in~\eqref{eq: private channel coding theorem} is necessary in general since the private information is proved to be non-additive~\cite{Smith2008structured} and an unbounded number of channel uses may be required to achieve its private capacity~\cite{Elkouss2015}. 

Despite of their importance in understanding the fundamental limits of quantum key distributions \cite{Bennett1984}, much less is known about the converse bounds on private capacities, mostly due to their inherently involved settings. The squashed entanglement of a channel was proposed in~\cite{Takeoka2014} and proved to be a converse bound on the two-way assisted private capacity. But its strong converse is still left open and the quantity itself is difficult to be computed exactly~\cite{Huang2014}. The entanglement cost of a channel was introduced in~\cite{Berta2015e} and shown to be a strong converse bound on the two-way assisted private capacity~\cite{Christandl2016}. But it was not given by a single-letter formula. A closely related quantity to this part is the relative entropy of entanglement of a channel ($E_R$), which was proved as a (weak) converse bound on the two-way assisted private capacity for channels with ``covariant symmetry''~\cite{Pirandola2015b}. This was later strengthened in~\cite{Wilde2016c} as a strong converse bound on the unassisted private capacity for general quantum channels and a strong converse bound on the two-way assisted private capacity for channels with ``covariant symmetry''. Moreover, the max-relative entropy of entanglement of a channel ($E_{\max}$) was proved as a strong converse bound on the two-way assisted private capacity in general~\cite{Christandl2016}.

\subsection{Summary of results}

In this part, we extend the techniques used in the previous sections to the task of private communication and aim to improve the max-relative entropy of entanglement of a channel in both assisted and unassisted scenarios. The structure is organized as follows (see also a schematic diagram in Figure~\ref{fig: private communiation summary}).

In Section~\ref{sec: Maximal Renyi divergence of entanglement} we discuss the unassisted private communication. While the relative entropy of entanglement $E_R$ established the best known strong converse in this case, the difficulties of its evaluation are two-folds: the optimization over the set of separable states and the minimax optimization of the Umegaki relative entropy. The first difficulty will be automatically removed for qubit channels since separability can be completely characterized by the positive partial transpose conditions~\cite{Horodecki1996}. The second can be handled by relaxing the Umegaki relative entropy to a semidefinite representable one, such as the max-relative entropy. Based on a notion of the generalized relative entropy of entanglement of a channel, we exhibit that the entanglement of a channel induced by the geometric \Renyi divergence ($\widehat E_{\a}$) lies between $E_R$ and $E_{\max}$. That is, we show that
\begin{align*}
P(\cN) \leq P^\dagger(\cN)\leq E_R(\cN) \leq \widehat E_{\a}(\cN) \leq E_{\max}(\cN),
\end{align*}
where $P(\cN)$ and $P^{\dagger}(\cN)$ denote the unassisted private capacity of channel $\cN$ and its corresponding strong converse capacity, respectively. Moreover, $\widehat E_{\a}(\cN)$ is given by a conic program in general and reduces to a semidefinite program for all qubit channels.

In Section~\ref{sec: Maximal Renyi Sigma-information}, we study the private communication with two-way classical communication assistance. We introduce the \emph{generalized Sigma-information} which is a new variant of channel information inspired by the channel resource theory. More precisely, we define the generalized Sigma-information as a ``channel distance'' to the class of entanglement breaking subchannels. We show that the max-relative entropy of entanglement $E_{\max}$ coincides with the generalized Sigma-information induced by the max-relative entropy $E_{\max,\Sigma}$, i.e., $E_{\max} = E_{\max,\Sigma}$, thus providing a completely new perspective of understanding the former quantity. Moreover, we prove that the generalized Sigma-information induced by the geometric \Renyi divergence ($\widehat E_{\a,\Sigma}$) is a strong converse on the two-way assisted private capacity by utilizing an ``amortization argument'', improving the previously best-known result of the max-relative entropy of entanglement~\cite{Christandl2016} in general. That is, we show that
\begin{align*}
P^{\leftrightarrow}(\cN) \leq  P^{\leftrightarrow,\dagger}(\cN) \leq \widehat E_{\a,\Sigma}(\cN) \leq E_{\max}(\cN),
\end{align*}
where $P^{\leftrightarrow}(\cN)$ and $P^{\leftrightarrow,\dagger}(\cN)$ denote the two-way assisted private capacity of channel $\cN$ and its corresponding strong converse capacity, respectively. Moreover, $\widehat E_{\a,\Sigma}(\cN)$ is given by a conic program in general and reduces to a semidefinite program for all qubit channels.


\begin{figure}[H]
\centering
\begin{tikzpicture}
\draw[fill=gray!10,opacity=0.3] (7,-7.5) rectangle node[gray,opacity=1,midway,shift={(0,2.2)}]{{\small SDP computable }} node[gray,opacity=1,midway,shift={(0,1.8)}]{{\small for qubit channels}} (10.8,-10.5);
\node (Ptwoway) at (0,-8) {$P^{\leftrightarrow}$};
\node[circle,fill=magenta!10,inner sep=1pt,minimum size=2pt] (Ptwowaydagger) at (2.5,-8) {$P^{\leftrightarrow,\dagger}$};
\node (Pcppp) at (0,-10) {$P$};
\node[circle,fill=magenta!10,inner sep=1pt,minimum size=2pt] (Pcpppdagger) at (2.5,-10) {$P^{\dagger}$};
\node[red] (Ea) at (7.5,-10) {$\widehat E_\a$};
\node[red] (EaSigma) at (7.5,-8) {$\widehat E_{\a,\Sigma}$};
\node (ER) at (5,-10) {$E_R$};
\node[red] (EmaxSigma) at (10,-8) {$E_{\max,\Sigma}$};
\node (Emax) at (10,-10) {$E_{\max}$};
\draw[very thick,->,dotted,gray] (ER) -- node[black,rotate = -40,midway,shift={(-0.1,0.15)}] {\footnotesize cov.} node[black,rotate=-39,midway,shift={(-0.1,-0.2)}] {\scriptsize \cite{Wilde2016c}} (Ptwowaydagger);
\draw[thick,->] (Ptwowaydagger) -- (Ptwoway);
\draw[thick,->] (Ptwoway) -- (Pcppp);
\draw[thick,->] (Ptwowaydagger) -- (Pcpppdagger);
\draw[thick,->] (Pcpppdagger) -- (Pcppp);
\draw[red,very thick,->] (EaSigma) -- node[black,midway,shift={(0.1,0.2)}] {\footnotesize Thm.~\ref{thm: main result private assisted}} (Ptwowaydagger);
\draw[thick,->] (ER) -- node[midway,shift={(0.1,0.2)}] {\scriptsize \cite{Wilde2016c}}(Pcpppdagger);
\draw[thick,->] (EaSigma) -- node[rotate=90,midway,shift={(0.1,0.2)}] {\footnotesize Prop.~\ref{prop: private Ea Easigma}} node[rotate = 90,midway,shift={(0,-0.2)}] {\footnotesize $\boldsymbol\neq$}  (Ea);
\draw[thick,->] (Ea) -- node[midway,shift={(0.1,0.2)}] {\footnotesize Lem.~\ref{thm: divergence chain inequality}} node[midway,shift={(0.1,-0.2)}] {\footnotesize $\boldsymbol\neq$}  (ER);
\draw[thick,->] (EmaxSigma) -- node[midway,shift={(0.1,0.2)}] {\footnotesize Lem.~\ref{thm: divergence chain inequality}} node[midway,shift={(0.1,-0.2)}] {\footnotesize $\boldsymbol\neq$}  (EaSigma);
\draw[thick,<->] (EmaxSigma) -- node[rotate=90,midway,shift={(0.1,0.2)}] {\footnotesize Prop.~\ref{prop: private Ea Easigma}} (Emax);
\draw[thick,->] (Emax) -- node[midway,shift={(0.1,0.2)}] {\footnotesize Lem.~\ref{thm: divergence chain inequality}} node[midway,shift={(0.1,-0.2)}] {\footnotesize $\boldsymbol\neq$}  (Ea);
\draw[thick,->] (Emax.east) -- (11,-10) -- (11,-6) -- node[midway,shift={(0,0.2)}] {\scriptsize \cite{Christandl2016}} (6,-6) -- (3,-7.8);

\end{tikzpicture}
\caption{\small Relations between different converse bounds for private communication. $P^*$ and $P^{*,\dagger}$ are the private capacity with assistance $*$ and its corresponding strong converse capacity, respectively. $E_R$, $\widehat E_\a$ and $E_{\max}$ are the generalized relative entropy of entanglement of a channel induced by different quantum divergences. $\widehat E_{\a,\Theta}$ and $E_{\max,\Theta}$ are the generalized Sigma-information induced by different quantum divergences. The circled quantities are those of particular interest in quantum information theory. The key quantities and the main contributions in this section are marked in red. The quantity at the start point of an arrow is no smaller than the one at the end point. The double arrow represents that two quantities coincide. The inequality sign represents two quantities are not the same in general. The dotted arrow represents that the relation holds under certain restrictions, where ``cov.'' stands for ``covariant''. The parameter $\a$ is taken in the interval $(1,2]$. The quantities in the shaded area are given by conic programs and are SDP computable for all qubit channels (or more generally channels with dimension $|A||B| \leq 6$).}
\label{fig: private communiation summary}
\end{figure}
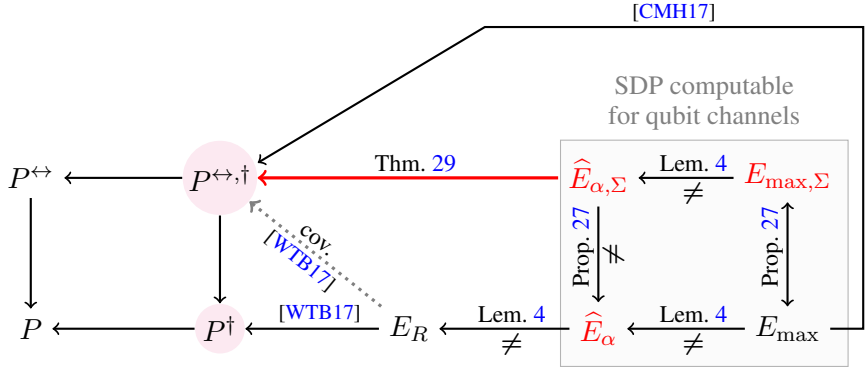

\subsection{Unassisted private capacity}
\label{sec: Maximal Renyi divergence of entanglement}

In this section we discuss converse bounds on the unassisted private capacity~\footnote{We refer to the work~\cite[Section V.A]{Wilde2016c} for rigorous definitions of the unassisted private capacity and its strong converse.}.

\begin{definition}
    For any generalized divergence $\bD$, the generalized relative entropy of entanglement of a quantum state $\rho_{AB}$ is defined as
    \begin{align}
        \boldsymbol E(\rho_{AB}) \equiv \min_{\sigma_{AB}\in \SEP_{\bullet}(A:B)} \boldsymbol D(\rho_{AB}\|\sigma_{AB}),
    \end{align}
    where $\SEP_{\bullet}(A:B)$ is the set of sub-normalized separable states between $A$ and $B$.
\end{definition}

If the generalized divergence satisfies the dominance property, i.e., $\bD(\rho\|\sigma) \geq \bD(\rho\|\sigma')$ if $\sigma \leq \sigma'$, then the optimal solution of the above minimization problem can always be taken at a normalized separable states. Since the dominance property is generic for most divergences of concern, the above definition is consistent with the one defined over the set of normalized separable states (e.g.~\cite{Vedral1998}).

\begin{definition}
    For any generalized divergence $\bD$, the generalized relative entropy of entanglement of a quantum channel $\cN_{A'\to B}$ is defined as 
\begin{align}
    \boldsymbol E(\cN) \equiv \max_{\rho_A \in \cS(A)} \boldsymbol E(\cN_{A'\to B}(\phi_{AA'})) = \max_{\rho_A \in \cS(A)}  \min_{\sigma_{AB}\in\SEP_{\bullet}(A:B)} \bD(\cN_{A'\to B}(\phi_{AA'})\|\sigma_{AB}),
\end{align}
where $\phi_{AA'}$ is a purification of quantum state $\rho_A$.
\end{definition}
In particular, the relative entropy of entanglement for a channel is induced by the Umegaki relative entropy~\cite{Pirandola2015b},
\begin{align}
    E_R(\cN)  = \max_{\rho_A \in \cS(A)}  \min_{\sigma_{AB}\in \SEP_{\bullet}(A:B)} D(\cN_{A'\to B}(\phi_{AA'})\|\sigma_{AB}).
\end{align}
The max-relative entropy of entanglement for a channel is induced by the max-relative entropy~\cite{Christandl2016},
\begin{align}
    E_{\max}(\cN)  = \max_{\rho_A \in \cS(A)}  \min_{\sigma_{AB}\in \SEP_{\bullet}(A:B)} D_{\max}(\cN_{A'\to B}(\phi_{AA'})\|\sigma_{AB}).
\end{align}
These two quantities are known as strong converse bounds for private capacity with and without classical communication assistance, respectively. That is, 
\begin{align}
    P^\dagger(\cN) \leq E_R(\cN)\quad ~\cite{Wilde2016c} \quad \text{and} \quad P^{\leftrightarrow,\dagger}(\cN) \leq E_{\max}(\cN)\quad ~\cite{Christandl2016}.
\end{align}

The computability of $E_R$ is usually restricted to qubit covariant channels where the input state $\rho_A$ can be taken as the maximally mixed states and the set of separable states coincides with the set of PPT states~\cite{Horodecki1996}. The following result relaxes $E_R$ to its geometric \Renyi version $\widehat E_\a$, which is SDP computable for \emph{all qubit channels} and is tighter than $E_{\max}$ in general.

\begin{theorem}[Main result 4]\label{thm: main result private unassisted}
    For any quantum channel $\cN_{A'\to B}$ and $\a \in (1,2]$, it holds
        \begin{align}\label{eq: main result 3}
            P(\cN) \leq P^\dagger(\cN)\leq E_R(\cN) \leq \widehat E_{\a}(\cN) \leq E_{\max}(\cN),
        \end{align}
      where $P(\cN)$ and $P^{\dagger}(\cN)$ denote the unassisted private capacity  and its corresponding strong converse capacity, respectively.
    Moreover, the bound $\widehat E_{\a}(\cN)$ with $\a(\ell) = 1+2^{-\ell}$ and $\ell \in \mathbb N$ can be given as 
\begin{align}
    \widehat E_{\a}(\cN) = \ell \cdot 2^\ell - (2^\ell+1)\log(2^\ell+1) + (2^\ell +1) \log T_\a(\cN),
\end{align}
with $T_\a(\cN)$ given by the following conic program
\begin{gather}
    T_\a(\cN) = \max \ \tr\left[\left(\plsdagger{K} - \ssum_{i=1}^\ell  W_i \right)\boldsymbol\cdot J_{\cN}\right] \quad \text{\rm s.t.}\quad \dblbig{ K,\{ Z_i\}_{i=0}^\ell},\dbhbig{\{ W_i\}_{i=1}^\ell,\rho}, \notag\\
    \dbp{\begin{matrix}
        \rho\ox \1 &  K\\ {K}^\dagger & \plsdagger{ Z}_{\ell}
    \end{matrix}},
    \left\{\dbp{\begin{matrix}
         W_{i} &  Z_{i}\\
         Z_{i}^\dagger & \plsdagger{ Z}_{i-1}
    \end{matrix}}\right\}_{i=1}^\ell,\, 
    \rho\ox \1 - \plsdagger{Z}_0 \in \BP(A:B),\dbebigg{\tr \rho - 1}\label{Renyi Ea max},
\end{gather}
where $J_{\cN}$ is the Choi matrix of $\cN$ and $\BP(A:B)$ is the set of block-positive operators which reduces to a semidefinite cone if $|A||B| \leq 6$.
\end{theorem}

\begin{proof}
The first inequality in~\eqref{eq: main result 3} follows by definition. The second inequality in~\eqref{eq: main result 3} was proved in~\cite{Wilde2016c}. The last two inequalities in~\eqref{eq: main result 3} are direct consequences of Lemma~\ref{thm: divergence chain inequality}. The derivation of the conic program~\eqref{Renyi Ea max} follows the same steps as Proposition~\ref{prop: maximal renyi rains SDP formula}. The block positive cone $\BP(A:B)$ is the dual cone of the set of separable operators. When the channel dimension satisfies $|A||B| \leq 6$, this cone admits a semidefinite representation as $\BP(A:B)=\{X+Y^{\sfT_B}\,|\, X \geq 0, Y \geq 0\}$~\cite[Table 2.2]{aubrun2017alice}. Thus the conic program~\eqref{Renyi Ea max} reduces to a semidefinite program.
\end{proof}

\subsection{Two-way assisted private capacity}
\label{sec: Maximal Renyi Sigma-information}

In this section we discuss converse bounds on the two-way assisted private capacity~\footnote{We refer to~\cite[Section V.A]{Wilde2016c} for rigorous definitions of the two-way assisted private capacity and its strong converse.}.

A quantum channel $\cN_{A'\to B}$ is called \emph{entanglement breaking} if its output state $\cN_{A'\to B}(\rho_{AA'})$ is separable for any input $\rho_{AA'}$ or equivalently if its Choi matrix is separable~\cite{horodecki2003entanglement}. Since every entanglement breaking channel can be simulated by a measurement-preparation scheme~\cite{horodecki2003entanglement,Holevo1998coding}, any two-way assisted private communication protocol via entanglement breaking channel will end up with a separable state. This indicates that these channels are useless for private communication. With this in mind, we consider the set of \emph{entanglement breaking subchannels} as
\begin{align}
    \EBset \equiv \big\{\cM \in \CP(A:B)\,\big|\, J_{\cM}\in \SEPcone(A:B),\, \tr_B J_{\cM} \leq \1_A\big\},
\end{align} 
where $\SEPcone(A:B)$ denotes the cone of separable operators.

\begin{definition}[Sigma-info.]
    For any generalized divergence $\bD$, the generalized Sigma-information of a quantum channel $\cN_{A'\to B}$
   is defined as 
\begin{align}
    \boldsymbol E_{\Sigma}(\cN) \equiv \min_{\cM \in \EBset}\bD(\cN\|\cM) =   \min_{\cM \in \EBset} \max_{\rho_A \in \cS(A)} \bD(\cN_{A'\to B}(\phi_{AA'})\|\cM_{A'\to B}(\phi_{AA'})),
\end{align}
where $\phi_{AA'}$ is a purification of quantum state $\rho_A$.
\end{definition}
As mentioned in Remark~\ref{Rains theta swap}, the min and max in the above definition can be swapped. 

Analogous to Proposition~\ref{prop: Rains and Theta information}, the following result establishes the relation between the generalized Sigma-information $\boldsymbol E_{\Sigma}$ and the generalized relative entropy of entanglement of a channel $\boldsymbol E$.

\begin{proposition}\label{prop: private Ea Easigma}
    For any generalized divergence $\bD$ and any quantum channel $\cN$, it holds 
    \begin{align}
        \boldsymbol E(\cN)\leq \boldsymbol E_{\Sigma}(\cN).
    \end{align}
  Moreover, for the max-relative entropy the equality always holds, i.e,
  \begin{align}\label{eq: Emax sigma information}
    E_{\max}(\cN)= E_{\max,\Sigma}(\cN).
  \end{align}
\end{proposition}
\begin{proof}
This first inequality can be proved in a similar manner as Proposition~\ref{prop: Rains and Theta information} by using the fact that $\cM_{A'\to B}(\phi_{AA'}) = \sqrt{\rho_A} J_{\cM}\sqrt{\rho_A} \in \SEP_{\bullet}(A:B)$ for any $\rho_A \in \cS(A)$ and $\cM \in \EBset$. We now prove the equation~\eqref{eq: Emax sigma information}. It has been shown in~\cite[Lemma 7]{Berta2017a} that
\begin{align}
  E_{\max}(\cN) = \log \min \Big\{\|\tr_B Y_{AB}\|_{\infty}\,\Big|\, J_{\cN} \leq Y_{AB},\, Y_{AB} \in \SEPcone(A:B)\Big\}.
\end{align}
Using the semidefinite representation of infinity norm and replacing $Y_{AB} = t J_{\cM}$, we have
\begin{align}
  E_{\max}(\cN) 
  & = \log \min \Big\{t\,\Big|\, J_{\cN} \leq t J_{\cM},\, J_{\cM} \in \SEPcone(A:B), \tr_B J_{\cM} \leq \1_A\Big\}.
\end{align}
By the definition of $D_{\max}$ and $\EBset$, we have
\begin{align}
  E_{\max}(\cN)  = \min_{\cM \in \EBset} D_{\max}(J_\cN\|J_\cM) = \min_{\cM \in \EBset} D_{\max}(\cN\|\cM)  = E_{\max,\Sigma}(\cN),
\end{align}
where the second equality follows from Eq.~\eqref{eq: Dmax channel divergence}.
\end{proof}

\begin{remark}
    The idea of considering the set of entanglement breaking channels also appears in~\cite[Theorem V.2]{Christandl2016}, where an upper bound of $E_{\max}(\cN)$ is given as $ E_{\max}(\cN) \leq B_{\max}(\cN)$ with
    \begin{align}
    B_{\max}(\cN)\equiv \min \big\{D_{\max}(J_{\cN}\|J_{\cM})\,\big|\, \cM\ \text{is an entanglement breaking quantum channel}\big\}.
    \end{align} 
    However, the key difference here is that $E_{\max,\Sigma}$ is minimizing over all the entanglement breaking \emph{subchannels} which is a strictly superset of entanglement breaking channels. Such extension is essential to get the equality $E_{\max}(\cN) = E_{\max,\Sigma}(\cN)$ instead of an upper bound.
\end{remark}

\vspace{0.2cm}
We further consider the Sigma-information induced by the geometric \Renyi divergence. Following a similar argument as Proposition~\ref{amortization proposition}, we can have the amortization property.

\begin{proposition}[Amortization]\label{amortization proposition priviate capacity}
    For any quantum state $\rho_{A'AB'}$ and quantum channel $\cN_{A\to B}$ and the parameter $\a \in (1,2]$, it holds
    \begin{align}
    \widehat E_\a(\omega_{A':BB'}) \leq \widehat E_\a(\rho_{A'A:B'}) + \widehat E_{\a,\Sigma}(\cN)\quad \text{with} \quad \omega_{A':BB'} = \cN_{A\to B}(\rho_{A'A:B'}).
\end{align}
\end{proposition}
\begin{proof}
The proof follows the same as Proposition~\ref{amortization proposition}. We only need to show that for any sub-normalized state $\sigma_{A'AB'} \in \SEP_{\bullet}(A'A:B')$ and map $\cM_{A\to B} \in \bcV_{\Sigma}$, it holds $\gamma_{A'BB'}\equiv \cM_{A\to B}(\sigma_{A'AB'}) \in \SEP_{\bullet}(A':BB')$. This can be checked as follows. First it is clear that $\tr \gamma_{A'BB'} \leq 1$ since both $\cM$ and  $\sigma$ are sub-normalized. Denote the tensor product decomposition  $\sigma_{A'AB'} = \sum_{i,j} \sigma^i_{A'A}\ox \sigma^j_{B'}$ and $J^{\cM}_{SB} = \sum_{k,\ell} J^k_{S}\ox J^\ell_B$. Let $\ket{\Phi}_{SA}$ be the unnormalized maximally entangled state. Then we have
\begin{align}
    \gamma_{A'BB'} 
    & = \big\<\Phi_{SA}\big|J^{\cM}_{SB} \ox \sigma_{A'AB'}\big|\Phi_{SA}\big\>
     = \sum_{i,j,k,\ell} \big\<\Phi_{SA}\big| \sigma^i_{A'A}\ox  J^k_{S}\,\big|\Phi_{SA}\big\> \ox \sigma^j_{B'}\ox  J^\ell_B,
\end{align}
where the r.h.s. belongs to $\SEPcone(A':BB')$. This completes the proof.
\end{proof}

\vspace{0.2cm}
Combining the amortization inequality and a similar argument in~\cite[Theorem IV.1.]{Christandl2016}, we can obtain an improved strong converse as follows:

\begin{theorem}[Main result 5]\label{thm: main result private assisted}
    For any quantum channel $\cN_{A'\to B}$ and $\a \in (1,2]$, it holds
        \begin{align}
            P^{\leftrightarrow}(\cN) \leq  P^{\leftrightarrow,\dagger}(\cN) \leq \widehat E_{\a,\Sigma}(\cN) \leq E_{\max}(\cN),
        \end{align}
    where $P^{\leftrightarrow}(\cN)$ and $P^{\leftrightarrow,\dagger}(\cN)$ denote the two-way assisted private capacity of channel $\cN$ and its corresponding strong converse capacity, respectively.
    Moreover, the bound $\widehat E_{\a,\Sigma}(\cN)$ with $\a(\ell) = 1+2^{-\ell}$ and $\ell \in \mathbb N$ can be given by a conic program
\begin{gather}
    \widehat E_{\a,\Sigma}(\cN) = 2^\ell\cdot \log \min \ y \quad \text{\rm s.t.}\quad \dbhbig{M,\{N_i\}_{i=0}^\ell,y},\, N_0 \in \SEPcone(A:B)\notag\\[2pt]
     \dbp{\begin{matrix}
        M & J_{\cN}\\
        J_{\cN} & N_{\ell}
    \end{matrix}},
    \left\{\dbp{\begin{matrix}
        J_{\cN} & N_{i} \\
        N_{i} & N_{i-1}
    \end{matrix}}\right\}_{i=1}^\ell,\, 
    \dbpbigg{\1_A - \tr_B N_0}, \dbpbigg{y\1 - \tr_B M},
\end{gather}
which reduces to a semidefinite program if the dimension  satisfies $|A||B| \leq 6$.
\end{theorem}

\newpage

\section{Classical communication}
\label{sec: Classical communication}

\subsection{Backgrounds}

The \emph{classical capacity} of a quantum channel is the maximum rate at which it can reliably transmit classical information over asymptotically many uses of the channel. Since classical messages are of major concern here, the communication assistance is usually given by the shared entanglement instead of the two-way classical communication discussed in the quantum/private communication scenarios.  The entanglement-assisted classical capacity has been completely solved as the mutual information of the channel~\cite{Bennett2002}, which is believed to be a natural counterpart in the classical Shannon theory. In this sense, shared entanglement simplifies the quantum Shannon theory. 

When it comes to the unassisted classical capacity, the best-known characterization is given by the classical capacity theorem, which states that the classical capacity of a quantum channel is given by its regularized Holevo information~\cite{Holevo1998,Schumacher1997},
\begin{align}\label{eq: classical channel coding theorem}
    C(\cN) = \lim_{n\to \infty} \frac{1}{n} \chi(\cN^{\ox n}) = \sup_{n\in \NN} \frac{1}{n} \chi(\cN^{\ox n}),
\end{align}
where $\chi(\cN) \equiv \max_{\mathscr E} \chi(\mathscr E, \cN) $ is the Holevo information with the maximization taken over all quantum state ensembles $\mathscr E = \{p_i, \rho_i\}$, $\chi(\mathscr E, \cN) \equiv H(\sum_i p_i\cN(\rho_i)) - \sum_i p_i H(\cN(\rho_i))$ is the Holevo information of the ensemble $\mathscr E$, and $H$ is the von Neumann entropy.
An impressive work by Hastings~\cite{Hastings2009} shows that the Holevo information is not additive in general, indicating the necessity of the regularization in \eqref{eq: classical channel coding theorem}. Moreover, as computing $\chi$ itself is already NP-complete~\cite{Beigi2007}, its regularized quantity for a general quantum channel is expected to be more difficult to evaluate. Even for the qubit amplitude damping channel, its unassisted classical capacity is still unknown~\cite{Wang2016g}. 

Deriving a single-letter expression for the classical capacity of a quantum channel remains a major open problem in quantum information theory. Several general converse bounds are given in~\cite{Leditzky2018} by an ``continuity argument'', extending the idea in~\cite{Sutter2014} from quantum capacity to classical capacity. However, those bounds are not known to be strong converse and typically work well only if the channel possesses certain structures, such as close to entanglement breaking channel or sufficiently covariant~\footnote{These are two main limitations of converse bounds established by using the continuity of the channel capacities.}. Two best-known strong converse bounds are given by $C_\b$ and $C_\zeta$ in~\cite{Wang2016g}, and both bounds are SDP computable. An attempt to improve the bound $C_\b$ is discussed in~\cite{wang2019converse} by a notion called Upsilon-information~($\U$), similar to the Theta-information and Sigma-information in the previous parts. 
However, a (weak) sub-additivity of the Upsilon-information induced by the sandwiched \Renyi divergence is required for showing $\U$ as a strong converse for general quantum channels. This sub-additivity was only proved in~\cite{wang2019converse} for covariant channels while the general case was left open.

\subsection{Summary of results}

In this part, we aim to push forward the analysis in~\cite{wang2019converse} by considering the geometric \Renyi divergence and improve both of the two strong converse bounds $C_\b$ and $C_\zeta$ in general. The structure of this part is organized as follows (see also a schematic diagram in Figure~\ref{fig: classical communiation summary}).

In Section~\ref{sec: Maximal Renyi Upsilon information}, we first study the generalized Upsilon-information induced by the max-relative entropy ($\U_{\max}$) and prove that it is no greater than $C_\b$ and $C_\zeta$ in general. We then discuss the generalized Upsilon-information induced by the geometric \Renyi divergence~($\widehat \U_{\a}$) and show that it is a strong converse on classical capacity by proving its sub-additivity. Due to the relation that $\widehat D_{\a} \leq D_{\max}$, we have $\widehat \U_\a \leq \U_{\max}$. Then we have an improved strong converse bound $\widehat \U_\a$ satisfying
\begin{align*}
C(\cN) \leq C^\dagger(\cN) \leq \widehat \U_\a(\cN) \leq \min\big\{C_\b(\cN),C_\zeta(\cN)\big\} \quad \text{with}\ \widehat \U_{\a}(\cN)\ \text{SDP computable},
\end{align*}
where $C(\cN)$ and $C^{\dagger}(\cN)$ denote the unassisted classical capacity of channel $\cN$ and its corresponding strong converse capacity, respectively.

In Section~\ref{sec: classical capacity Examples},
we investigate several fundamental quantum channels, demonstrating the efficiency of our new strong converse bounds. It turns out that our new bounds work exceptionally well and exhibit a significant improvement on previous results for almost all cases.

\begin{figure}[H]
    \centering
\begin{tikzpicture}
\draw[fill=gray!10,opacity=0.3] (7,0.5) rectangle node[gray,opacity=1,midway,shift={(0,-1.2)}] {\small SDP computable} (13,-2.5);
\node (C) at (0,0) {$C$};
\node[circle,fill=magenta!10,inner sep=2pt,minimum size=2pt] (Cdagger) at (2.5,0) {$C^\dagger$};
\node (U) at (5,-2) {$\U$};
\node[red] (Ua) at (7.5,0) {$\widehat \U_\a$};
\node[red] (Umax) at (10,0) {$\U_{\max}$};
\node (Cbeta) at (12.5,0) {$C_\b$};
\node (Cxi) at (12.5,-2) {$C_\zeta$};
\node (chi)at (1,-2) {$\chi$};

\draw[thick,->] (C) -- (chi);
\draw[very thick,->,dotted,gray] (U) -- node[black,midway,rotate=-42,shift={(0,0.2)}] {\footnotesize cov.} node[black,rotate=-42,midway,shift={(0,-0.2)}] {\scriptsize \cite{wang2019converse}} (Cdagger);
\draw[thick,->] (Umax) -- node[midway,shift={(0.1,0.2)}] {\footnotesize Lem.~\ref{thm: divergence chain inequality}} node[midway,shift={(0.1,-0.2)}] {\footnotesize $\boldsymbol\neq$}  (Ua);
\draw[thick,->] (Cbeta) -- node[midway,shift={(0.1,0.2)}] {\footnotesize Prop.~\ref{prop: comparison of Umax cbeta cxi}} node[midway,shift={(0,-0.2)}] {\footnotesize $\boldsymbol\neq$} (Umax);
\draw[thick,->] (Cxi) -- node[rotate = -40,midway,shift={(0.15,0.2)}] {\footnotesize Prop.~\ref{prop: comparison of Umax cbeta cxi}} node[rotate = -43,midway,shift={(0.1,-0.2)}] {\footnotesize $\boldsymbol\neq$} (Umax);
\draw[thick,->] (Ua) -- (7.5,-2) node[rotate = 90,midway,shift={(0,0.2)}] {\footnotesize Lem.~\ref{thm: divergence chain inequality}} node[rotate=90, midway,shift={(0,-0.2)}] {\footnotesize $\boldsymbol\neq$} -- (U);
\draw[thick,->] (U) -- node[midway,shift={(0.1,0.2)}] {\scriptsize \cite{wang2019converse}} (chi);
\draw[thick,->] (Cdagger) -- (C);
\draw[red,very thick,->] (Ua)  -- node[black,midway,shift={(0,0.2)}] {\footnotesize Thm.~\ref{thm: main result classical}} (Cdagger);
\draw[thick] (Cxi) -- (13.2,-2) -- (13.2,0) -- (Cbeta);
\draw[thick,->] (13.2,-1) -- (13.5,-1) -- (13.5,1.5) -- node[midway,shift={(0.2,0.2)}] {\scriptsize \cite{Wang2016g}}(2.5,1.5) -- (Cdagger);
\end{tikzpicture}
\caption{\small Relations between different converse bounds for classical communication. $C$ and $C^{\dagger}$ are the classical capacity and the strong converse capacity, respectively. $\U$, $\widehat \U_{\a}$ and $\U_{\max}$ are the generalized Upsilon-information induced by different quantum divergences. $C_\b$ and $C_\zeta$ are the SDP strong converse bounds in~\cite{Wang2016g}. $\chi$ is the Holevo information. The circled quantities is the one of particular interest in quantum information theory. The key quantities and the main contributions in this section are marked in red. The quantity at the start point of an arrow is no smaller than the one at the end point. The inequality sign represents two quantities are not the same in general. The dotted arrow represents that the relation holds under certain restrictions, where ``cov.'' stands for ``covariant''. The parameter $\a$ is taken in the interval $(1,2]$. The quantities in the shaded area are SDP computable in general.}
\label{fig: classical communiation summary}
\end{figure}

\subsection{Unassisted classical capacity}
\label{sec: Maximal Renyi Upsilon information}

In this section we discuss converse bounds on the unassisted classical capacity of a quantum channel~\footnote{We refer to~\cite[Section IV.A]{Wang2016g} for rigorous definitions of the unassisted classical capacity and its strong converse.}.

A quantum channel is called \emph{constant channel} or \emph{replacer channel} if it always output a fixed quantum state, i.e., there exists $\sigma_B \in \cS(B)$ such that $\cN_{A\to B}(\rho_A) = \sigma_B$ for all $\rho_A \in \cS(A)$. Unlike quantum or private communication where the sets of useless channels are not completely determined yet, the useless set for classical communication is fully characterized by the set of constant channels. That is, $C(\cN) = 0$ if and only if $\cN$ is a constant channel~\footnote{This can be easily seen from the radius characterization of the Holevo capacity $\chi(\cN) = \min_{\sigma} \max_{\rho} D(\cN(\rho)\|\sigma)$~\cite{ohya1997capacities}.}. As a natural extension, the work~\cite{wang2019converse} proposed to consider the set of \emph{constant-bounded subchannels},
\begin{align}
    \bcV_{cb}\equiv\big\{\cM \in \CP(A:B)\,\big|\, \exists\, \sigma_B \in \cS(B)\ \text{s.t.}\ \cM_{A\to B}(\rho_A)\leq \sigma_B, \forall \rho_A\in \cS(A)\big\}.
\end{align}
It seems not easy to find a semidefinite representation for the set $\bcV_{cb}$ directly. Thus a restriction of $\bcV_{cb}$ was given in~\cite{wang2019converse} as
\begin{gather}
    \bcV_\b  \equiv \big\{\cM\in \CP(A:B)\,\big|\, \beta(J_{\cM})\leq 1\big\}\ \ \text{with}\ \notag\\[2pt]
    \beta(J_{\cM}) \equiv \min\left\{\tr S_B\,\Big|\, 
        R_{AB} \pm J_{\cM}^{\sfT_B} \geq 0,\,
         \1_A \ox S_B \pm R_{AB}^{\sfT_B} \geq 0\,\right\}.\label{beta set definition}
\end{gather}
This subset can be seen as the zero set~\footnote{It makes no difference by considering $\b(J_{\cM}) \leq 1$ or $\b(J_{\cM}) = 1$.} of the strong converse bound $C_\b(\cN) \equiv \log \b(J_{\cN})$~\cite{Wang2016g}. As discussed  in Appendix~\ref{app: A complete hierarchy for the set of constant-bounded maps}, we will see that $\bcV_{cb}$ can be approximated by a complete semidefinite hierarchy, where the subset $\bcV_\b$ can be considered as a symmetrized version of its first level.
In the following, we will proceed our analysis, without loss of generality, over the set $\bcV_\b$. A more detailed discussion of $\bcV_{cb}$ can be found in Appendix~\ref{app: A complete hierarchy for the set of constant-bounded maps}.

\begin{definition}[Upsilon-info.~\cite{wang2019converse}]
For any generalized divergence $\bD$, the generalized Upsilon-information of a quantum channel $\cN_{A'\to B}$ with respect to the set $\bcV_\b$ is defined as
\begin{align}
    \mathbf \U(\cN)\equiv  \min_{\cM \in \bcV_\b}\bD(\cN\|\cM) = \min_{\cM \in \bcV_\b} \max_{\rho_A\in \cS(A)} \bD(\cN_{A'\to B} (\phi_{AA'})\|\cM_{A'\to B}(\phi_{AA'})),
\end{align}
where $\phi_{AA'}$ is a purification of quantum state $\rho_A$.
\end{definition}
As mentioned in Remark~\ref{Rains theta swap}, the min and max in the above definition can be swapped.

\vspace{0.2cm}
Let us first consider the generalized Upsilon-information induced by the max-relative entropy $\U_{\max}$.

\begin{proposition}
    For any quantum channel $\cN$, the generalized Upsilon-information induced by the max-relative entropy $\U_{\max}(\cN)$ is given as an SDP,
    \begin{align}
        \U_{\max}(\cN)  = \log \min \Big\{\tr S_B\,\Big|\, J_{\cN} \leq K_{AB},\,R_{AB} \pm K_{AB}^{\sfT_B} \geq 0,\,
        \1_A \ox S_B \pm R_{AB}^{\sfT_B} \geq 0\,\Big\}.
    \end{align}
\end{proposition}
\begin{proof}
By definition we have
$\U_{\max}(\cN)\equiv \min_{\cM \in \boldsymbol{\cV_\beta}} D_{\max}(\cN\|\cM) = \min_{\cM \in \bcV_\beta} D_{\max}(J_{\cN}\|J_{\cM})$,
where the second equality follows from Eq.~\eqref{eq: Dmax channel divergence}.
Then it is clear that 
\begin{align}
\U_{\max}(\cN) & = \log \min \Big\{ t\, \Big|\, J_{\cN} \leq t J_{\cM},\,
\tr G\leq 1,\,W \pm J_{\cM}^{\sfT_B} \geq 0,\,\1 \ox G \pm W^{\sfT_B} \geq 0\Big\}.
\end{align}
Replacing $K = t J_{\cM}$, $S = t G$ and $R = tW$, we have the desired result.
\end{proof}

\vspace{0.2cm}
Besides the bound $C_\b$, there is another SDP strong converse bound given in~\cite{Wang2016g} as
\begin{align}
    C_\zeta(\cN) \equiv \log \min \left\{\tr S_B \Big|\, J_{\cN} \leq K_{AB},\, \1_A \ox S_B \pm K_{AB}^{\sfT_B} \geq 0\,\right\}.
\end{align}
We can show that $\U_{\max}$ is no greater than both of these quantities in general.

\begin{proposition}\label{prop: comparison of Umax cbeta cxi}
For any quantum channel $\cN$, it holds
$\U_{\max}(\cN) \leq \min\big\{C_\b(\cN),C_\zeta(\cN)\big\}$.
\end{proposition}
\begin{proof}
The result is clear by comparing their SDP formulas. Specifically, by restricting $K_{AB} = J_{\cN}$ in $\U_{\max}$, we can retrieve $C_\b$. By restricting $R_{AB} = \1_A \ox S_B$ in $\U_{\max}$, we can retrieve $C_\zeta$.
\end{proof}

\vspace{0.2cm}
We further consider the generalized Upsilon-information induced by the geometric \Renyi divergence $\widehat \U_\a$. The following sub-additivity is a key ingredient to prove the strong converse of $\widehat \U_\a$ in Theorem~\ref{thm: main result classical}.

\begin{proposition}[Sub-additivity]\label{prop: sub-additivity renyi upsilon}
    For any quantum channels $\cN_1$, $\cN_2$ and $\a \in (1,2]$, it holds
    \begin{align}
        \widehat \U_{\a}(\cN_1\ox \cN_2) \leq \widehat \U_{\a}(\cN_1) + \widehat \U_{\a}(\cN_2).
    \end{align}
\end{proposition}
\begin{proof}
This is a direct consequence of the additivity of the geometric \Renyi channel divergence in Lemma~\ref{lem: maximal Renyi channel divergence additivity} and the sub-additivity of the quantity $\b(\cdot)$ in~\eqref{beta set definition}. More specifically, suppose the optimal solution of $\widehat \U_\a(\cN_1)$ and $\widehat \U_\a(\cN_2)$ are taken at $\cM_1 \in \bcV_\b^1$ and $\cM_2 \in \bcV_\b^2$ respectively. Then we can check that $\cM_1\ox \cM_2 \in \bcV_\b^{12}$ which is a feasible solution for $\widehat \U_\a(\cN_1\ox \cN_2)$. Thus we have 
\begin{align}
    \widehat \U_\a(\cN_1\ox \cN_2) \leq \widehat D_\a(\cN_1\ox \cN_2 \|\cM_1\ox \cM_2) = \widehat D_\a(\cN_1\|\cM_1) + \widehat D_\a(\cN_2\|\cM_2) = \widehat \U_\a(\cN_1) + \widehat \U_\a(\cN_2),\notag
\end{align}
where the last inequality follows by the optimality assumption of $\cM_1$ and $\cM_2$. 
\end{proof}

\vspace{0.2cm}
Based on the sub-additivity, we are ready to show our improved strong converse bound.

\begin{theorem}[Main result 6]\label{thm: main result classical}
    For any quantum channel $\cN$ and $\a \in (1,2]$, it holds
    \begin{align}
        C(\cN) \leq C^\dagger(\cN) \leq \widehat \U_\a(\cN) \leq \U_{\max}(\cN) \leq \min\big\{C_\b(\cN),C_\zeta(\cN)\big\},
    \end{align}
    where $C(\cN)$ and $C^{\dagger}(\cN)$ denote the unassisted classical capacity of channel $\cN$ and its corresponding strong converse capacity, respectively.
\end{theorem}
\begin{proof}
    The first inequality holds by definition. The third inequality follows from Lemma~\ref{thm: divergence chain inequality}. The last inequality was proved in Proposition~\ref{prop: comparison of Umax cbeta cxi}. It remains to prove the second inequality $C^\dagger(\cN) \leq \widehat \U_\a(\cN)$. For any classical communication protocol with a triplet $(r,n,\ve)$, it holds by a standard argument~\cite[Proposition 20]{wang2019converse} that
    \begin{align}\label{eq: main result 3 proof tmp 1}
        1 - \ve \leq 2^{-n\left(\frac{\a-1}{\a}\right)\left[r-\frac{1}{n}\widetilde \U_\a(\cN^{\ox n})\right]},
    \end{align}
    where $\widetilde \U_\a$ is the Upsilon information induced by the sandwiched R\'{e}nyi divergence $\widetilde D_\a$.
    Due to the sub-additivity of $\widehat \U_\a$ in Proposition~\ref{prop: sub-additivity renyi upsilon} and the inequality in Lemma~\ref{thm: divergence chain inequality}, we have
    \begin{align}\label{eq: main result 3 proof tmp 2}
    n \widehat \U_\a(\cN) \geq \widehat \U_\a(\cN^{\ox n}) \geq \widetilde \U_\a(\cN^{\ox n}).
    \end{align}
Combining~\eqref{eq: main result 3 proof tmp 1} and~\eqref{eq: main result 3 proof tmp 2}, we have
\begin{align}
    1 - \ve \leq  2^{-n\left(\frac{\a-1}{\a}\right)\left[r-\frac{1}{n}\widetilde \U_\a(\cN^{\ox n})\right]} \leq  2^{-n\left(\frac{\a-1}{\a}\right)\left[r-\widehat \U_\a(\cN)\right]}.
\end{align}
This implies that if the communication rate $r$ is strictly larger than $\widehat \U_\a(\cN)$, the success probability of the transmission $1-\ve$ decays exponentially fast to zero as the number of channel use $n$ increases. Or equivalently, we have the strong converse inequality $C^\dagger(\cN) \leq \widehat \U_\a(\cN)$ and completes the proof.
\end{proof}

\vspace{0.2cm}
Finally, we present how to compute the geometric \Renyi Upsilon information.
\begin{proposition}[SDP formula]
For any quantum channel $\cN_{A'\to B}$ and $\a(\ell) = 1+2^{-\ell}$ with $\ell \in \mathbb N$, the geometric \Renyi Upsilon information can be computed by the following SDP:
\begin{gather}
    \widehat \U_\a(\cN)= 2^\ell\cdot \log \min \ y \quad \text{\rm s.t.}\quad \dbhbig{M, \{N_i\}_{i=0}^\ell,R,S,y}, \ \dbpbig{y\1 - \tr_B M} \notag\\[2pt]
     \dbp{\begin{matrix}
        M & J_{\cN}\\
        J_{\cN} & N_{\ell}
    \end{matrix}},
    \left\{\dbp{\begin{matrix}
        J_{\cN} & N_{i} \\
        N_{i} & N_{i-1}
    \end{matrix}}\right\}_{i=1}^\ell, 
    \dbp{R \pm N_0^{\sfT_B}},
    \dbp{\1\ox S \pm R^{\sfT_B}},
     \dbpBig{1 - \tr S}.\label{Renyi Upsilon information SDP min}
\end{gather}
\end{proposition}
\begin{proof}
This directly follows from Lemma~\ref{lem: SDP representation of the channel information measure} and the definition of the set $\bcV_\b$ in~\eqref{beta set definition}.
\end{proof}

\subsection{Examples}
\label{sec: classical capacity Examples}

In this section, we study several fundamental quantum channels as well as their compositions. We use these toy models to test the performance of our new strong converse bounds, demonstrating the improvement on the previously known results.  

Consider the depolarizing channel $\cD_p$ defined in~\eqref{DP channel definition}, the erasure channel $\cE_p$ defined in~\eqref{ER channel definition} and the dephrasure channel $\cN_{p,q}(\rho)\equiv (1-q)[(1-p)\rho+pZ\rho Z] + q \tr(\rho)\ket{e}\bra{e}$, where $\ket{e}$ is an erasure flag orthogonal to the input Hilbert space. Since these channels are covariant with respect to the unitary group, their Upsilon informations are known as strong converse bounds~\cite[Proposition 20]{wang2019converse} and can be computed via the algorithm in~\cite{fawzi2018efficient,Fawzi2017}. As for the generalized amplitude damping (GAD) channel $\cA_{\gamma,N}$ defined in~\eqref{GAD definition}, its Upsilon information is not known as a valid converse bound.

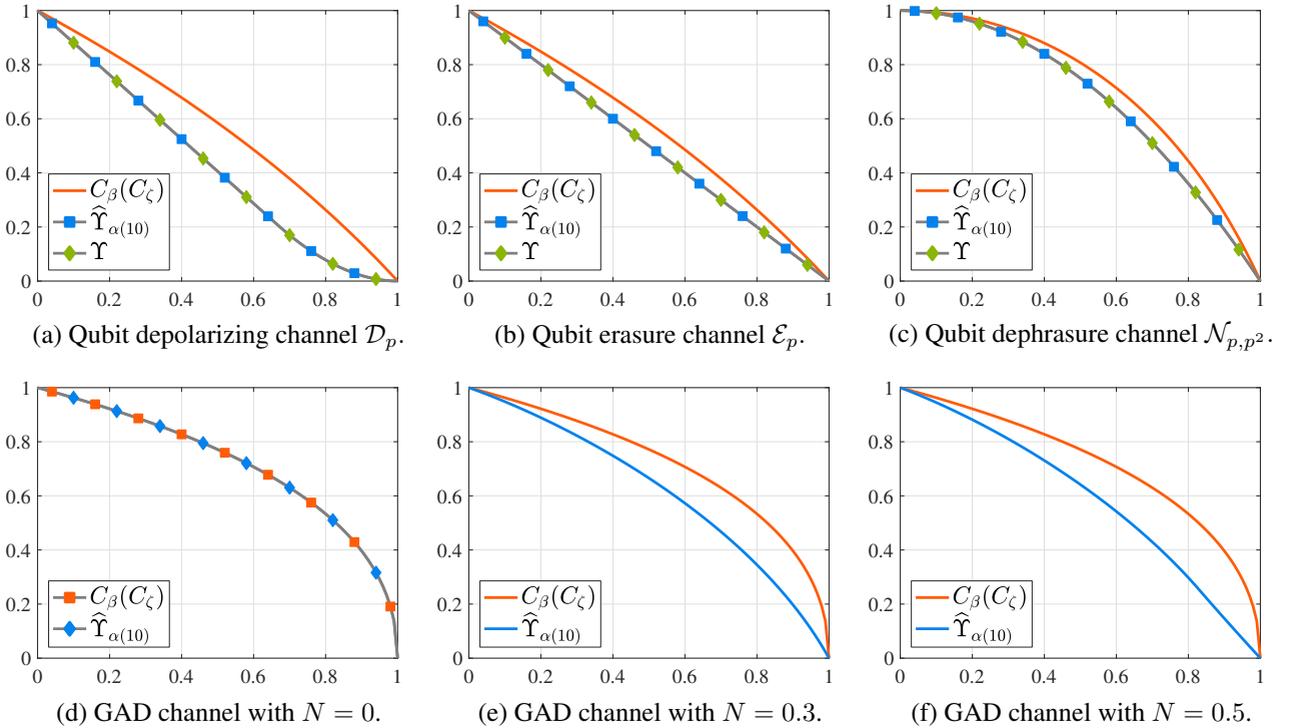
\begin{figure}[H]
\centering
\begin{adjustwidth}{-0.5cm}{0cm}
\begin{tikzpicture}
\node at (-5.7,0) {\includegraphics[width = 5.3cm]{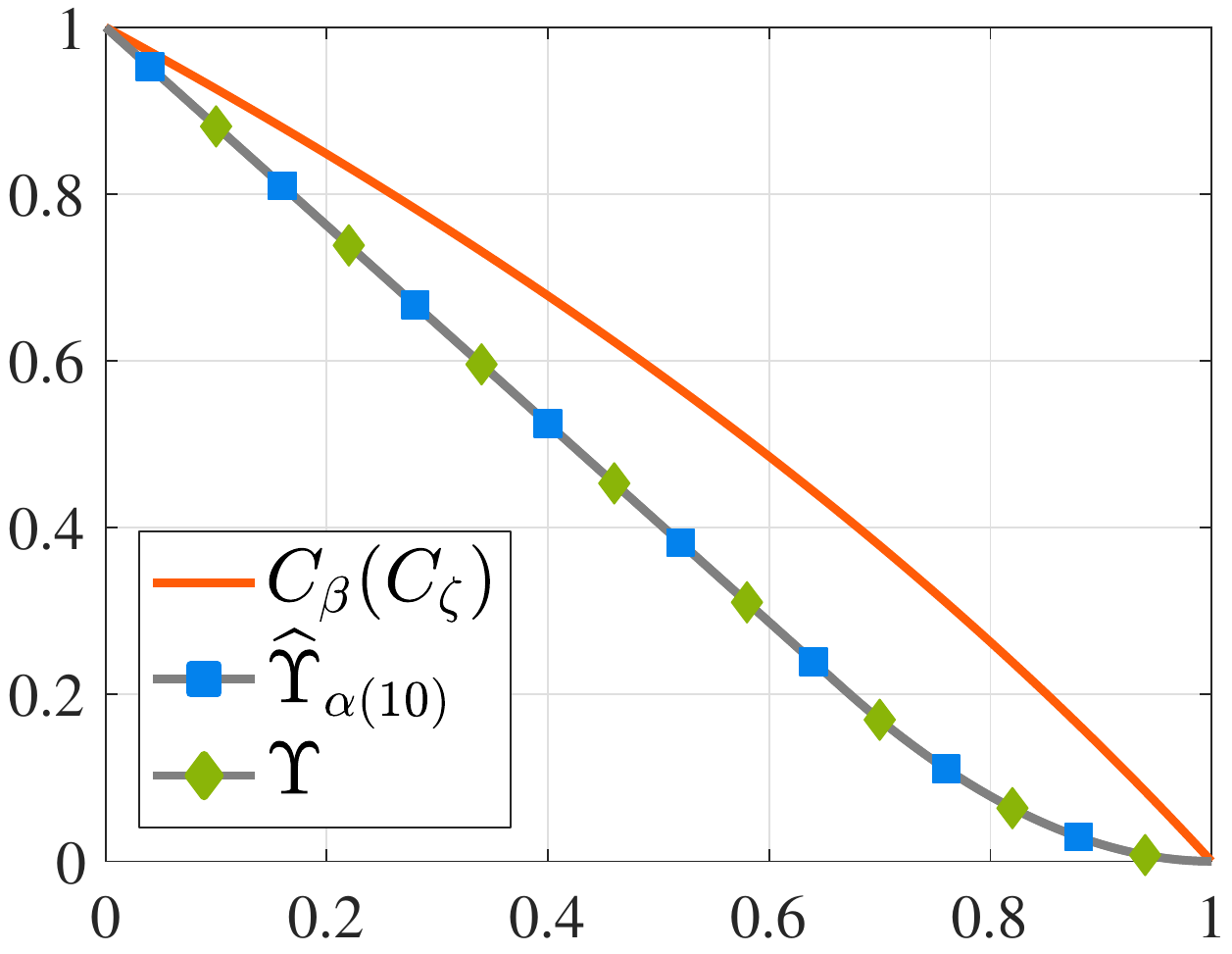}};
\node at (-5.5,-2.4) {\small (a) Qubit depolarizing channel $\cD_p.$};

\node at (0,0) {\includegraphics[width = 5.3cm]{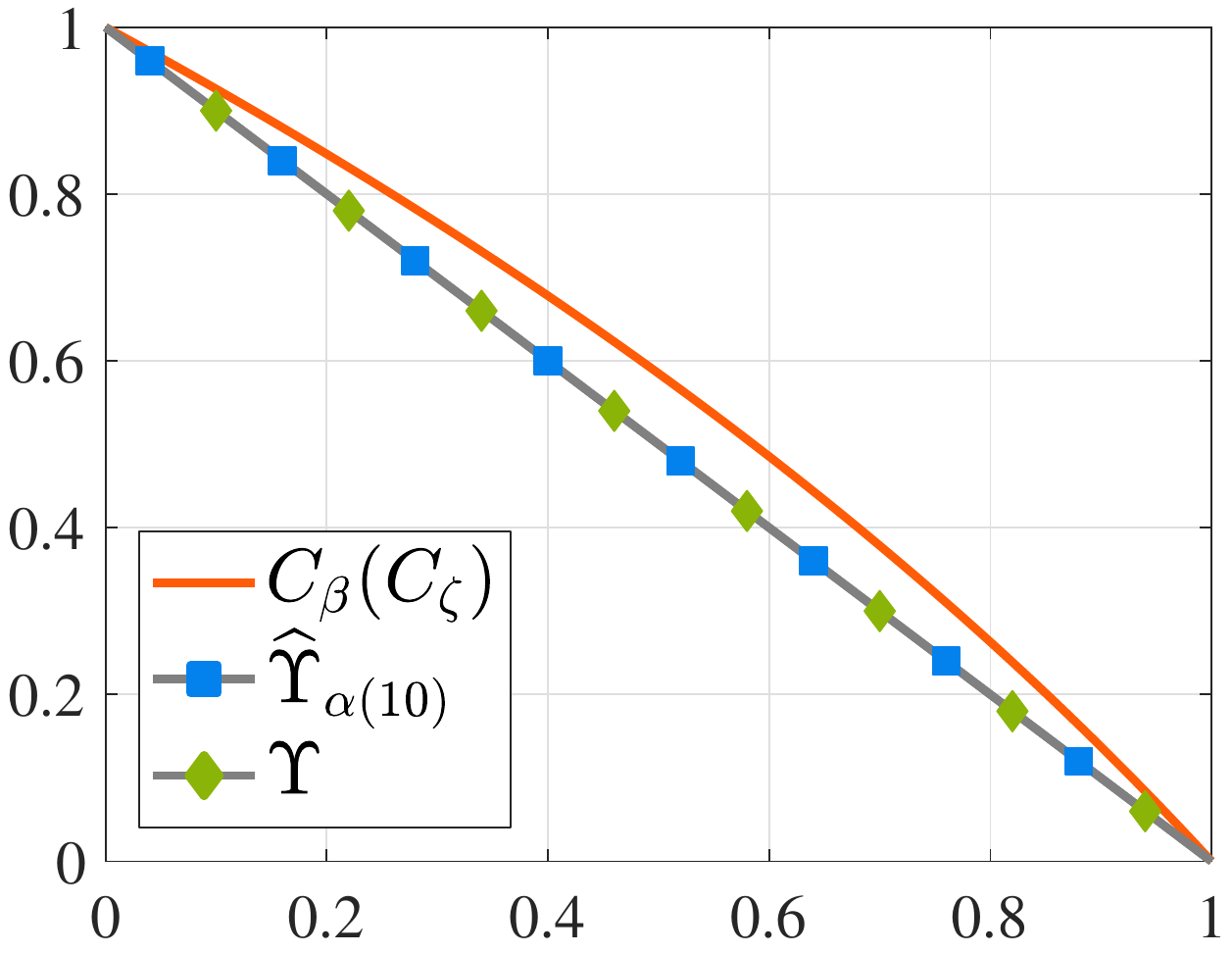}};
\node at (0.2,-2.4) {\small (b) Qubit erasure channel $\cE_p.$};

\node at (5.7,0) {\includegraphics[width = 5.3cm]{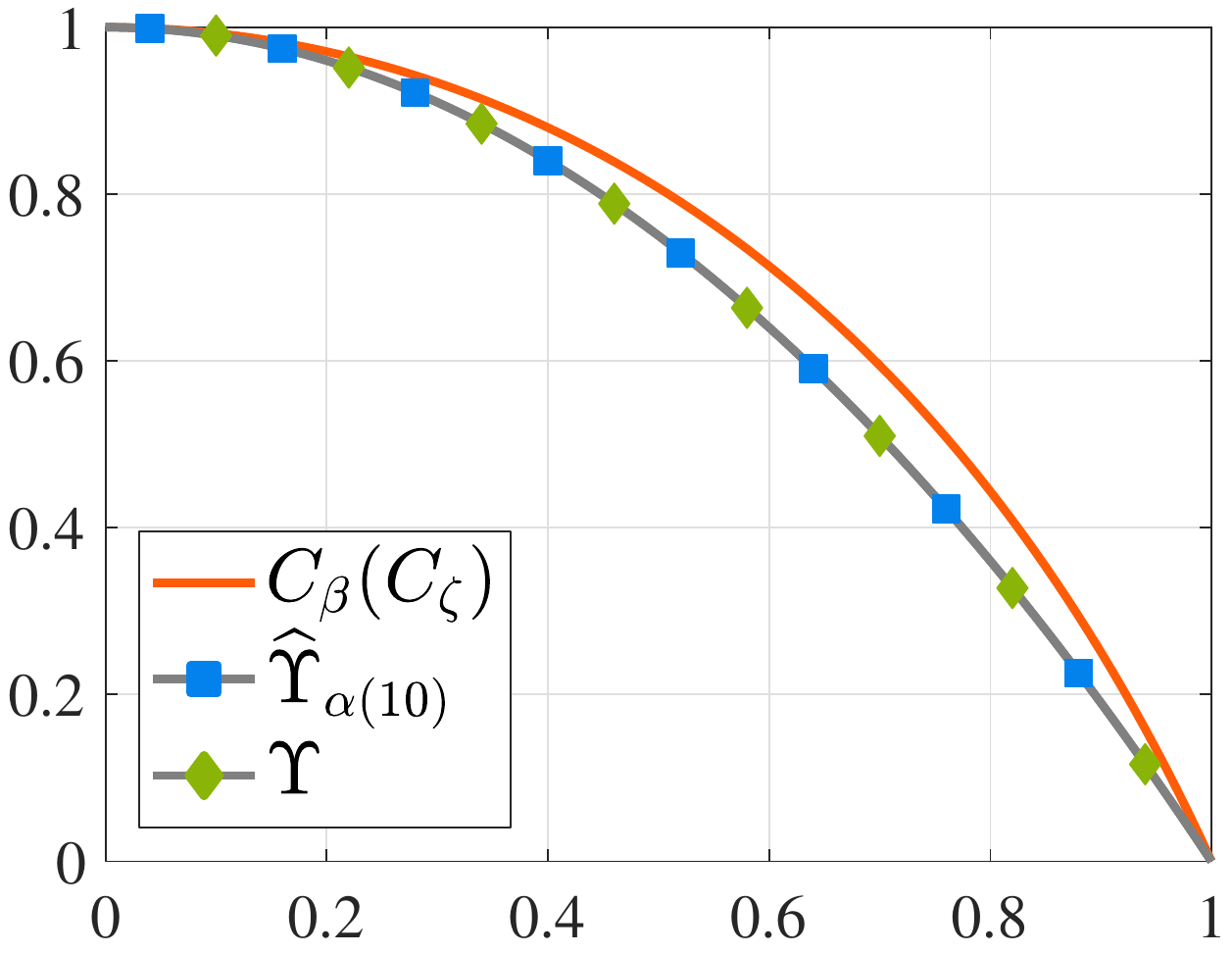}};
\node at (5.9,-2.4) {\small (c) Qubit dephrasure channel $\cN_{p,p^2}.$};

\begin{scope}[shift={(0,-5)}]
\node at (-5.7,0) {\includegraphics[width = 5.3cm]{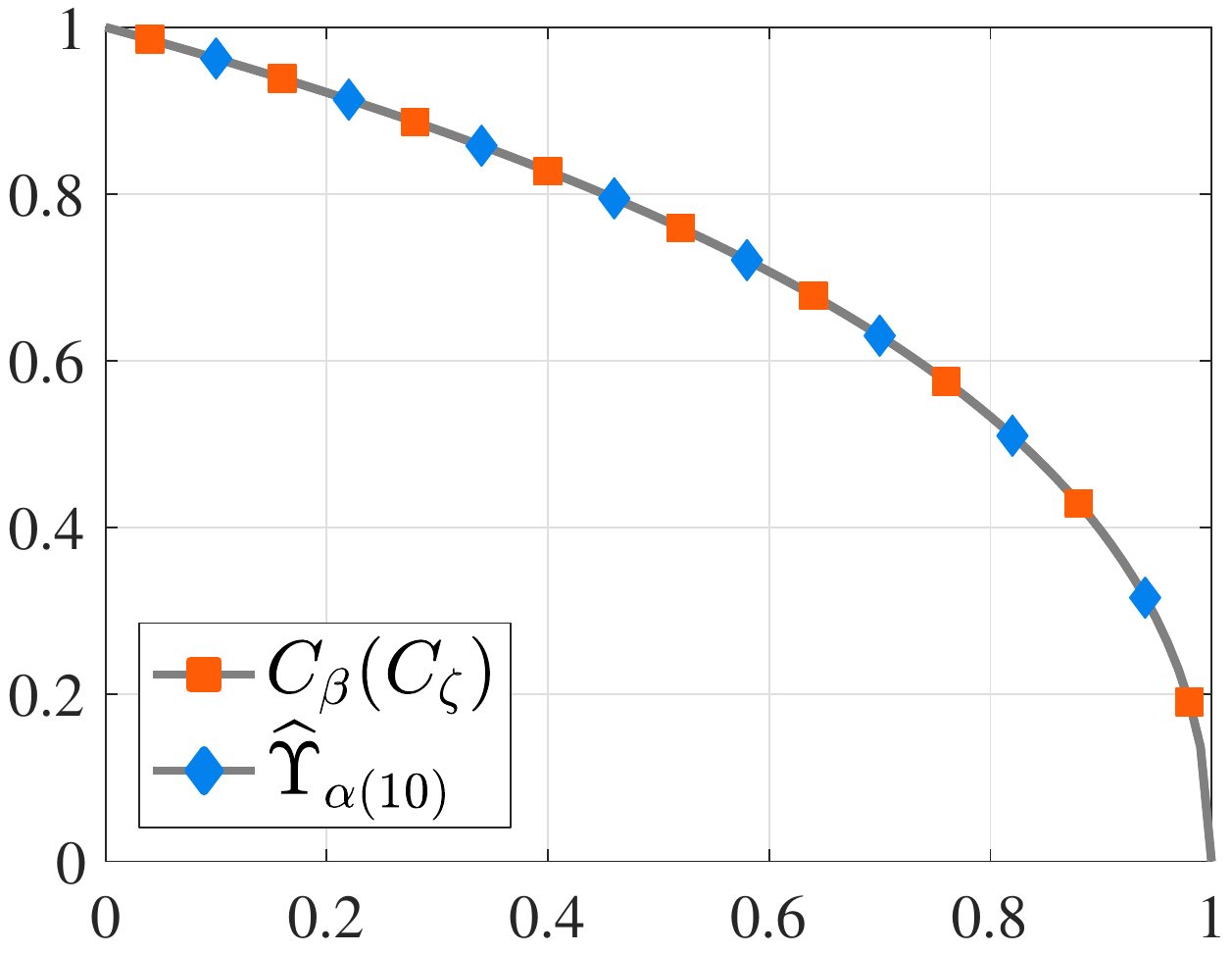}};
\node at (-5.5,-2.4) {\small (d) GAD channel with $N = 0$.};

\node at (0,0) {\includegraphics[width = 5.3cm]{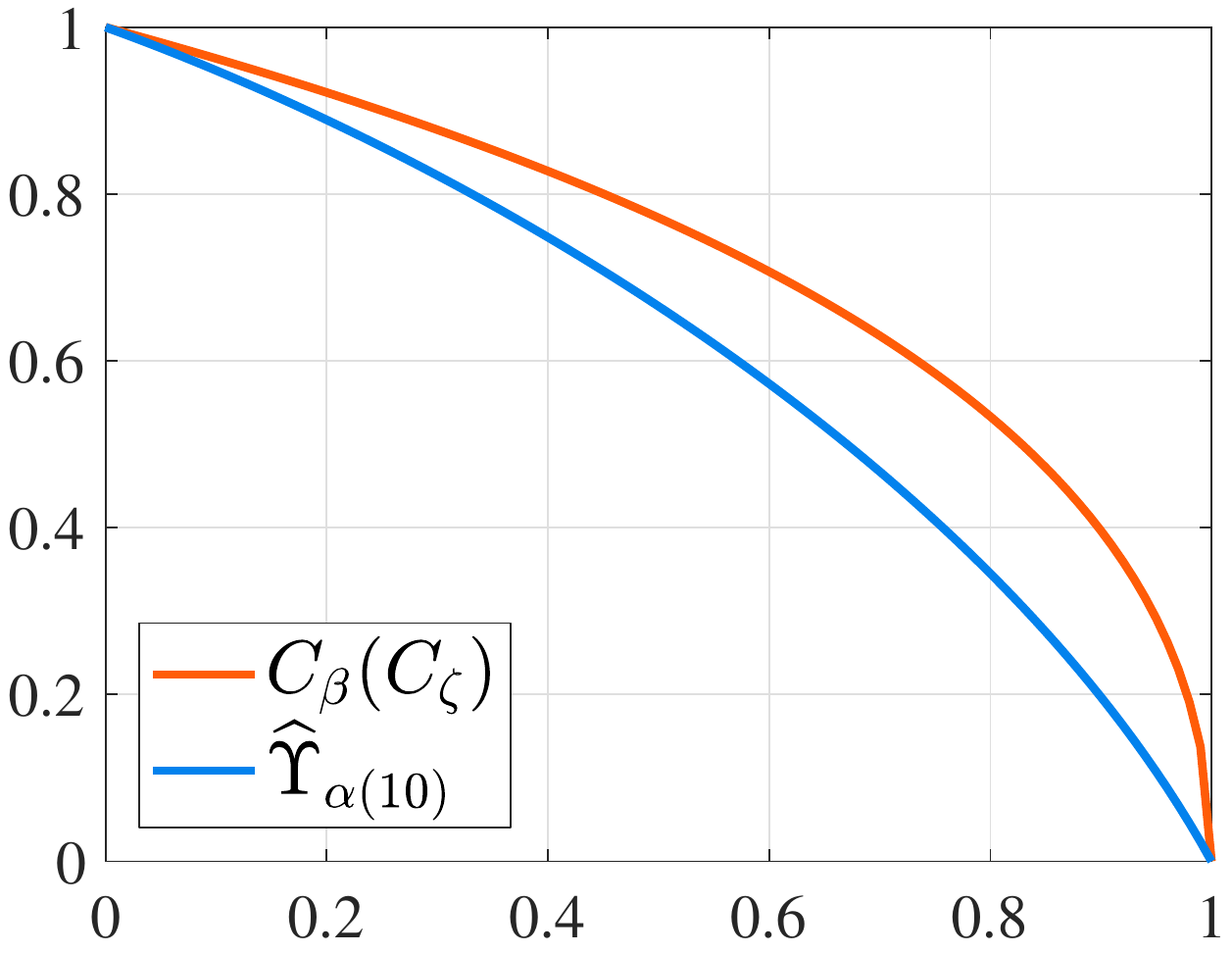}};
\node at (0.2,-2.4) {\small (e) GAD channel with $N = 0.3$.};

\node at (5.7,0) {\includegraphics[width = 5.3cm]{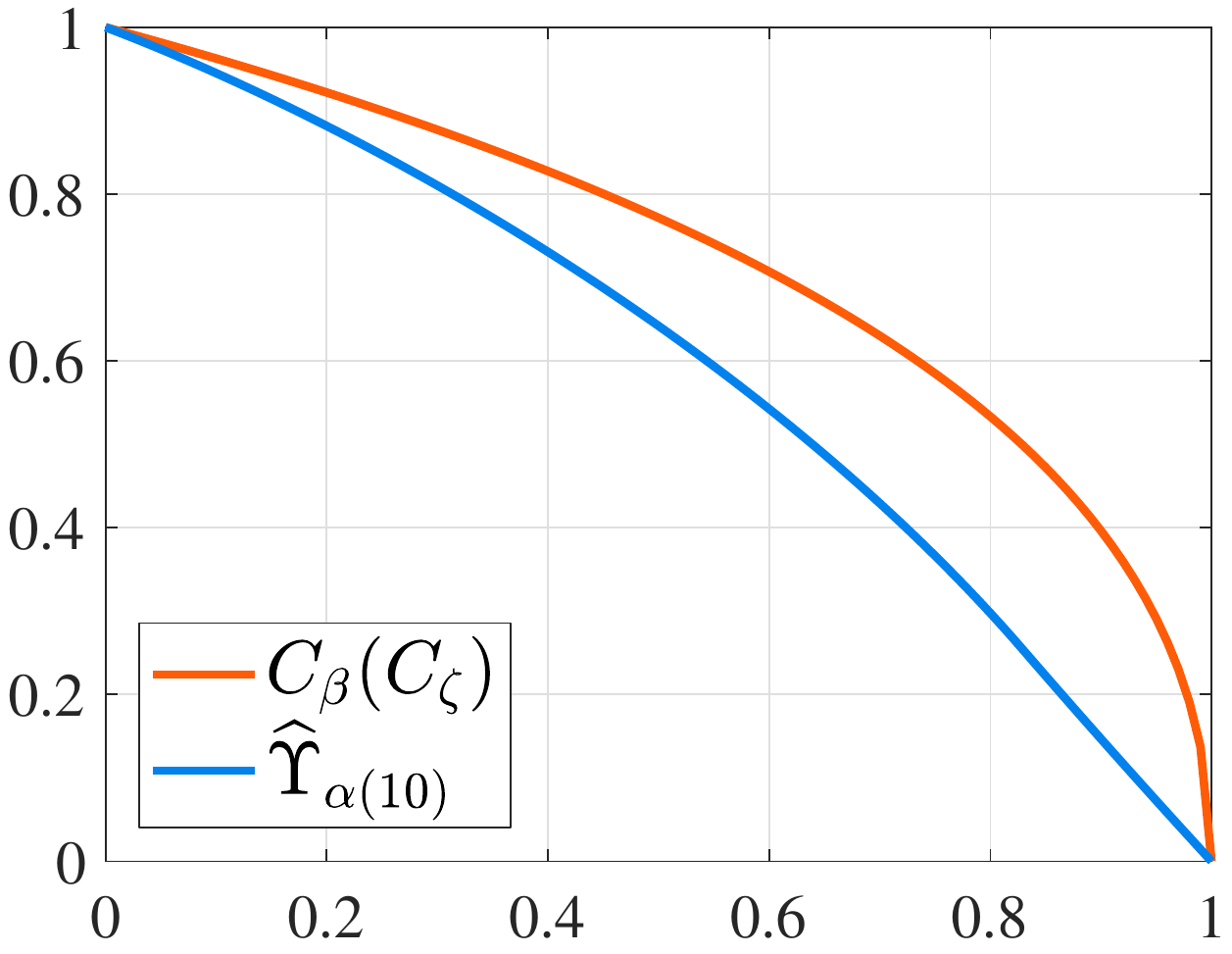}};
\node at (5.9,-2.4) {\small (f) GAD channel with $N = 0.5$.};
\end{scope}
\end{tikzpicture}
\end{adjustwidth}
\caption{\small Comparison of the strong converse bounds on the classical capacity for the qubit depolarizing channel $\cD_p$, the qubit erasure channel $\cE_p$, the qubit dephrasure channel $\cN_{p,q}$ with $q = p^2$, and the generalized amplitude damping channels $\cA_{p,N}$ with different parameters. Note that $C_\beta$ and $C_\zeta$ coincide for all these examples. The horizontal axis takes value of $p \in [0,1]$.}
\label{classical capacity compare 1}
\end{figure}

The comparison results~\footnote{A detailed comparison of the GAD channels with other weak converse bounds in~\cite{Khatri2019} is given in Appendix~\ref{app: Detailed comparison for generalized amplitude damping channel}.} are shown in Figure~\ref{classical capacity compare 1}. It is clear that $\widehat \U_{\a {\scriptscriptstyle (10)}}$ demonstrates significant improvements over $C_\b$ and $C_{\zeta}$ for all these channels except for one particular case $\cA_{p,0}$ in subfigure~(d) where all bounds coincide.  It is interesting to note that an analytical expression of the bounds $C_\beta(\cA_{\gamma,N}) = C_\zeta(\cA_{\gamma,N}) = \log (1+\sqrt{1-\gamma})$ is given in~\cite[Proposition 6]{Khatri2019}, which is independent on the parameter $N$. However, this is clearly not the case for our new bound $\widehat \U_\a$.  For covariant channels $\cD_p$, $\cE_p$ and $\cN_{p,p^2}$, the bound $\widehat \U_{\a {\scriptscriptstyle (10)}}$ also coincides with the Upsilon information $\U$ in subfigures (a-c). In particular, $\widehat \U_{\a {\scriptscriptstyle (10)}}$ is given by $1-p$ in subfigure (b), witnessing again the strong converse property of the qubit erasure channel $C(\cE_p) = C^\dagger(\cE_p) = 1-p$~\cite{Wilde2014b}. Such tightness can also be observed here for the dephrasure channel $\cN_{p,q}$ and it would be easy to show that $\chi(\cN_{p,q}) = C(\cN_{p,q}) = C^{\dagger}(\cN_{p,q}) = 1-q$ which is independent of the dephasing noise parameter $p$. 

From Figure~(\ref{classical capacity compare 1}\,e), $\widehat \U_\a$ does not give improvement for the amplitude damping channel $\cA_{\gamma,0}$. However, when considering the composition channel $\cM_p \equiv \cA_{p,0}\circ \cZ_p$, which was studied by Aliferis et al.~\cite{Aliferis2009} in the context of fault-tolerant quantum computation, the strong converse bound $\widehat \U_\a$ works unexpectedly well, as shown in Figure~\ref{classical capacity compare 3} . 

Since $\cM_p$ is an entanglement-breaking channel at $p = 1/2$, it is expected that $\cM_p$ is approximately entanglement-breaking around this point. Therefore, a converse bound $C_{\rm EB}$ was established in~\cite[Corollary III.7]{Leditzky2018} by using a continuity argument of the classical capacity. It has been shown in~\cite[Figure 5]{Leditzky2018} that this continuity bound $C_{\rm EB}$ gives certain improvement on $C_\b$ for the interval round $p = 1/2$. However, Figure~\ref{classical capacity compare 3} shows that the new strong converse bound $\widehat \U_{\a{\scriptscriptstyle (10)}}$ is much tighter than both $C_{\rm EB}$ and $C_\b$ for all $\cM_p$ with $p\in[0,0.75]$.
The Holevo information $\chi$ is also numerically computed by utilizing the algorithm~\footnote{The MATLAB codes we use are given from~\cite{Leditzky2018}.} in~\cite{Sutter2016b}. We observe that the upper bound $\widehat \U_{\a {\scriptscriptstyle (10)}}$ and the lower bound $\chi$ are very close, leading to a good estimation to the classical capacity of $\cM_p$.

\begin{figure}[H]
\centering
\includegraphics[width=10cm]{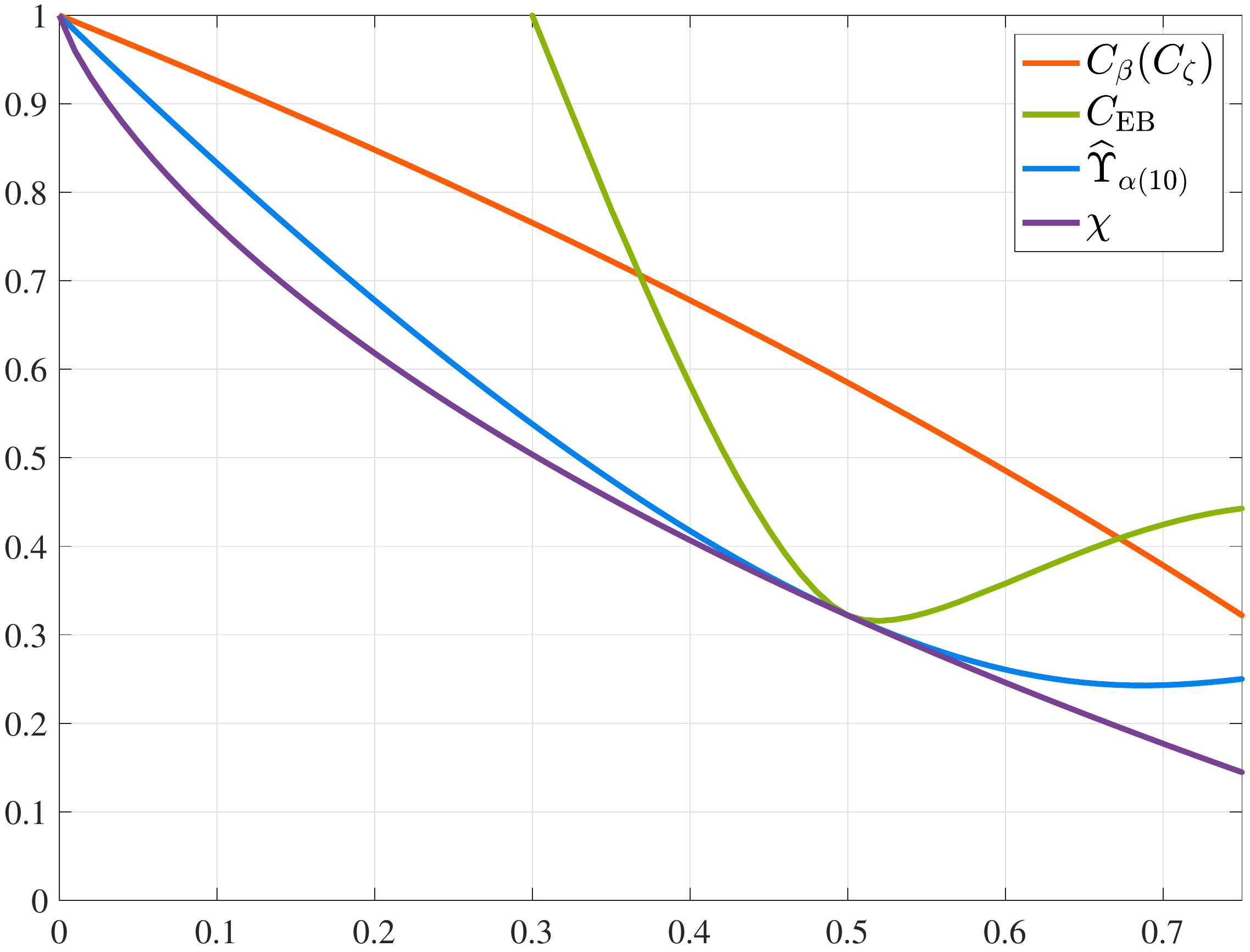}
\caption{\small Upper and lower bounds on the classical capacity of the composition channel $\cM_p = \cA_{p,0}\circ \cZ_p$ with the amplitude damping channel $\cA_{p,0}$ and the dephasing channel $\cZ_p$. All of $C_\b$, $C_\zeta$ and $\widehat \U_{\a{\scriptscriptstyle (10)}}$ are strong converse upper bounds. $C_\b$ coincides with $C_\zeta$ in this case. $C_{\text{EB}}$ is known as a weak converse bound given in~\cite[Corollary III.7]{Leditzky2018}. The Holevo information $\chi$ is a lower bound. The horizontal axis takes value of $p \in [0,0.75]$.}
\label{classical capacity compare 3}
\end{figure}

\newpage
\section{Magic state generation}

\subsection{Backgrounds}

The idea of fault-tolerant quantum computation proposes a reliable framework to implement practical quantum computation against noise and decoherence (e.g.~\cite{Shor1996,kitaev2003fault,knill2004fault,raussendorf2007fault}). Due to the Gottesman-Knill theorem~\cite{gottesman1997stabilizer,aaronson2004improved}, quantum circuits constructed by stabilizer operations 
can be efficiently simulated by a classical computer. Therefore, to fully power the universal quantum computation the stabilizer operations must be supplemented with some other fault-tolerant non-stabilizer resource. A celebrated scheme for this is given by the state injection technique that allows us to implement non-stabilizer operations by mixing the stabilizer operations with a key ingredient called ``magic states''~\cite{gottesman1999demonstrating,zhou2000methodology}. These are non-stabilizer states that must be prepared using the experimentally costly process of magic state distillation (e.g.~\cite{Bravyi2012,Hastings2018}). While extensive efforts have been devoted to construct efficient distillation codes (e.g.~\cite{gottesman1997stabilizer,gottesman1998,Aharonov2008,Hastings2018}), recent study in~\cite{Wang2019channelmagic} as well as~\cite{seddon2019quantifying} initiate the investigation of magic state generation via a general quantum channel, aiming to exploit the power and the limitations of a noisy quantum channel in the scenario of fault-tolerant quantum computation. 

Of particular interest is the work~\cite{Wang2019channelmagic} which identifies a larger class of operations (CPWP operations), that can be efficiently simulated via classical algorithms. Based on this notion of free operations, the authors established a complete resource theory framework and introduced two efficiently computable magic measures for quantum channels, named \emph{mana} ($\mana$) and \emph{max-Thauma} ($\theta_{\max}$) respectively. They proved several desirable properties of these two measures, and further showcased that these channel measures provided strong converse bounds for the task of magic state generation as well as lower bounds for the task of quantum channel synthesis.

\subsection{Summary of results}

In this part, we aim to push forward the study in~\cite{Wang2019channelmagic} by considering the Thauma measure induced by the geometric \Renyi divergence. Our results can be summarized as follows.

In Section~\ref{Geometric Renyi Thauma of a channel}, we prove that the \emph{geometric \Renyi Thauma} of a channel ($\widehat \theta_{\a}$) possesses all the nice properties that are held by the mana and max-Thauma, including the reduction to states, monotonicity under CPWP superchannels, faithfulness, amortization inequality, subadditivity under channel composition, additivity under tensor product as well as a semidefinite representation.

In Section~\ref{Magic state generation capacity}, we exhibit that the geometric \Renyi Thauma of a channel is a strong converse on the magic state generation capacity (the maximum number of magic state that can be produced per channel use of $\cN$ via adaptive protocols), improving the max-Thauma in general. More precisely, we show that 
\begin{align*}
C_\psi(\cN)  \leq C_\psi^{\dagger}(\cN) \leq \frac{\widehat \theta_{\a}(\cN)}{\theta_{\min}(\psi)} \leq \frac{\theta_{\max}(\cN)}{\theta_{\min}(\psi)},\quad \text{with} \ \widehat \theta_{\a}(\cN) \ \text{SDP computable},
\end{align*}
where $C_\psi(\cN)$ and $C_\psi^\dagger(\cN)$ denote the capacity of a channel $\cN$ to generate magic state $\psi$ and its corresponding strong converse capacity, respectively, and $\theta_{\min}(\psi)$ is a constant coefficient for given $\psi$.

In Section~\ref{Quantum channel synthesis}, we show that the geometric \Renyi Thauma can also provide lower bounds for the task of quantum channel synthesis. That is, we prove that the number of channel $\cN'$ required to implement another channel $\cN$ is bounded from below by ${\widehat \theta_{\a}(\cN)}/{\widehat \theta_{\a}(\cN')}$ for all $\a \in (1,2]$, complementing to the previous results by mana and max-Thauma.

\subsection{Preliminaries of the resource theory of magic}
We first review some basic formalism of the resource theory of magic.
Throughout this part, a Hilbert space implicitly has an odd dimension, and if the dimension is not prime, it should be understood to be a tensor product of Hilbert spaces each having odd prime dimension. Let $\{\ket{j}\}_{j=0}^{d-1}$ be the standard computational basis. For a prime number $d$, the generalized Pauli operator (or sometimes called the shift and boost operators) $X,Z$ are respectively defined as
\begin{align}
  X\ket{j} = \ket{j\oplus 1},\quad Z\ket{j} = \o^j \ket{j},\quad \text{with}\quad \o = e^{2\pi i/d},
\end{align}
where $\oplus$ deontes the addition modulo $d$.
The Heisenberg-Weyl operators is defined as~\footnote{The definition here is sightly different from some literatures. We adopt the same notion in~\cite{Wang2018magicstates}.}
\begin{align}
  T_{\bu} = \tau^{-a_1a_2} Z^{a_1} X^{a_2}, \quad\text{with}\quad \tau = e^{(d+1)\pi i /d},\quad \bu = (a_1,a_2) \in \mathbb Z_d \times \mathbb Z_d.
\end{align}
For a system with composite Hilbert space $\cH_A \ox \cH_B$, the Heisenberg-Weyl operators are the tensor product of the Heisenberg-Weyl operators on subsystems $T_{\bu_A \otimes \bu_B} = T_{\bu_A} \ox T_{\bu_B}$. 
For each point $\bu \in \ZZ_d \times \ZZ_d$ in the discrete phase space, there is a corresponding operator
\begin{align}
  \mA_{\bu} \equiv T_{\bu} \mA_0 T_{\bu}^\dagger \quad \text{with}\quad \mA_0 \equiv \frac{1}{d} \sum_{\bu} T_{\bu}
\end{align}
The value of the discrete Wigner representation of a quantum state $\rho$ at $\mA_{\bu}$ is given by
\begin{align}
  W_\rho(\bu) \equiv \frac{1}{d}\tr [\mA_{\bu} \rho].
\end{align}
The Wigner trace and Wigner spectral norm of an Hermitian operator $V$ are defined as
\begin{align}\label{eq: wigner trace norm infinity norm}
  \|V\|_{W,1} \equiv \sum_{\bu} |W_V(\bu)|,\quad \text{and}\quad \|V\|_{W,\infty} \equiv d \max_{\bu} |W_V(\bu)|,
\end{align}
respectively. For any Hermiticity-preserving map $\cN$, its discrete Wigner function is defined as
\begin{align}
  W_{\cN}(\bv|\bu)\equiv \frac{1}{d_B} \tr [\mA_{\bv_B} \cN(\mA_{\bu_A})] = \frac{1}{d_B} \tr [J_{\cN} (\mA_{\bu_A} \ox \mA_{\bv_B})],
\end{align}
with $J_{\cN}$ being the Choi matrix of $\cN$.
The set of quantum states with a non-negative Wigner function is denoted as 
\begin{align}
  \cW_+ \equiv \{\rho \,|\, \rho \geq 0, \tr \rho = 1, W_\rho(\bu) \geq 0, \forall \bu\}.
\end{align}
A quantum operation $\cE$ is \emph{completely positive Wigner preserving} (CPWP) if the following holds for any system $R$ with odd dimension~\cite{Wang2019channelmagic}  
\begin{align}
  \id_R \ox \cE_{A\to B} (\rho_{RA}) \in \cW_+\quad \forall \rho_{RA} \in \cW_+.
\end{align}

\begin{definition}[Mana]
The mana of a quantum state $\rho$ is defined as~\cite{Veitch_2014}
\begin{align}
  \mana(\rho)\equiv \log \|\rho\|_{W,1}.
\end{align}
  The mana of a quantum channel $\cN_{A\to B}$ is defined as~\footnote{This can be seen as an analog of Holevo-Werner bound for quantum capacity of a channel or log-negativity of a quantum state.}~\cite{Wang2019channelmagic}
  \begin{align}\label{eq: mana of a channel}
    \mana(\cN_{A\to B}) \equiv \log \max_{\bu_A} \|\cN_{A\to B}(\mA_{\bu_A})\|_{W,1} = \log \max_{\bu_A} \sum_{\bv_B} \frac{1}{d_B}|\tr J_{\cN} (\mA_{\bu_A}\ox \mA_{\bv_B})|.
  \end{align}
\end{definition}

\begin{definition}[Thauma]
Let $\bD$ be a generalized quantum divergence. The generalized Thauma of a quantum state $\rho$ is defined as~\cite{Wang2018magicstates}
\begin{align}
  \btheta(\rho)\equiv \min_{\sigma \in \cW} \bD(\rho\|\sigma),
\end{align}
where $\cW \equiv \{\sigma\,|\, \mana(\sigma) \leq 0, \sigma \geq 0\}$ is the set of sub-normalized states with non-positive mana.
The generalized Thauma of a quantum channel $\cN_{A\to B}$ is defined as~\cite{Wang2019channelmagic} 
  \begin{align}
    \btheta(\cN)\equiv \min_{\cE \in \bcV_{\cM}} \bD(\cN\|\cE),
  \end{align}
  where $\bcV_{\mana}\equiv \{\cE \in \CP(A:B)| \mana(\cE) \leq 0\}$ is the set of subchannels with non-positive mana.
\end{definition}
In particular, the max-Thauma of a channel is induced by the max-relative entropy~\cite{Wang2019channelmagic}
\begin{align}
  \theta_{\max}(\cN)\equiv \min_{\cE \in \bcV_{\mana}} D_{\max}(\cN\|\cE).
\end{align}

\subsection{Geometric \Renyi Thauma of a channel}
\label{Geometric Renyi Thauma of a channel}

In this section, we investigate the generalized Thauma induced by the geometric \Renyi divergence: 
\begin{align}
  \widehat \theta_{\a}(\rho) \equiv \min_{\sigma \in \cW} \widehat D_{\a}(\rho\|\sigma)\quad \text{and}\quad \widehat \theta_{\a}(\cN) \equiv \min_{\cE\in \bcV_\mana} \widehat D_{\a}(\cN\|\cE).
\end{align}
The authors in \cite{Wang2019channelmagic} proved that the mana and max-Thauma of a channel possess several nice properties, as listed in Table~\ref{tab: thauma properties}. Here we aim to show that all the desirable properties are also held by the geometric \Renyi Thauma as well. These basic properties will be utilized in the next two sections for improving the converse bound on magic state generation capacity and the efficiency of quantum channel synthesis.

\begin{table}[H]
\centering
\begin{tabular}{c|c|c|c}
\toprule[2pt]
 \diagbox{Property}{Quantifier} &  geometric R\'{e}nyi Thauma ($\widehat \theta_{\a}$) & max-Thauma ($\theta_{\max}$) & Mana ($\mana$)\\
 \hline
Reduction to states & \cmark \quad Lem.~\ref{prop: reduction monotone faithful} \hphantom{($\alpha = 2$)} & \cmark & \cmark\\
Monotonicity under CPWP & \cmark \quad Lem.~\ref{prop: reduction monotone faithful} \hphantom{($\alpha = 2$)} & \cmark & \cmark\\
Faithfulness & \cmark \quad Lem.~\ref{prop: reduction monotone faithful} \hphantom{($\alpha = 2$)} & \cmark & \cmark\\
Amortization & \cmark \quad Lem.~\ref{prop: amortization} \hphantom{($\alpha = 2$)} & \cmark  & \cmark \\
Subadditivity under $\circ$ & \cmark \quad Lem.~\ref{prop: subadditivity} \hphantom{($\alpha = 2$)} & \cmark & \cmark \\
SDP computable & \cmark \quad Lem.~\ref{prop: SDP formula for maximal renyi thauma of a channel} \hphantom{($\alpha = 2$)} & \cmark & \cmark\\
Additivity under $\otimes$ & \cmark \quad Lem.~\ref{prop: additivity} ($\alpha = 2$) & \cmark & \cmark \\
\bottomrule[2pt]  
\end{tabular}
\caption{\small Comparison of properties for the geometric \Renyi Thauma, max-Thauma and mana of quantum channels.}
\label{tab: thauma properties}
\end{table}

\begin{lemma}\label{prop: reduction monotone faithful}
The following properties hold for the geometric \Renyi Thauma of a channel when $\alpha \in (1,2]$:
\begin{itemize}
\item ({Reduction to states}): Let $\cN(\rho) = \tr[\rho] \sigma$ be a replacer channel with fixed $\sigma$ for any $\rho$. Then 
\begin{align}
  \widehat \theta_{\a}(\cN) = \widehat \theta_{\a}(\sigma).
\end{align}
\item ({Monotonicity}): Let $\cN$ be a quantum channel and $\Gamma$ be a CPWP superchannel. Then 
\begin{align}
  \widehat \theta_{\a}(\Gamma(\cN)) \leq \widehat \theta_{\a}(\cN).
\end{align}
\item ({Faithfulness}): $\widehat \theta_{\a}(\cN)$ is nonnegative for any quantum channel $\cN$ and 
\begin{align}
  \widehat \theta_{\a}(\cN) = 0 \quad \text{if and only if} \quad \cN \in \CPWP.
\end{align}
\end{itemize}
\end{lemma}
\begin{proof}
The first two properties directly follow from the argument for the generalized Thauma in~\cite[Proposition 9 and 10]{Wang2019channelmagic}. The third property follows from the argument in~\cite[Proposition 11]{Wang2019channelmagic} and the fact that the geometric \Renyi divergence is continuous and strongly faithful (i.e, $\widehat D_\a(\rho\|\sigma) \geq 0$ in general and $\widehat D_{\a}(\rho\|\sigma) = 0$ if and only if $\rho = \sigma$).
\end{proof}

\begin{lemma}[Amortization]\label{prop: amortization}
  For any quantum state $\rho_{RA}$, any quantum channel $\cN_{A\to B}$ and the parameter $\a\in (1,2]$, it holds
  \begin{align}
    \widehat \theta_{\a}(\cN_{A\to B}(\rho_{RA})) - \widehat \theta_{\a}(\rho_{RA}) \leq \widehat \theta_{\a}(\cN_{A\to B}).
  \end{align}
\end{lemma}
\begin{proof}
The proof follows the similar steps as Proposition~\ref{amortization proposition}. We only need to show that for any sub-normalized state $\sigma_{RA} \in \cW$ and any subchannel $\cE \in \bcV_{\mana}$, it holds $\gamma_{RB}\equiv \cE_{A\to B}(\sigma_{RA}) \in \cW$. This can be checked as follows:
  \begin{align}
    \|\gamma_{RB}\|_{W,1} & = \sum_{\bu_R,\bv_B} |W_{\gamma_{RB}}(\bu_R,\bv_B)|\\
    & = \sum_{\bu_R,\bv_B} \Big|\sum_{\by_A} W_{\cE}(\bv_B|\by_A) W_{\sigma_{RA}}(\by_A,\bu_R)\Big|\\
    & \leq \sum_{\bu_R,\bv_B,\by_A} \Big| W_{\cE}(\bv_B|\by_A)\Big| \Big| W_{\sigma_{RA}}(\by_A,\bu_R)\Big|\\
    & =  \sum_{\bu_R,\by_A}  \bigg[\sum_{\bv_B} \Big| W_{\cE}(\bv_B|\by_A)\Big|\bigg] \Big| W_{\sigma_{RA}}(\by_A,\bu_R)\Big|\\
    & \leq \sum_{\bu_R,\by_A} \Big| W_{\sigma_{RA}}(\by_A,\bu_R)\Big|\\
    & \leq 1.
  \end{align}
  The first line is the definition of the Wigner trace norm in~\eqref{eq: wigner trace norm infinity norm}. The second line is a chain relation in~\cite[Lemma 1]{Wang2019channelmagic}. The third line follows from the triangle inequality of the absolute value function. The fourth line follows by grouping the components with respect to index $\bv_B$. The fifth line follows since $\cE\in \bcV_\mana$ and thus $\sum_{\bv_B} \big| W_{\cE}(\bv_B|\by_A)\big| \leq \max_{\by_A}\sum_{\bv_B} \big| W_{\cE}(\bv_B|\by_A)\big| \leq 1$. The last line follows since $\sigma_{RA} \in \cW$. Thus we can conclude that $\gamma_{RB} \in \cW$. This completes the proof.
\end{proof}

\begin{lemma}[Sub-additivity]\label{prop: subadditivity}
  For any two quantum channels $\cN_1$, $\cN_2$ and $\a \in (1,2]$, it holds
  \begin{align}
    \widehat \theta_{\a}(\cN_2\circ \cN_2) \leq \widehat \theta_{\a}(\cN_1) + \widehat \theta_{\a}(\cN_2).
  \end{align}
\end{lemma}
\begin{proof}
  Suppose the optimal solution of $\widehat \theta_{\a}(\cN_1)$ and $\widehat \theta_{\a}(\cN_2)$ are taken at $\cE_1$ and $\cE_2$, respectively. By the subadditivity of the mana under composition, we have $\mana(\cE_2\circ \cE_1) \leq 0$ (see~\cite[Proposition 5]{Wang2019channelmagic}). Thus $\cE_2\circ \cE_1$ is a feasible solution for $\widehat \theta_{\a}(\cN_2\circ \cN_2)$ and we have
  \begin{align}
    \widehat \theta_{\a}(\cN_2\circ \cN_2) \leq \widehat D_{\a}(\cN_2\circ \cN_2\|\cE_2\circ \cE_1) \leq \widehat D_{\a}(\cN_1\|\cE_1) + \widehat D_{\a}(\cN_2\|\cE_2) = \widehat \theta_{\a}(\cN_1) + \widehat \theta_{\a}(\cN_2),
  \end{align}
  where the second inequality follows from Lemma~\ref{lem: subadditivity of D alpha}, the last equality follows from the optimality assumption of $\cE_1$ and $\cE_2$.
\end{proof}

\begin{lemma}[SDP formula]
\label{prop: SDP formula for maximal renyi thauma of a channel}
  For any quantum channel $\cN$ and $\a(\ell) = 1+2^{-\ell}$ with $\ell \in \mathbb N$, it holds
\begin{gather}
  \widehat \theta_{\a}(\cN)= 2^\ell\cdot \log \min \ y \quad \text{\rm s.t.}\quad \dbhbig{M,\{N_i\}_{i=0}^\ell,y},\ \dbpbig{y\1 - \tr_B M}\notag\\[2pt]
   \dbp{\begin{matrix}
    M & J_{\cN}\\
    J_{\cN} & N_{\ell}
  \end{matrix}},
  \left\{\dbp{\begin{matrix}
    J_{\cN} & N_{i} \\
    N_{i} & N_{i-1}
  \end{matrix}}\right\}_{i=1}^\ell, 
  \dbpbigg{1 - \frac{1}{d_B}\sum_{\bv_B} \big|\tr N_0 (\mA_{\bu_A} \ox \mA_{\bv_B})\big|}, \forall \bu_A.
  \label{eq: SDP formula for maximal renyi thauma of a channel}
\end{gather}
\end{lemma}
\begin{proof}
This directly follows from Lemma~\ref{lem: SDP representation of the channel information measure} and the definition of mana in~\eqref{eq: mana of a channel}. Note that the absolute value conditions can be written as semidefinite conditions by introducing slack variables.
\end{proof}

\begin{lemma}[Additivity]
\label{prop: additivity}
  The geometric \Renyi Thauma at $\alpha =2$ is additive under tensor product. That is, for any two quantum channels $\cN_1$, $\cN_2$,  it holds
\begin{align}
  \widehat \theta_2(\cN_1\ox \cN_2) =  \widehat \theta_2(\cN_1) + \widehat \theta_2(\cN_2).
\end{align}
\end{lemma}
\begin{proof}
We first prove the sub-additivity. Suppose the optimal solution of $\widehat \theta_2(\cN_1)$ and $\widehat \theta_2(\cN_2)$ are taken at $\cE_1$ and $\cE_2$ respectively. Then we have $\mana(\cE_1\ox \cE_2) = \cM(\cE_1) + \cM(\cE_2) \leq 0$ since mana is additive under tensor product~\cite[Proposition 4]{Wang2019channelmagic}. This implies that $\cE_1\ox \cE_2$ is a feasible solution for $\widehat \theta_2(\cN_1\ox \cN_2)$. Thus we have
\begin{align}
   \widehat \theta_2(\cN_1\ox \cN_2) \leq \widehat D_2(\cN_1\ox \cN_2\|\cE_1\ox \cE_2) = \widehat D_2(\cN_1\|\cE_1) + \widehat D_2(\cN_2\|\cE_2) = \widehat \theta_2(\cN_1) + \widehat \theta_2(\cN_2),
 \end{align}
 where the first equality follows from Lemma~\ref{lem: maximal Renyi channel divergence additivity} and the second equality follows from the optimality assumption of $\cE_1$ and $\cE_2$. 

We now show the super-additivity by utilizing the dual formula of~\eqref{eq: SDP formula for maximal renyi thauma of a channel}.
According to the Lagrangian method, we have the dual problem as
\begin{gather}
  \widehat \theta_{2}(\cN)= \log \max  \tr \left[J_{\cN} \left(K+K^\dagger\right) \right]- \sum_{\bu} f_{\bu} \quad \text{s.t.} \notag\\
    \begin{bmatrix}
    \rho_A \ox \1_B &  K\\ {K}^\dagger & Z
  \end{bmatrix} \geq 0,
  |\tr Z (\mA_{\bu} \ox \mA_{\bv})/d_A| \leq f_{\bu},\forall \bu,\bv, \tr \rho = 1.\label{eq: additivity tmp1}
\end{gather}
It is easy to check that the strong duality holds.
Note that if we replace $K$ as $x K$ with $|x| = 1$, the optimization is unchanged. Thus we can choose scalar $x = \tr (J_{\cN}K)^\dagger/ |\tr J_{\cN}K|$ to make the term $\tr J_{\cN} (xK) = |\tr J_{\cN}K|$ to a real scalar. Thus optimization~\eqref{eq: additivity tmp1} is equivalent to
\begin{gather}
  \widehat \theta_{2}(\cN)= \log \max 2 |\tr J_{\cN} K |- \sum_{\bu} f_{\bu} \quad \text{s.t.} \notag\\
    \begin{bmatrix}
    \rho_A \ox \1_B &  K\\ {K}^\dagger & Z
  \end{bmatrix} \geq 0,
  |\tr Z (\mA_{\bu} \ox \mA_{\bv})/d_A| \leq f_{\bu},\forall \bu,\bv, \tr \rho = 1.\label{eq: additivity tmp2}
\end{gather}
Again, by replacing $\widetilde K = K/w$, $\widetilde Z = Z/w^2$ and $\widetilde f_{\bu} = f_{\bu}/w^2$, we have 
\begin{gather}
  \widehat \theta_{2}(\cN)= \log \max 2 w |\tr J_{\cN} \widetilde K |- w^2 \sum_{\bu} \widetilde f_{\bu} \quad \text{s.t.} \notag\\
    \begin{bmatrix}
    \rho_A \ox \1_B &  \widetilde K\\ {\widetilde K}^\dagger & \widetilde Z
  \end{bmatrix} \geq 0,
  |\tr \widetilde Z (\mA_{\bu} \ox \mA_{\bv})/d_A| \leq \widetilde f_{\bu},\forall \bu,\bv, \tr \rho = 1.\label{eq: additivity tmp3}
\end{gather}
For any fixed $|\tr J_{\cN} \widetilde K |$ and $\sum_{\bu} \widetilde f_{\bu}$, we can quickly check that the optimal solution of the objective function is always taken at $w = |\tr J_{\cN} \widetilde K |/ (\sum_{\bu} \widetilde f_{\bu})$ with the optimal value $|\tr J_{\cN} \widetilde K |^2/ (\sum_{\bu} \widetilde f_{\bu})$. Thus the optimization~\eqref{eq: additivity tmp3} is equivalent to
\begin{gather}
  \widehat \theta_{2}(\cN)= \log \max |\tr J_{\cN} K |^2/ \left(\sum\nolimits_{\bu} f_{\bu}\right) \quad \text{s.t.} \notag\\
    \begin{bmatrix}
    \rho_A \ox \1_B &  K\\ {K}^\dagger & Z
  \end{bmatrix} \geq 0,
  |\tr Z (\mA_{\bu} \ox \mA_{\bv})/d_A| \leq f_{\bu},\forall \bu,\bv, \tr \rho = 1. \label{eq: magic additivity proof tmp1}
\end{gather}
Suppose the optimal solution of $\widehat \theta_2(\cN_1)$ and $\widehat \theta_2(\cN_2)$ are taken at $\{K_1,Z_1, f^1_{\bu},\rho_1\}$ and $\{K_2,Z_2, f^2_{\bv},\rho_2\}$ respectively. We can check that their tensor product $\{K_1\ox K_2, Z_1 \ox Z_2, f^1_{\bu} f^2_{\bv},\rho_1\ox \rho_2\}$ forms a feasible solution for $\widehat \theta_2(\cN_1\ox \cN_2)$ in~\eqref{eq: magic additivity proof tmp1}. Thus we have
\begin{align}
  \widehat \theta_2(\cN_1\ox \cN_2) \geq \log \frac{|\tr (J_{\cN_1}\ox J_{\cN_2}) (K_1\ox K_2)|^2}{\sum_{\bu,\bv} f_{\bu}^1 f_{\bv}^2} =  \widehat \theta_2(\cN_1) + \widehat \theta_2(\cN_2),
\end{align}
which completes the proof.
\end{proof}

\begin{remark}
  Based on numerical observations, we expect that the additivity of the geometric \Renyi Thauma holds for general $\a \in (1,2]$. However, the current proof seems to only work for $\a = 2$.
\end{remark}

\subsection{Magic state generation capacity}
\label{Magic state generation capacity}

In~\cite{Wang2019channelmagic}, the authors study an information task which uses a quantum channel to produce magic states and quantifies the ``magic of channel'' by the amount of magic state generated per channel use. Here, we simply dub it as the \emph{magic state generation capacity}~\footnote{This is analogous to the name of entanglement/coherence generation capacity in the existing literature.} as it characterizes the capability of a channel to generate magic states. The most general protocol to produce a magic state can be proceeded as follows (see Figure~\ref{fig: magic state generation scheme}). 

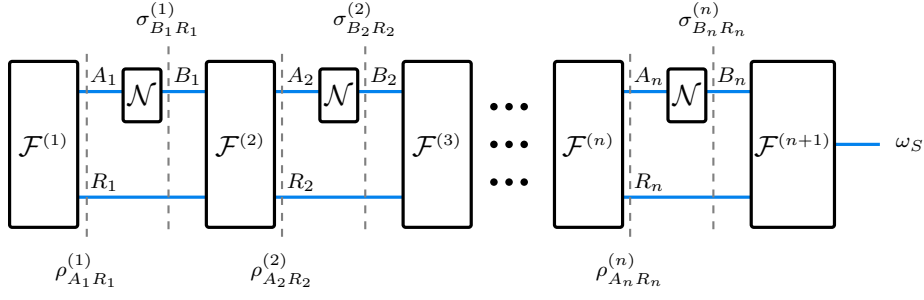
\begin{figure}[H]
\centering
\begin{tikzpicture}

\draw[very thick, colorthree] (-5.6,-0.7) -- (-3.9,-0.7) node[black,midway,shift={(-0.5,0.2)}] {\scriptsize $R_1$};
\draw[very thick, colorthree] (-3,-0.7) -- (-1.3,-0.7) node[black,midway,shift={(-0.5,0.2)}] {\scriptsize $R_2$};
\draw[very thick, colorthree] (1.6,-0.7) -- (3.3,-0.7) node[black,midway,shift={(-0.5,0.2)}] {\scriptsize $R_n$};

\draw[very thick, colorthree] (-5.6,0.7) -- (-5,0.7)  node[black,midway,shift={(0.05,0.2)}] {\scriptsize $A_1$};
\draw[very thick, colorthree] (-3,0.7) -- (-2.4,0.7) node[black,midway,shift={(0.05,0.2)}] {\scriptsize $A_2$};
\draw[very thick, colorthree] (1.6,0.7) -- (2.2,0.7) node[black,midway,shift={(0.05,0.2)}] {\scriptsize $A_n$};

\draw[very thick, colorthree] (-4.5,0.7) -- (-3.9,0.7) node[black,midway,shift={(0.05,0.2)}] {\scriptsize $B_1$};
\draw[very thick, colorthree] (-1.9,0.7) -- (-1.3,0.7) node[black,midway,shift={(0.05,0.2)}] {\scriptsize $B_2$};
\draw[very thick, colorthree] (2.7,0.7) -- (3.3,0.7) node[black,midway,shift={(0.05,0.2)}] {\scriptsize $B_n$};

\draw[very thick, colorthree] (4.4,0) -- (5,0);

\node at (5.4,0) {\scriptsize $\omega_S$};

\draw[thick,gray,dashed] (-5.48,1.2) -- (-5.48,-1.3) node[black,below] {\scriptsize $\rho^{(1)}_{A_1R_1}$};
\draw[thick,gray,dashed] (-2.9,1.2) -- (-2.9,-1.3) node[black,below] {\scriptsize $\rho^{(2)}_{A_2R_2}$};
\draw[thick,gray,dashed] (1.7,1.2) -- (1.7,-1.3) node[black,below] {\scriptsize $\rho^{(n)}_{A_nR_n}$};

\draw[thick,gray,dashed] (-4.4,1.3) node[black,above] {\scriptsize $\sigma^{(1)}_{B_1R_1}$} -- (-4.4,-1.2);
\draw[thick,gray,dashed] (-1.8,1.3) node[black,above] {\scriptsize $\sigma^{(2)}_{B_2R_2}$} -- (-1.8,-1.2);
\draw[thick,gray,dashed] (2.8,1.3) node[black,above] {\scriptsize $\sigma^{(n)}_{B_nR_n}$} -- (2.8,-1.2);

\draw[very thick, rounded corners = 1] (-6.5,1.1) rectangle (-5.6,-1.1) node[midway] { $\cF^{\scriptscriptstyle (1)}$};
\draw[very thick, rounded corners = 1] (-3.9,1.1) rectangle (-3,-1.1) node[midway] { $\cF^{\scriptscriptstyle (2)}$};
\draw[very thick, rounded corners = 1] (-1.3,1.1) rectangle (-0.4,-1.1) node[midway] { $\cF^{\scriptscriptstyle (3)}$};
\draw[very thick, rounded corners = 1] (0.7,1.1) rectangle (1.6,-1.1) node[midway] { $\cF^{\scriptscriptstyle (n)}$};
\draw[very thick, rounded corners = 1] (3.3,1.1) rectangle (4.4,-1.1) node[midway] { $\cF^{\scriptscriptstyle (n+1)}$};

\begin{scope}[shift={(0,0.7)}]
\draw[very thick, rounded corners = 1] (-5,0.3) rectangle (-4.5,-0.4) node[midway] {$\cN$};
\draw[very thick, rounded corners = 1] (-2.4,0.3) rectangle (-1.9,-0.4) node[midway] {$\cN$};
\draw[very thick, rounded corners = 1] (2.2,0.3) rectangle (2.7,-0.4) node[midway] {$\cN$};
\end{scope}
\node at (-0.1,0) [circle,fill,inner sep=1pt]{};
\node at (0.1,0) [circle,fill,inner sep=1pt]{};
\node at (0.3,0) [circle,fill,inner sep=1pt]{};

\node at (-0.1,0.5) [circle,fill,inner sep=1pt]{};
\node at (0.1,0.5) [circle,fill,inner sep=1pt]{};
\node at (0.3,0.5) [circle,fill,inner sep=1pt]{};

\node at (-0.1,-0.5) [circle,fill,inner sep=1pt]{};
\node at (0.1,-0.5) [circle,fill,inner sep=1pt]{};
\node at (0.3,-0.5) [circle,fill,inner sep=1pt]{};

\end{tikzpicture}
\caption{\small A schematic diagram for the magic state generation protocol that uses a quantum channel $n$ times. Every channel use is interleaved by a free CPWP operation $\cF^{\scriptscriptstyle (i)}$. The goal of such a protocol is to produce an approximate magic state $\o_{S}$ in the end.}
 \label{fig: magic state generation scheme}
 \end{figure}

First, we start from preparing a quantum state $\rho^{\scriptscriptstyle{(1)}}_{R_1 A_1}$ via a free CPWP operation $\cF^{\scriptscriptstyle{(1)}}_{\emptyset \to R_1A_1}$. Then we apply the given channel $\cN$ on system $A_1$ and obtain a quantum state $\sigma_{R_1B_1}^{\scriptscriptstyle{(1)}} = \cN_{A_1\to B_1}(\rho^{\scriptscriptstyle{(1)}}_{R_1 A_1})$. After this, we perform another free CPWP operation $\cF^{\scriptscriptstyle{(2)}}_{R_1 B_1 \to R_2A_2}$ and then apply the channel $\cN$ again. These processes can be conducted iteratively $n$ times, and we obtain a quantum state $\sigma^{\scriptscriptstyle{(n)}}_{R_n B_n}$. At the end of such a protocol, a final free CPWP operation $\cF^{\scriptscriptstyle{(n+1)}}_{R_nB_n \to S}$ is performed, producing a quantum state $\omega_S$. 

For any error tolerance $\ve \in [0,1]$, the above procedure defines an $(n,k,\ve)$ $\psi$-magic state generation protocol, if the final state $\omega$ has a sufficiently high fidelity with $k$ copies of the target magic state $\psi$,
\begin{align}
  \tr \omega_S \ket{\psi}\bra{\psi}^{\ox k} \geq 1-\ve.
\end{align}
A rate $r$ is achievable if for all $\ve \in (0,1]$   and $\delta > 0$ and sufficiently large $n$, there exists an $(n,n(r-\delta),\ve)$ $\psi$-magic state generation protocol as depicted above. Then the $\psi$-\emph{magic state generation capacity} of the channel $\cN$ is defined as the supremum of all achievable rates and is denoted as $C_{\psi}(\cN)$. On the other hand, $r_0$ is called a strong converse rate if for every $r > r_0$, the fidelity $1-\ve$ of any generation protocol will decays to zero as the number of rounds $n$ increases. The strong converse capacity, denoted as $C_{\psi}^{\dagger}(\cN)$ is the infimum of all strong converse rates. By definition, we have $C_{\psi}(\cN) \leq C_{\psi}^{\dagger}(\cN)$ in general.

\vspace{0.2cm}
Based on the amortization inequality in Lemma~\ref{prop: amortization}, a similar argument as~\cite[Proposition 20]{Wang2019channelmagic} will give us the following improved bound on the magic state generation capacity:
\begin{theorem}[Main result 7]
\label{thm: main result magic}
  For any quantum channel $\cN$ and $\a \in (1,2]$, it holds
  \begin{align}
    C_\psi(\cN)  \leq C_\psi^{\dagger}(\cN) \leq \frac{\widehat \theta_{\a}(\cN)}{\theta_{\min}(\psi)} \leq \frac{\theta_{\max}(\cN)}{\theta_{\min}(\psi)},    
  \end{align}
  where $\theta_{\min}(\psi) = \min_{\sigma \in \cW} D_{\min}(\psi\|\sigma)$ is the min-Thauma of the magic state $\psi$. 
\end{theorem}
\begin{proof}
The first inequality holds by definition. The last inequality is a direct consequence of the relation $\widehat D_{\a}(\rho\|\sigma) \leq D_{\max}(\rho\|\sigma)$ proved in Lemma~\ref{thm: divergence chain inequality}. It remains to show the second inequality. The main ingredient to prove this is the amortization property of the geometric \Renyi Thauma in Lemma~\ref{prop: amortization}. Consider $n$ round magic state generation protocol as shown in Figure~\ref{fig: magic state generation scheme}. For each round,  denote the input state of the channel $\cN$ as $\rho^{\scriptscriptstyle (i)}_{R_iA_i}$ and the output state as $\sigma^{\scriptscriptstyle (i)}_{R_iB_i}$. The final state after $n$ round operations is denoted as $\omega_S$. Thus we have
\begin{align}
  \widehat \theta_{\a}(\omega_S) & \leq \widehat \theta_{\a}(\sigma^{\scriptscriptstyle (n)}_{R_nB_n})\\
  & = \widehat \theta_{\a}(\sigma^{\scriptscriptstyle (n)}_{R_nB_n}) - \widehat \theta_{\a}(\rho^{\scriptscriptstyle (1)}_{R_1A_1})\\
  & \leq \widehat \theta_{\a}(\sigma^{\scriptscriptstyle (n)}_{R_nB_n}) + \sum\nolimits_{i=1}^{n-1} \left[\widehat \theta_{\a}(\sigma^{\scriptscriptstyle (i)}_{R_iB_i}) - \widehat \theta_{\a}(\rho^{\scriptscriptstyle (i+1)}_{R_iA_i}) \right] - \widehat \theta_{\a}(\rho^{\scriptscriptstyle (1)}_{R_1A_1})\\
  & = \sum\nolimits_{i=1}^{n} \left[\widehat \theta_{\a}(\sigma^{\scriptscriptstyle (i)}_{R_iB_i}) - \widehat \theta_{\a}(\rho^{\scriptscriptstyle (i)}_{R_iA_i})\right]\\
  & \leq n \widehat \theta_{\a}(\cN).\label{magic renyi strong converse tmp}
\end{align}
The first and third lines follow from the monotonicity of the geometric \Renyi Thauma of a quantum state under CPWP operations. The second line follows since  $\widehat \theta_{\a}(\rho^{\scriptscriptstyle (1)}_{R_1A_1}) = 0$. The last line follows from Lemma~\ref{prop: amortization}.

For any $\psi$-magic state generation protocol with triplet $(n,k,\ve)$, denote $r = k/n$. By definition, we have
$\tr \left[\omega_S \ket{\psi}\bra{\psi}^{\ox k}\right] \geq 1-\ve$ . Moreover, for any $\sigma_S \in \cW$, it holds
$\tr \left[\sigma_S \ket{\psi}\bra{\psi}^{\ox k}\right]  \leq 2^{-nr \theta_{\min}(\psi)}$ \cite{Wang2018magicstates}. Without loss of generality, we can assume that $\ve \leq 1-2^{-nr \theta_{\min}(\psi)}$. Otherwise, the strong converse would already hold for any rates above the capacity since $1-\ve < 2^{-nr \theta_{\min}(\psi)}$. Thus for any $\sigma_S \in \cW$, we have the inequalities
\begin{align}
1- \tr  \left[\omega_S \ket{\psi}\bra{\psi}^{\ox k}\right] \leq \ve \leq 1- 2^{-nr \theta_{\min}(\psi)} \leq 1 - \tr \left[\sigma_S \ket{\psi}\bra{\psi}^{\ox k}\right].
\end{align}
Consider a quantum channel $\cN(\gamma) = \left[\tr \ket{\psi}\bra{\psi}^{\ox k} \gamma\right] \ket{0}\bra{0} + \left[ 1 - \tr \ket{\psi}\bra{\psi}^{\ox k} \gamma\right] \ket{1}\bra{1}$. Due to the data-processing inequality, we have
\begin{align}
\widehat D_\a(\omega\|\sigma) & \geq \widehat D_\a(\cN(\omega)\|\cN(\sigma))\notag\\
 & = \delta_\a\Big( 1-\tr \ket{\psi}\bra{\psi}^{\ox k} \omega \Big\|1-\tr \ket{\psi}\bra{\psi}^{\ox k} \sigma\Big) \geq \delta_\a\left(\ve\, \Big\|1- 2^{-nr \theta_{\min}(\psi)}\right),
\end{align}
where $\delta_\a(p\|q)\equiv \frac{1}{\a-1} \log \big[p^\a q^{1-\a} + (1-p)^\a (1-q)^{1-\a}\big]$. The last inequality follows from the monotonicity property that $\delta_\a(p'\|q) \leq \delta_\a(p\|q)$ if $p \leq p' \leq q$ and $\delta_\a(p\|q') \leq \delta_\a(p\|q)$ if $p \leq q' \leq q$~\cite{Polyanskiy2010b}.
Then we have
\begin{align}
\widehat \theta_\a(\omega_S) = \min_{\sigma \in \cW} & \widehat D_{\a}(\omega\|\sigma)  \geq \delta_\a \left(\ve\Big\|1- 2^{-nr \theta_{\min}(\psi)}\right)\notag\\
& \geq \frac{1}{\a-1} \log (1-\ve)^\a \Big(2^{-nr \theta_{\min}(\psi)}\Big)^{1-\a} = \frac{\a}{\a-1} \log (1-\ve) + nr \theta_{\min}(\psi).
\label{magic renyi strong converse tmp1}
\end{align}
Combining Eqs.~\eqref{magic renyi strong converse tmp} and~\eqref{magic renyi strong converse tmp1}, we have
\begin{align}
\frac{\a}{\a-1} \log (1-\ve) + nr \theta_{\min}(\psi) \leq n \widehat \theta_{\a}(\cN),
\end{align}
which is equivalent to 
\begin{align}
1-\ve \leq 2^{-n \theta_{\min}(\psi)\left(\frac{\a-1}{\a}\right)\left[r - \widehat \theta_{\a}(\cN)/\theta_{\min}(\psi)\right]}.
\end{align}
This implies that if the generation rate $r$ is strictly larger than $\widehat \theta_{\a}(\cN)/\theta_{\min}(\psi)$, the fidelity of the generation protocol $1-\ve$ decays exponentially fast to zero as the number of rounds $n$ increases. Or equivalently, we have $C_\psi^{\dagger}(\cN) \leq {\widehat \theta_{\a}(\cN)}/{\theta_{\min}(\psi)}$ and completes the proof.
\end{proof}

\begin{remark}
If the target magic state is $T$ state or $H_+$ state, we have $\theta_{\min}(T)= \log (1+2\sin(\pi/18))$ and $\theta_{\min}(H_+) = \log (3 - \sqrt{3})$, respectively~\cite[Proposition 2]{Wang2018magicstates}.
\end{remark}

\begin{figure}[H]
  \centering
  \begin{tikzpicture}
    \node at (0,0) {\includegraphics[width=6.5cm]{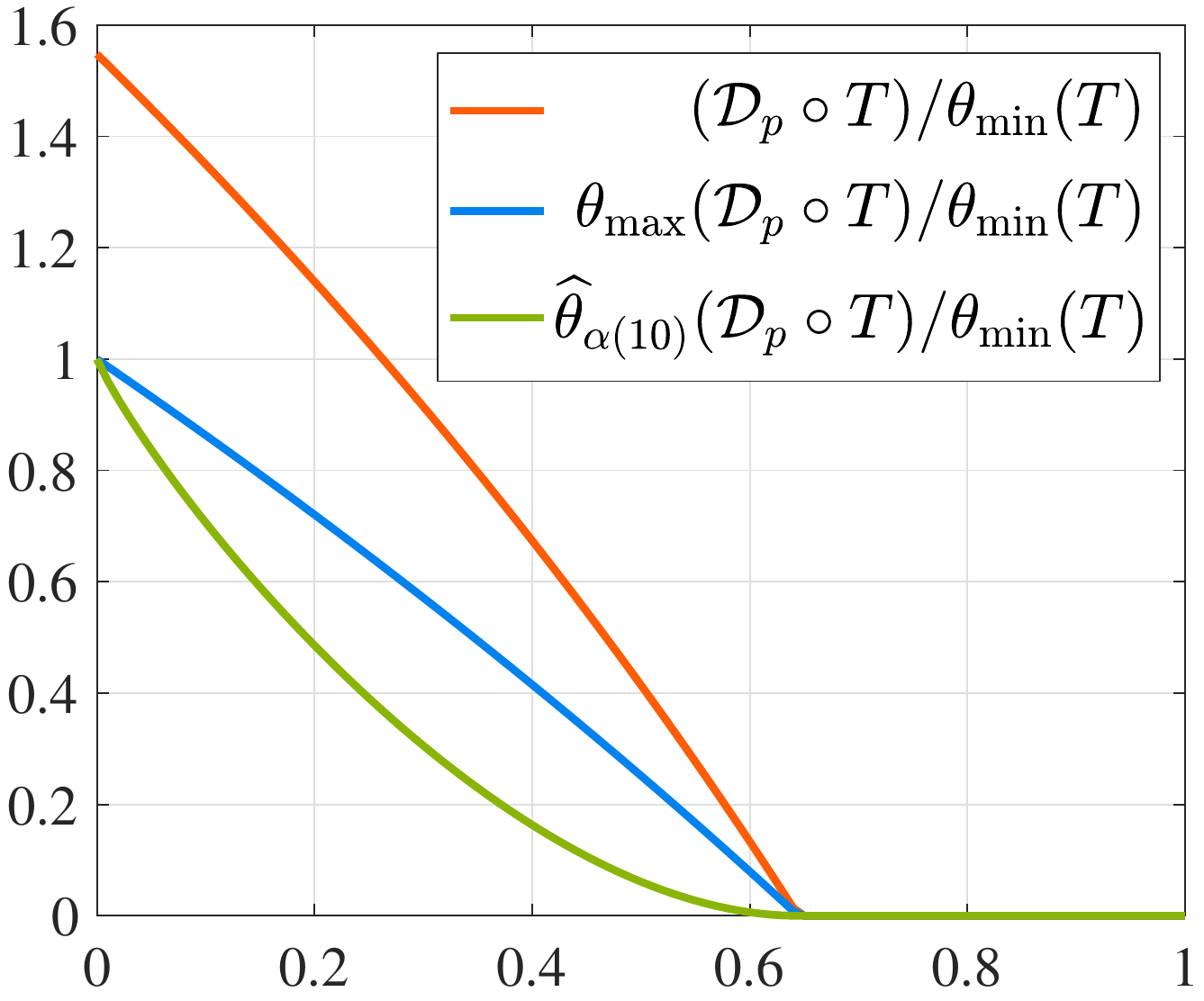}};
    \node at (0.25,2.12) {\normalsize $\mana$};
    \node at (0,-3) {$p$};
  \end{tikzpicture}
  \caption{\small Comparison of the strong converse bounds on the $T$-magic state generation capacity of the qutrit quantum channel $\cD_p \circ T$, where the depolarizing noise parameter $p \in [0,1]$ and $\alpha(10) = 1+2^{-10}$.}
  \label{fig: magic capacity bound compare}
\end{figure}

\vspace{0.2cm}

Consider a qutrit quantum channel $\cD_p\circ T$ composed by a $T$-gate with a qutrit depolarizing noise $\cD_p$. The above Figure~\ref{fig: magic capacity bound compare} compares different converse bounds on the $T$-magic state generation capacity of the channel $\cD_p\circ T$. It is clear that our new bound based on the geometric \Renyi divergence is significantly tighter than the others.

\subsection{Quantum channel synthesis}
\label{Quantum channel synthesis}

Another fundamental question in the resource theory of magic asks how many instances of a given quantum channel $\cN'$ are required to simulate another quantum channel $\cN$, when supplemented with free CPWP operations. Such a general scheme is illustrated in Figure~\ref{fig: channel synthesis}. Denote $S(\cN'\to \cN)$ as the smallest number of $\cN'$ channels required to implement the target channel $\cN$ exactly.

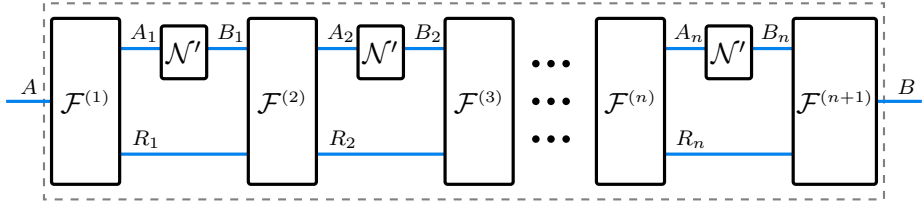
\begin{figure}[H]
\centering
\begin{tikzpicture}
\draw[thick,dashed, gray] (-6.6,1.3) rectangle (4.5,-1.3);
\draw[very thick, colorthree] (-5.6,-0.7) -- (-3.9,-0.7) node[black,midway,shift={(-0.5,0.2)}] {\scriptsize $R_1$};
\draw[very thick, colorthree] (-3,-0.7) -- (-1.3,-0.7) node[black,midway,shift={(-0.5,0.2)}] {\scriptsize $R_2$};
\draw[very thick, colorthree] (1.6,-0.7) -- (3.3,-0.7) node[black,midway,shift={(-0.5,0.2)}] {\scriptsize $R_n$};

\draw[very thick, colorthree] (-5.6,0.7) -- (-5.05,0.7)  node[black,midway,shift={(0.05,0.2)}] {\scriptsize $A_1$};
\draw[very thick, colorthree] (-3,0.7) -- (-2.45,0.7) node[black,midway,shift={(0.05,0.2)}] {\scriptsize $A_2$};
\draw[very thick, colorthree] (1.6,0.7) -- (2.15,0.7) node[black,midway,shift={(0.05,0.2)}] {\scriptsize $A_n$};

\draw[very thick, colorthree] (-4.45,0.7) -- (-3.9,0.7) node[black,midway,shift={(0.05,0.2)}] {\scriptsize $B_1$};
\draw[very thick, colorthree] (-1.85,0.7) -- (-1.3,0.7) node[black,midway,shift={(0.05,0.2)}] {\scriptsize $B_2$};
\draw[very thick, colorthree] (2.75,0.7) -- (3.3,0.7) node[black,midway,shift={(0.05,0.2)}] {\scriptsize $B_n$};

\draw[very thick, colorthree] (4.4,0) -- (5,0);
\node at (4.8,0.2) {\scriptsize $B$};
\draw[very thick, colorthree] (-6.5,0) -- (-7.1,0);
\node at (-6.8,0.2) {\scriptsize $A$};



\draw[very thick, rounded corners = 1] (-6.5,1.1) rectangle (-5.6,-1.1) node[midway] { $\cF^{\scriptscriptstyle (1)}$};
\draw[very thick, rounded corners = 1] (-3.9,1.1) rectangle (-3,-1.1) node[midway] { $\cF^{\scriptscriptstyle (2)}$};
\draw[very thick, rounded corners = 1] (-1.3,1.1) rectangle (-0.4,-1.1) node[midway] { $\cF^{\scriptscriptstyle (3)}$};
\draw[very thick, rounded corners = 1] (0.7,1.1) rectangle (1.6,-1.1) node[midway] { $\cF^{\scriptscriptstyle (n)}$};
\draw[very thick, rounded corners = 1] (3.3,1.1) rectangle (4.4,-1.1) node[midway] { $\cF^{\scriptscriptstyle (n+1)}$};

\begin{scope}[shift={(0,0.7)}]
\draw[very thick, rounded corners = 1] (-5.05,0.3) rectangle (-4.45,-0.4) node[midway] {$\cN'$};
\draw[very thick, rounded corners = 1] (-2.45,0.3) rectangle (-1.85,-0.4) node[midway] {$\cN'$};
\draw[very thick, rounded corners = 1] (2.15,0.3) rectangle (2.75,-0.4) node[midway] {$\cN'$};
\end{scope}
\node at (-0.1,0) [circle,fill,inner sep=1pt]{};
\node at (0.1,0) [circle,fill,inner sep=1pt]{};
\node at (0.3,0) [circle,fill,inner sep=1pt]{};

\node at (-0.1,0.5) [circle,fill,inner sep=1pt]{};
\node at (0.1,0.5) [circle,fill,inner sep=1pt]{};
\node at (0.3,0.5) [circle,fill,inner sep=1pt]{};

\node at (-0.1,-0.5) [circle,fill,inner sep=1pt]{};
\node at (0.1,-0.5) [circle,fill,inner sep=1pt]{};
\node at (0.3,-0.5) [circle,fill,inner sep=1pt]{};

\end{tikzpicture}
\caption{\small A schematic diagram for the quantum channel synthesis protocol that uses $n$ times of the channel $\cN'$. Every channel use is interleaved by a free CPWP operation $\cF^{\scriptscriptstyle (i)}$. The goal of such a protocol is to make the effective channel in the dashed box simulate a target channel $\cN$.}
\label{fig: channel synthesis}
 \end{figure}

The following result establishes new fundamental limits on the quantum channel synthesis problem by employing the geometric \Renyi Thauma of the resource and target channels, respectively.
\begin{theorem}
Let $\cN'$ and $\cN$ be two qudit quantum channels. Then the number of channel $\cN'$ required to implement $\cN$ is bounded from below as 
  \begin{align}
    S(\cN' \to \cN) \geq {\widehat \theta_{\a}(\cN)}/{\widehat \theta_{\a}(\cN')},\quad \forall \a \in (1,2].
  \end{align}
\end{theorem}
\begin{proof}
  Suppose the optimal simulation protocol requires to use the resource channel $n = S(\cN'\to \cN)$ times and the protocol is given by
  \begin{align}
    \cN = \cF^{\scriptscriptstyle (n+1)} \circ \cN' \circ \cF^{\scriptscriptstyle (n)} \circ \cdots \circ \cF^{\scriptscriptstyle (2)} \circ \cN' \circ \cF^{\scriptscriptstyle (1)},
  \end{align}
  with $\cF^{\scriptscriptstyle (i)}$ being CPWP operations.
  Using the subadditivity of the geometric \Renyi Thauma in Lemma~\ref{prop: subadditivity} iteratively, we have
  \begin{align}
    \widehat \theta_{\a}(\cN) \leq n \widehat \theta_{\a}(\cN') + \sum_{i=1}^{n+1} \widehat \theta_{\a}(\cF^{\scriptscriptstyle (i)}) = n \widehat \theta_{\a}(\cN'),
  \end{align}
  where the equality follows from the faithfulness of the geometric \Renyi Thauma in Lemma~\ref{prop: reduction monotone faithful}. Therefore, we have $S(\cN' \to \cN) = n \geq {\widehat \theta_{\a}(\cN)}/{\widehat \theta_{\a}(\cN')}$, which concludes the proof.
\end{proof}

\vspace{0.2cm}
Together with the previous result in~\cite[Proposition 23]{Wang2019channelmagic}, we have
\begin{corollary}
Let $\cN'$ and $\cN$ be two qudit quantum channels. Then the number of channel $\cN'$ required to implement $\cN$ is bounded from below as  
  \begin{align}
    S(\cN' \to \cN) \geq \max\left\{\frac{\mana(\cN)}{\mana(\cN')}, \frac{\theta_{\max}(\cN)}{\theta_{\max}(\cN')}, \frac{\widehat \theta_{\a}(\cN)}{\widehat \theta_{\a}(\cN')}\right\},\quad \forall \a \in (1,2].
  \end{align}
\end{corollary}

\begin{remark}
  Note that each lower bound is given by a fraction of two quantities. It is thus not known which fraction is tighter in general, despite that $\widehat \theta_{\a}(\cN) \leq \theta_{\max}(\cN) \leq \mana(\cN)$. 
\end{remark}

\section{Conclusions}
\label{sec: Conclusions}

We have established several fundamental properties of the geometric \Renyi divergence as well as its channel divergence. We further demonstrated the usefulness of these properties in the study of quantum channel capacity problems, strengthening the previously best-known result based on the max-relative entropy in general. We expect that the technical tools established in this work can find a diverse range of applications in other research areas, such as quantum network theory and quantum cryptography. For example, we illustrate one more application of the geometric \Renyi divergence in the task of quantum channel discrimination in  Appendix~\ref{app: Quantum channel discrimination}.

There are also some interesting problems left for future investigation. The Umegaki relative entropy is the most commonly studied quantum divergence  because of its operational interpretation as an optimal error exponent in the hypothesis testing problem (known as the quantum Stein's lemma)~\cite{Ogawa2000,Hiai1991}. One open question is to know whether the geometric \Renyi divergence as well as the Belavkin-Staszewski relative entropy have any operational interpretation.


\paragraph{Acknowledgment.}
We would like to thank Omar Fawzi for bringing to our attention the open question by Berta et al.~\cite[Eq.~(55)]{Berta2018} and for suggesting the chain rule for the geometric \Renyi divergence, which helped us simplify and unify the proofs of Proposition~\ref{amortization proposition},~\ref{prop: amortization bidirectional},~\ref{amortization proposition priviate capacity},~\ref{prop: amortization}. We also thank David Sutter for encouraging us to make some proof steps more precise. KF and HF acknowledge the support of the University of Cambridge Isaac Newton Trust
Early Career grant RG74916.

\bibliographystyle{alpha_abbrv}

{\small
\bibliography{bib_GRD_capacity}}

\newpage

\appendix

\section{Technical Lemmas}
\label{app: technical lemmas}

In this section, we present several technical lemmas that are used in the main text.


\begin{lemma}[\cite{fawzi2017lieb}]
\label{geometric SDP general lemma}
  For any positive semidefinite operators $X$ and $Y$ with $X \ll Y$, Hermitian operator $M$ and $\a(\ell) = 1 + 2^{-\ell}$ with $\ell \in \mathbb N$, the matrix inequality $G_{1-\a}(X, Y) \leq M$ holds if and only if
  \begin{align}\label{eq: SDP representation condition general}
    \exists\ \dbh{N_0,N_1, \cdots N_\ell},\ \text{\rm s.t.}\
    \dbp{\begin{matrix}
      M & X \\ X & N_\ell
    \end{matrix}},
    \left\{\dbp{
    \begin{matrix}
      X & N_{i} \\ N_i & N_{i-1}   \end{matrix}}\right\}_{i=1}^\ell, \dbe{N_0-Y}.   
  \end{align}
  When $\ell = 0$, the conditions in the loop are taken as trivial. Here the short notation that $\dbp{X}$, $\dbe{X}$ and $\dbh{X}$ represent the positive semidefinite condition $X \geq 0$, the equality condition $X = 0$ and the Hermitian condition $X = X^\dagger$, respectively.
\end{lemma}

The following lemma proves a transformer inequality of the weighted geometric matrix means. Here we require this result to hold for a specific range of the weighting parameter that to the best of our knowledge has not been stated properly before. 
\begin{lemma}[Transformer inequality] \label{lem_transformer}
  Let $X$ and $Y$ be two positive operators, $K$ be any linear operator, and $\alpha \in (1,2]$. Then it holds
  \begin{align}
    G_{1-\alpha}(K X K^\dagger, K Y K^\dagger) \leq K G_{1-\alpha}(X,Y) K^\dagger \, .
  \end{align}
  Furthermore, if $K$ is invertible the statement above holds with equality.
\end{lemma}
\begin{proof}
Before proving the assertion of the lemma we need to collect some basic properties.
We start by recalling the known result~\cite{Kubo1980} that for $\beta \in [0,1]$, we have
  \begin{align}\label{eq: fact0}
    G_{\beta}(K X K^\dagger, K Y K^\dagger) \geq K G_{\beta}(X,Y) K^\dagger \, .
  \end{align}
  As a next preparatory fact we show that the desired statement is correct for $\alpha = 2$, i.e.,
          \begin{align}\label{eq: fact1}
    G_{-1}( K X K^\dagger, K Y K^\dagger) \leq K G_{-1}(X,Y) K^\dagger \, .
  \end{align}
  To see this we recall that by Schur's complement~\cite[Theorem~1.3.3]{bhatia_psd_book} we have
  \begin{align}\label{eq: tmp1}
  G_{-1}(X,Y) = X Y^{-1}X \leq M \iff \begin{pmatrix}
    M & X\\ X & Y
  \end{pmatrix} \geq 0.
\end{align}
Choosing $M = G_{-1}(X,Y)$ thus gives
\begin{align}
  \begin{pmatrix}
    G_{-1}(X,Y) & X \\ X & Y
  \end{pmatrix} \geq 0 \, ,
\end{align}
which then implies
\begin{align}
  \begin{pmatrix}
    K G_{-1}(X,Y) K^\dagger & K X K^\dagger \\ K X K^\dagger & K Y K^\dagger
  \end{pmatrix} \geq 0 \, ,
\end{align} 
because $Z \mapsto K Z K^\dagger$ is a positive map~\cite[Exercise~3.2.2]{bhatia_psd_book}. 
Using \eqref{eq: tmp1} again then implies~\eqref{eq: fact1}.
Because the maps $t \mapsto t^{-1}$ is operator anti-monotone~\cite[Table~2.2]{Sutter2018} we have
  \begin{align}\label{eq: fact2}
    Y \geq \omega \Rightarrow G_{-1}(X,Y) \leq G_{-1}(X,\omega) \, .
  \end{align}
As a final property we recall a fact from~\cite[Equation~19]{Fawzi2017} stating that
  \begin{align}\label{eq: fact3}
    G_{s}(X,G_{t}(X,Y)) = G_{st}(X,Y) \, .
  \end{align}
Now we are ready to prove the assertion of the lemma. For any $\beta \in [-1,0)$, using~\eqref{eq: fact3} we have 
  \begin{align}
    G_{\beta}(KX K^\dagger, KY K^\dagger) 
    & = G_{-1}(KX K^\dagger,G_{-\beta}(KX K^\dagger,KY K^\dagger))\\
    & \leq G_{-1}(KX K^\dagger,K G_{-\beta}(X ,Y)K^\dagger)\\
    & \leq K G_{-1}(X,G_{-\beta}(X ,Y))K^\dagger\\
    & = K G_{\beta}(X,Y) K^\dagger \, ,
  \end{align}
where the first inequality step follows from \eqref{eq: fact0} and \eqref{eq: fact2}. The second inequality is implied by~\eqref{eq: fact1}. The final step uses~\eqref{eq: fact3} again.

The fact that the transformer inequality holds with equality in case $K$ is invertible follows by applying the inequality twice as
\begin{align}
G_{1-\alpha}(K X K^\dagger, K Y K^\dagger) 
&\leq K G_{1-\alpha}(X,Y) K^\dagger  \\
&=  K G_{1-\alpha}\big(K^{-1} K X K^\dagger (K^\dagger)^{-1} ,K^{-1} K Y K^\dagger (K^\dagger)^{-1} \big) K^\dagger \\
&\leq G_{1-\alpha}(K X K^\dagger, K Y K^\dagger) \, ,
\end{align}
which proves that the two inequalities above actually hold with equality.
\end{proof}

\begin{corollary}\label{coro: transformer inequality log}
  Let $X$ and $Y$ be two positive operators, $K$ be any linear operator. Let $D_{op}(X,Y) = X^{\frac{1}{2}} \log \big(X^{\frac{1}{2}} Y^{-1} X^{\frac{1}{2}}\big) X^{\frac{1}{2}}$ be the operator relative entropy.  Then the $D_{op}$ satisfies the transformer inequality:
  \begin{align}
    D_{op}(KX K^\dagger, KY K^\dagger) \leq K D_{op}(X,Y) K^{\dagger}.
  \end{align}
  Furthermore, if $K$ is invertible the statement above holds with equality.
\end{corollary}
\begin{proof}
  Due to the fact that $\lim_{\gamma \to 0} -\frac{1}{\gamma} (x^{\gamma} - 1) = \log (x)$, we have the limit identity 
  \begin{align}\label{eq: op relative entropy limit}
    \lim_{\gamma \to 0} -\frac{1}{\gamma} (G_{\gamma}(X,Y) - X) = D_{op}(X,Y).
  \end{align}
  Then we have
  \begin{align}
    D_{op}(KX K^\dagger, KY K^\dagger) & = \lim_{\alpha \to 1} \frac{1}{\alpha-1}\left[ G_{1-\alpha}(KX K^\dagger, KY K^\dagger) -  KX K^\dagger\right]\\
    & \leq \lim_{\alpha \to 1} \frac{1}{\alpha-1}\left[ K G_{1-\alpha}(X, Y ) K^\dagger-  KX K^\dagger\right]\\\
    & = K \lim_{\alpha \to 1} \frac{1}{\alpha-1}\left[  G_{1-\alpha}(X, Y ) -  X \right] K^\dagger\\
    & = K D_{op}(X,Y) K^{\dagger},
  \end{align}
  where the first and last equalities follow from Eq.~\eqref{eq: op relative entropy limit}, the inequality follows from Lemma~\ref{lem_transformer}.
\end{proof}

\section{A hierarchy for constant-bounded subchannels}

\label{app: A complete hierarchy for the set of constant-bounded maps}
In this section we discuss the set of constant-bounded subchannels
\begin{align}
    \bcV_{cb} \equiv\big\{\cM \in \text{CP}(A:B)\,\big|\, \exists\, \sigma_B \in \cS(B)\ \text{s.t.}\ \cM_{A\to B}(\rho_A)\leq \sigma_B, \forall \rho_A\in \cS(A)\big\}.
\end{align}
Denote $\cN_\sigma$ as the constant map induced by the state $\sigma$. For any $\cM \in \bcV_{cb}$ the condition $\cM(\rho)\leq \sigma$ for all $\rho$ is equivalent that $\cN_\sigma - \cM$ is a positive map. In terms of their Choi matrices, we have $\1_A\ox \sigma_B - J_{\cM} \in \BP(A:B)$ where $\BP(A:B)$ is the cone of block positive operators. Thus we have
\begin{align}
    \bcV_{cb} =\big\{\cM \in \text{CP}(A:B)\,\big|\, \exists\, \sigma_B \in \cS(B)\ \text{s.t.}\ \1_A \ox \sigma_B - J_{\cM} \in \BP(A:B)\big\}.
\end{align}
Due to the difficulty of finding a semidefinite representation for $\BP$~\cite{fawzi2019set}, we do not expect that there is a semidefinite representation for the set $\bcV_{cb}$.  Nevertheless, the cone $\BP$ can be approximated by a complete hierarchy from the inside as
\begin{align}
    \DPScone_1^* \subseteq \DPScone_2^* \subseteq \cdots \subseteq \DPScone_k^* \subseteq \cdots \subseteq \BP,
\end{align}
where $\DPScone_k^*$ is the dual cone of the well-known DPS hierarchy~\cite{Doherty2002,Doherty2004} and is given by the semidefinite representation~\cite{Fang2019}
\begin{align}
\DPScone_k^* = \Biggl\{ M_{AB_1} \Bigg| & \;\; M_{AB_1}\ox \1_{B_{[2:k]}} = \big(Y_{AB_{[k]}} - \Pi_kY_{AB_{[k]}} \Pi_k\big) + \sum_{s=0}^k W_{s,AB_{[k]}}^{\sfT_{B_{[s]}}} \notag\\
& \quad \quad \text{ where } Y_{AB_{[k]}} \in \herm,\, W_{s,AB_{[k]}} \geq 0, \forall s \in [0:k] \Biggr\},
\end{align}
where the index $[s_1:s_2]\equiv \{s_1,s_1+1,\cdots s_2\}$, $[s]\equiv [1:s]$ and $\Pi_k$ is the projector on the symmetry subspace of $\cH_B^{\ox k}$.
Then we can construct a complete semidefinite hierarchy for the set $\bcV_{bc}$ as
\begin{align}
    \bcV_{cb}^1 \subseteq \bcV_{cb}^2  \subseteq \cdots \subseteq \bcV_{cb}^k  \subseteq \cdots \subseteq \bcV_{cb},
\end{align}
with each level given by
\begin{align}
    \bcV_{cb}^k =\left\{\cM \in \text{CP}(A:B)\,\big|\, \exists\, \sigma_B \in \cS(B)\ \text{s.t.}\ \1_A \ox \sigma_B - J_{\cM} \in \DPScone_k^*\right\},
\end{align}

Consider the first level of the hierarchy 
\begin{align}
    \bcV_{cb}^1 = \left\{\cM \in \text{CP}(A:B)\,\Big|\, \exists\, \sigma_B \in \cS(B), W_0, W_1 \geq 0\ \text{s.t.}\ \1_A \ox \sigma_B - J_{\cM} = W_0 + W_1^{\sfT_B}\right\}
\end{align}
Denote $R = \1_A \ox \sigma_B - W_0$, we obtain
\begin{align}
    \bcV_{cb}^1 = \left\{\cM \in \text{CP}(A:B)\,\Big|\, \exists\, \sigma_B \in \cS(B),\ \text{s.t.}\ \1_A\ox \sigma_{B} - R \geq 0, R^{\sfT_B} - J_{\cM}^{\sfT_B} \geq 0\right\}
\end{align}
By symmetrizing the conditions $X - Y \geq 0$ to $X \pm Y \geq 0$, we will retrieve the set $\bcV_\b$. 
Similarly, by using a different way of variable replacement $V = \1_A \ox \sigma - W_1$, we have
\begin{align}
    \bcV_{cb}^1 = \left\{\cM \in \text{CP}(A:B)\,\Big|\, \exists\, \sigma_B \in \cS(B),\ \text{s.t.}\ \1_A\ox \sigma_{B} - V^{\sfT_B} \geq 0, V - J_{\cM} \geq 0\right\}.
\end{align}
Then symmetrizing the conditions, we will obtain a set of subchannels $\bcV_\zeta$ which is exactly the zero set of the strong converse bound $C_\zeta$. 

The restriction to symmetric conditions ensures that the set is closed under tensor product, which is a key ingredient to proving the sub-additivity of the geometric \Renyi Upsilon-information in Proposition~\ref{prop: sub-additivity renyi upsilon}. Thus a further improvement of the result in the main text can be made by considering a symmetric restriction of a higher level set $\bcV_{cb}^k$. But we also note that the size of the SDP will exponentially increase  in the level of $k$.

\newpage
\section{A detailed comparison for generalized amplitude damping channels}
\label{app: Detailed comparison for generalized amplitude damping channel}

In this section we give a detailed comparison of our new strong converse bounds with previously known results for the generalized amplitude damping (GAD) channels. This class of channels has been systematically investigated in~\cite{Khatri2019}, with several converse bounds being established therein based on the data-processing inequality and the continuity of channel capacities as well as a few different techniques. 

Due to the covariance symmetry of the GAD channels under the Pauli-$z$ operator, the quantities introduced in this work do not provide advantage over the Rains information $R$ (resp. the relative entropy of entanglement $E_R$) in terms of the unassisted quantum (resp.  private) capacity. In the following, we will focus our comparison for the two-way assisted scenarios where both $R$ and $E_R$ are not known as valid converse bounds. The comparison result for the two-way assisted quantum capacity is given in Figure~\ref{GAD quantum capacity weak converse compare}. The red solid line is the previously best-known converse bound composed by several different quantities. It is clear that our new strong converse bound $\widehat R_{\a {\scriptscriptstyle (10)},\Theta}$ can be much tighter in most cases. Note that for the two-way assisted private capacity, we will obtain exactly the same result as Figure~\ref{GAD quantum capacity weak converse compare} by replacing $R_{\max}$ with $E_{\max}$ and $\widehat R_{\a{\scriptscriptstyle (10)},\Theta}$ with $\widehat E_{\a{\scriptscriptstyle (10)},\Sigma}$ respectively.

\begin{figure}[H]
\centering
\begin{adjustwidth}{-0.6cm}{0cm}
\begin{tikzpicture}
\node at (-5.5,-2.3) {\small $N = 0.1$};
\node at (0.2,-2.3) {\small $N = 0.2$};
\node at (5.9,-2.3) {\small $N = 0.3$};
\node at (-5.5,-7.3) {\small $N = 0.4$};
\node at (0.2,-7.3) {\small $N = 0.45$};
\node at (5.9,-7.3) {\small $N = 0.5$};

\node at (-5.7,0) {\includegraphics[width = 5.3cm]{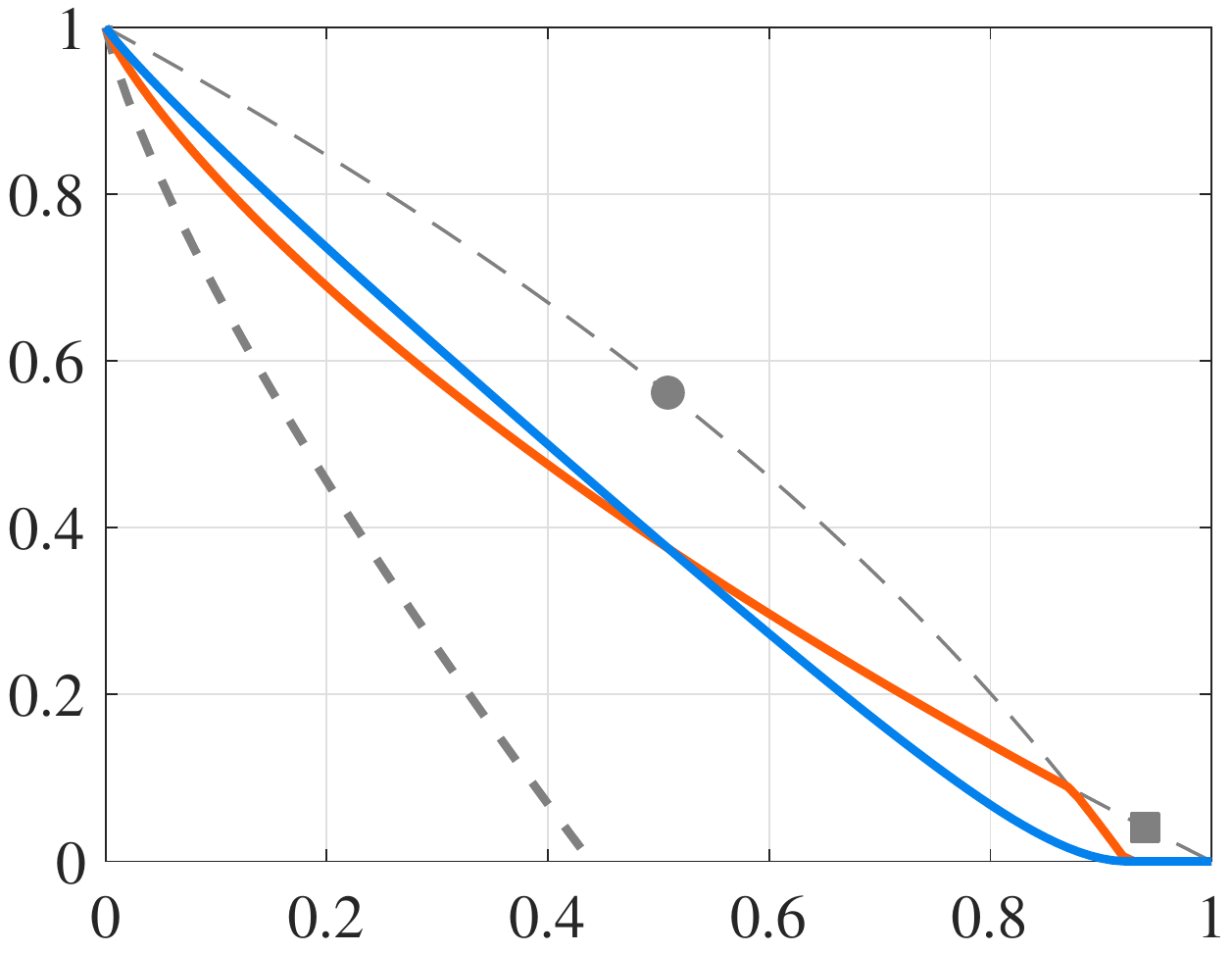}};

\node at (0,0) {\includegraphics[width = 5.3cm]{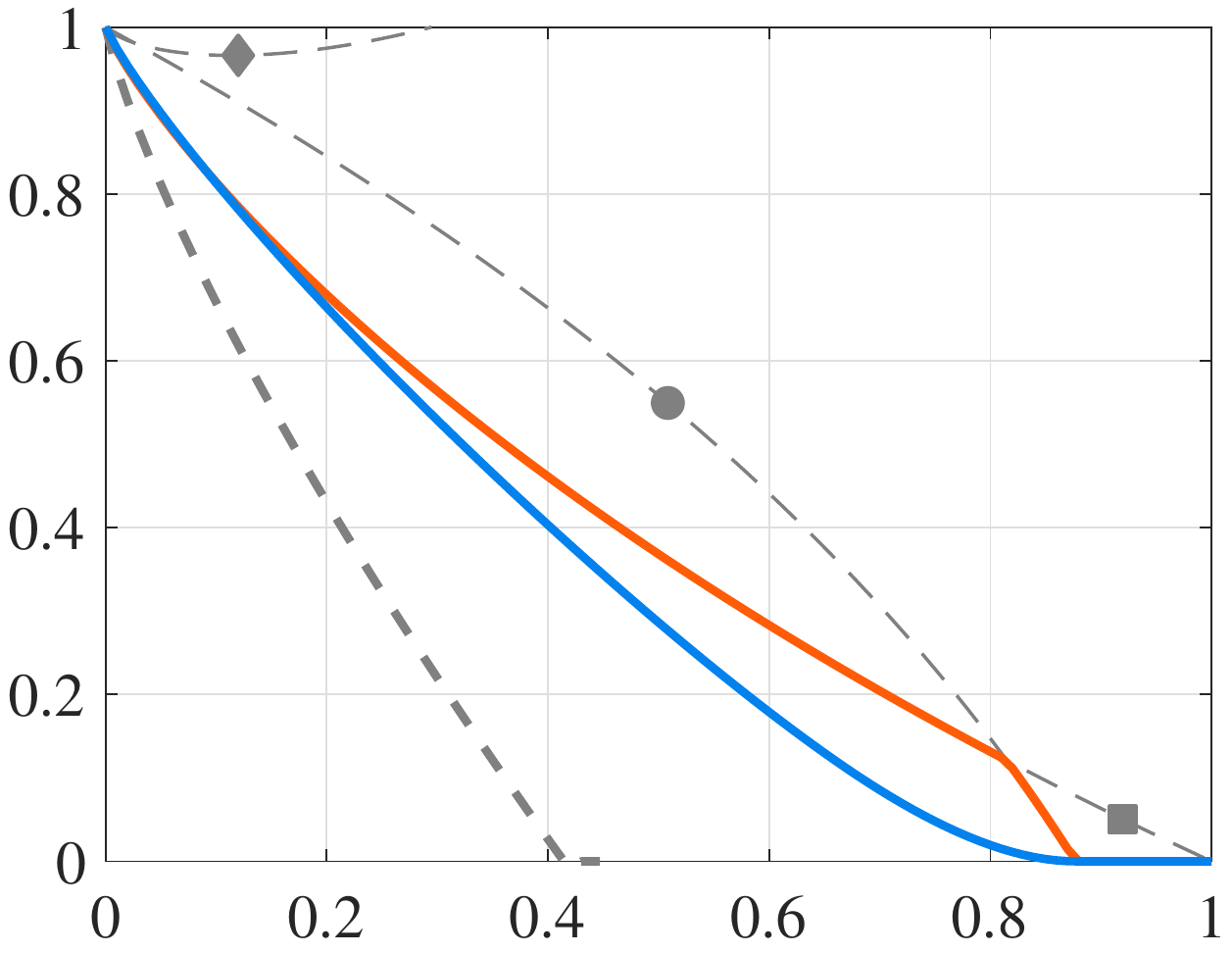}};

\node at (5.7,0) {\includegraphics[width = 5.3cm]{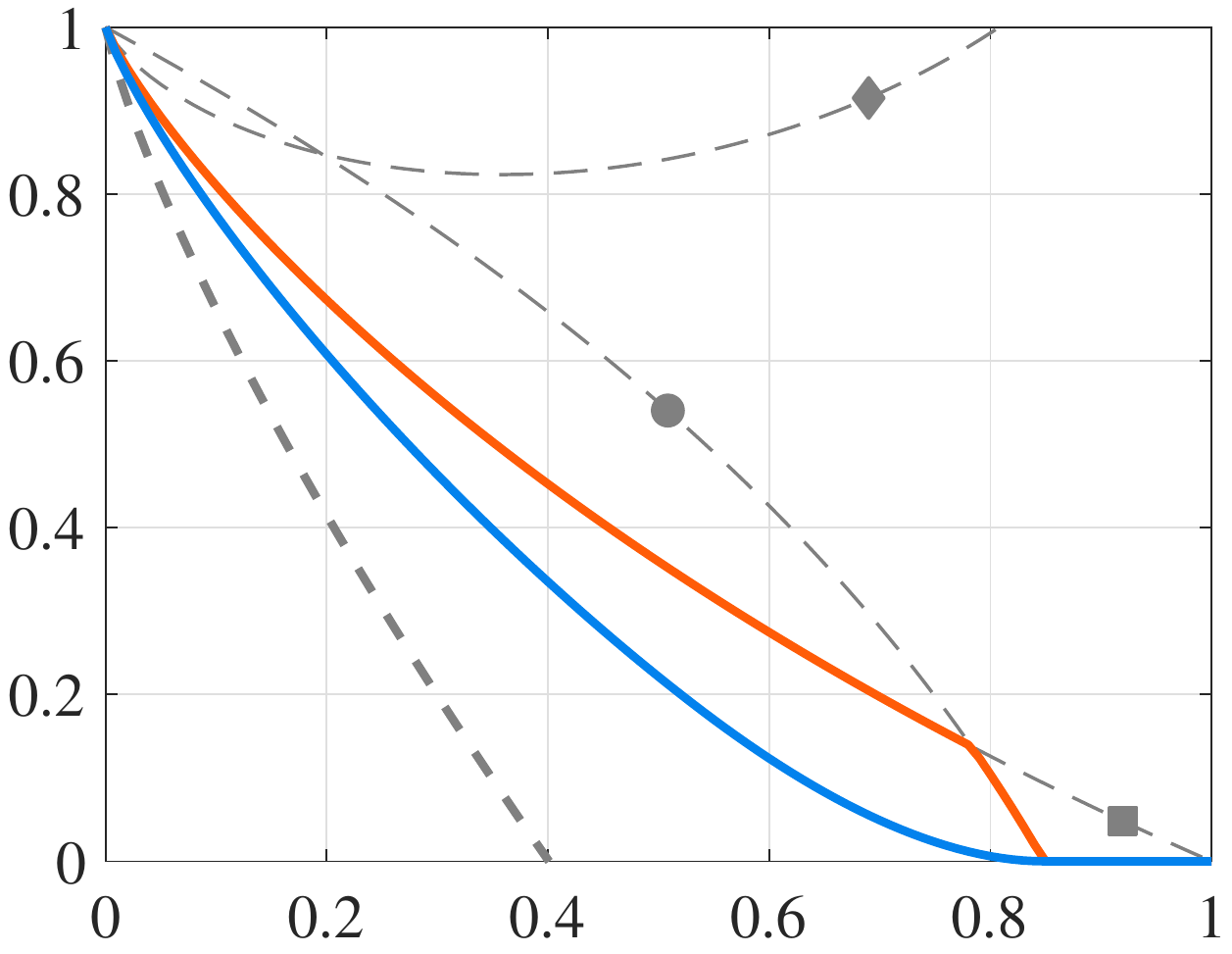}};

\node at (-5.7,-5) {\includegraphics[width = 5.3cm]{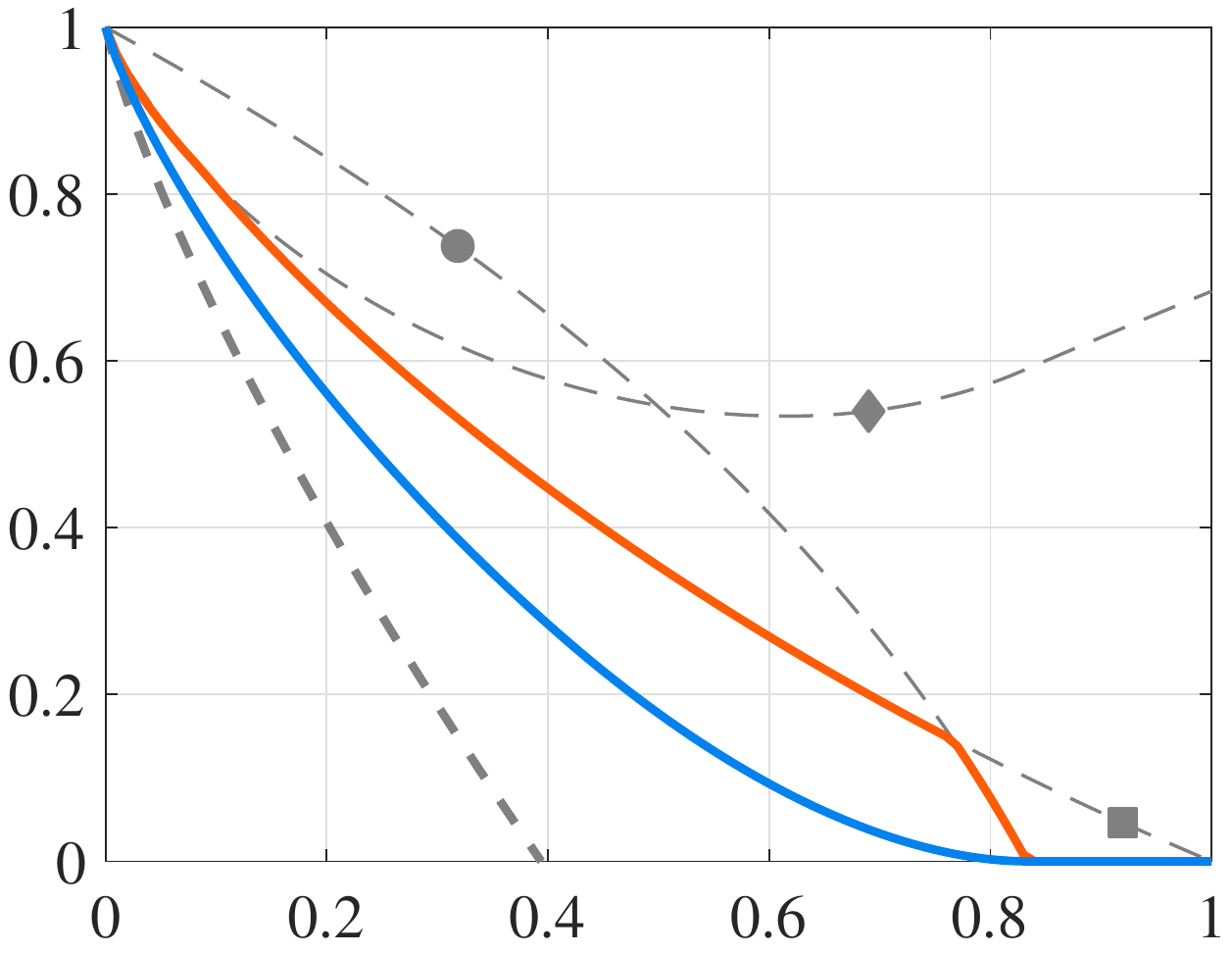}};

\node at (0,-5) {\includegraphics[width = 5.3cm]{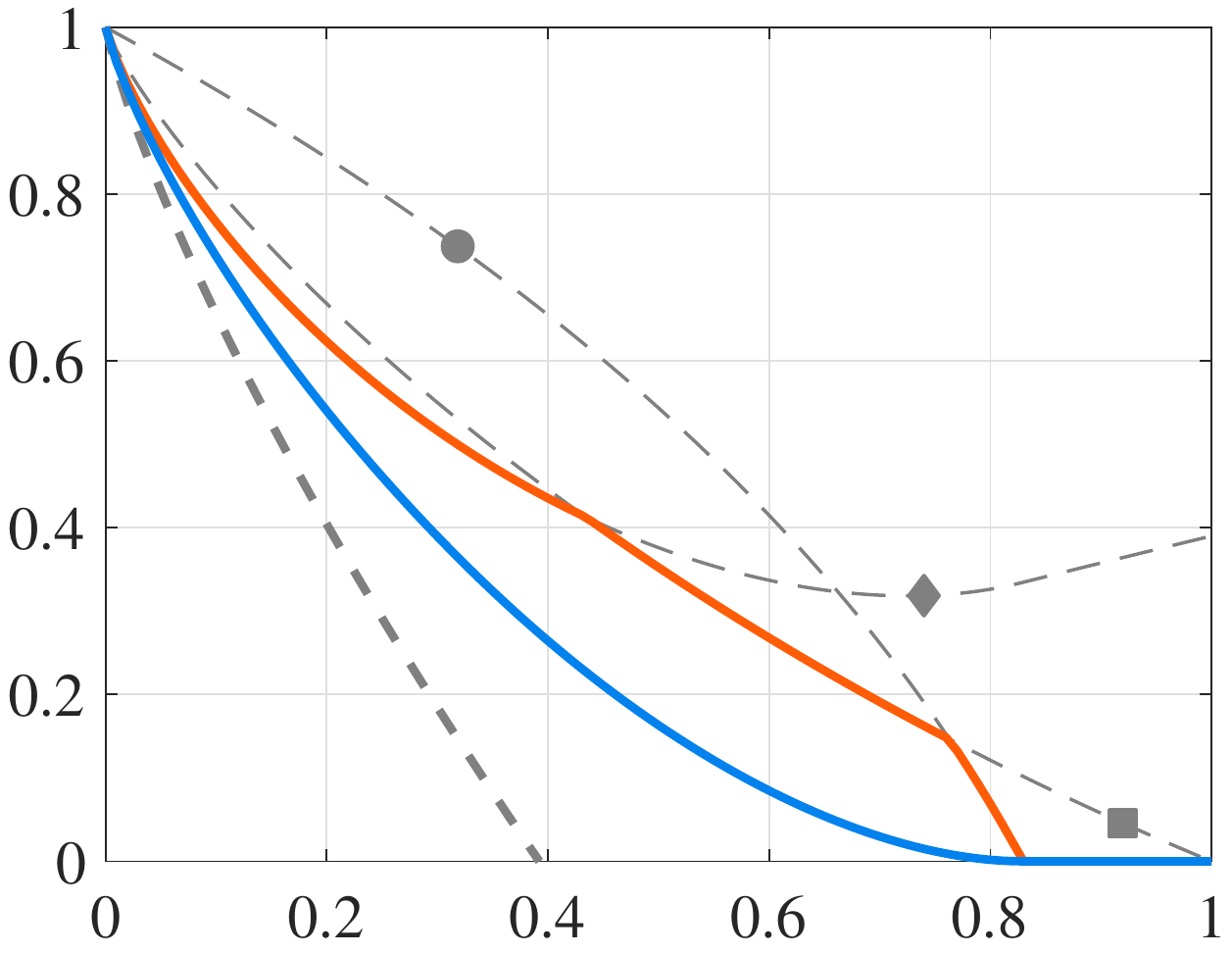}};

\node at (5.7,-5) {\includegraphics[width = 5.3cm]{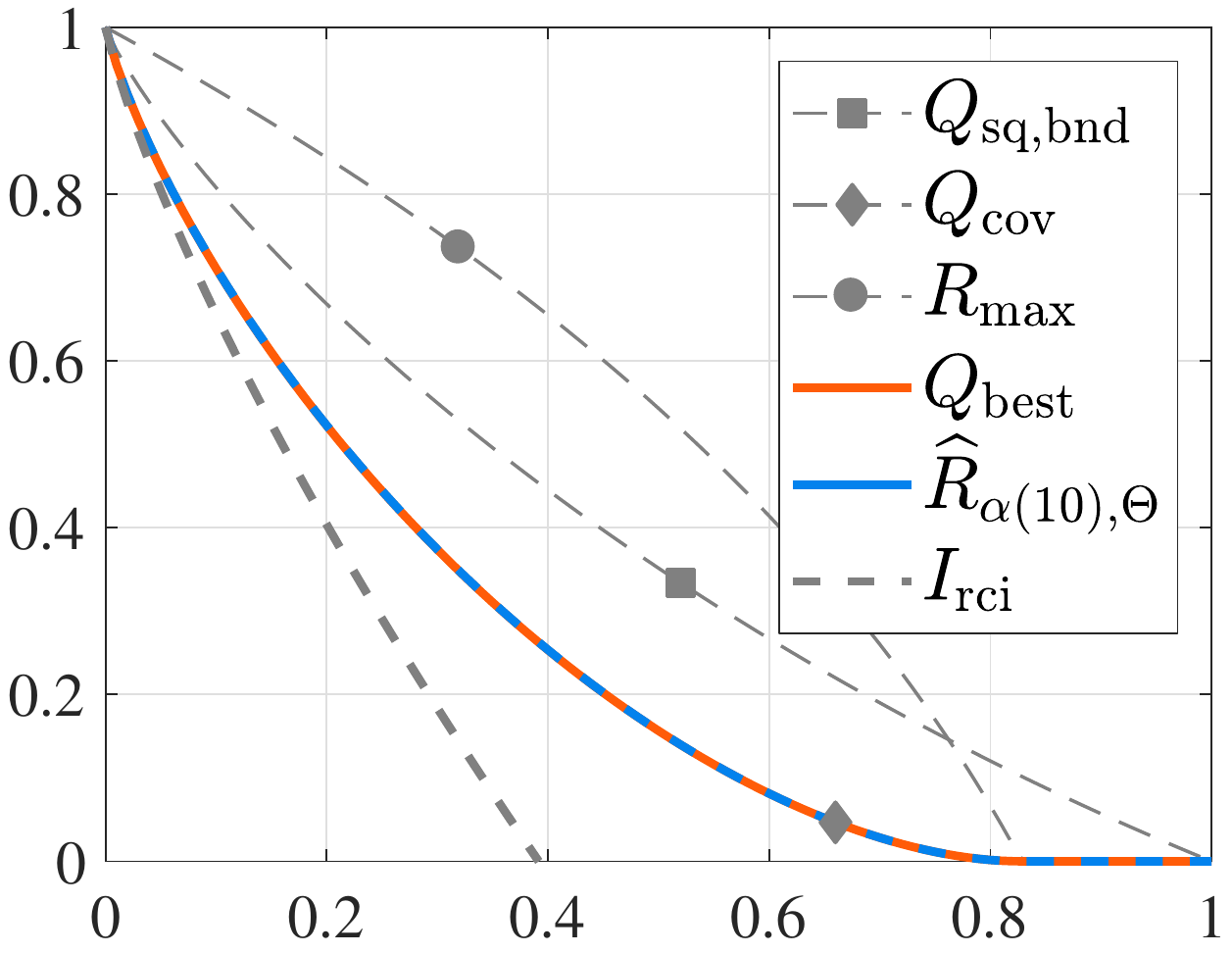}};

\end{tikzpicture}
\end{adjustwidth}
\caption{\small Comparison of the converse bounds on two-way assisted quantum capacities of the GAD channels. The quantity $Q_{\rm sq,bnd}$ is the squashed entanglement upper bound in~\cite[Eq.~(176)]{Khatri2019}. The quantity $Q_{\rm cov}$ is the approximate covariance upper bound in~\cite[Eq.~(205)]{Khatri2019}. The quantity $R_{\max}$ is the max-Rains information in~\cite{Wang2016a,Wang2017d} which is later proved to be a strong converse for two-way assisted quantum capacity in~\cite{Berta2017a}. The red solid line $Q_{\rm best}$ is the previously tightest upper bound composed by $Q_{\rm sq,bnd}$, $Q_{\rm cov}$ and $R_{\max}$, i.e., $Q_{\rm best} = \min\{Q_{\rm sq,bnd}, Q_{\rm cov},R_{\max}\}$. The blue solid line $\widehat R_{\a{\scriptscriptstyle (10)},\Theta}$ is our new strong converse bound in~[this work, Eq.~\eqref{eq:  SDP formula for maximal Rains theta info}] with level $\ell = 10$. The quantity $I_{\rm rci}$ is the reverse coherent information in~\cite[Eq.~(188)]{Khatri2019} which is a lower bound on the two-way assisted quantum capacity.}
\label{GAD quantum capacity weak converse compare}
\end{figure}
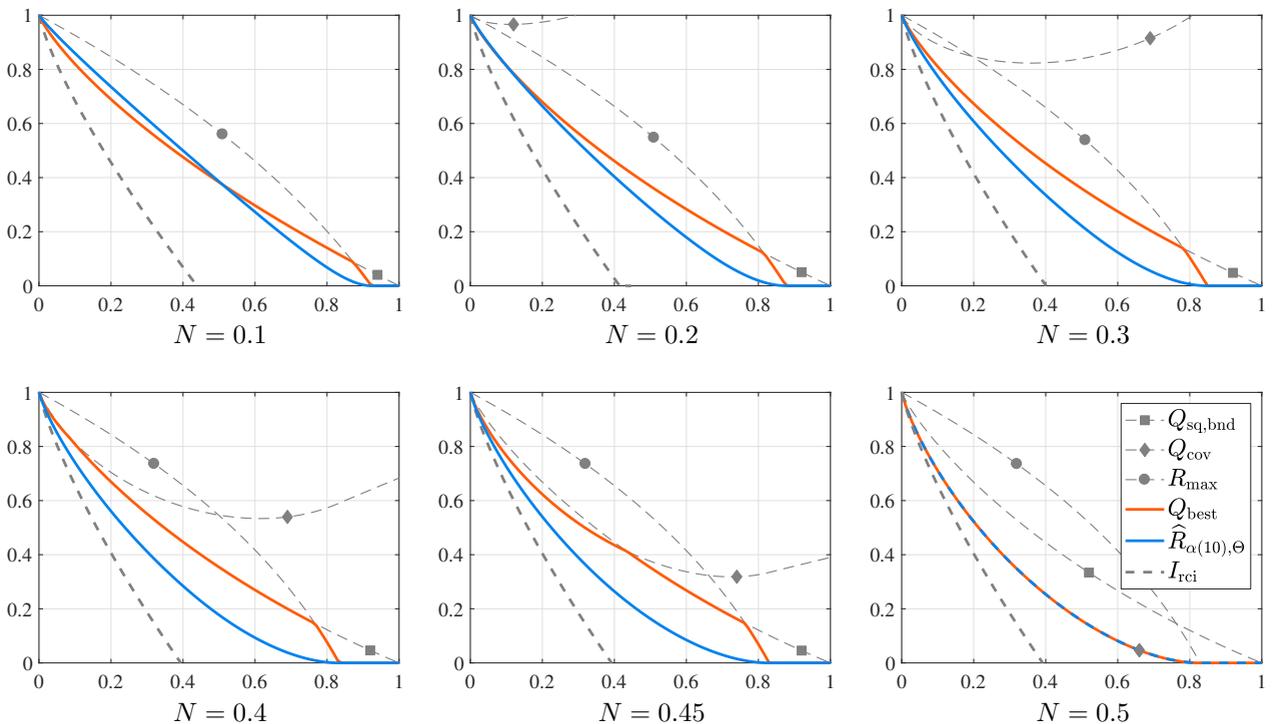


The comparison result for the classical capacity is given in Figure~\ref{GAD classical capacity weak converse compare}. The red solid line is the  previously best-known converse bound composed by several different quantities. It is clear that our new strong converse bound $\widehat \U_{\a{\scriptscriptstyle (10)}}$ can make further improvement at some parameter range, particularly for low to medium amplitude damping noise. In the range of high noise, the GAD channel becomes entanglement-breaking. Thus the $\ve$-entanglement breaking upper bound $C_{\rm EB}$ will be the tightest one, as expected.
We do not show the plot for $N=0.5$, because the channel becomes a qubit unital channel and thus its Holevo information is already tight~\cite{king2002additivity}.

\begin{figure}[H]
\centering
\begin{adjustwidth}{-0.6cm}{0cm}
\begin{tikzpicture}
\node at (-5.5,-2.3) {\small $N = 0.1$};
\node at (0.2,-2.3) {\small $N = 0.2$};
\node at (5.9,-2.3) {\small $N = 0.3$};
\node at (-5.5,-7.3) {\small $N = 0.4$};
\node at (0.2,-7.3) {\small $N = 0.45$};

\node at (-5.7,0) {\includegraphics[width = 5.3cm]{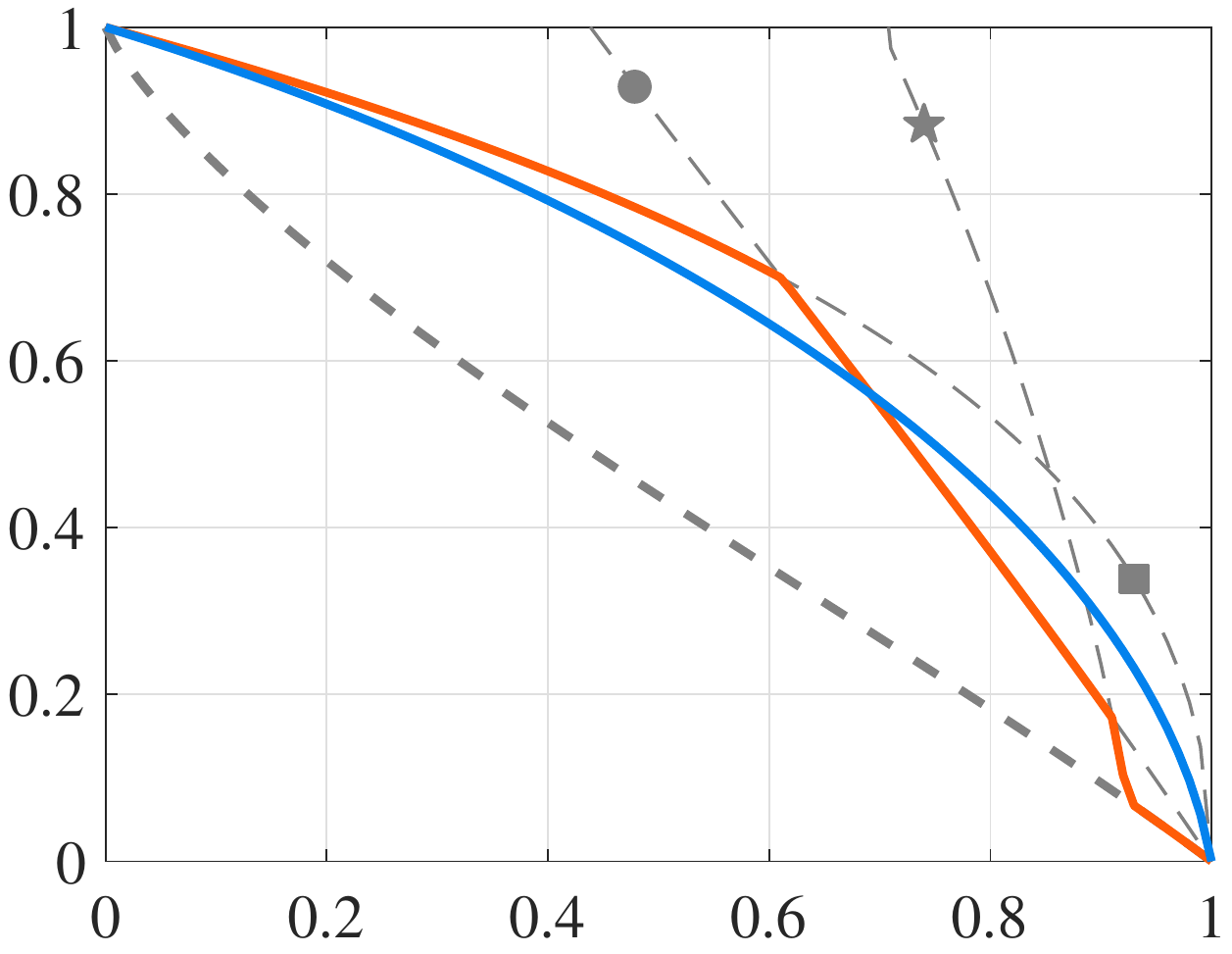}};

\node at (0,0) {\includegraphics[width = 5.3cm]{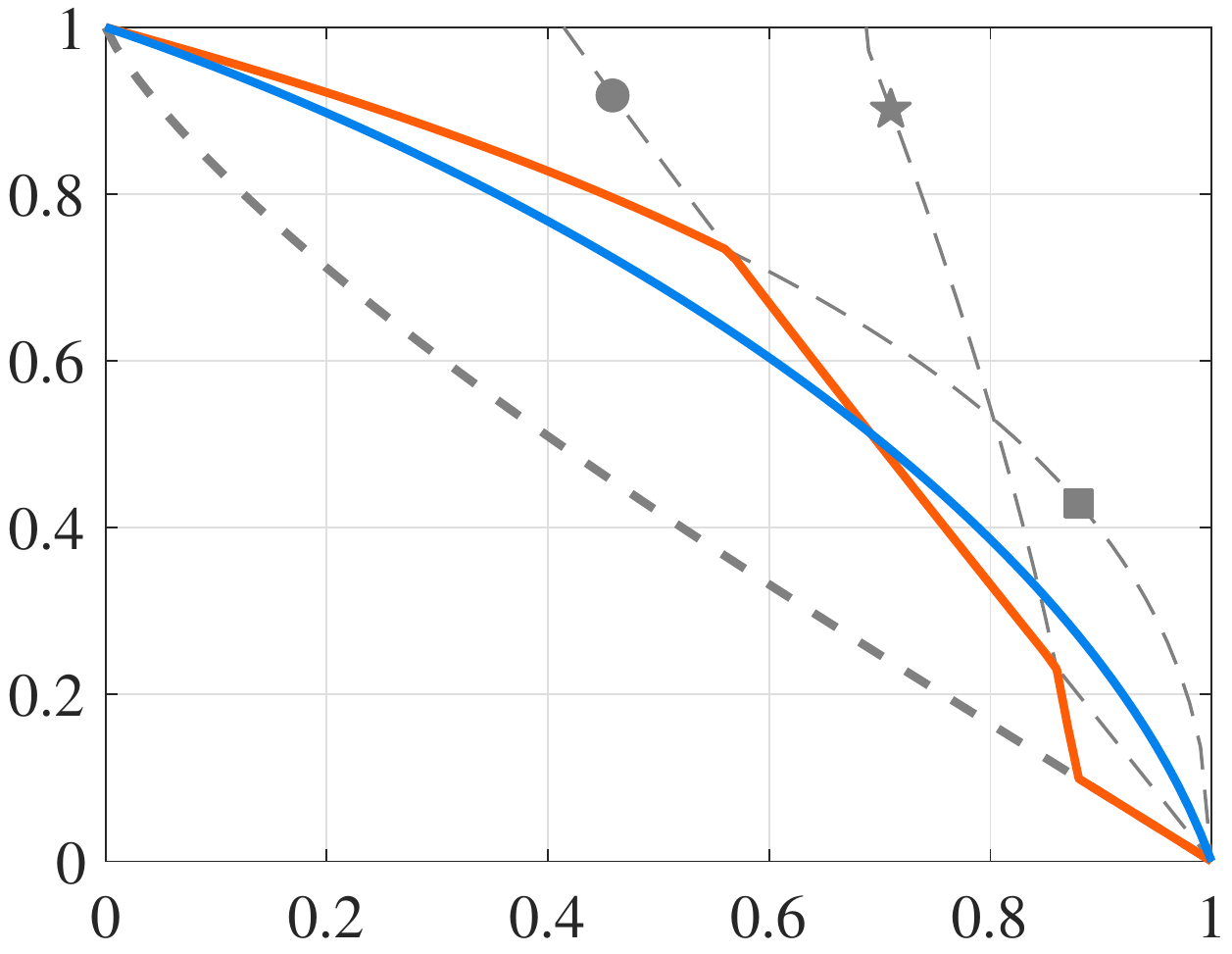}};

\node at (5.7,0) {\includegraphics[width = 5.3cm]{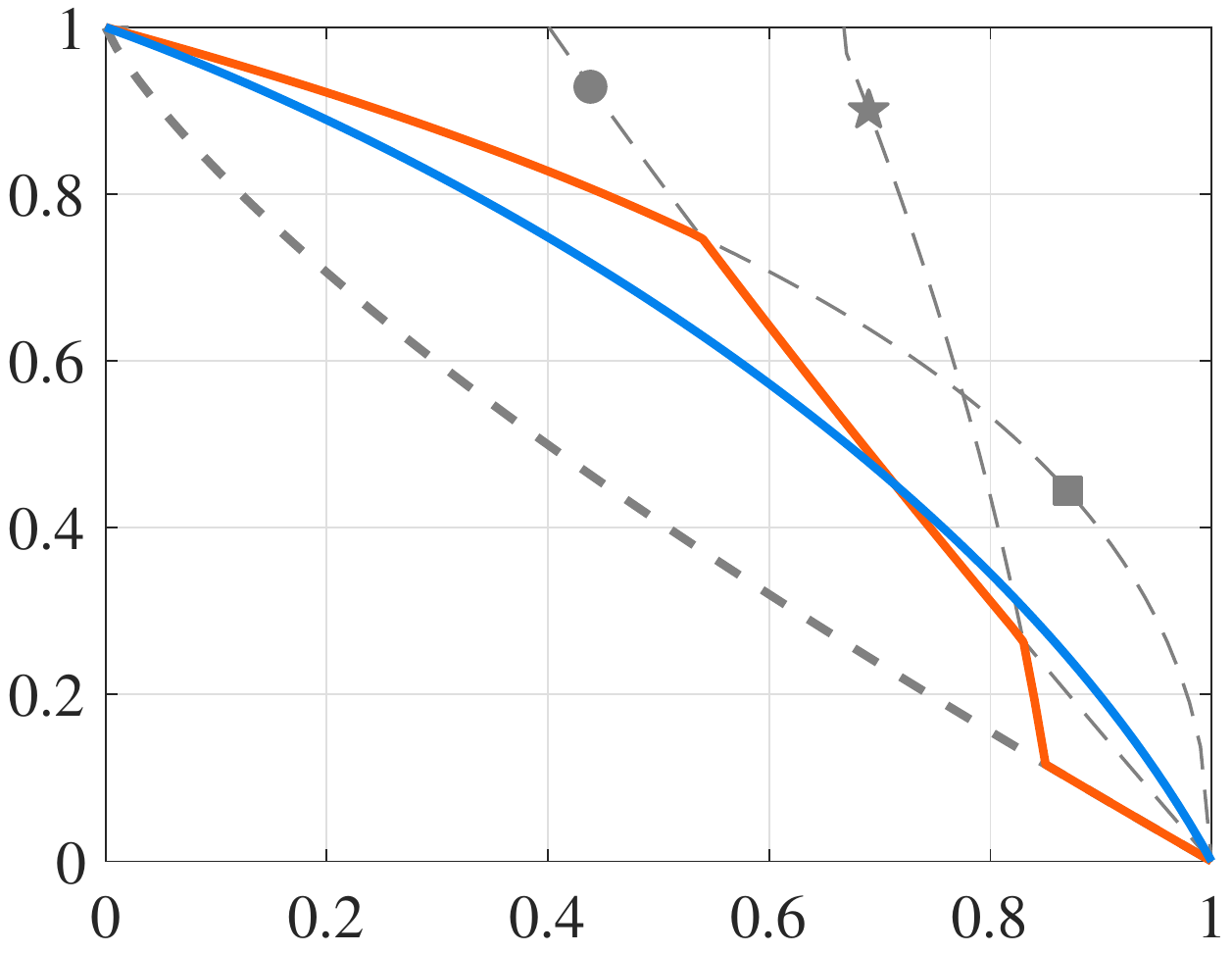}};

\node at (-5.7,-5) {\includegraphics[width = 5.3cm]{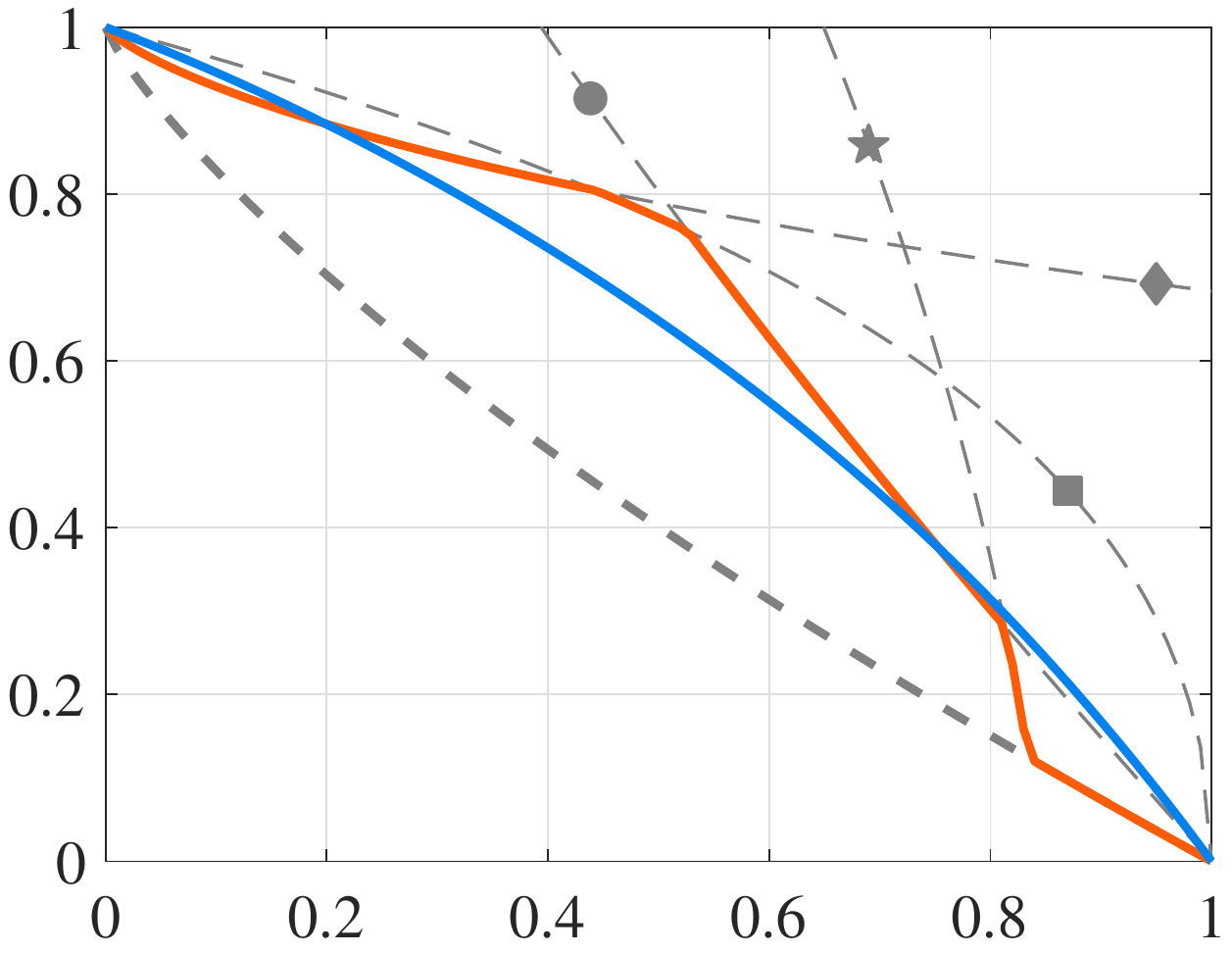}};

\node at (0,-5) {\includegraphics[width = 5.3cm]{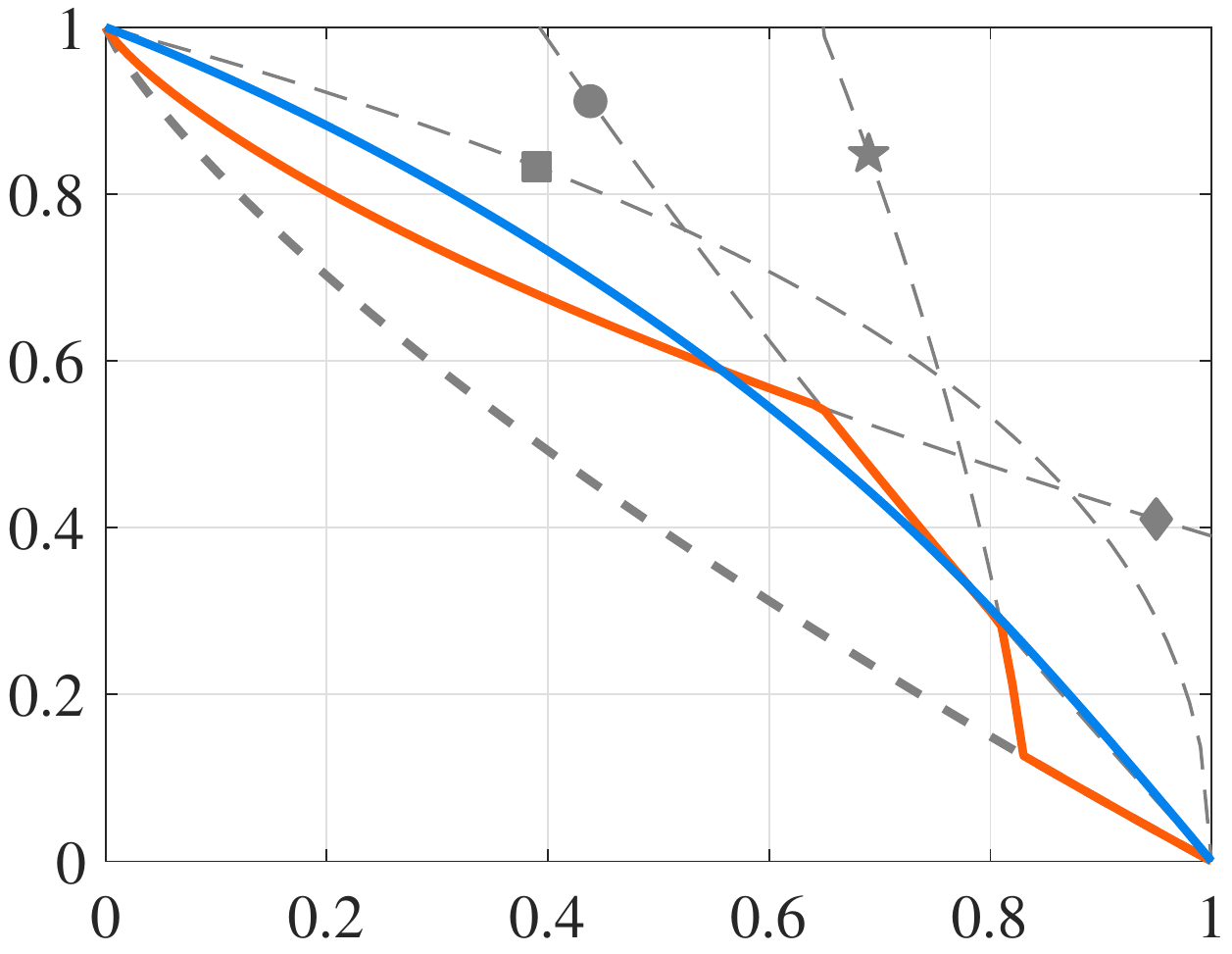}};

\node at (5.9,-4.9) {\includegraphics[height = 3.6cm]{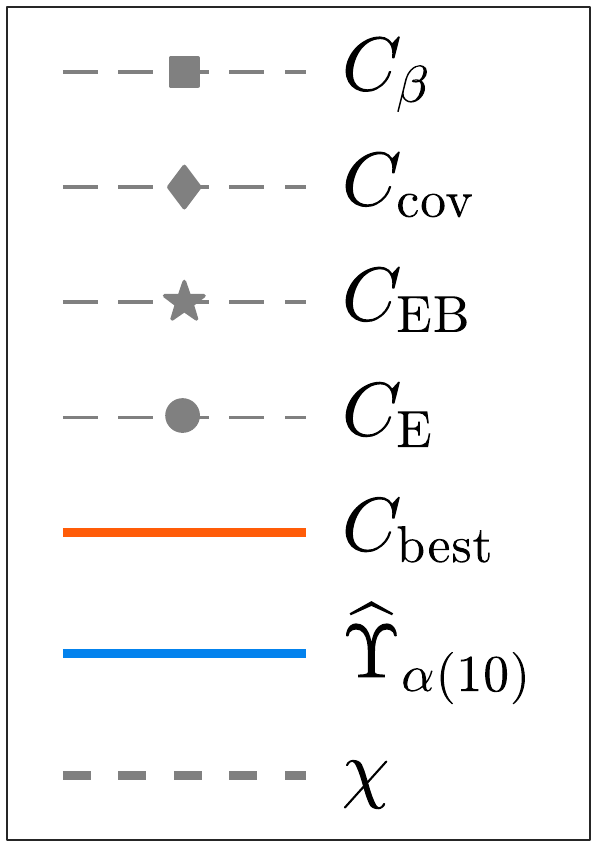}};

\end{tikzpicture}
\end{adjustwidth}
\caption{\small Comparison of the converse bounds on the classical capacity of the GAD channels. The quantity $C_\b$ is the strong converse bound in~\cite{Wang2016g}. The quantity $C_{\rm cov}$ is the $\ve$-covariance upper bound in~\cite[Eq.~(113)]{Khatri2019}. The quantity $C_{\rm EB}$ is the $\ve$-entanglement breaking upper bound in~\cite[Eq.~(112)]{Khatri2019}. The quantity $C_{\rm E}$ is the entanglement-assisted classical capacity in~\cite{li2007entanglement}. The red solid line $C_{\rm best}$ is the previously tightest upper bound composed by $C_\b$, $C_{\rm cov}$, $C_{\rm EB}$ and $C_{\rm E}$, i.e, $C_{\rm best} = \min\{C_\b, C_{\rm cov}, C_{\rm EB},C_{\rm E}\}$. The blue solid line $\widehat \U_{\a{\scriptscriptstyle (10)}}$ is our new strong converse bound in [this work, Eq.~\eqref{Renyi Upsilon information SDP min}] with level $\ell = 10$. The quantity $\chi$ is the Holevo information given in~\cite{li2007holevo} which is a lower bound on the classical capacity.}
\label{GAD classical capacity weak converse compare}
\end{figure}

\section{Application in quantum channel discrimination}
\label{app: Quantum channel discrimination}

A fundamental problem in quantum information theory is to distinguish between two quantum channels $\cN$ and $\cM$. In the asymmetric hypothesis testing setting (Stein's setting), we aim to minimize the type II error probability, under the condition that the type I error probability does not exceed a constant $\ve \in (0,1)$. 
More precisely, for any given two quantum channels $\cN$ and $\cM$, denote the corresponding type I and type II error of the adaptive protocol $\{Q,\cA\}$ as  $\alpha_n(\{Q,\cA\})$ and $\beta_n(\{Q,\cA\})$. Then the asymmetric distiguishibility is defined as
\begin{align}
  \zeta_n(\ve,\cN,\cM):=\sup_{\{Q,\cA\}} \left\{-\frac{1}{n}\log \beta_n(\{Q,\cA\}) \Big| \alpha_n(\{Q,\cA\}) \leq \ve\right\}.
\end{align}
Its asymptotic quantities are defined as
\begin{align}
  \underline{\zeta}(\ve,\cN,\cM):=\liminf_{n\to \infty} \zeta_n(\ve,\cN,\cM),\quad \overline{\zeta}(\ve,\cN,\cM):=\limsup_{n\to \infty} \zeta_n(\ve,\cN,\cM)
\end{align}
The best-known single-letter strong converse on $\overline{\zeta}(\ve,\cN,\cM)$ is given by the channel's max-relative entropy $D_{\max}(\cN\|\cM)$~\cite[Corollary 18]{Berta2018}, i.e.,
\begin{align}\label{eq: channel known result}
  D(\cN\|\cM) \leq \underline{\zeta}(\ve,\cN,\cM) \leq \overline{\zeta}(\ve,\cN,\cM) \leq D_{\max}(\cN\|\cM).
\end{align}
In the following, we sharpen this upper bound by the geometric \Renyi channel divergence in general. This gives a more accurate estimation of the fundamental limits of channel discrimination under adaptive strategies.

\begin{theorem}
  Let $\cN$ and $\cM$ be two quantum channels and $\ve \in (0,1)$, $\alpha \in (1,2]$. It holds
  \begin{align}
    D(\cN\|\cM) \leq \underline{\zeta}(\ve,\cN,\cM) \leq \overline{\zeta}(\ve,\cN,\cM) \leq \widehat D_{\alpha}(\cN\|\cM) \leq D_{\max}(\cN\|\cM).
  \end{align}
  Moreover, $\widehat D_\a(\cN\|\cM)$ is also a strong converse bound.
\end{theorem}
\begin{proof}
  The first two inequalities follow from~\eqref{eq: channel known result}. The last inequality follows from the relation that $\widehat D_{\a}(\rho\|\sigma) \leq D_{\max}(\rho\|\sigma)$ in Lemma~\ref{thm: divergence chain inequality}. In the following, we show that $\overline{\zeta}(\ve,\cN,\cM) \leq \widehat D_{\alpha}(\cN\|\cM)$. This follows a similar step as~\cite[Proposition 17]{Berta2018}. 

  Let $\{Q,\cA\}$ be an arbitary adaptive protocol for discriminating $\cN$ and $\cM$. Let $p$ and $q$ denote the final decision probabilities. As argued in \cite[Proposition 17]{Berta2018}, we can take $\a_n(\{Q,\cA\}) = \ve$. Then 
  \begin{align}
    \widehat D_{\alpha}(p\|q) \geq  \frac{1}{\a-1} \log p^\alpha q^{1-\alpha} = \frac{1}{\a-1} \log (1-\ve)^\alpha q^{1-\alpha}  = \frac{\alpha}{\alpha-1} \log (1-\ve) - \log q.
  \end{align}
  By applying the meta-converse in \cite[Lemma 14]{Berta2018} as well as the chain rule of the geometric \Renyi divergence, we have
  \begin{align}\label{eq: channel tmp1}
    -\frac{1}{n} \log q \leq \widehat D_{\alpha}(\cN\|\cM) + \frac{1}{n} \frac{\alpha}{\alpha-1} \log \frac{1}{1-\ve}.
  \end{align}
  Since Eq.~\eqref{eq: channel tmp1} holds for any channel discrimination protocol, we have $\overline{\zeta}(\ve,\cN,\cM) \leq \widehat D_{\alpha}(\cN\|\cM)$.
\end{proof}

\vspace{0.2cm}
Note that our new strong converse bound is also single-letter and efficient computable via semidefinite program (it even admits a closed-form expression as shown in Lemma~\ref{lem: maximal Renyi channel divergence SDP}). The following example of the GAD channels $\cA_{0.8,N_1}$ and $\cA_{0.7,N_2}$ demonstrates that $\widehat D_{\a}$ is much tighter than $D_{\max}$.

\begin{figure}[H]
  \centering
  \begin{tikzpicture}
    \node at (0,0) {\includegraphics[width=8cm]{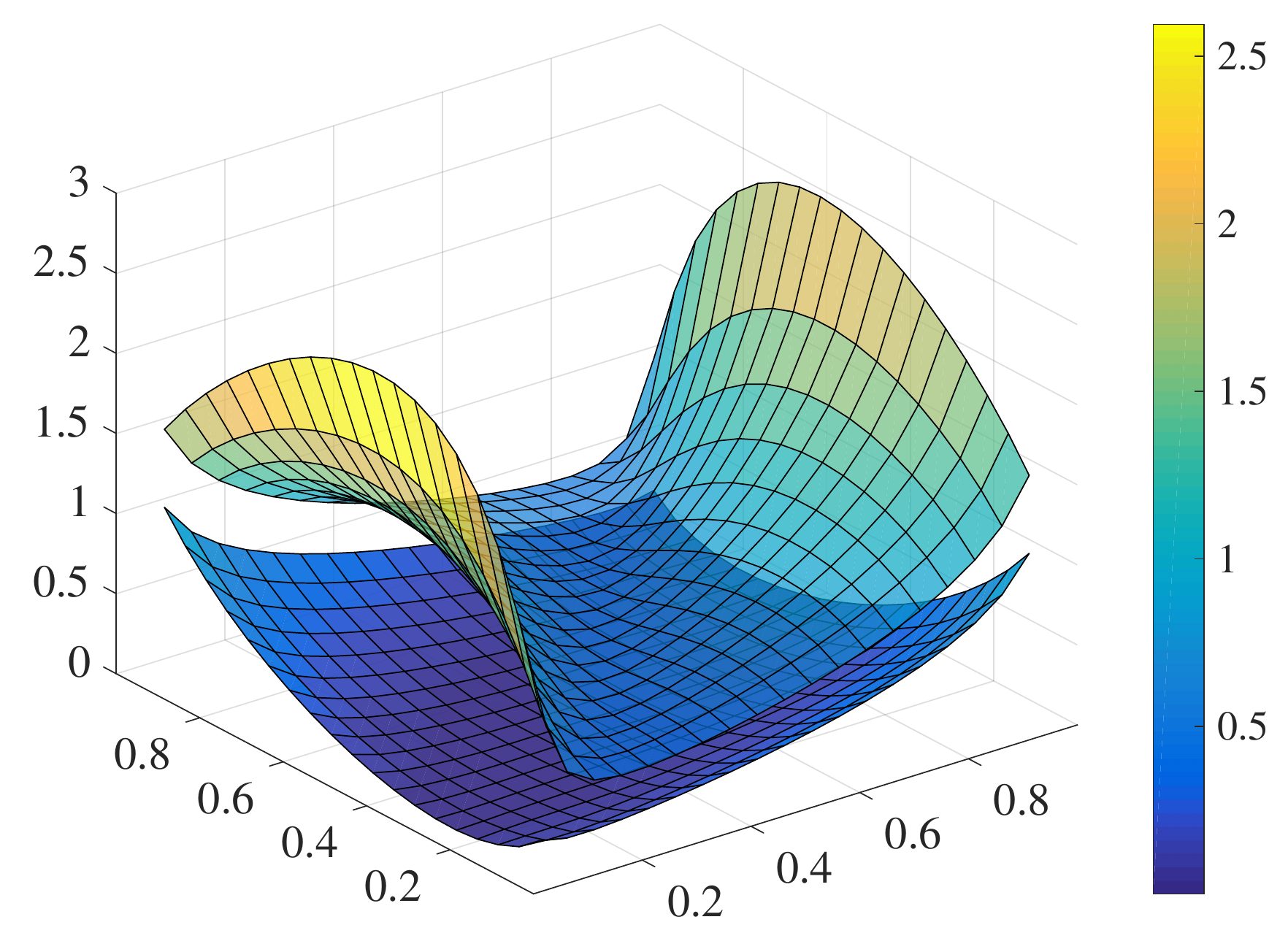}};
    \node at (2,-2.6) {\footnotesize $N_1$};
  \node at (-2.8,-2.5) {\footnotesize $N_2$};
  \end{tikzpicture}  
  \caption{\small This figure displays the difference between the upper and lower bounds in the Stein setting for the GAD channels $\cA_{0.8,N_1}$ and $\cA_{0.7,N_2}$. We vary the parameter $N_1,N_2 \in [0,1]$. The upper surface is the difference $D_{\max}(\cA_{0.8,N_1}\|\cA_{0.7,N_2}) - D(\cA_{0.8,N_1}\|\cA_{0.7,N_2})$. The lower surface is the difference $\widehat D_{\alpha(10)}(\cA_{0.8,N_1}\|\cA_{0.7,N_2}) - D(\cA_{0.8,N_1}\|\cA_{0.7,N_2})$.}
\end{figure}

\end{document}